\documentclass{article}
\usepackage{arxiv}

\usepackage {amsthm}
\usepackage{mathpazo}
\usepackage[utf8]{inputenc} 
\usepackage[T1]{fontenc}    
\usepackage{hyperref}       
\usepackage{url}            
\usepackage{booktabs}       
\usepackage{amsfonts}       
\usepackage{nicefrac}       
\usepackage{microtype}      
\usepackage{fancyhdr}       
\usepackage{graphicx}       
\usepackage{float} 
\usepackage[T1]{fontenc}
\usepackage{dcolumn}
\usepackage{bm}
\usepackage{siunitx}
\usepackage{mathptmx}
\usepackage{etoolbox}
\usepackage{amsmath}
\usepackage{amssymb}
\usepackage{tabularx}
\usepackage{soul}
\setstcolor{red}
\usepackage{pdfpages}
\usepackage{float}
\usepackage{physics}
\AtBeginDocument{\RenewCommandCopy\qty\SI}
\def\dbar{{\mathchar'26\mkern-12mu d}}
\usepackage{cancel}
\newtheorem{theorem}{Theorem}[section]
\newtheorem{postulate}{Postulate}[section]
\newtheorem{corollary}{Corollary}[postulate]
\newtheorem{lemma}[postulate]{Lemma}
\newtheorem{axiom}{Axiom}
\DeclareSIUnit\angstrom{\text {Å}}
\usepackage{multirow}
\usepackage[style=chem-acs]{biblatex}
\usepackage{cleveref}

\title{Uncertainty-Aware Liquid State Modeling from Experimental Scattering Measurements}

\author{Brennon L. Shanks \\
Department of Chemical Engineering, University of Utah \\
\texttt{Brennon L. Shanks} {brennon.shanks@chemeng.utah.edu}
}

\begin{document}

\maketitle

\begin{abstract}

This dissertation is founded on the central notion that structural correlations in dense fluids — such as dense gases, liquids, and glasses — are directly related to fundamental interatomic forces. Identified early in the development of statistical theories of fluids through the mathematical formulations of Gibbs in the 1910s, it took nearly 80 years before practical implementations of structure-based theories became widely used for interpreting and understanding the atomic structures of fluids from experimental X-ray and neutron scattering data. Breakthroughs in structure-potential methods can be largely attributed to advancements in molecular mechanics simulations and improving computational resources. Consequently, pioneers in the field, such as Putzai and McGreevy, Schommers, and Soper, were able to develop successful hybrid statistical mechanics and molecular simulation techniques, enabling the analysis of experimental scattering data with physics-guided models. 

Despite advancements in understanding the relationship between structure and interatomic forces, a significant gap remains. Current techniques for interpreting experimental scattering measurements are widely used, yet there is little evidence that they yield physically accurate predictions for interatomic forces. In fact, it is generally assumed that these methods produce interatomic forces that poorly model the atomistic and thermodynamic behavior of fluids, rendering them unreliable and non-transferable. This thesis aims to address these limitations by refining the statistical theory, computational methods, and philosophical approach to structure-based analyses, thereby developing more robust and accurate techniques for characterizing structure-potential relationships.

In summary, the central theme of this dissertation is the idea that rigorously quantifying uncertainty in thermophysical properties can enhance predictive accuracy and deepen our conceptual understanding of the liquid state. This work explores several key concepts:

\begin{enumerate}
    \item Probabilistic Iterative Potential Refinement: Utilizing Gaussian process regression allows us to reconstruct interatomic forces from structural correlations while maintaining thermodynamic consistency (Chapter 2).

    \item Bayesian Uncertainty Quantification and Propagation: An accelerated Bayesian method is proposed and implemented to quantify uncertainty in pair potential reconstructions from scattering measurements (Chapter 3).

    \item Error Propagation in Neutron Scattering Measurements: The application of Bayesian methods demonstrates that even random errors in neutron scattering measurements can impede our ability to accurately infer interatomic forces. Furthermore, that modern neutron instruments can successfully extract forces due to their sufficiently low random noise (Chapter 4).
\end{enumerate}

\noindent Overall, this dissertation asserts that structural analysis is more nuanced and practically useful than previously believed.

\end{abstract}

\tableofcontents

\newpage

\section{Introduction}

\subsection{The Origins of Liquid Structure Analysis}

The liquid state, being the intermediate phase between the well-ordered solid and the chaotic gas, has been described as a "statistical mechanical jungle" reserved for only the most foolhardy of academics \cite{croxton_liquid_2009}. For me, liquid state theory seems at times to be a study in cryptology, riddled with strange symbols that you might find on the walls of a Masonic temple let alone a graduate textbook in statistical physics. One only needs to turn to page 3 of Croxton's \textit{Liquid State Physics} and take a glance at graphical representations of cluster integral expansions (Figure \ref{fig:cluster_integrals}) to see what I mean! Nevertheless, the reward for fighting through these strange notations and difficult concepts is a beautiful and concise liquid state theory. This theory is the foundation for modern molecular simulations, has motivated the design and advancement of multi-billion dollar particle scattering facilities, and forms the cornerstone of modern thermodynamics. 

\begin{figure}
    \centering
    \includegraphics[width = 11 cm]{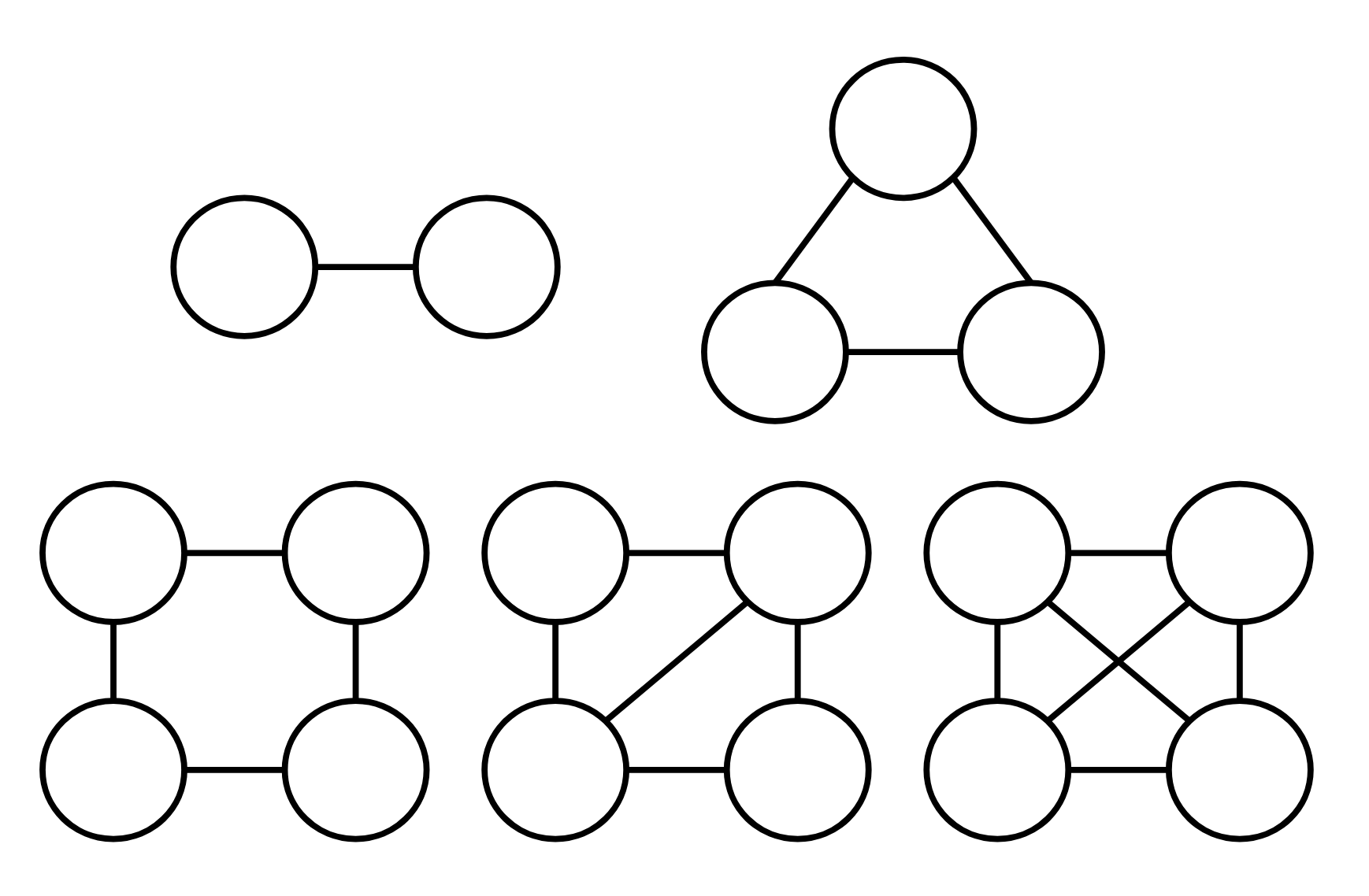}
    \caption{The statistical mechanical version of a ball and stick model.}
    \label{fig:cluster_integrals}
\end{figure}

The most exciting aspect for me (and likely for others in this field) is that there is still so much about liquids that remains unknown. Indeed, liquid state theory is not yet advanced enough to allow one to take a molecular description of a liquid and fully describe its structural, thermophysical, and flow properties. This challenge was recognized by theorists in the early 1960s, and owing to advancements in computing machines and numerical methods, led to the development of computer simulations of the liquid state. Computer simulations of liquids have since become one of the most widely used and successful methods for understanding chemical processes fundamental to energy storage and transfer, biological function and health, and even the behavior of interstellar bodies. Clearly, if a general liquid state theory were discovered, many of us molecular dynamicists might soon find ourselves out of a job!

Despite lacking a unified theory of the liquid state, what we know is that there are fundamental relationships between the arrangement and organization of molecules and emergent thermodynamic properties. The arrangement of molecules, which can be charted by a spherical coordinate system with vectors corresponding to the radial, polar, and azimuthal positions of atoms in the system, will be referred to as the \textit{structure} of a liquid. Of course, our limited senses and technology make it practically impossible to know the structure at any instant, or snapshot, in time. However, measurement techniques such as X-ray and neutron scattering can be used to estimate the structure that the liquid takes on average. It is these averages of the atomic coordinates that can be used within existing liquid state theories to establish a connection to interatomic forces and thermodynamic behavior. 

Correctly linking liquid structure with theoretical statistical mechanics holds great promise in engineering novel liquid state materials and advancing our fundamental understanding of the liquid state. Such a method could obviate the need for designing and training expensive computer simulations, inform engineers on subtle, atomic scale structure-property relationships, and possibly be used in tandem to \textit{ab initio} electronic structure calculations to understand how quantum mechanical effects impact condensed phase properties. Finally, studying the connection between structure, atomic scale forces, and thermodynamics is of great theoretical interest. It pushes the boundaries of what can be learned from existing theories and has the potential to identify key problems with predominant schools of thought on the connection between atomic and thermodynamic behavior. Ultimately, such an approach could form a unified theory of liquids that comprehensively describes molecular, thermophysical, and continuum-scale behavior from a molecular perspective. \textbf{Towards this ultimate goal, this dissertation aims to propose a novel interpretation of structure-potential modeling through the application of Bayesian probability theory.}

To this aim, the remainder of this introduction will be organized as follows:

\begin{enumerate}
    \item Measuring Liquid Structure with Experimental Neutron Scattering - Overview of liquid structure measurement techniques using neutron scattering experiments. It discusses fundamental mathematical relationships between observed quantities (\textit{e.g.}, the structure factor) and derived quantities (\textit{e.g.}, radial distribution functions), essential for interpreting liquid structure. Furthermore, it addresses key experimental and modeling challenges involved in extracting interatomic forces from structural data.

    \item A Comprehensive Review of Structure-Potential Analysis - Literature review of state-of-the-art applications of liquid state theory in the analysis of neutron scattering data. It outlines key methods such as the Ornstein-Zernike integral relation, the Henderson inverse theorem, and empirical potential structure refinement. Additionally, a detailed yet concise proof of the Henderson inverse theorem is included, as it serves as the primary evidence for the hypothesis that a unique pair potential can be derived from scattering measurements for real liquids. This theorem also underpins the variational method of structure optimized potential refinement described in Chapter 2.

    \item Modern Scattering Analysis: A New Perspective - Introduces the central thesis of the dissertation, emphasizing the potential of a novel perspective centered on uncertainty quantification. It defines and justifies the adoption of uncertainty quantification using Bayesian analysis and outlines its application within the text. Specifically, it discusses the implementation of Bayesian techniques for both continuous functions (Gaussian processes - Chapters 2-4) and discrete variables (Bayesian parameter optimization - Chapters 3 and 4).
\end{enumerate}

\subsection{Measuring Liquid Structure with Experimental Neutron Scattering}

In statistical mechanics, the structure of a liquid is characterized using the radial distribution function, $g(r)$. This function is computed by counting the number of particles surrounding a reference particle and constructing a radial shell of thickness $dr$ around this particle. The average value of the radial distribution function at $r+dr$ is the particle number density within the shell divided by the bulk particle density of the material. This process is repeated for each particle in the system and over time, and the results are subsequently averaged. Thus, $g(r)$ represents a radial, time, and particle average density distribution. A visualization of a single snapshot of this computation is shown in Figure \ref{fig:rdfvis} (radial distribution function was taken from Yarnell (1974) \cite{yarnell_structure_1973}).

Neutron scattering is the gold-standard technique to measure radial distribution functions from sub-angstrom to micron length scales for systems with light-atoms (such as hydrogen) \cite{willis_experimental_2017}. When a neutron beam is directed through soft matter, incident neutrons collide with atomic nuclei and scatter in various directions (Figure \ref{fig:scatter}). A detector mounted behind the sample container is designed to quantify single neutron interaction events as a function of momentum transfer, yielding an experimental observable known as the differential scattering cross section, $\frac{d \sigma}{d \Omega}$. The scattering cross section is the ratio of scattered neutrons per second into solid angle $d\Omega$ divided by the number of neutrons incident to $d\Omega$ (which is just $d\Omega$ times the incident flux and has units of barns/steradian) and contains contributions from a wide array of neutron-atom interactions, including elastic, inelastic, incoherent, and multiple scattering \cite{soper_inelasticity_2009}. 

\begin{figure}
    \centering    \includegraphics[width = 14cm]{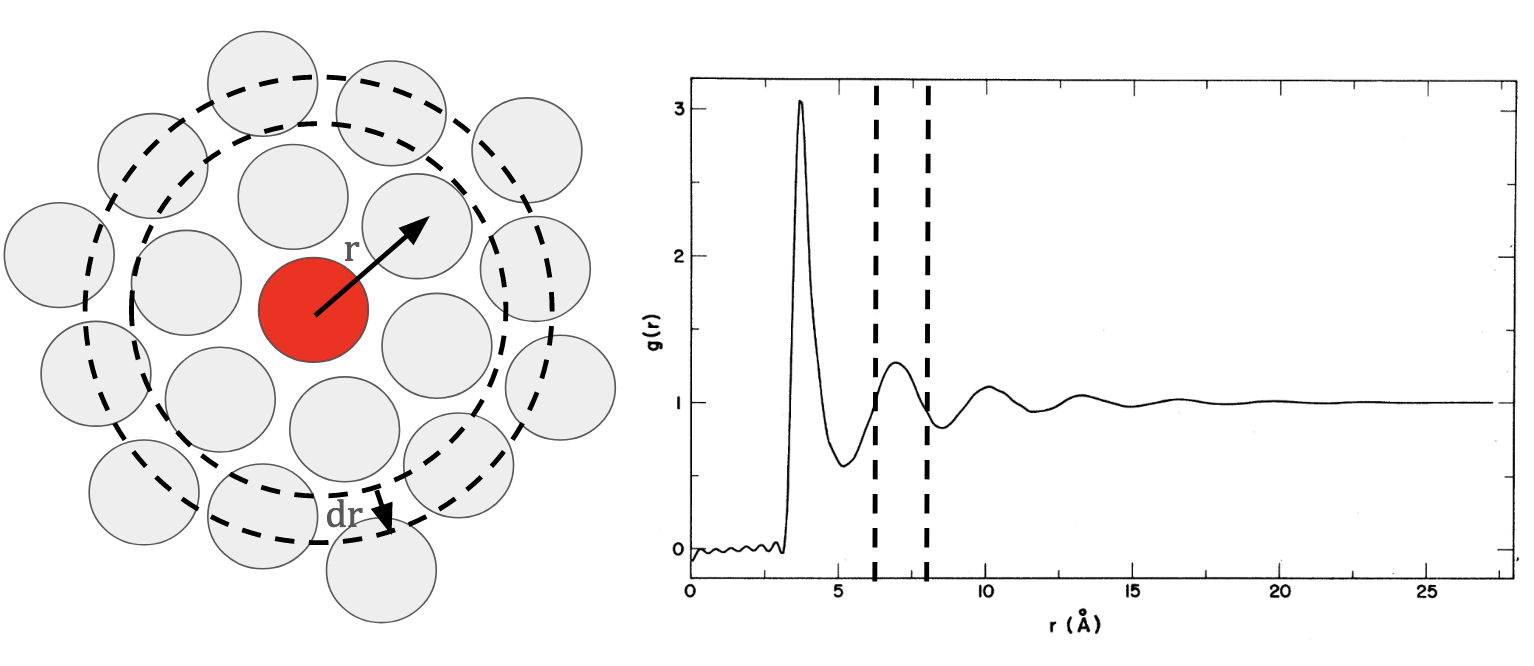}
    \caption{The radial distribution function keeps track of the particle density of a system as a function of radius away from a reference atom. The dashed radial shell in the particle picture (left) is represented as a radial interval in the radial distribution function (right).}
    \label{fig:rdfvis}
\end{figure}

\begin{figure}
    \centering
    \includegraphics[width = 12cm]{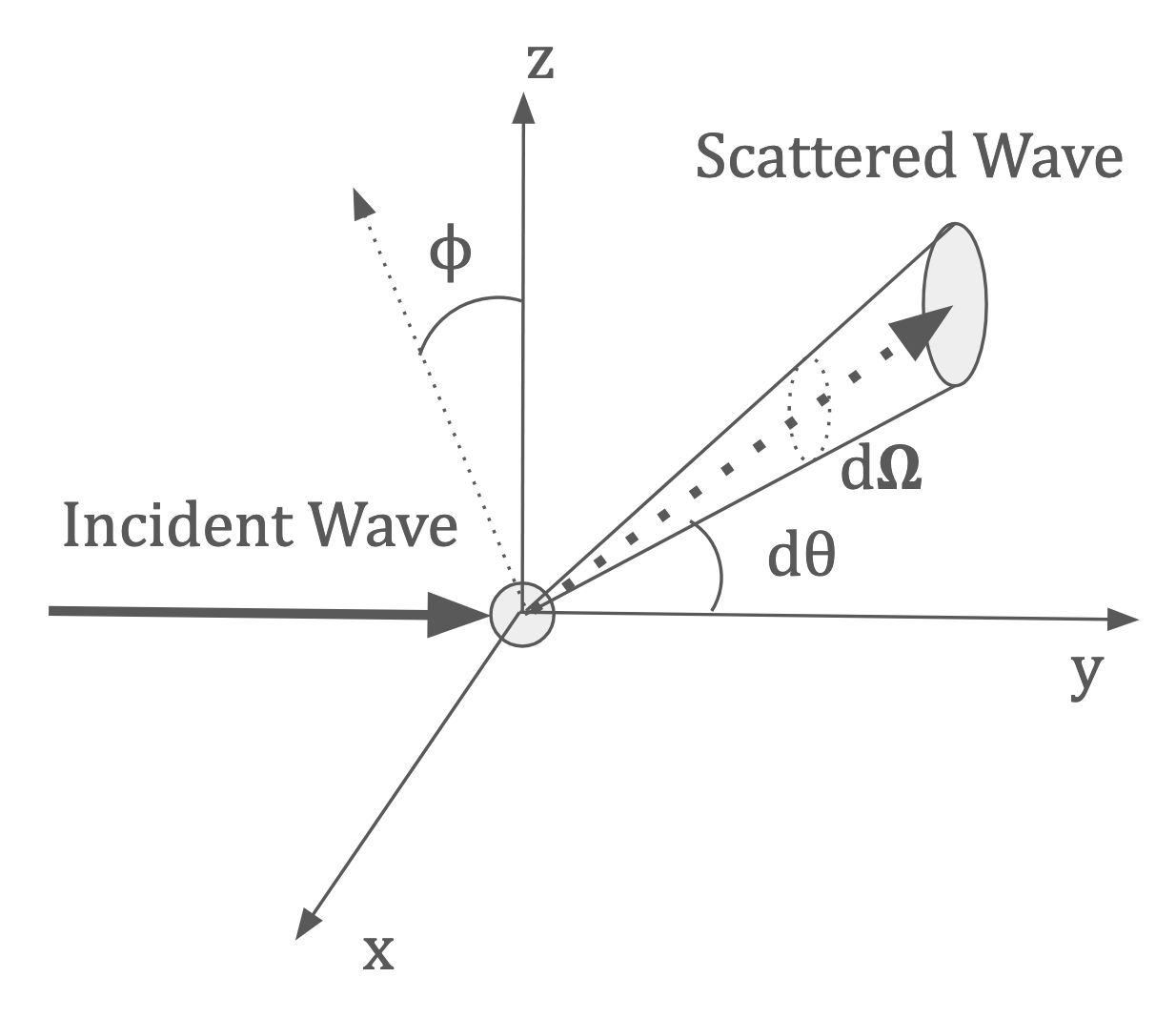}
    \caption{Incident scattering vector $\mathbf{k_i}$ scattered through solid angle $d \Omega$. }
    \label{fig:scatter}
\end{figure}

The time-averaged, elastic contribution to the scattering cross section, named the static structure factor, is the linear density response of the material to the neutron wave propagation momentum, $\hslash \mathbf{q}$. By the fluctuation-dissipation theorem, the linear density response of a perturbed system can be expressed in terms of equilibrium fluctuations in the unperturbed system. Therefore, the static structure factor measures the time-averaged, equilibrium particle density distribution in momentum space \cite{kubo_fluctuation-dissipation_1966}. The static structure factor in a monatomic system with no long-range order is related to the more familiar radial distribution function, $g(r)$, through the radial Fourier transform, 

\begin{equation}\label{eq:radFT}
    S(q) = 1 + \frac{4 \pi}{q \rho} \langle b \rangle^2 \int_0^\infty r [g(r) - 1] \sin (qr) dr 
\end{equation} 

\noindent where $q$ is the momentum transfer, $b$ is the scattering length density, and $\rho$ is the atomic number density \cite{sivia_elementary_2011}.

In mixtures or molecular liquids, the static structure factor, referred to in this case as the total structure factor $F(q)$, can be expressed as a combination of of site-site partial structure factors, $s_{ij}$, between atoms $i,j$ such that,

\begin{equation} \label{eq:faber}
    F(\mathbf{q}) = \sum_{i\geq j} [2 - \delta_{ij}]w_{ij}s_{ij}(\mathbf{q})
\end{equation}

\noindent where $w_{ij}$ is a weighting factor depending on the scattering length density and atomic concentration of the $i,j$ pair and $\delta_{ij}$ is the Kronecker delta. Partial structure factors can be Fourier transformed with Eq. \eqref{eq:radFT} to obtain real space site-site pair distribution functions. These site-site radial distribution functions show how the atomic density of type 1 within a spherical shell around any atom of type 2 in the system changes with respect to the radius of the shell. For example, in liquid water, an O-H partial distribution function describes atomic fluctuations of oxygen atoms around any arbitrary hydrogen atom. Eq. \eqref{eq:faber}, referred to as the Faber-Ziman approximation, is ill-posed (since it has no unique solution) and can only be approximated via iterative molecular simulation approaches.

While eqs \eqref{eq:radFT} and \eqref{eq:faber} hold true in theory, practical implementation of these models face several challenges. First, the finite size of individual neutron detectors constrains structure factor measurements to discrete momentum transfer, $\Delta q = q_i - q_{i-1}$, values which can result in aliasing if the sampling efficiency is $<1$ (\textit{i.e.} the Peterson-Middleton theorem \cite{petersen_sampling_1962}). Second, finite detector coverage windows the measurement to a range between some $q_\text{min}$ and $q_\text{max}$, preventing the evaluation of the full integral specified in Eq. \eqref{eq:radFT}. Finally, measurement uncertainty of neutron counts and momentum transfer positions (\textit{i.e.} time-of-flight uncertainty) introduce noise that can corrupt the underlying signal \cite{neuefeind_nanoscale_2012}. These limitations mean that we can only compute a discrete radial Fourier transform over uncertain observations,

\begin{equation}
    g(r) \approx 1 + \frac{1}{2\pi^2\rho} \sum_{q_i=q_\text{min}+q_1}^{q_\text{max}} \frac{\mathcal{S}(q_{i-1}) - \mathcal{S}(q_{i})}{2}(q_i - q_{i-1}) 
\end{equation}

\noindent where we have introduced the notation $\mathcal{S}(q) = S(q)\frac{\sin(qr)}{qr}q^2$ for brevity. The key problem here is that this discrete Fourier transform can introduce systematic deviations in the predicted $g(r)$ from the ground truth one.

The most well studied issue in prior literature is addressing the $q_\text{max}$ cutoff using so-called modification functions. The essential idea here is to smoothly transition the structure factor from a data dominated section (as measured by the neutron detector) to a model driven section (dictated by prior physical knowledge of the structure factor). Modification functions are designed to force the contribution of the experimental data to 0 near $q_\text{max}$, effectively nullifying any features in the data and strictly relying on the physical model alone. Usually the data is transitioned into is the Poisson point process ideal gas model (i.e. $S_\text{Ideal}(q) = 1$)\cite{torquato_hyperuniform_2018}. Mathematically, this adjusts the integral seen in equation \ref{eq:radFT} into,

\begin{equation}
    g(r) = 1 + \frac{1}{2\pi^2\rho}\int_{0}^{\infty} (S(q)-1)M(q)\frac{\sin(qr)}{qr}q^2 dq \label{eq:modtransform}
\end{equation} 

\noindent where $M(q)$ is the modification function as a function of $q$. Common choices for the modification function are the first Bessel function \cite{lorch_neutron_1969}, second Bessel function \cite{soper_use_2012,soper_extracting_2011}, cosine cutoff \cite{bellissent-funel_neutron_1992}, and dynamic functions \cite{skinner_benchmark_2013}. However, as pointed out by J.E. Proctor and co-workers \cite{proctor_comparison_2023}, this is an approximate Bayesian predictive model where the modification function is used to decide where a prior model of the $S(q)$ should be preferred over the data. This can be seen by rewriting equation $\ref{eq:modtransform}$ so that,

\begin{equation}
    = 1 + \frac{1}{2\pi^2\rho}\int_{0}^{\infty} \bigg(\underbrace{(S(q)-1)M(q)}_\text{Data Driven Predictive}+ \underbrace{(S_\text{Ideal}(q)-1)(1-M(q))}_\text{Model Driven Predictive}\bigg)\frac{\sin(qr)}{qr}q^2 dq.\label{eq:approxbayestransform}
\end{equation} 

If you view $M(q)$ as discrete posterior probability mass, then this expression can be thought of as writing the structure factor as a weighted mixture of the two outcomes, either data or model, at each $q$ value. 

\subsection{A Comprehensive Review of Structure-Potential Analysis}

At this point, it should be clear that the determination of radial distribution functions are obscured by experimental, model, and measurement uncertainty. However, assuming that the radial distribution functions can be determined accurately, there are a few fundamental liquid state theories that allow us to predict interatomic forces between atoms in the system and consequently their thermodynamic behavior. One of the most important and far reaching results is the Ornstein-Zernike (OZ) integral relation \cite{ornstein_accidental_1914}.

The OZ relation can be proven by noticing that the excess part of the free energy functional generates a set of \textit{direct} correlation functions that can be related to the density correlation function \cite{hansen_theory_2013}. Conceptually, all that we are doing is a statistical mechanical "book keeping" of the direct and indirect atomic correlations. Unfortunately, on its own the OZ relation is not enough to take an experimental scattering measurement and determine an interatomic potential. Instead, we introduce some approximation of how the direct correlation function is related to the interatomic potential using a \textit{closure relation}. Practically, OZ integral relation methods generally have slow convergence; and further, owing to their approximate nature, often do not give accurate predictions for interatomic potentials and thermodynamic properties for real systems \cite{levesque_pair_1985,hansen_theory_2013}. It is notable, however, that the emergence of machine learning in liquid state theory has shown promise in mitigating these challenges \cite{wu_perfecting_2023}. For example, neural networks have been implemented to solve both the forward and inverse Ornstein-Zernike integral relations for simple liquids \cite{carvalho_radial_2020,chen_physics-informed_2024}. Further note that there are other integral relations from liquid state theory, including the Yvon-Born-Green equation and the Bogoliubov–Born–Green–Kirkwood–Yvon (BBGKY) hierarchy, but these have seen little use for scattering analysis.

The Henderson inverse theorem, first published in 1974, is another important and more practical result on the relationship between the radial distribution function and pairwise additive potential in a statistical ensemble. The Henderson inverse theorem states that, given a fixed density, homogeneous system with pairwise additive Hamiltonian and the same radial distribution function, their pair potentials can differ by at most a trivial constant \cite{henderson_uniqueness_1974}. The importance of this result lies in the fact that, if one can find a potential that reproduces a target radial distribution function, this must be the unique (up to an additive constant) potential that models that system. Supposing that this pair potential was sufficient to model the interatomic forces, it is possible, in principle, to recover thermodynamic consistency with liquid state theories such as the Ornstein-Zernike relation.

To prove this theorem, we first need to establish the Gibbs and Gibbs-Bogoliubov inequalities for a quantum system. The Gibbs inequality is an important result about the information entropy of a system while the Gibbs-Bogoliubov inequality establishes a relationship between the free energy and entropy in the canonical ensemble.

\begin{lemma}
    Let $\rho_1$ and $\rho_2$ be positive, trace-class, and linear density operators on a Hilbert space, $H$, such that $Tr(\rho_i) = 1$. Then, 
    $$
        Tr(\rho_1 \log (\rho_2)) \leq Tr(\rho_1 \log (\rho_1))
    $$
\end{lemma}

\begin{proof} 

\noindent We can express the states $\rho_1$ and $\rho_2$ in an arbitrary basis of $H$ (\textit{c.f} Riesz's Lemma \cite{riesz_functional_2012}) such that,

$$
\rho_1 = \sum_\alpha p_\alpha \ket{\alpha} \bra{\alpha}
$$

$$
\rho_2 = \sum_\alpha q_\alpha \ket{\alpha} \bra{\alpha}
$$

\noindent We then compute the difference between the cross entropy, $\rho_1 \log (\rho_2)$, and information entropy of $\rho_1$, $\rho_1 \log (\rho_1)$,

$$
\rho_1 \log (\rho_2) - \rho_1 \log (\rho_1) = \sum_{\alpha,\beta} [p_\alpha \ket{\alpha} \bra{\alpha} \log (q_\beta) \ket{\beta} \bra{\beta} - p_\alpha \ket{\alpha} \bra{\alpha} \log (p_\beta) \ket{\beta} \bra{\beta}]
$$

\noindent and since ${\alpha}$ and ${\beta}$ are orthonormal bases,

$$
= \sum_{\alpha} [p_\alpha \log (q_\alpha) - p_\alpha \log (p_\alpha)]\ket{\alpha} \bra{\alpha}.
$$

\noindent Taking the trace of this operator we obtain,

$$
Tr(\rho_1 \log (\rho_2) - \rho_1 \log (\rho_1)) = \sum_{\alpha} [p_\alpha \log (q_\alpha) - p_\alpha \log (p_\alpha)] = \sum_{\alpha} p_\alpha \log \frac{q_\alpha}{p_\alpha}.
$$

\noindent Note that since $\log x \leq x - 1$,

$$
\sum_{\alpha} p_\alpha \log \frac{q_\alpha}{p_\alpha} \leq \sum_{\alpha} [q_\alpha - p_\alpha] = 0
$$

\noindent and finally, since the trace is a linear operator, this means that,

$$
Tr(\rho_1 \log (\rho_2)) \leq Tr(\rho_1 \log (\rho_1))
$$.

\end{proof}

\begin{lemma}

Let $\rho_1$ and $\rho_2$ be positive, trace-class, and linear density operators on a Hilbert space, $H$, such that $Tr(\rho_i) = 1$. Then, in the canonical ensemble where, $\rho = \exp(-\beta \mathcal{H}) / Z$, where $\beta$ is the inverse thermal energy, $\mathcal{H}$ is the Hamiltonian and $Z$ is the canonical partition function, then,  

$$
F_2 \leq F_1 + \langle \mathcal{H_2} - \mathcal{H_1} \rangle_1.
$$

\end{lemma}

\begin{proof} 

Suppose we take the state $\rho_2$ in the canonical ensemble so that,

$$
\rho_2 = \exp(-\beta \mathcal{H_2}) / Z
$$

\noindent where $\beta$ is the inverse thermal energy, $\mathcal{H_2}$ is the Hamiltonian, and $Z$ is the partition function. Then for some $\rho_1$ we have,

$$
Tr(\rho_1 (-\beta \mathcal{H_2} - \log Z)) \leq Tr(\rho_1 \log (\rho_1)) = - S_1 / k_B
$$

\noindent where $S_1 = -k_B Tr(\rho_1 \log (\rho_1))$ is the entropy of system 1 and $k_B$ is the Boltzmann constant. Since the trace is a linear operator, we can separate the argument of the trace on the left hand side and divide both sides by the thermodynamic $\beta$ to obtain,

$$
Tr(\rho_1\mathcal{H_2}) + k_BT \log Tr(Z) \geq + T S_1.
$$

\noindent But $Tr(\rho_1\mathcal{H_2})$ is just the expectation of $\mathcal{H_2}$ over system state 1 and $-k_BT \log Tr(Z)$ is the definition of the Helmholtz free energy in the Canonical ensemble. Thus,

$$
F_2 \leq \langle \mathcal{H_2} \rangle_1  - T S_1
$$

\noindent and for system 1, 

$$
F_1 \leq \langle \mathcal{H_1} \rangle_1  - T S_1.
$$

\noindent Combining the two expressions gives us the inequality,

$$
F_2 \leq F_1 + \langle \mathcal{H_2} - \mathcal{H_1} \rangle_1.
$$

\end{proof}

\noindent The content of Henderson's inverse theorem can now be stated as follows:

\begin{theorem}
Two systems with Hamiltonian's of the form,

$$
  \mathcal{H} = \sum_i \frac{p_i^2}{2m} + \frac{1}{2} \sum_{i \neq j} u(|\mathbf{r_i} - \mathbf{r_j}|)
$$

\noindent with the same radial distribution function, $g^{(2)}(\mathbf{r_i},\mathbf{r_j})$,

$$
g^{(2)}(\mathbf{r_i},\mathbf{r_j}) = \frac{1}{\rho^2}\bigg\langle \sum_i \sum_j \delta(\mathbf{r} - \mathbf{r_i}) \delta(\mathbf{r}' - \mathbf{r_j}) \bigg\rangle
$$

\noindent have pair potentials, $u(|\mathbf{r_i} - \mathbf{r_j}|)$, that differ by at most a trivial constant.

\end{theorem}

\begin{proof} 

Suppose that two systems with a pairwise additive Hamiltonian have equal radial distribution functions and $u_1 - u_2 \neq c$ where $c$ is some constant. Then,

$$
\langle \mathcal{H_2} - \mathcal{H_1} \rangle_1 \neq c
$$

\noindent and since the Helmholtz free energies are constants,

$$
F_2 - F_1 < \langle \mathcal{H_2} - \mathcal{H_1} \rangle_1
$$

\noindent where we lose the possibility of equality from the Gibbs-Bogoliubov inequality. Now, we can expand the expectation of the Hamiltonian in terms of the radial distribution function (since the system is pairwise additive) so that,

$$
F_2 - F_1 < \frac{n}{2} \int [u_2 - u_1] g_1(\mathbf{r}) d^3\mathbf{r}
$$

\noindent and the same holds for a swap of the indices,

$$
F_1 - F_2 < \frac{n}{2} \int [u_1 - u_2] g_2(\mathbf{r}) d^3\mathbf{r}.
$$

\noindent Combining these two equations gives,

$$
0 < 0
$$

\noindent a contradiction. Therefore, our premise that the radial distribution functions are equal while the pairwise additive potential energies differ by a trivial constant must be false. The only other possible difference between the potential energies is constant, so this must be true to satisfy the Gibbs-Bogoliubov inequality.

\end{proof}

Analogous results can be proven more generally using the methods of relative entropy \cite{shell_relative_2008} and functional analysis (see Appendix D) \cite{hanke_well-posedness_2018,frommer_note_2019,frommer_variational_2022}, although it is notable that there is no guarantee that a unique potential exists for a given radial distribution function. And furthermore, these results take us no further in practice since the radial distribution function can only be inferred and not known exactly. 

The Henderson inverse theorem forms the basis of modern numerical structure inversion methods including Schommers algorithm \cite{schommers_pair_1983}, iterative Boltzmann inversion (IBI) \cite{moore_derivation_2014,shakiba_machine-learned_2024}, and empirical potential structure refinement (EPSR) \cite{soper_empirical_1996}. Currently, IBI is generally used for training coarse-grained force fields from more computationally expensive all-atom force fields models, while EPSR is exclusively used to find molecular configurations that are consistent with experimental neutron scattering patterns \cite{,hammond_effect_2017,fheaden_structures_2018}.

IBI and EPSR are iterative predictor-corrector algorithms that can be summarized as follows: (1) a predictor is applied to 'predict' the underlying interatomic potential given an experimental pair correlation function, (2) a corrector performs a molecular simulation (Monte Carlo molecular mechanics) with the predicted potential to evaluate the quality of fit between the simulated and experimental pair correlation functions, and (3) the corrector results are fed back to the predictor to refine or 'correct' the prior potential. Steps 1-3 are iterated until the simulated and experimental pair correlation functions converge. While these methods can be used for coarse-graining or to determine molecular configurations consistent with a given scattering pattern, there is little evidence that interatomic potentials obtained from these techniques can reliably predict thermodynamic behavior for real liquids. For example, Soper showed that O-O and H-H site-site interatomic potentials derived from EPSR applied to scattering data of liquid water predict a 4 times more negative excess internal energy compared to the experimental value \cite{soper_empirical_1996} and later concluded that EPSR cannot be used to derive a reliable set of site-site pair potentials for a given system \cite{soper_tests_2001}.

In our search for a method that can fit experimental scattering data, estimate interatomic forces and model liquid state thermodynamics, clearly EPSR falls short. First, the molecular configurations and potentials generated from these methods can be wildly non-physical \cite{pruteanu_krypton_2022} and non-unique \cite{soper_uniqueness_2007} (Figure \ref{fig:nonunique}). This lack of robustness maybe due to the specific construction of EPSRs Monte Carlo molecular mechanics step or because its solutions fluctuate on the whim of user-specified model inputs. Second, the fact that the determination of site-site partial structure factors is an ill-posed inverse problem performed on uncertain experimental data means that there is no unique molecular configuration that explains a given scattering measurement. The problem here is quite grave since, even if we did find a molecular configuration that matches the scattering data, there is no guarantee that this configuration makes any sense at all. Finally, it is well-known that there are significant quantum mechanical (both nuclear and electronic) and many-body effects that influence structural behavior in liquids that EPSR does not take into account in its molecular models. It is becoming increasingly evident that such effects significantly influence solvent behavior and self-assembly \cite{pereyaslavets_accurate_2022}. As methods to model these behaviors evolve and become more computationally accessible, the physical models underlying liquid structure prediction should evolve accordingly.

\begin{figure}
    \centering
    \includegraphics[width = 12cm]{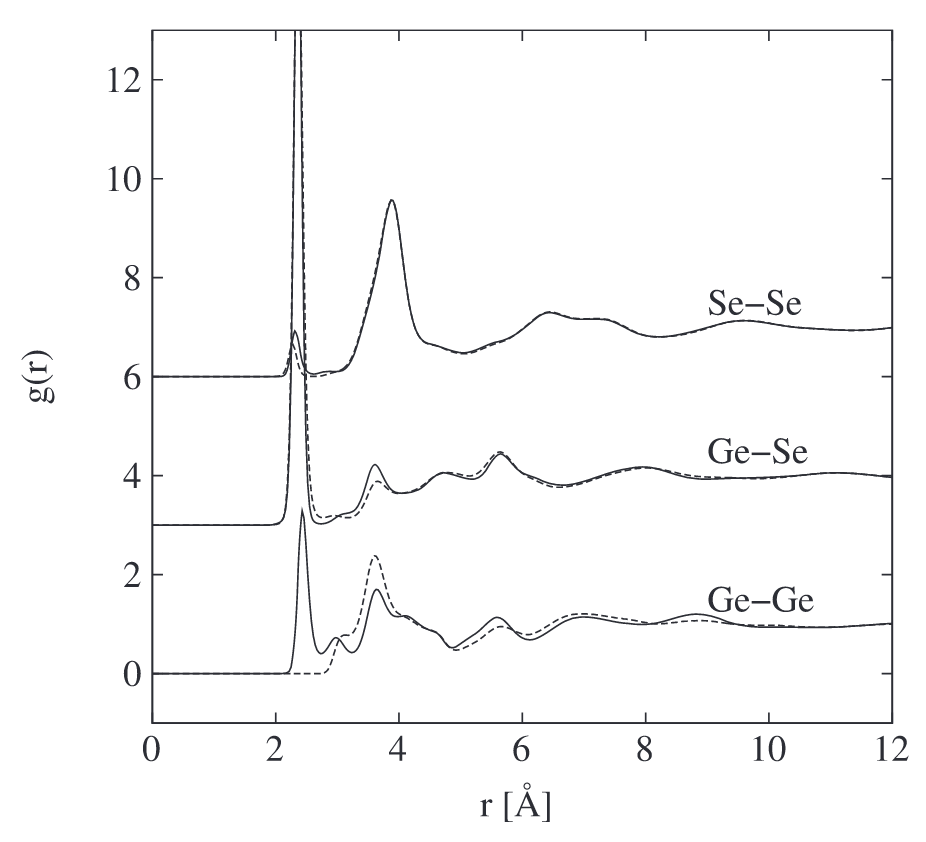}
    \caption{Site-site radial distributions for amorphous $GeSe_2$ determined from different EPSR runs differ by $\sim$16\%. Figure reproduced from Alan Soper, \textit{On the uniqueness of structure extracted from diffraction experiments on liquids and glasses}, Journal of Physics : Condensed Matter, Volume 19, Issue 41, Page 12, 01/01/1989. © IOP Publishing. Reproduced with permission. All rights reserved.}
    \label{fig:nonunique}
\end{figure}

To summarize the current state-of-the-art, what we have is a scientific landscape where we can accurately measure atomic scale correlations between atoms in complex molecular liquids with neutron scattering experiments, along with multiple and sound theoretical results allowing us to connect this information to fundamental interatomic forces and thermodynamics. And yet, a clear demonstration of this connection has evaded the scientific literature for over a century. While interatomic potentials can be derived with existing tools like EPSR, the fact that the results are not robust, can change from run to run, and do not reproduce thermodynamic behavior suggests that existing techniques are generally unreliable and not suitable as a bridge between structure, interatomic forces, and thermodynamic properties. Therefore, what is needed is a careful reexamination of the philosophical, theoretical, and computational approaches to this complex inverse problem.

One could say that the problem of building thermodynamically consistent structure-potential models for liquids is similar to "finding a needle in a hay stack". But in fact, the situation is much worse since we wouldn't be able to tell the difference between the needle and the hay even if we found it! A more accurate description could be summarized as "finding a specific needle in a needle stack".

\subsection{Modern Scattering Analysis: A New Perspective}

Based on the latest advancements in liquid structure modeling, the prevailing challenges that remain in the field can be outlined as follows:

\begin{enumerate}
    \item Experimental scattering data is subject to numerous sources of random errors, arising from uncontrolled effects and improper data corrections and manipulations (e.g., Fourier transforms). Consequently, interpreting scattering measurements in real space is non-trivial and can introduce uncertainty into the observed structures used in liquid state theories.

    \item The non-uniqueness of real-space structure interpretation implies that commonly used scattering analysis tools, such as EPSR, are sensitive to model parameter inputs and do not guarantee the same solution for each run. Additionally, the lack of uncertainty quantification in existing methods means we cannot gauge the reliability of the results.     

    \item Current state-of-the-art methods for the structure-potential inverse problem, such as the Ornstein-Zernike relation and Henderson inverse theorem based methods, are rigorously derived from statistical mechanics. These equations necessitate approximate closure relations, molecular simulations, or advanced machine learning algorithms (e.g., neural networks) to solve. Thus, these methods are often slow to converge and have primarily been successful only in simple fluid models, thereby limiting their practical utility in studying real liquids. 
\end{enumerate}

Clearly, uncertainty is the main thread that occludes every step towards a consistent structure-property model of liquids. First, the scattering data is uncertain, making it impossible to confirm if a 'correct' measurement is being applied as the structure target. The statistical mechanical model and its parameters are also uncertain. For example, the Ornstein-Zernike equation may be rigorous for a system of classical particles, yet contains no description of important quantum mechanical effects of the electrons or nuclei. Moreover, the resulting solution lacks guaranteed uniqueness, leaving us unable to verify its accuracy without external validation through molecular simulations or thermodynamic calculations.

In this dissertation, it is my aim to outline an alternative philosophy to liquid state theory that centers on the key idea of making decisions in the face of uncertainty. Here \textit{uncertainty} will refer to our current state of knowledge (\textit{i.e.} the value of an observable within some credibility interval) given model, parametric, experimental, and computational uncertainties as defined below:

\begin{enumerate}

\item \textit{Model uncertainty} refers to the fact that there is never a 'perfect' model of nature that we can use to predict a quantity-of-interest. In the context of molecular simulations, model uncertainty is associated with choosing a specific force field, or choosing to use a path integral molecular dynamics method rather than an \textit{ab initio} electron structure method. 

\item \textit{Parametric uncertainty} refers to the uncertainty in the parameters we use within a given model. For instance, the Lennard-Jones force field parameters for argon are $\sigma = 3.4$ \AA and $\epsilon = 0.24$ kcal/mol. But how sure can we be that these are exactly correct parameters? What if I choose a slightly smaller or larger $\sigma$? Does it really effect the results of the molecular simulation? In reality, there is a distribution of parameters that can model a given quantity-of-interest, and we can think of this distribution as representing parametric uncertainty. 

\item \textit{Experimental uncertainty} refers to the fact that measurements are subject to numerous sources of error that make the data deviate from the ground truth. A straightforward example is noise in signals or systematic error from a thermocouple being poorly calibrated.

\item \textit{Computational uncertainty} refers to uncertainty in numerical calculations. Although these errors can be quite small, floating point errors or misallocation of memory in parallel algorithms can cause deviations from true solutions, particularly in complicated numerical problems that require multiprocessing. 

\end{enumerate}

Bayesian methods are the gold-standard in treating these obscure types of uncertainty in a mathematically rigorous way. The idea is to learn the \textit{posterior} probability distribution, $p(\theta|D)$ (read probability of $\theta$ given $D$), of some quantity-of-interest $\theta$ (this could be a parameter, model, function, field, etc), given observed data $D$ according to the following equation,

\begin{equation}
    p(\theta|D) = \frac{p(D|\theta)p(\theta)}{p(D)}
\end{equation}

\noindent where $p(D|\theta)$ is the \textit{likelihood} that the data is well modeled by $\theta$, $p(\theta)$ is our \textit{prior}, and $p(D)$ is the probability of $D$ being observed at all \cite{gelman_bayesian_1995}. Note that no restriction has been imposed on the form of $\theta$, $D$, the likelihood or prior aside from the fact that any quantity $p(\cdot)$ must be a probability distribution (\textit{i.e.} it must be non-negative and integrate to one over all possible events). In fact, $\theta$ can be a single quantity, tensor, function, or field and $D$ can be a list of many observations or a combination of completely different observations without loss of generality. If the quantity-of-interest is a function, then stochastic processes (\textit{e.g.} Gaussian processes) are typically invoked using Bayesian nonparametrics \cite{hjort_bayesian_2010} (see Appendix C).

Although Bayes' Theorem is a straightforward statement of conditional probability, it is conceptually powerful. It provides a method for updating our beliefs about a hypothesis based on new evidence by combining prior knowledge with new data to revise the probability of the hypothesis being true. The flexibility in choosing likelihood and prior probability distributions has led some to criticize Bayesian statistics as being too subjective \cite{finetti_theory_1979}. However, this flexibility allows for the incorporation of expert knowledge, ensuring that models remain realistic and physically plausible. This interpretability gives Bayesian inference a significant advantage over other black-box machine learning methods, such as neural networks or variational autoencoders, for function approximation and uncertainty quantification \cite{zuckerman_bayesian_2024}.

Bayesian inference finds diverse applications, ranging from training models to predicting outcomes with credibility estimates, and even forecasting functional and field distributions \cite{bishop_pattern_2006}. In essence, it offers a structured approach to compute both discrete probability mass densities and continuous probability distribution functions across parameters and quantities-of-interest. These probabilistic assessments are instrumental in decision theory applications, facilitating risk and loss quantification, as well as sensitivity analysis of model parameters in predicting outcomes. Moreover, they serve as invaluable guides in research, pinpointing gaps in existing models and directing further investigations. Remarkably, Bayesian field theory even reveals fundamental connections with statistical physics and quantum mechanics, offering insights into the complex phenomena of atomic systems \cite{lemm_bayesian_2003}.

At first, uncertainty can be a confusing concept, as it is not immediately clear how to express a 'lack of knowledge' in a rigorous way. For this, we will need the tools of probability theory \cite{finetti_theory_1979} and its numerical counterpart, probabilistic machine learning \cite{bishop_pattern_2006}. Probabilistic machine learning has experienced a renaissance in the last few decades, with applications becoming commonplace across a broad landscape of contemporary science including astronomy \cite{thrane_introduction_2019}, ecology \cite{tredennick_practical_2021}, and genetics \cite{sivaganesan_bayesian_2008,sivaganesan_improved_2010,russel_model_2019,sivaganesan_fecal_2024}, among others. Liquid structure analysis is no different. In fact, it has been speculated that novel approaches to neutron scattering would likely be Bayesian in nature \cite{soper_uniqueness_2007}, and we are just starting to see probabilistic machine learning methods applied to neutron diffraction \cite{doucet_machine_2020}. Uncertainties have been reported for the estimated partial radial distributions for water \cite{soper_radial_2013} (see Figure \ref{fig:soper_unc}), although this is likely an underestimate since the analysis did not consider experiment, model or parameter uncertainty.

\begin{figure}
    \centering
    \includegraphics[width = 10cm]{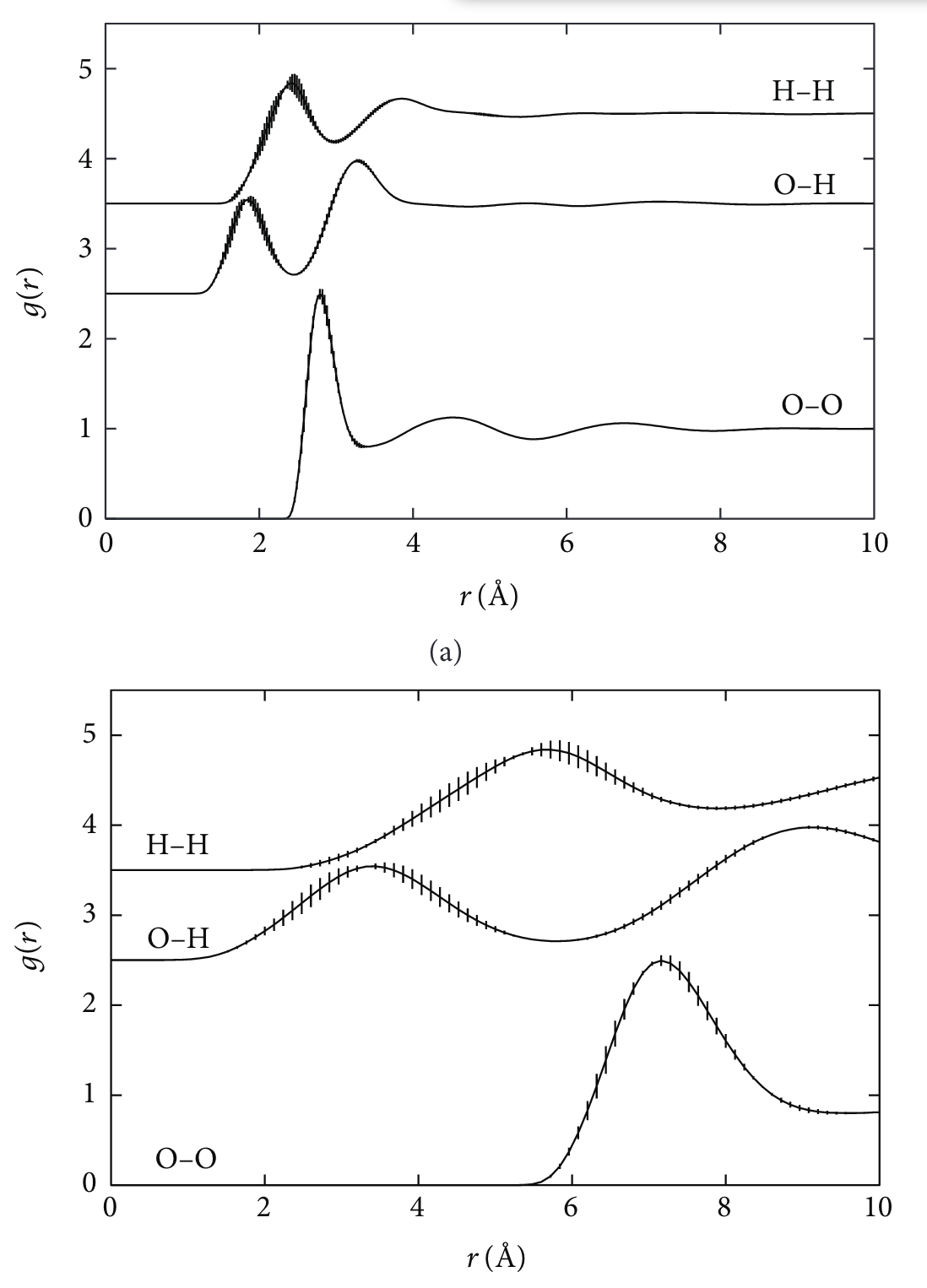}
    \caption{Partial radial distribution functions for HH, OH, OO atom pairs with uncertainty estimates from 6 EPSR runs (vertical bars). Figure reproduced from Alan Soper, \textit{The Radial Distribution Functions of Water as Derived from Radiation Total Scattering Experiments: Is There Anything We Can Say for Sure?}, ISRN Physical Chemistry, Volume 2013, 11/28/2013.}
    \label{fig:soper_unc}
\end{figure}

This dissertation hypothesizes that rigorous uncertainty quantification using Bayesian inference can enhance the capabilities of liquid state theory and address its primary challenges. For instance, applying Bayesian uncertainty quantification to experimental scattering data can reduce the risk of overfitting to poor data. Additionally, recognizing that model results come from a distribution of possible solutions allows for the quantification of non-uniqueness and lack of robustness in a given model. Using Bayesian Gaussian processes to quantify this distribution enables us to learn model outputs with uncertainty, even for complex functions like the pair potential (Chapter 2) and structural correlation functions (Chapters 3 and 4). Furthermore, physics-informed Gaussian process design can ensure that the model predictions abide by physically justified principles such as continuity and differentiability, mitigating non-physical solutions often observed from EPSR. Bayesian optimization can also estimate model and parameter uncertainty within a selected model framework (Chapters 3 and 4), helping to determine whether a model choice is adequate for explaining a given quantity of interest. 

\subsection{Outline and Scope of the Thesis}

The main content of this dissertation includes three chapters in which a novel application of Bayesian inference is applied to advancing the start-of-the-art in structure-property models. Chapter 2 is a reproduction of my first paper, "Transferable force fields from experimental scattering data with machine learning assisted structure refinement" originally published in the Journal of Physical Chemistry Letters in December, 2022. The chapter describes how to implement nonparametric Bayesian methods, specifically Gaussian process regression, within an iterative Boltzmann inversion framework to learn interatomic potentials from neutron scattering data in noble liquids. Chapter 3 shifts gears and focuses on a practical method to train molecular models given experimental scattering data using Bayesian inference. The idea is to use machine learned "surrogate" models to replace the molecular dynamics step in the model parameter training, effectively reducing the force field training time by over a million-fold. Chapter 3 is a reproduction of the paper, "Accelerated Bayesian inference for molecular simulations using local Gaussian process surrogate models", published in the Journal of Chemical Theory and Computation in March, 2024. Chapter 4 explores the questions (1) how accurate does scattering data need to be to learn underlying interatomic forces and (2) do random errors from scattering measurements corrupt the structure measurement to a degree that we can no longer recover interatomic force information? A question originally investigated by Verlet in 1968, this chapter demonstrates that the prior conclusion that it was not possible to quantify interatomic forces from scattering data may need to be overturned. Chapter 5: Conclusions and Future Work outlines key results and motifs from the thesis and expands on ongoing and future research on the horizon of Bayesian liquid state theory. 

Finally, four appendices are included for reference to specific topics and/or mathematical concepts addressed in the main chapters: (A) Quantum and Many-Body Effects from Neutron Scattering expands on ongoing work attempting to decipher what physical phenomenon induce specific features in SOPR potentials. (B) Statistical Mechanics of Liquid Phase Systems is a collection of personal notes starting with classical Hamiltonian mechanics through to statistical mechanics of the liquid state. These notes were compiled and prepared as a lecture series for the course, "Molecular Simulations for Chemical Engineers" taught in Fall, 2023 at the University of Utah and have an accompanying video lecture series on YouTube \hyperlink{https://www.youtube.com/playlist?list=PLmVLYa8E5SDff-uE7ehWBsX5sQzVRYiG6}{here}. The course was taught by the support of the Utah Teaching Assistantship program. (C) Principles of Bayesian Statistics outlines the basics of probability theory and Bayesian methods as implemented in the main chapters. Although this appendix is not suitable to replace a full text on the subject (I would recommend Gelman's \textit{Bayesian Analysis} \cite{gelman_bayesian_1995}), the hope was to briefly outline key terminology and results that were not expanded on in the main chapters. Finally, Appendix (D) Introduction to Functional Analysis is included as a reference to terminology and theorems relevant to understanding some of the most recent results in liquid state theory. Particularly relevant are a series of papers outlining the applicability, limitations, and practical variational implementations of the Henderson inverse theorem \cite{hanke_well-posedness_2018,frommer_note_2019,frommer_variational_2022} as well as providing a basis to understand Bayesian field theory \cite{lemm_bayesian_2003}. Although not directly used in the main chapters, I felt that it was important to compile and include these notes for future reference as the methods of functional analysis may offer exciting new ways of conceptualizing important problems in the field.

\section{Learning Interaction Potentials from Scattering Data}

Structure and interatomic forces are fundamentally linked. Although these relationships can be rigorously established for model fluids, the situation is less optimistic for real systems due to significant many-body interactions and quantum mechanical effects. However, fundamental results from statistical mechanics, such as the Ornstein-Zernike integral relations, inverse Kirkwood-Buff theory, and the Henderson inverse theorem (see Appendix B), hint at the possibility of determining unique and accurate interaction potentials to model thermodynamic and structural property relationships simultaneously. These force fields would be better suited to studying complex behaviors of liquids, such as self-assembly or vapor-liquid equilibrium, and possibly give insight into the nature of physical interactions.

We have developed an algorithm called structure-optimized potential refinement (SOPR) designed to extract interaction potentials from experimental scattering data. SOPR is a numerical method for the Henderson inverse theorem assisted by probabilistic machine learning to address challenges such as over-fitting to uncertain data and numerical instability. The probabilistic machine learning step involves Gaussian process regression with physics-guided priors of the interaction potential estimated with Henderson inverse theorem. SOPR has been applied to radial distribution functions on noble gases and found to be transferable to the prediction of vapor-liquid equilibria, demonstrating for the first that potentials derived directly from a scattering measurements can reproduce thermodynamic properties of real fluids. In this chapter, the original publication describing the foundational principles and SOPR method is reproduced from the Journal of Physical Chemistry Letters, 13 (49), 11512-11520 with permission of the publisher.

\subsection{Abstract}

Deriving transferable pair potentials from experimental neutron and X-ray scattering measurements has been a longstanding challenge in condensed matter physics. State-of-the-art scattering analysis techniques estimate real-space microstructure from reciprocal-space total scattering data by refining pair potentials to obtain agreement between simulated and experimental results. Prior attempts to apply these potentials with molecular simulations have revealed inaccurate predictions of thermodynamic fluid properties. In this letter, a machine learning assisted structure-inversion method applied to neutron scattering patterns of the noble gases (Ne, Ar, Kr, and Xe) is shown to recover transferable pair potentials that accurately reproduce both microstructure and vapor-liquid equilibria from the triple to critical point. Therefore, it is concluded that a single neutron scattering measurement is sufficient to predict macroscopic thermodynamic properties over a wide range of states and provide novel insight into local atomic forces in dense monoatomic systems. 

\subsection{Introduction}

Advances in neutron and X-ray scattering analysis have significantly furthered our understanding of self-assembly and dynamic transport properties in dense fluid systems \cite{glotzer_anisotropy_2007,wu_high-temperature_2017}. Scattering analysis is therefore an important and necessary component in the development and validation of atomistic force fields aimed at predicting both micro- and macroscopic thermodynamic properties over a wide range of states. However, strikingly contradictory predictions between experimental and simulated microstructures have been reported in relatively simple systems, including monoatomic liquid metals \cite{itami_structure_2003}, aromatic hydrocarbons \cite{fheaden_structures_2018}, and water \cite{amann-winkel_x-ray_2016}. Given the proliferation of accessible neutron and X-ray scattering instrumentation, advances in computational analysis, and development of machine learning approaches, it is relevant to revisit whether scattering data can improve force fields for fluid property predictions and provide insight into local atomic forces. 

One approach to benchmark force fields to scattering data is to calculate the underlying interatomic potentials from the experimental pair correlation functions, the so-called \textit{inverse problem} of statistical mechanics. A number of well-established inversion techniques have been proposed, including Ornstein-Zernike (OZ) integral relation methods \cite{ornstein_accidental_1914,percus_analysis_1958,percus_approximation_1962,wang_use_2010,agoodall_data-driven_2021}, Yvon-Born-Green (YBG) theories \cite{mullinax_generalized_2009,mullinax_generalized_2010,ellis_generalized-yvonborngreen_2011,dunn_bottom-up_2016,delyser_bottom-up_2020}, Schommer's algorithm \cite{schommers_pair_1983}, hypernetted chain methods \cite{levesque_pair_1985,delbary_generalized_2020}, the generalized Lyubartsev–Laaksonen approach \cite{lyubartsev_calculation_1995,toth_determination_2001,toth_iterative_2003}, empirical potential structure refinement (EPSR) \cite{soper_empirical_1996}, and a neural network \cite{toth_pair_2005}. However, there is little evidence that interatomic potentials obtained from these techniques can reliably predict thermodynamic behavior for real liquids. For example, Soper showed that O-O and H-H site-site interatomic potentials derived from EPSR applied to scattering data of liquid water predict a 4 times more negative excess internal energy compared to the experimental value \cite{soper_empirical_1996} and later concluded that EPSR cannot be used to derive a reliable set of site-site pair potentials for a given system \cite{soper_tests_2001}. A recent scattering study on supercritical krypton found rapid short-range oscillations in EPSR-derived interatomic potentials that led the authors to conclude that augmentation of the EPSR algorithm is required to obtain a more accurate representation of the real physical system \cite{pruteanu_krypton_2022}. Additionally, the remaining studies on structure-inversion of real liquids reported no validation of the interatomic potentials to predict fluid properties \cite{reatto_iterative_1986,kahl_inversion_1994,soper_partial_2005,youngs_dissolve:_2019,zhao_structural_2021}. Notably, in a review of structure-inversion methods it is opined that the general purpose of these techniques is not to derive or evaluate interatomic potentials, but rather to determine molecular configurations that are consistent with the scattering data \cite{toth_interactions_2007}. Therefore, it remains to be shown if scattering derived potentials can predict atomic trajectories consistent with experimental scattering measurements while also accurately modeling other thermodynamic properties.  

The atomic length scale probed by experimental scattering measurements also confers an additional opportunity, specifically whether it is possible to learn details of the local interactions independent of assumptions on a specific model potential form (e.g., 12-6 Lennard-Jones). For example, the rate of short-range repulsive decay indicates the propensity of an atom to deform in a collision, such that relaxation from an infinitely steep potential wall to a finite exponential or power-law decay represents the transition from hard- to soft-particle collision dynamics. The approximate collision diameter may be estimated by the radial position where the potential energy intersects zero, and the pairwise radial separation of zero force describes the effective dispersion energy. Provided structure-optimized potentials demonstrate the ability to predict emergent thermodynamic properties, their application provides a bridge between local atomic physics and continuum behavior. 

\subsection{Results and Discussion}

In this letter, force fields were determined for four noble liquids (Ne, Ar, Kr, Xe) using a machine-learning augmented Schommer's algorithm, referred to as structure-optimized potential refinement (SOPR), to refine pair potentials and obtain convergence between simulated and experimental pair distribution functions. Modifications to an initial reference potential are informed by numerical implementation of the point-wise Henderson's inverse theorem and augmented via Gaussian process regression with a squared-exponential kernel function described in Equations \eqref{refinement} and \eqref{eq:gp}, respectively. The structure-optimized potentials predict excellent representations of both the experimental pair distribution functions (Figure \ref{fig:rdfs}) and saturated vapor-liquid fluid properties. Consequently, structure-optimized potentials are validated using experimentally-consistent observations on both the micro- and macroscopic length scales, motivating the analysis of specific properties of the generated potentials. Additionally, the monoatomic structure and spherical symmetry of the noble gas system facilitates the comparison of the structure-optimized potentials (Figure \ref{fig:pots}) to reference \textit{ab initio} potentials obtained in the low-density state from coupled cluster theory \cite{hellmann_ab_2008,patkowski_argon_2010,jager_state---art_2016,hellmann_state---art_2017}, referred to as reference quantum dimer potentials. This comparison reveals state-dependent changes of many-body forces present in the experimental systems that were collected at states with varying reduced temperatures ($T_r$) relative to the critical point.

\begin{figure}
    \centering
    \includegraphics[width=14cm]{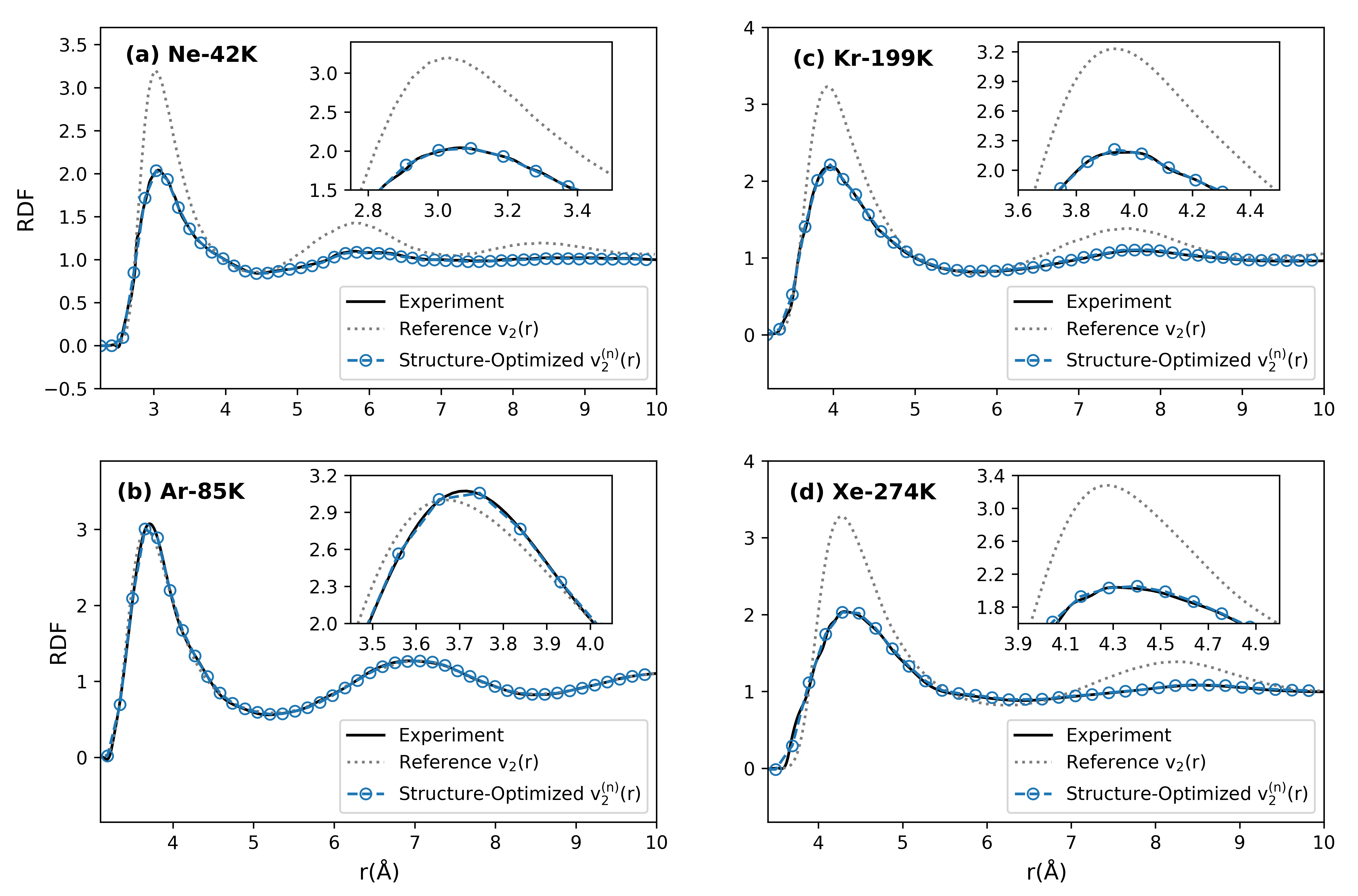}
    \caption{\label{fig:rdf} Reference potential radial distribution functions (grey dotted line) compared to the converged simulated radial distribution functions (blue circles) and experimental radial distribution functions (black line). Inset figures show the first solvation shell of the radial distribution function. }
    \label{fig:rdfs}
\end{figure}

The structure-optimized potentials collected for fluids near their critical point (Ne-42K, Kr-199K, Xe-274K at $T_r = 0.95$) exhibit softer repulsive decay, insignificant change to the collision diameter, and a substantial reduction in dispersion energy with respect to reference quantum dimer potentials. Thus, the ensemble averaged many-body behavior near the critical point results in softer particle collisions with decreased particle attraction. Near the triple point (Ar-85K at $T_r = 0.56$), structure-optimized potentials show no significant change in the repulsive exponent, a 1.5 $\%$ increase in collision diameter, and a reduction in dispersion energy (Figure \ref{fig:pots}(b)) compared to the quantum dimer potential. Many-body effects therefore had a negligible effect on the particle stiffness while decreasing particle attraction near the triple point. The observation that the dispersion energy correction was relatively smaller for the near triple point potential compared to the near critical point potentials suggests that the dispersion energy is a function of the thermodynamic state, which is discussed in the context of temperature-dependent many-body effects later. 

\begin{figure}
    \centering
    \includegraphics[width=15cm]{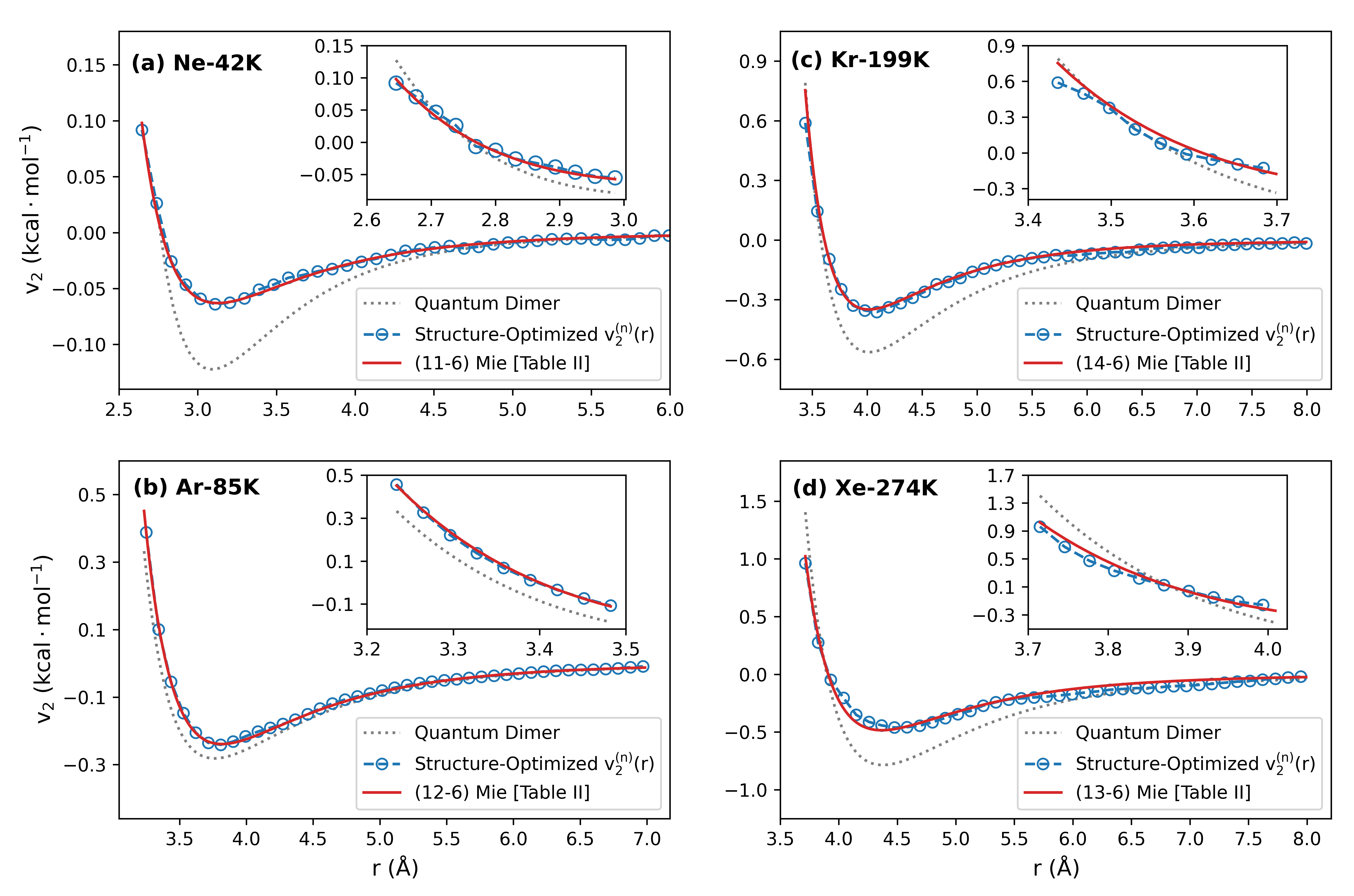}
    \caption{Tabulated structure-optimized potentials (blue) and $(\lambda-6)$ Mie potentials determined with Bayesian regression (red) are shown with reference quantum dimer potentials (grey). Inset figures show the short-range repulsive region of the corresponding structure-optimized potential.}
    \label{fig:pots}
\end{figure}

It is instructive to compare the structure-optimized potentials to widely employed transferable pair potential functions, such as the ($\lambda-6)$ Mie potential, 

\begin{equation}
    v^{Mie}_2(r) = \frac{\lambda}{\lambda-6}\bigg(\frac{\lambda}{6}\bigg)^{\frac{6}{\lambda-6}} \epsilon \bigg[ \bigg(\frac{\sigma}{r}\bigg)^\lambda - \bigg(\frac{\sigma}{r}\bigg)^6 \bigg]
\end{equation}

\noindent where $\lambda$ is the short-range repulsion exponent, $\sigma$ is the collision diameter (\AA), and $\epsilon$ is the dispersion energy (kcal $\cdot$ mol$^{-1}$) \cite{mie_zur_1903}. The ($\lambda-6)$ Mie potential offers increased flexibility over the standard $(12-6)$ Lennard-Jones potential since the repulsion exponent may be varied to produce a wider array of potential shapes. Structure-optimized potentials were fit to the ($\lambda-6)$ Mie function via Bayesian regression and plotted as red lines in Figure \ref{fig:pots}. Note that the excellent quality-of-fit of structure-optimized potentials to the ($\lambda-6)$ Mie function indicates that the fitted parameters (listed in Table \ref{tab:correct}) can closely approximate the thermodynamic predictions of the tabulated structure-optimized potentials.

\begin{table}
\centering
\caption{\label{tab:correct} Summary of ($\lambda-6)$ Mie potential parameters determined by Bayesian linear regression and modifications to the reference quantum dimer potentials in terms of the ($\sigma, \epsilon)$ parameterization. $\Delta \sigma$ and $\Delta \epsilon$ are shown as percent deviations from the \textit{ab initio} parameter values given in Table \ref{tab:refstate}.
}
\begin{tabular}{l c c c c c}
\hline
\textrm{Element}&
\textrm{$\lambda$}&
\textrm{$\sigma$ (\AA)}&
\textrm{$\epsilon$ (kcal/mol)}&
\textrm{$\Delta \sigma$ (\%)}&
\textrm{$\Delta \epsilon$ (\%)} \\
\hline
Ne  & 11 & 2.77 & 0.063 & 0.31  & -48.4\\
Ar  & 12 & 3.40 & 0.239 & 1.50  & -16.7\\
Kr  & 14 & 3.58 & 0.359 & -0.08 & -38.3\\
Xe  & 13 & 3.91 & 0.484 & 0.51 &  -40.3\\
\hline
\end{tabular}
\end{table}

Transferability of the potentials was assessed by performing vapor-liquid equilibrium (VLE) calculations from the triple to critical point using histogram-reweighting grand canonical Monte Carlo (GCMC) simulations in the GPU-Optimized Monte Carlo (GOMC) simulation package \cite{nejahi_gomc_2019} (see Supporting Information). Figure \ref{fig:vle} shows vapor and liquid densities for structure-optimized potentials fit to $(\lambda-6)$ Mie potentials (red triangles) compared with experimental data (black lines) compiled from the National Institute of Standards and Technology (NIST) \cite{lemmon_thermophysical_1998}. The Ne-42K structure-optimized force field predicts liquid densities within 0.1-2.5$\%$ relative error between 30-40K, on par with the top-performing Lennard-Jones force field from Vrabec \textit{et al.} \cite{vrabec_set_2001} and outperforming the next closest model \cite{brown_interatomic_1966} by as much as 10$\%$. The Ar-85K, Kr-199K, and Xe-274K force fields are less accurate, with liquid density relative errors of 0.2-5$\%$ (85-140K), 6.2-10.1$\%$ (120-180K) and 4.7-8.4$\%$ (190-260K), respectively. Simulated critical points determined with the Ising-type critical point scaling law \cite{sengers_thermodynamic_1986} and law of rectilinear diameters \cite{widom_new_1970} are provided in Table \ref{tab:vle}. 

\begin{table}
\centering
\caption{\label{tab:vle} Simulated critical temperatures ($T^{sim}_C$) and densities ($\rho^{sim}_C$) with statistical uncertainty calculated from 5 independent GCMC simulations. Percent error between simulated and experimental critical temperature and density are also shown.}
\begin{tabular}{l r r r r}
\hline
\textrm{Element}&
\textrm{$T^{sim}_C$ (K)}&
\textrm{$T^{err}_C$ (\%)}&
\textrm{$\rho^{sim}_C$ (kg/m$^{3}$)}&
\textrm{$\rho^{err}_C$ (\%)}\\ 
\hline
Ne & 43.84  $\pm$ 0.06 & -1.33 & 488.4  $\pm$ 1.08 & 0.91\\
Ar & 154.27 $\pm$ 0.16 &  2.40 & 528.9  $\pm$ 0.84 & -1.32\\
Kr & 216.58 $\pm$ 0.60 &  3.36 & 952.7  $\pm$ 3.62 & 5.04\\
Xe & 300.99 $\pm$ 0.28 &  3.98 & 1142.9 $\pm$ 1.76 & 4.85\\
\hline
\end{tabular}
\end{table}

\begin{figure}
    \centering
    \includegraphics[width=12cm]{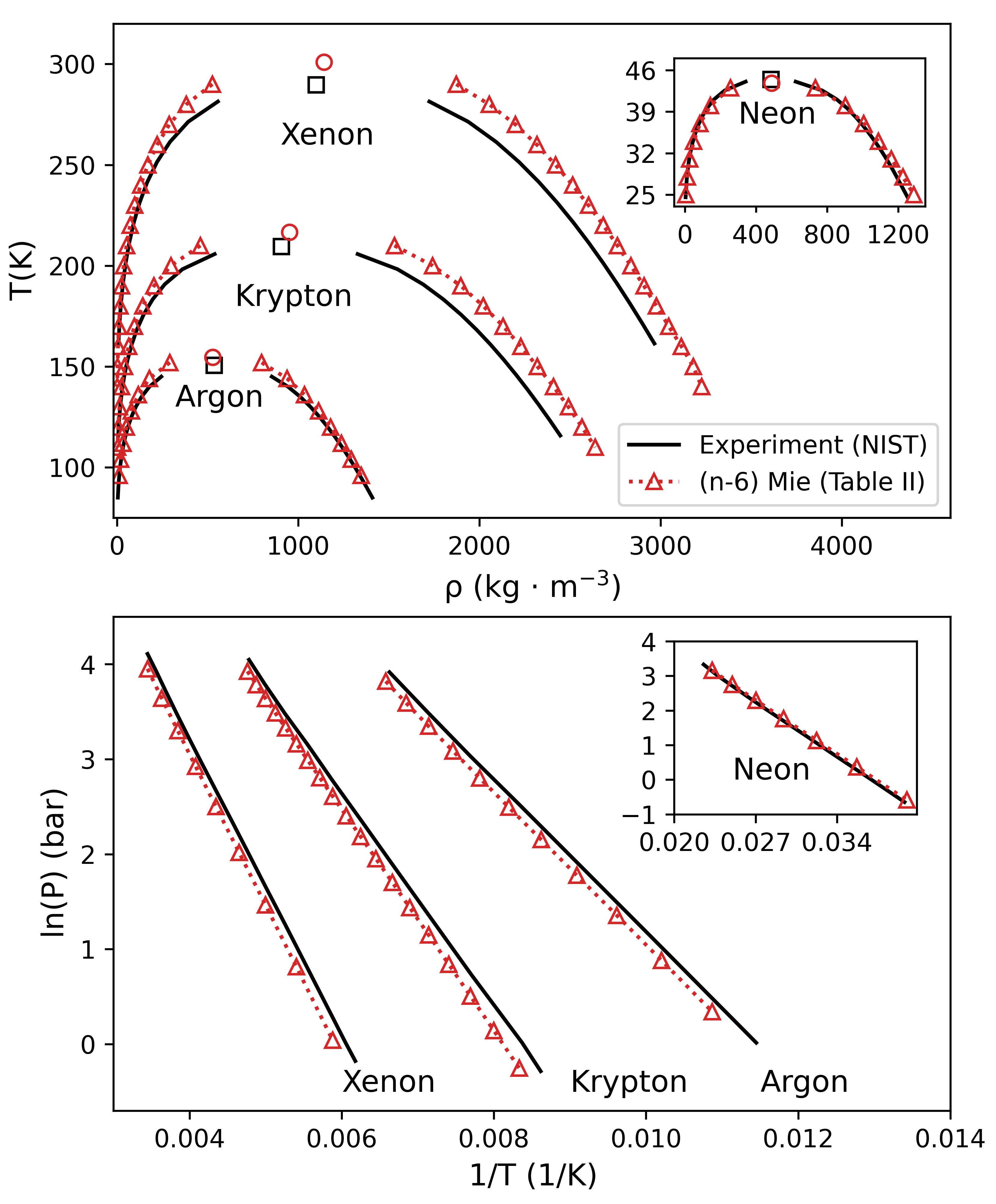}
    \caption{\label{fig:vle} (Top) Simulated phase coexistence curves (red triangles) are shown with experimental phase data (black lines). Simulated and experimental critical points are given by the open red circle and open black square, respectively. (Bottom) Clausius-Clapeyron plots of the simulated pressure (red triangles) compared to experimental pressures (black lines).}
\end{figure}

A recently developed series of $(\lambda-6)$ Mie force fields benchmarked to noble gas vapor-liquid equilibrium (VLE) provides an excellent comparison to the structure-optimized force fields proposed in this work. In general, both force fields predict similar repulsion exponents ($\lambda$) and dispersion energies ($\epsilon$). One interesting observation is that both the structure and VLE-optimized force fields predict an increase in the repulsive exponent with increasing atomic weight. Mick \textit{et al.} \cite{mick_optimized_2015} demonstrated that varying the repulsive exponent improved the simultaneous prediction of saturated densities and vapor pressures, supporting the conclusion that static structure is sensitive to subtle variations in the pair potential. The most consequential difference between the two models is that the structure-optimized force field gives systematically lower predictions for the collision diameter ($\sigma$), causing the simulated liquid density to be overestimated compared to the experimental value. It is notable that the reported differences in the collision diameter (0.01-0.05 \AA) are approximately one order of magnitude smaller than the real-space resolution of the experimental diffraction data (0.4-0.7 \AA). Modern, high-flux scattering instruments can achieve real-space resolution of approximately 0.05 \AA, on the same order as the error in the collision diameter, suggesting that more accurate potentials may be determined from repeating neutron scattering measurements on noble gases with modern instruments \cite{neuefeind_nanoscale_2012}.

The observation that structure-optimized potentials can estimate vapor-liquid coexistence behavior from a single neutron scattering measurement suggests that structure-inversion may be a promising approach to develop force fields for materials where experimental phase behavior is absent or impractical to obtain. The phase behavior and criticality of liquid uranium (U) is one of many important and unresolved examples relating to nuclear reactor design and safety analysis. Neutron diffraction data on solid $\beta$-U exists to temperatures as high as 1045K \cite{lawson_structure_1988} and X-ray diffraction patterns of $\gamma$-U close to the melting point (1405K) are well-characterized \cite{wilson_structures_1949}, but the phase coexistence line and critical point remain unknown, with critical temperature predictions ranging from 5000-13000K \cite{iosilevskiy_uranium_2005}. However, our results suggest that scattering measurements of liquid U may enable estimation of vapor-liquid phase coexistence via GCMC simulations with a structure-optimized embedded atom model or set of state dependent structure-optimized potentials.

Structure-inversion may also enable quantification of liquid state many-body effects. Quantum calculations on Kr-Ar-Ar and Ar-Ar-Kr trimers have revealed that noble gases experience two important many-body effects: (1) 3-body exchange repulsion and (2) the exchange/dispersion quadropole induced dipole \cite{elrod_many-body_1994}. The averaged pairwise influence of these many-body interactions decreases the short-range repulsion exponent and the dispersion energy \cite{guillot_theoretical_1988}, which is in agreement with the behavior observed in the structure-optimized potentials obtained near the critical point. Self-consistent field calculations of electron distributions in Ar clusters have shown that the electron cloud is compressed at higher densities, and that this compression reduces the probability of exchange repulsion \cite{lesar_ground-_1983}. Additionally, experimental results from collision-induced depolarized light scattering on compressed H$_2$ demonstrated that exchange-dependent many-body interactions become more prominent with increasing temperature \cite{bafile_third_1991}. It is therefore expected that many-body effects in noble gases should be less dominant near the triple point, explaining why the short-range decay rate for the Ar-85K structure-optimized potential was unchanged and the well-depth correction smaller than in the near critical point states. This conclusion is also supported by analyzing trends in the structural many-body correction for each fluid at near triple and critical point conditions (see Supporting Information). Further analysis with quantum mechanical and explicit 3-body dispersion models, such as the Axilrod-Teller potential \cite{axilrod_interaction_1943,sadus_exact_1998,marcelli_molecular_1999}, are reserved for future study on high-resolution scattering data sets obtained with state-of-the-art neutron techniques. 

We demonstrate that transferable pair potentials can be reconstructed from a single neutron scattering measurement for monoatomic liquids, and further; that structure-inversion techniques have fundamental and interdisciplinary applications bridging experimental scattering, molecular simulation, and quantum mechanics. Of particular interest is the prediction of thermodynamic properties at extreme conditions, such as high temperature and pressure materials, molten salts, and liquid metals. The inclusion of experimental diffraction results for optimizing effective pair potentials may also facilitate improvements to local structure predictions for fluid mixtures and molecular liquids. However, incoherent and inelastic scattering corrections \cite{soper_inelasticity_2009}, as well as non-uniqueness of the partial structure factor decomposition, will need to be addressed to extend the presented techniques to complex liquids. Finally, the methods presented in this letter may be applied to benchmark force fields for coarse-grained simulations, which has seen a growing interest in structure-inversion techniques \cite{rosenberger_comparison_2016,bernhardt_iterative_2021}.

\subsection{Theory and Computational Methods} \label{methods}

The following section provides the relevant definitions, statistical mechanics, and necessary computational details of the proposed machine learning assisted structure refinement method. First, the microstructure is considered as the local atomic density correlation and is formalized by counting the number of atomic neighbors as a function of position with respect to a reference atom and taking the ensemble average,

\begin{equation}
    g(\mathbf{r}) = \frac{1}{\rho} \bigg \langle \frac{1}{N} \sum_{i=1}^N \sum_{j=1}^N \delta^3(\mathbf{r} - \mathbf{r}_j + \mathbf{r}_i) \bigg \rangle
\end{equation}

\noindent where $g(\mathbf{r})$ is referred to as the radial distribution function, $\delta^3$ is the three-dimensional Dirac delta function, $\rho$ is the thermodynamic density and $N$ is the total number of particles in the system. The radial distribution function and pair correlation function, $h(\mathbf{r})$, are related by, $h(\mathbf{r}) = g(\mathbf{r}) - 1$. Due to the lack of long-range order in liquids, the isotropically-averaged radial distribution function is related to the experimentally observed structure factor, $S(Q)$, which for a monoatomic liquid is given by, 

\begin{equation}
    S(Q) = 1 + \frac{4 \pi}{Q \rho} \frac{\langle b \rangle^2 }{\langle |f(Q)|^2 \rangle} \int_0^\infty r [g(r) - 1] \sin (Qr) dr 
\end{equation} 

\noindent where $Q$ is the momentum transfer, $b$ is the scattering length density, $f(Q)$ is the form factor, and $\rho$ is the atomic number density \cite{sivia_elementary_2011}. 

The potential energy can be written as a sum of $n$-body potential energy terms such that,   

\begin{equation}
\label{exp}
    U(\mathbf{r}) = \overbrace{\sum_i^N v_1(\mathbf{r}_i)}^{\text{external field}} + \overbrace{\sum_{i = 1}^N \sum_{j \neq i}^N v_2(\mathbf{r}_{ij})}^{\text{two-body}} + \overbrace{\sum_{i}^N \sum_{j \neq i}^N \sum_{k \neq j}^N v_3(\mathbf{r}_{ijk}) + ...}^{\text{many-body}}
\end{equation}

\noindent where $v_p(\mathbf{r}_{1,...,p})$ is a position dependent function that assigns a potential energy to a subset containing $p \leq N$ atoms for a given configuration $\mathbf{r}_{i,...,p}$ \cite{allen_computer_2017}. We further simplify this expression by neglecting the external field contribution ($p=1$) and averaging higher-order many-body terms ($p \geq 3$) into a state-dependent pair term,

\begin{equation} \label{trunc}
\begin{split}
    U(\mathbf{r}; \rho,T) & =  \sum_{i = 1}^N \sum_{j \neq i}^N \bigg[ v_2(\mathbf{r}_{ij}) +  v_2^{m}(\mathbf{r}_{ij}; \rho, T) \bigg]
\end{split}
\end{equation}

\noindent such that $v_2^{m}(\mathbf{r}_{ij}; \rho, T)$ is explicitly dependent on the atomic positions and implicitly dependent on the physical state (temperature, density, etc). The bracketed quantity in Equation \eqref{trunc} is defined as the effective pair potential,

\begin{equation} \label{decomp}
    v_2^{eff}(\mathbf{r}_{ij}; \rho, T) = v_2(\mathbf{r}_{ij}) +  v_2^{m}(\mathbf{r}_{ij}; \rho, T)
\end{equation} 

\noindent which cannot be determined exactly for a state-dependent ensemble \cite{louis_beware_2002} but can be optimized to reproduce a set of experimentally observed thermodynamic properties, such as structure \cite{barocchi_neutron_1993}, heat of vaporization \cite{robertson_improved_2015}, or vapor-liquid equilibrium \cite{nejahi_gomc_2019}. The pair potential defined in Equation \eqref{decomp} is the most common non-bonded term in the Hamiltonian of classical force fields and is typically modeled as a hard-particle, Lennard-Jones, $(\lambda-6)$ Mie, Buckingham, or Yukawa potential. 

Pairwise additivity imposes important theoretical constraints on the relationship between the potential energy and pair correlation function. Henderson proved that for pairwise additive, constant density ensembles that there exists a one-to-one map between the effective pair potential and the radial distribution function up to an additive constant \cite{henderson_uniqueness_1974,zwicker_when_1990,frommer_note_2019}. In monoatomic liquids with spherical symmetry, the structure-potential uniqueness theorem on a finite radial interval $[r', r'']$ such that $r'' > r'$ and $r',r'' \in \mathbb{R}_0^+$, is equivalent to,

\begin{equation} \label{uniqueness}
    \int_{r'}^{r''} \Delta v_2(r)\Delta g(r) \, dr \leq 0 
\end{equation}

\noindent where $\Delta v(r)$ and $\Delta g(r)$ are the difference between a model (M) and target (T) pair potential and radial distribution function, respectively \cite{gray_theory_1990}.

\begin{equation}
  \begin{aligned}
  \Delta v_2(r) &= v_2^{M}(r) - v_2^{T}(r)  \\
  \Delta g(r) &= g^M(r) - g^T(r) 
  \end{aligned}
\end{equation}

\noindent Note that the structure-potential uniqueness theorem in the form of Equation \eqref{uniqueness} is written in terms of the $r$-coordinate only due to spherical symmetry. Initially, Equation \eqref{uniqueness} appears uninformative since the inequality prevents direct calculation of the target potential at any $r$ value. The situation is amenable under the assumption that the integrand is continuous and differentiable, so that Equation \eqref{uniqueness} can be rewritten using the mean value theorem,

\begin{equation}
   (r'' - r') \bigg \langle \Delta v_2(r)\Delta g(r) \bigg \rangle_{r'}^{r''} \leq 0 \label{average}  
\end{equation}

\noindent where the bracketed quantity represents the average of $\Delta v_2(r)\Delta g(r)$ over finite interval $[r', r'']$. Notice that for this inequality to be satisfied in the limit $(r'' - r') \to 0$, it must hold at any point $r_o \in r$ so that,

\begin{equation}
    \Delta v_2(r_o)\Delta g(r_o) \leq 0 \label{converge}
\end{equation}

\noindent which is true only when $\Delta v_2(r_o) \neq 0$ and $\Delta g(r_o) \neq 0$. The practicality of this point-wise structure-potential uniqueness theorem is now clear, since Equation \eqref{converge} prescribes what direction that an initial guess for the model potential should be corrected given the difference between the model and target experimental radial distribution function at any point $r_o \in r$; namely, by decreasing the potential if $\Delta g(r_o)$ is negative and increasing the potential if $\Delta g(r_o)$ is positive. While Henderson's structure-potential uniqueness condition has been implemented previously to obtain empirical estimates of pair potential functions in Schommer's algorithm and EPSR, this derivation demonstrates the validity of its use at an arbitrary point without the potential of mean force approximation $g(r) = \exp [-\beta v_2^{eff}(r_{ij}; \rho, T)]$ where $\beta = \frac{1}{k_BT}$, which only holds in the dilute limit \cite{hansen_theory_2013}. 

The structure-potential uniqueness condition is implemented via iterative refinement of a reference potential, $v_2^{0}(r_i)$, with an energy scaled, continuous sum of the radial distribution function error such that,

\begin{equation} \label{refinement}
    v_{2}^{(n)'}(r_i) = v_2^{0}(r_i) +  \gamma \beta^{-1} \sum_n \Delta g^{(n)'} (r_i)
\end{equation}

\noindent where $i$ is the radial index of the tabulated potential, $v_{2}^{(n)'}(r_i)$ is the predictor estimated pair potential at iteration $n$, $\beta$ is the inverse thermal energy $(k_B T)^{-1}$, and $0 < \gamma \leq 1$ is an empirical scaling constant to dampen the potential correction. Comparing the refinement Equation \eqref{refinement} to Equation \eqref{decomp}, it is clear that if $v_2^{0}(r_i)$ is selected as the quantum dimer potential that $v_{2}^{(n)'}(r_i)$ is the estimated effective pair potential and $\gamma \beta^{-1} \sum_n \Delta g^{(n)'} (r_i)$ is the pair averaged many-body term. Note that the prime notation in $v_{2}^{(n)'}(r_i)$ denotes that the pair potential is the predictor estimate before smoothing and treatment of numerical and experimental uncertainty.

In a standard Schommer's algorithm, the potential predicted by Equation \eqref{refinement} is passed to the next iteration without smoothing or uncertainty quantification, which has been shown to reduce the methods robustness \cite{levesque_pair_1985,toth_interactions_2007}. Here a squared-exponential kernel Gaussian process (GP) is applied to the predictor estimate to account for numerical fluctuations arising from the molecular dynamics simulations as well as systematic over-fitting to uncertain experimental data. A GP is a non-parametric stochastic process, equivalent to an infinitely wide neural network of a single layer, that generalizes the concept of probability distributions to functions \cite{rasmussen_gaussian_2006}. In this implementation, the GP takes the potential estimated by Equation \eqref{refinement} as an input and returns a Gaussian probability distribution of continuous and infinitely differentiable functions fitting the predictor estimate \cite{gelman_bayesian_1995}. Thus, a GP acts as an uncertainty propagator and smoothing function that, by nature of its Gaussian form, inherits an analytical Fourier transform that equivalently represents the data in real- or reciprocal-space without introducing significant truncation error \cite{ambrogioni_integral_2018}. Parallel techniques to enhance the accuracy of Fourier transforms in inverse problems, such as fitting structure factors to Poisson series expansions implemented in the EPSR and \textit{Dissolve} packages and T\'{o}th's Gauss-Newton parameterization and Golay–Savitzky smoothing \cite{toth_determination_2001,toth_iterative_2003}, can therefore be replaced with GP regression. Notably, SE-GP regression can be integrated into any existing iterative predictor-corrector without modifications to the base algorithm.

The predicted structure-optimized potential is then expressed as a $k$-multivariate normal distribution ($\mathcal{N}$) of random variables such that,

\begin{equation} \label{random}
    v^{(n)'}_{2}(r_1), v^{(n)'}_{2}(r_2), ... v^{(n)'}_{2}(r_k) \sim \mathcal{N}(\mu(\mathbf{r}), \mathbf{K}(r_i,r_j))
\end{equation}

\noindent where $r_1 < r_2 < ... < r_k$ are the radii positions for the potential, $k$ is the number of points in the structure-optimized potential, $\mu(\mathbf{r})$ is a mean function, and $\mathbf{K}(r_i,r_j)$ is a squared-exponential covariance function (or kernel) describing the relatedness of observations $v^{(n)'}_{2}(r_i)$ on $v^{(n)'}_{2}(r_j)$. Here the squared-exponential kernel is applied,    

\begin{equation}
    \mathbf{K}(r_i,r_j) = \bar{\sigma}^2e^{-\frac{(r_i - r_j)^2}{2\ell^2}} +\delta_{ij}\sigma_{noise}^2
\end{equation}

\noindent where $\bar{\sigma}^2$ is the expected variance of the interatomic potential, $\ell$ is the correlation length and $\sigma_{noise}^2$ is the variance due to numerical latent effects. Notice that if the distance between two points $r_i$ and $r_j$ is small, $\exp(r_i - r_j)^2/2\ell^2$ approaches unity and $v^{(n)'}_{2}(r_i)$ and $v^{(n)'}_{2}(r_j)$ are strongly correlated. As the distance between $r_i$ and $r_j$ increases, $\exp (r_i - r_j)^2/2\ell^2$ vanishes such that $v^{(n)'}_{2}(r_i)$ and $v^{(n)'}_{2}(r_j)$ are uncorrelated. The hyperparameters ($\bar{\sigma}$, $\ell$, $\sigma_{noise}$) are optimized by maximizing the marginal likelihood (model evidence) $p(v^{(n)'}_{2}(r) | r, \bar{\sigma}$, $\ell$, $\sigma_{noise})$.

Regression of $v_{2}^{(n)'}(r)$ over an arbitrary set of radii $r' = \{r_i'\} \ni r_1 \leq r'_1 < r'_2 < ... < r'_m \leq r_k$ is equal to the mean of the $k$-variate normal distribution,

\begin{equation}\label{eq:gp}
    v^{(n)}_{2}(r') = [\mathbf{K}^T_{r',r} -  \sigma_{noise}^2\mathbf{I}]\mathbf{K}_{r,r}^{-1}v^{(n)'}_{2}(r)
\end{equation}

\noindent where $v^{(n)}_2(r')$ is the final structure-optimized potential at iteration $n$ and $\mathbf{K}_{r',r}$ is the squared-exponential covariance matrix between coordinate representations $r'$ and $r$. Figure \ref{fig:gp} shows that GP regression smooths numerical artifacts in the interatomic force when the length scale hyperparameter is on the order of $\ell \sim 1$ \AA. A detailed comparison between a standard and GP assisted Schommer's algorithm is provided in the Supporting Information.

\begin{figure}
    \centering
    \includegraphics[width=14cm]{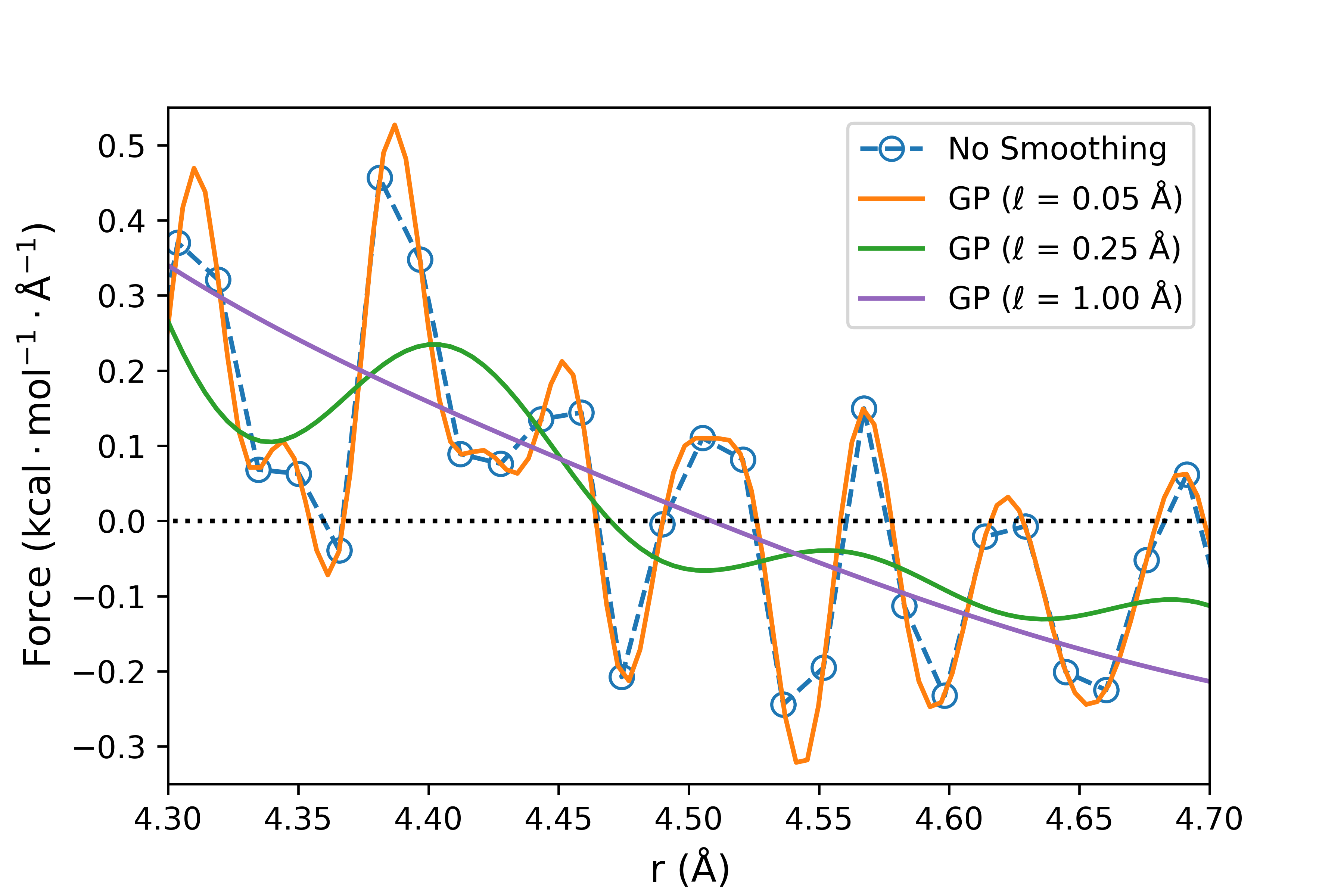}
    \caption{The force calculated from the Xe-274K tabulated potential (blue circles) compared to a non-parametric GP regressed potential at three different length-scale parameters ($\ell$ = 0.05, 0.25, 1.00 \AA) (orange, green, and purple line).}
    \label{fig:gp}
\end{figure}

The GP regressed structure-optimized potential is then applied in the molecular simulation corrector to calculate a simulated radial distribution function, $g^{(n)}(r')$. The molecular simulation corrector is a Canonical ($NVT$) molecular dynamics simulation performed in HOOMD-Blue \cite{anderson_hoomd-blue_2020}. MD simulations were initiated with a 500 atom fcc-lattice at the experimental density and equilibrated with Langevin dynamics for $5\times10^3$ timesteps ($dt=5$ femtoseconds). Tabulated structure-optimized potentials were truncated at $3 r_{vdW}$ with analytical tail corrections, and simulated radial distribution functions were calculated with MDAnalysis \cite{michaud-agrawal_mdanalysis_2011,gowers_mdanalysis_2016} from $1 \times 10^5$ timestep trajectories sampled at $100$ timestep intervals. Convergence is checked against the average squared error between the simulated and experimental radial distribution function such that $\langle [\Delta g^{(n)}(r')]^2 \rangle < 10^{-4}$, which generally is satisfied within 5-10 iterations at scaling constant $\gamma = 0.2$. 

\begin{table}
\centering
\caption{\label{tab:refstate}
Reduced temperature ($T_r = T/T_c$) and atomic density ($\rho$) are listed for the neutron scattering experimental conditions. Reference van der Waal radii ($r_{vdW}$) are used to define pair potential truncation in molecular dynamics simulations. $\sigma_{ai}$ is defined as the radius where the quantum dimer pair potential transitions from positive to negative potential energy and $\epsilon_{ai}$ is the potential minimum.
}
\begin{tabular}{l c c c c c}
\hline
\textrm{Element}&
\textrm{$T_r$}&
\textrm{$\rho$ (1/\AA$^{3}$)}&
\textrm{$r_{vdW}$ (\AA)}&
\textrm{$\sigma_{ai}$ (\AA)}&
\textrm{$\epsilon_{ai}$ (kcal/mol)}\\
\hline
Ne & 0.95 & 0.02477 & 2.91 & 2.76 & 0.122\\
Ar & 0.56 & 0.02125 & 3.55 & 3.35 & 0.287\\
Kr & 0.95 & 0.01187 & 3.82 & 3.58 & 0.582\\
Xe & 0.95 & 0.00881 & 4.08 & 3.89 & 0.811\\
\hline
\end{tabular}
\end{table}

Structure-inversion was initiated with a target experimental radial distribution function and a reference (or model) pair potential, $v_2^{0}(r_i)$. Experimental radial distribution data determined with elastic neutron scattering \cite{bellissent-funel_neutron_1992,yarnell_structure_1973,barocchi_neutron_1993} were compiled at the thermodynamic state conditions listed in Table \ref{tab:refstate}. Reference quantum dimer potentials were obtained via couple-cluster theory/t-aug-cc-pV6Z quality basis sets with spin-orbit relativistic corrections. In practice, any of the numerous existing pair potentials for the noble gases may be applied as a reference potential with equivalent outcomes for the structure-optimized potential (see Supporting Information). However, selecting the quantum dimer pair potential as the reference guarantees that the structure-optimized refinement correction is equal to the pairwise many-body contribution to the effective pair potential, $v_2^{m}(\mathbf{r}_{ij}; \rho, T) = \gamma \beta^{-1} \sum_n \Delta g^{(n)} (r'_i)$. 

\subsection{Supporting Information}

\subsubsection{Grand Canonical Monte Carlo Simulations}

Vapor-liquid coexistence curves and vapor pressures were determined from histogram-reweighting Monte Carlo simulations in the grand canonical ensemble\cite{ferrenberg_new_1988,ferrenberg_optimized_1989,panagiotopoulos_monte_1999}. Simulations were performed with GPU Optimized Monte Carlo (GOMC), version 2.70\cite{nejahi_gomc_2019}. All calculations were performed in a cubic cell with a side length of 25 Å.  Initial configurations were generated with Packmol \cite{martinez_packmol_2009}. Psfgen was used to generate coordinate (*.pdb) and connectivity (*.psf) files\cite{humphrey_vmd_1996}. The Mie potentials were truncated at 10 Å and analytical corrections were applied to the energy and pressure\cite{mick_optimized_2017}.  A hard inner cutoff was used to reject any MC moves that placed atom centers closer than 0.5 Å. A move ratio of 40\% displacements and 60\% molecule transfers was used.  Configurational-bias Monte Carlo (CBMC) was used to improve the acceptance rate for molecule transfers \cite{martin_novel_1999}. Three trial locations were used for simulations near the critical temperature, while 8 trial sites were used for the lowest temperature simulations (Tr=0.65). Acceptance rates for molecule insertions in liquid phase simulations were between 0.5\% and 6.2\%, depending on, chemical potential, and temperature.

To generate the phase diagrams predicted by each parameter set, 9 to 10 simulations were performed; one simulation to bridge the gas and liquid phases near the critical temperature, four in the gas phase, and 5 to 6 liquid simulations. For all noble gases, 2x106 Monte Carlo steps (MCS) were used for equilibration, followed by a data production period of 1.8x107 steps or 4.8x107 steps for gas and liquid phase simulations, respectively.  Histogram data were collected as samples of the number of molecules in the simulation cell and the non-bonded energy of the system. Samples were taken on an interval of 500 MCS.  Histograms from the GCMC simulations were reweighted and properties calculated as described by Messerly \cite{messerly_histogram-free_2019}. Averages and statistical uncertainties were determined from five independent sets of simulations, where each simulation was started with a different random number seed. Phase coexistence data for each noble gas is provided in \cref{tab:ne_phase,tab:ar_phase,tab:kr_phase,tab:xe_phase} and compared to existing force field models in Figure \ref{fig:rel_err}.

\subsubsection{Tabulated SOPR Results}

Tab-delimited text files (.txt) for experimental radial distribution functions \cite{yarnell_structure_1973,barocchi_neutron_1993,bellissent-funel_neutron_1992} and tabulated structure-optimized potentials obtained in this study are included as supplemental documents. Units for the radial distribution function are in angstrom (Å) in the radius and dimensionless in the g(r). Units for the provided structure-optimized potentials are angstrom (Å) in the radius and kcal/mol in the potential energy.

\begin{figure}
    \centering
    \includegraphics[width = 15cm]{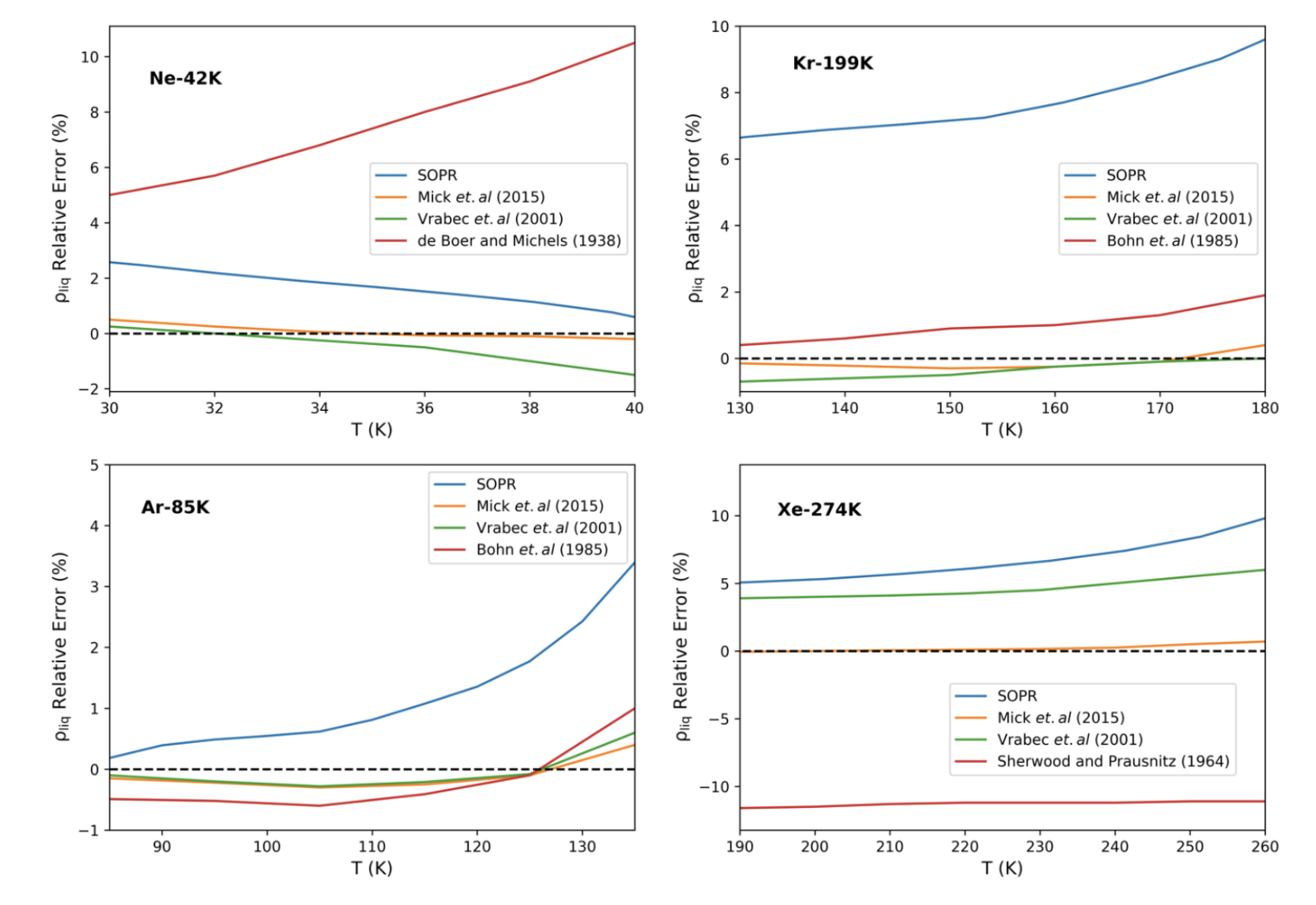}
    \caption{ Relative error in liquid density for the noble liquids. Relative error plots were generated with SOPR VLE data and compared to the reported error from other force fields.}
    \label{fig:rel_err}
\end{figure}

\begin{table}
\centering
\caption{Phase coexistence data from SOPR Ne-42K potential.}
\begin{tabular}{l c c c c c}
\hline
\textrm{T (K)}&
\textrm{$\rho_{liq}$ (kg/m$^3$)}&
\textrm{$\rho_{vap}$ (kg/m$^3$)}&
\textrm{P (bar)}&
\textrm{dHv (kJ/mol)}&
\textrm{Z} \\
\hline
43 & 733.714264 & 256.047536 & 23.57491 & 0.655222  & 0.519732\\
42 & 802.903123 & 211.968635 & 20.710539 & 0.809184 & 0.564664\\
41 & 858.831461 & 173.977971 & 18.089103 & 0.94601 & 0.615533  \\
40 & 902.873606 & 144.342911 & 15.718855 & 1.057575 & 0.660804 \\
39 & 940.575247 & 120.887088 & 13.581881 & 1.150318 & 0.699232\\
38&974.392979&101.635719&11.659854&1.230199&0.732772\\
37&1005.361735&85.477948&9.937776&1.300566&0.762674\\
36&1034.327748&71.748811&8.402114&1.363704&0.789545\\
35&1061.845082&59.979902&7.040359&1.421216&0.814005\\
34&1088.031234&49.834866&5.8411&1.473944&0.836737\\
33&1113.06391&41.095363&4.79364&1.522527&0.857957\\
32&1137.482411&33.601697&3.886811&1.567896&0.877381\\
31&1161.202543&27.203598&3.109184&1.609999&0.894878\\
30&1183.701132&21.764617&2.45&1.648361&0.910751\\
29&1205.099967&17.173856&1.898661&1.683608&0.925314\\
28&1225.73425&13.340107&1.444291&1.716293&0.938529\\
27&1246.346124&10.178756&1.075976&1.747296&0.950288\\
26&1267.238749&7.607599&0.782745&1.777414&0.960517\\
25&1287.667888&5.551135&0.554253&1.806016&0.969351\\
\hline
\end{tabular}
\label{tab:ne_phase}
\end{table}

\begin{table}
\centering
\caption{Phase coexistence data from SOPR Ar-85K potential.}
\begin{tabular}{l c c c c c}
\hline
\textrm{T (K)}&
\textrm{$\rho_{liq}$ (kg/m$^3$)}&
\textrm{$\rho_{vap}$ (kg/m$^3$)}&
\textrm{P (bar)}&
\textrm{dHv (kJ/mol)}&
\textrm{Z} \\
\hline
152&798.196968&291.533313&45.571786&2.243789&0.494141\\
150&836.656053&257.222877&42.268671&2.591589&0.526387\\
148&874.397214&226.119637&39.149769&2.934793&0.562104\\
146&908.89679&199.775981&36.212629&3.246104&0.596555\\
144&939.435052&177.960207&33.447597&3.515884&0.627142\\
142&966.572627&159.643406&30.843602&3.749191&0.65375\\
140&991.203415&143.860244&28.391126&3.955419&0.677329\\
138&1014.0377&129.965504&26.082236&4.142081&0.698752\\
136&1035.52168&117.564032&23.910432&4.313827&0.718554\\
134&1055.91477&106.402261&21.869793&4.473382&0.737011\\
132&1075.37469&96.299118&19.954885&4.62246&0.754289\\
130&1094.01521&87.114077&18.160762&4.762317&0.770526\\
128&1111.9393&78.733568&16.482801&4.894038&0.785862\\
126&1129.25415&71.065144&14.916567&5.01865&0.800437\\
124&1146.0727&64.035994&13.457917&5.137141&0.814362\\
122&1162.49057&57.588697&12.10266&5.250322&0.827692\\
120&1178.56105&51.677614&10.846588&5.358677&0.840417\\
118&1194.2921&46.262093&9.685515&5.462372&0.852513\\
116&1209.64317&41.304773&8.615326&5.561308&0.863972\\
114&1224.53168&36.770416&7.632034&5.65525&0.874829\\
112&1238.91774&32.627499&6.731743&5.744209&0.885141\\
110&1252.88688&28.848214&5.910461&5.828782&0.894943\\
108&1266.59482&25.408149&5.164272&5.909884&0.904268\\
106&1280.1958&22.28512&4.48897&5.988369&0.913082\\
104&1293.80881&19.457582&3.880413&6.064929&0.921381\\
102&1307.41921&16.905474&3.334503&6.139582&0.929153\\
100&1320.88493&14.609372&2.84726&6.211769&0.936441\\
98&1334.0903&12.551765&2.414737&6.281104&0.943246\\
96&1347.00602&10.715882&2.033112&6.347769&0.949619\\
94&1359.65999&9.086014&1.698543&6.412183&0.955572\\
92&1372.01604&7.647196&1.407326&6.474152&0.961156\\
\hline
\end{tabular}
\label{tab:ar_phase}
\end{table}

\begin{table}
\centering
\caption{Phase coexistence data from SOPR Kr-199K potential.}
\begin{tabular}{l c c c c c}
\hline
\textrm{T (K)}&
\textrm{$\rho_{liq}$ (kg/m$^3$)}&
\textrm{$\rho_{vap}$ (kg/m$^3$)}&
\textrm{P (bar)}&
\textrm{dHv (kJ/mol)}&
\textrm{Z} \\
\hline
210&1531.11817&459.974951&50.626437&3.765603&0.528218\\
205&1645.79806&366.381497&43.961845&4.596145&0.589887\\
200&1742.68564&297.413939&37.995176&5.276635&0.64375\\
195&1825.46461&245.594213&32.649925&5.824705&0.687084\\
190&1897.63413&204.286876&27.866868&6.283993&0.723559\\
185&1961.92976&170.227413&23.603813&6.681208&0.75537\\
180&2021.7309&141.580536&19.823791&7.036668&0.78395\\
175&2078.16333&117.248277&16.492711&7.359511&0.810074\\
170&2131.4784&96.538731&13.578873&7.653746&0.833855\\
165&2181.75136&78.900464&11.051246&7.921878&0.855509\\
160&2228.9557&63.882448&8.880364&8.166045&0.8756\\
155&2274.54894&51.154757&7.03666&8.393506&0.894387\\
150&2320.48457&40.456479&5.488944&8.611949&0.911566\\
145&2366.20448&31.546962&4.206358&8.819254&0.926747\\
140&2409.88941&24.20158&3.159392&9.010171&0.93975\\
135&2451.57431&18.215968&2.319986&9.186547&0.950773\\
130&2491.53902&13.40713&1.660498&9.350053&0.960133\\
125&2529.71653&9.614336&1.154729&9.501783&0.968318\\
120&2566.36151&6.690208&0.777101&9.643374&0.975484\\
115&2603.29991&4.497391&0.503879&9.781788&0.981806\\
110&2639.67485&2.90476&0.313072&9.915437&0.9874\\
\hline
\end{tabular}
\label{tab:kr_phase}
\end{table}

\begin{table}
\centering
\caption{Phase coexistence data from SOPR Xe-274K potential.}
\begin{tabular}{l c c c c c}
\hline
\textrm{T (K)}&
\textrm{$\rho_{liq}$ (kg/m$^3$)}&
\textrm{$\rho_{vap}$ (kg/m$^3$)}&
\textrm{P (bar)}&
\textrm{dHv (kJ/mol)}&
\textrm{Z} \\
\hline
290&1872.4037&526.262111&52.004613&5.440892&0.538123\\
285&1967.09269&447.677473&46.976381&6.239123&0.581441\\
280&2054.03333&382.799284&42.336581&6.963763&0.623764\\
275&2131.14867&330.772546&38.056685&7.58385&0.660695\\
270&2199.59475&288.319254&34.106135&8.11039&0.691874\\
265&2261.17401&252.55593&30.461055&8.56687&0.718742\\
260&2317.52979&221.658863&27.102719&8.97288&0.742652\\
255&2370.16211&194.549802&24.015796&9.342204&0.764464\\
250&2420.26327&170.55061&21.186091&9.684&0.784675\\
245&2468.44759&149.194894&18.600206&10.003237&0.803583\\
240&2514.7754&130.143936&16.245757&10.301796&0.821372\\
235&2559.24147&113.137502&14.110993&10.581039&0.838146\\
230&2602.16918&97.96734&12.183996&10.843475&0.853924\\
225&2643.96701&84.455192&10.452821&11.091671&0.86869\\
220&2684.68136&72.445131&8.905584&11.326431&0.882415\\
215&2724.1085&61.796925&7.530531&11.547532&0.895084\\
210&2762.2453&52.38534&6.31611&11.755633&0.906699\\
205&2799.37292&44.095262&5.25085&11.952533&0.917324\\
200&2835.94066&36.824133&4.323573&12.140646&0.927067\\
195&2872.3748&30.47986&3.522885&12.322267&0.935986\\
190&2908.51521&24.978722&2.837801&12.497539&0.944192\\
185&2943.61226&20.248163&2.257546&12.663904&0.951618\\
180&2977.54768&16.217716&1.771668&12.821135&0.958263\\
175&3011.05962&12.81936&1.369788&12.972135&0.964041\\
170&3045.27725&9.987163&1.041737&13.121276&0.968718\\
165&3080.69488&7.655546&0.777787&13.270025&0.972114\\
160&3115.40577&5.762415&0.568922&13.411029&0.974169\\
155&3148.6791&4.250562&0.40665&13.544108&0.974402\\
150&3180.41128&3.065916&0.283242&13.67061&0.97227\\
145&3207.49035&2.156153&0.191456&13.774698&0.966664\\
140&3227.16642&1.474877&0.125163&13.840011&0.956789\\
\hline
\end{tabular}
\label{tab:xe_phase}
\end{table}

\subsubsection{Convergence Stability}

In general, the reference potential can impact the stability of the molecular simulation as well as the interpretation of the structural correction term. In this study, it was possible to directly equate the structure refinement term and the ensemble averaged many-body term since an accurate quantum dimer reference potential was applied. However, if one requires only a structure-optimized potential and not a quantification of the many-body effects, the reference potential can be arbitrarily selected if it is stable within the molecular simulation. For example, Figure \ref{fig:robustness} shows Xe-274K structure- optimized potentials given five different reference potential conditions; namely, LJ parameters for Ne, Ar, Kr, Xe, and Ra \cite{gopal_intermolecular_1962,rutkai_how_2017}. Clearly, the predicted structure-optimized potential is independent of the reference potential in this system, although in principle this may not hold in complex liquids or for thermodynamic states near the amorphous glass transition where the structure does not explicitly depend on the non-bonded potential energy (e.g. as T$\to$ 0).

In addition to the reference potential, the scaling coefficient defined in Equation \eqref{refinement} can also impact the convergence rate and stability. For example, structure inversion runs for the Xe-274K system at varying scaling coefficient demonstrates that a scaling coefficient of 0.6 is ideal for rapid convergence and low relative deviation from the experimental structure at high iteration number (Figure \ref{fig:convergence}). Therefore, scaling coefficients can significantly impact computational performance and should be optimized for physical systems where molecular simulation is computationally expensive (e.g. high molecular weight liquids).

\begin{figure}
    \centering
    \includegraphics[width = 15cm]{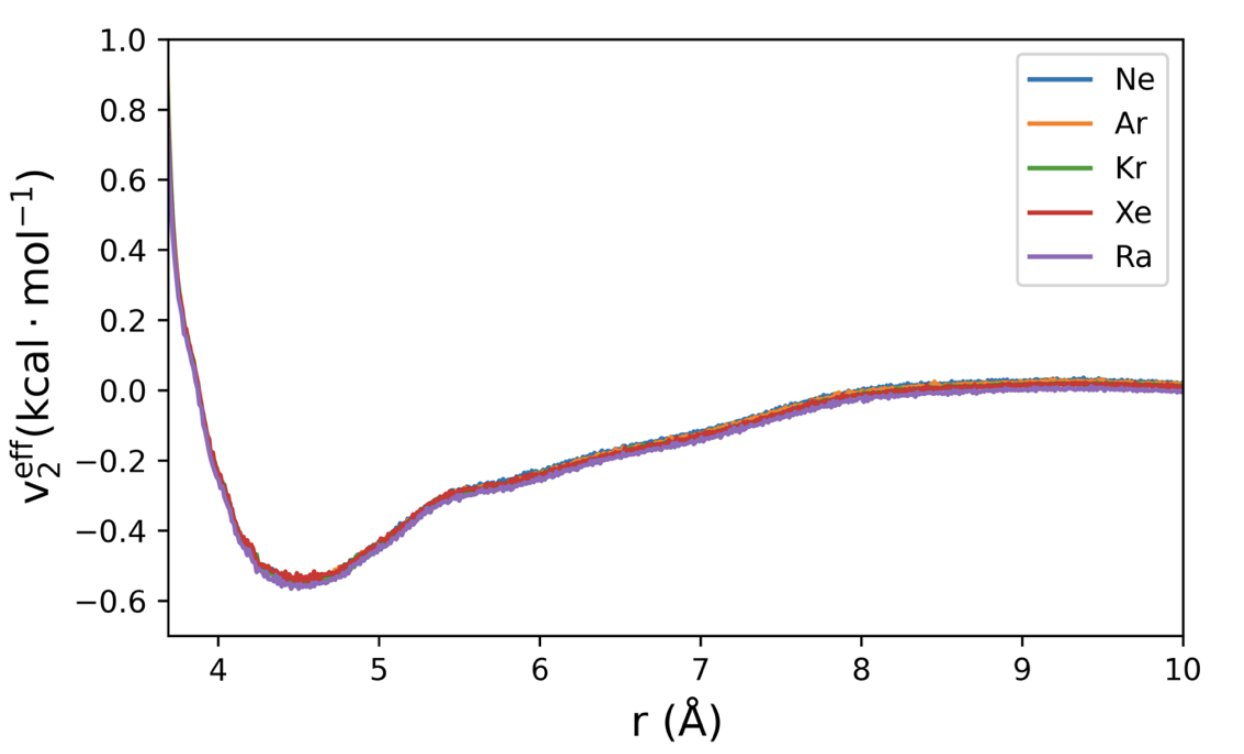}
    \caption{Structure-optimized potentials are unique up to negligible numerical fluctuations for the Xe-274K scattering data given a broad range of reference potentials (LJ parameters for Ne, Ar, Kr, Xe, and Ra).}
    \label{fig:robustness}
\end{figure}

\begin{figure}
    \centering
    \includegraphics[width = 15cm]{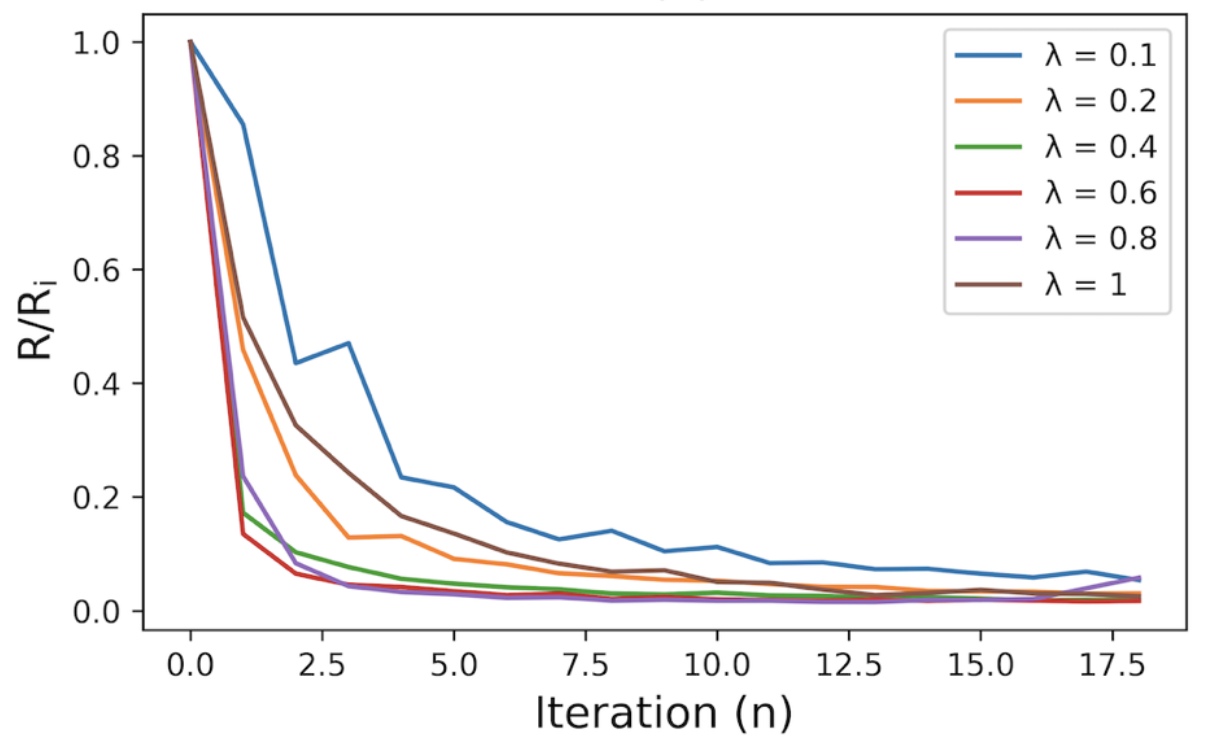}
    \caption{Normalized relative error in the simulated and experimental radial distribution function, defined as the sum-of-square residual (ssr) at iteration n divided by the initial ssr, R/Ri.}
    \label{fig:convergence}
\end{figure}

\subsubsection{Gaussian Process Regression}

Gaussian process (GP) regression is a powerful and robust method to perform non-parametric regression of complex functions. One of the main benefits of GP regression is that it has an assumption-free, analytical Fourier transform (FT) that is more accurate than a discrete Fourier transform (DFT) in noise-free and noise-corrupted signals. In fact, it has been shown that conventional DFT is a special case of the more general Bayesian formulation of Fourier transforms that is only possible with a GP regression \cite{ambrogioni_integral_2018}. With respect to scattering analysis, techniques such as fitting to Poisson series expansions in EPSR and Dissolve simulation packages \cite{soper_empirical_2017,youngs_dissolve:_2019} were implemented since they had analytical FTs that improve the convergence and accuracy of the inverse algorithm. Additionally, GP regression probabilistically smooths the predicted potential, which was shown in the main text to eliminate numerical fluctuations in the interatomic force that can propagate between refinement iterations. Figure \ref{fig:gpsupp} shows that structure-optimized potentials in a standard Schommer’s algorithm exhibit sub-angstrom fluctuations near equilibrium separation that are likely caused by the propagation of numerical/experimental noise introduced in the refinement procedure. However, the GP assisted Schommer’s algorithm smooths the potentials so that they adhere more to transferable potential forms while still providing an excellent quality-of- fit to the radial distribution function.

\begin{figure}[H]
    \centering
    \includegraphics[width = 15cm]{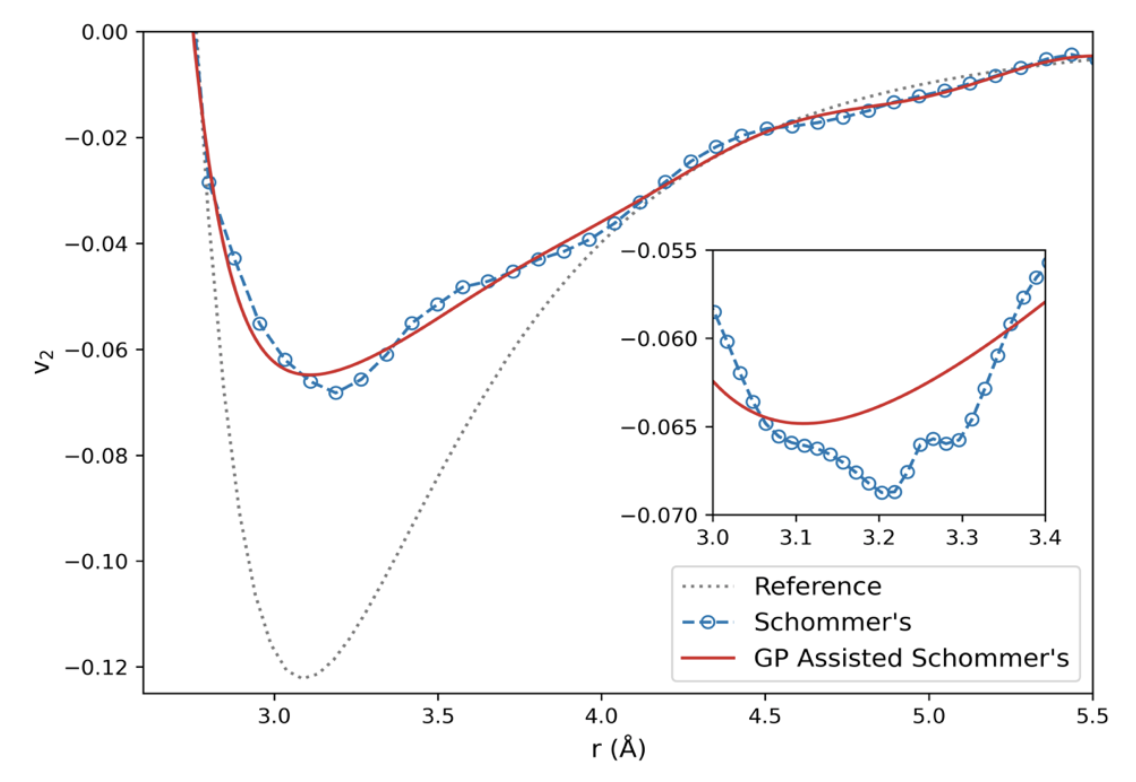}
    \caption{Structure-optimized potentials with and without GP regression.}
    \label{fig:gpsupp}
\end{figure}

\subsubsection{Triple and Critical Point Structural Many-Body Corrections}

SOPR potentials were generated at near triple and critical point conditions for each fluid where scattering data was available (Ne, Kr, and Xe). For each fluid, the dispersion energy decreased with increasing temperature (Figure \ref{fig:manybodysopr}) consistent with the expected increase in exchange repulsion effects, and the repulsive exponent decreased with increasing temperature (Figure \ref{fig:nearmanybodysopr}). The change in the short-range repulsion rate between the near triple and critical point states was found to be negligible in Ne, likely because the difference in absolute temperature is smaller (16 K) compared to the difference in absolute temperature for Kr (69 K) and Xe (105 K). Additionally, inelastic and incoherent scattering corrections have been shown to broaden the radial distribution function in hard spheres \cite{soper_inelasticity_2009} resulting in a decrease in both the height and slope of the first solvation shell. Furthermore, post- processing and low instrument resolution for the available scattering data propagate non-negligible error to the determined structure factors which is further confounded by discrete Fourier transform truncation error \cite{soper_partial_2005}. While these realities motivate the use of uncertainty quantification in structure inversion, it is likely that new scattering measurements with state-of-the-art neutron instruments must be obtained to draw significant, quantitative conclusions on the many-body corrections.

\subsection{Python Notebook Tutorial}

In addition to the manuscript, an example code for running SOPR in python was created and posted on GitHub \hyperlink{https://github.com/hoepfnergroup/SOPR}{here}. The notebook contains a detailed description of the code functions, basic theory, and results of the method. The example is for liquid neon which is a traditionally challenging liquid to model with molecular dynamics due to its low temperature and significant quantum mechanical effects. The notebook provides relevant background on iterative refinement algorithms, Gaussian process regression, simple liquid molecular simulations in HOOMD-Blue, and visualization of the SOPR results.

\clearpage

\begin{figure}[H]
    \centering
    \includegraphics[width = 12cm]{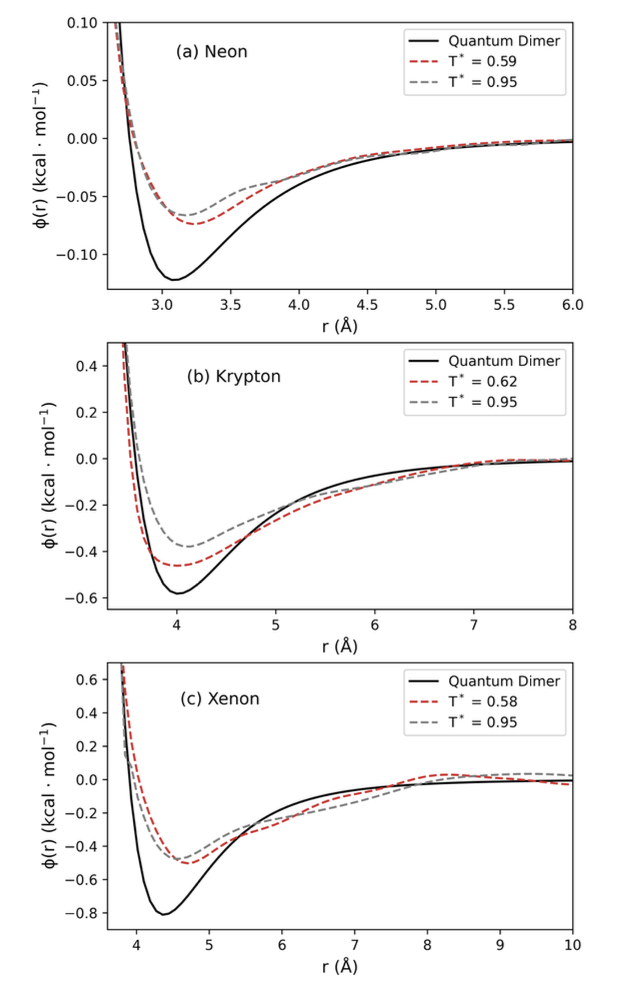}
    \caption{SOPR potentials for near triple and critical point experimental scattering data.}
    \label{fig:manybodysopr}
\end{figure}

\begin{figure}[H]
    \centering
    \includegraphics[width = 12cm]{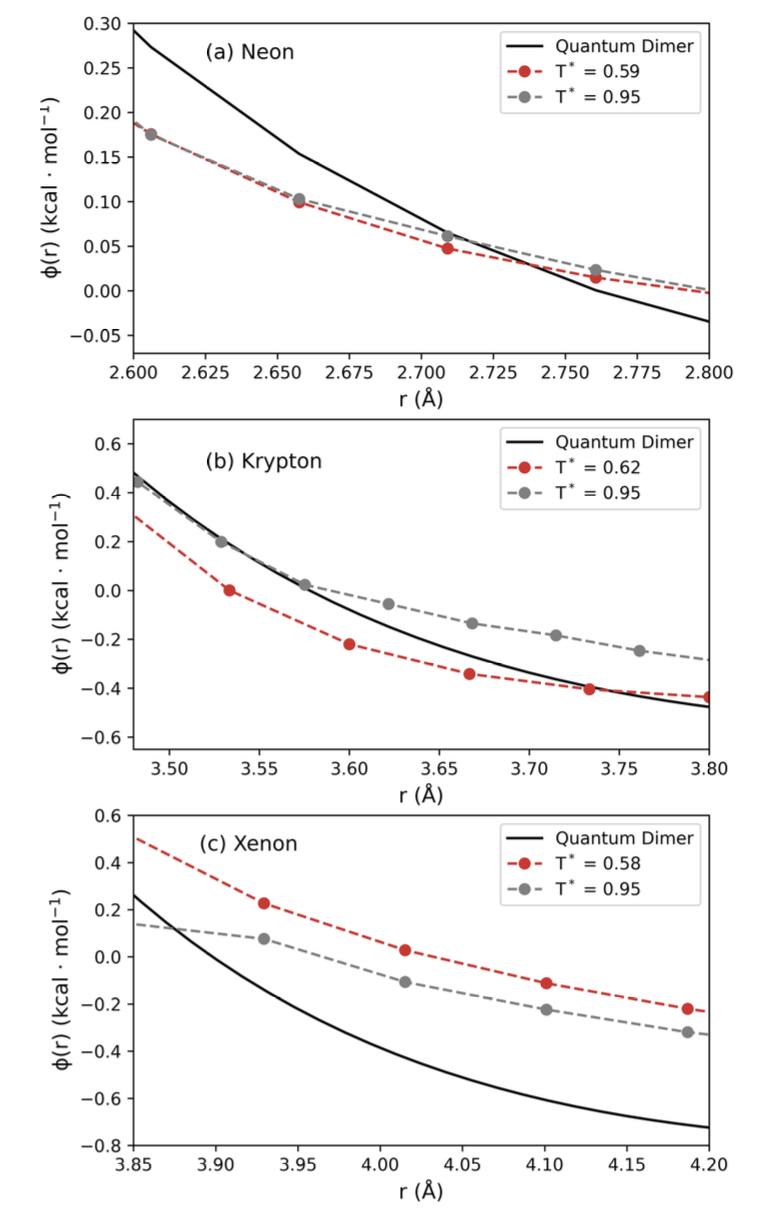}
    \caption{Repulsive region of the SOPR potentials for Ne, Kr, and Xe.}
    \label{fig:nearmanybodysopr}
\end{figure}

\section{Accelerated Bayesian Force Field Uncertainty Quantification for Structural Correlations}

The previous chapter explored the estimation of pair interactions from experimental neutron scattering data using Gaussian process (GP) priors conditioned on iterative potential refinement guided by the Henderson inverse theorem. GPs were shown to provide \textit{non-parametric} interaction potentials with uncertainty quantification, enabling the construction of flexible yet physically-justified pair potentials. However, there are instances where it is desirable to reproduce scattering data as accurately as possible given a \textit{parametric} potential form. For instance, parametric models are more widely adopted, easily implemented in existing software packages, and are faster to compute than non-parametric potentials. In this case, Bayesian inference over the model parameters can be performed with neutron data as the target quantity-of-interest (QoI). 

The main challenge with this approach is not in its construction, but rather the exceptional computational cost of Bayesian inference due to the curse of dimensionality. The main pinch point is that molecular simulations take a long time to perform (ranging from minutes to days) and Bayesian inference necessitates performing thousands to hundreds of thousands of such calculations. In practice, we can alleviate this problem by training machine learning \textit{surrogate} (or meta-) models on a smaller subset of simulations and then use this surrogate model in place of the more expensive molecular simulation. Some commonly used methods for designing surrogate models include polynomial chaos expansions, neural networks, and GPs. 

The built-in uncertainty quantification of GP surrogate models makes them particularly attractive for molecular systems, particularly when the amount of training data is small. GPs excel with little training data because a well-posed prior and kernel function can guide the prediction and enforce physically-justified behavior \textit{a priori}. Of course, these advantages come with a drawback, namely that the GP evaluation is slow compared to other methods. The content of this chapter is to describe a "greedy" approximation to a GP that provides highly accurate estimations of the molecular model prediction at a fraction of the computational time. In summary, the approximation involves splitting a GP along its independent variable inputs into a subset of GPs, resulting in a reduction in call time-complexity from cubic to linear in the number of independent variables (turning a standard GP into a speedy race car shown in Figure \ref{fig:racer}). The original publication describing a reliable and robust method for uncertainty quantification and propagation to complex experimental data (\textit{i.e.} radial distribution functions or electromagnetic spectra) using Bayesian inference is reproduced from the Journal of Chemical Theory and Computation, 2024, 20, 9, 3798–3808 with permission of the publisher. The article is also available on arXiv with article identifier arXiv:2310.19108 \cite{shanks_accelerated_2023}.

\begin{figure}
    \centering
    \includegraphics{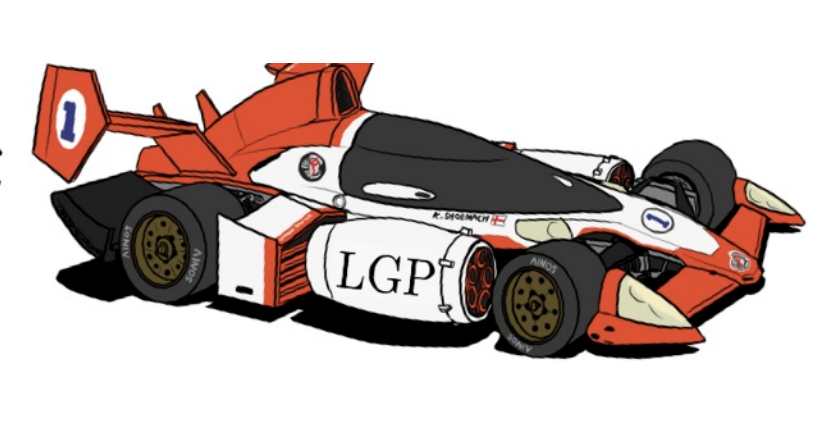}
    \caption{Local Gaussian processes are like a Formula 1 racer in comparison to standard GPs.}
    \label{fig:racer}
\end{figure}

\subsection{Abstract}

While Bayesian inference is the gold standard for uncertainty quantification and propagation, its use within physical chemistry encounters formidable computational barriers. These bottlenecks are magnified for modeling data with many independent variables, such as X-ray/neutron scattering patterns and electromagnetic spectra. To address this challenge, we employ local Gaussian process (LGP) surrogate models to accelerate Bayesian optimization over these complex thermophysical properties. The time-complexity of the LGPs scales linearly in the number of independent variables, in stark contrast to the computationally expensive cubic scaling of conventional Gaussian processes. To illustrate the method, we trained a LGP surrogate model on the radial distribution function of liquid neon and observed a 1,760,000-fold speed-up compared to molecular dynamics simulation, beating a conventional GP by three orders-of-magnitude. We conclude that LGPs are robust and efficient surrogate models, poised to expand the application of Bayesian inference in molecular simulations to a broad spectrum of experimental data.

\subsection{Introduction}

Molecular simulations are able to estimate a broad array of complex experimental observables, including scattering patterns from neutron and X-ray sources and spectra from near-infrared \cite{czarnecki_advances_2015}, terahertz \cite{schmuttenmaer_exploring_2004}, sum frequency generation \cite{nihonyanagi_ultrafast_2017,hosseinpour_structure_2020}, and nuclear magnetic resonance \cite{mishkovsky_principles_2009}. Recent interest in these experiments to study hydrogen bonding networks of water at interfaces \cite{roget_water_2022,li_hydrogen_2022}, electrolyte solutions \cite{wang_hydrogen-bond_2022}, and biological systems \cite{meng_modeling_2023} has motivated the continued advancement of simulations to calculate these properties from first-principles \cite{bally_quantum-chemical_2011,thomas_computing_2013,gastegger_machine_2017}. However, the ability to estimate these complex properties comes with a high computational cost. This barrier greatly limits our ability to quantify how experimental, model, and parametric uncertainty impact molecular simulation predictions, making it difficult to know whether a model is an appropriate representation of nature or if it is simply over-fitting to a given training set. Therefore, what is needed is a computationally efficient and rigorous uncertainty quantification/propagation (UQ/P) method to link molecular models to large and complex experimental datasets.

Bayesian methods are the gold standard for these aims \cite{psaros_uncertainty_2023}, with examples spanning from neutrino and dark matter detection \cite{eller_flexible_2023}, materials discovery and characterization \cite{todorovic_bayesian_2019,zuo_accelerating_2021,fang_exploring_2023,sharmapriyadarshini_pal_2024}, quantum dynamics \cite{vkrems_bayesian_2019,deng_bayesian_2020}, to molecular simulation \cite{frederiksen_bayesian_2004,cooke_statistical_2008,cailliez_statistical_2011,farrell_bayesian_2015,wu_hierarchical_2016,patrone_uncertainty_2016,messerly_uncertainty_2017,dutta_bayesian_2018,wen_uncertainty_2020,bisbo_efficient_2020,xie_uncertainty-aware_2023}. The Bayesian probabilistic framework is a rigorous, systematic approach to quantify probability distribution functions on model parameters and credibility intervals on model predictions, enabling robust and reliable parameter optimization and model selection \cite{gelman_bayesian_1995,shahriari_taking_2016}. Interest in Bayesian methods and uncertainty quantification for molecular simulation has surged \cite{musil_fast_2019,cailliez_bayesian_2020,vandermause_--fly_2020,kofinger_empirical_2021,vandermause_active_2022,blasius_uncertainty_2022} due to its flexible and reliable estimation of uncertainty, ability to identify weaknesses or missing physics in molecular models, and systematically quantify the credibility of simulation predictions. Additionally, standard inverse methods including relative entropy minimization, iterative Boltzmann inversion, and force matching have been shown to be approximations to a more general Bayesian field theory \cite{lemm_bayesian_2003}.

The biggest problem plaguing Bayesian inference is its massive computational cost. The two major pinch points are (1) sampling in high-dimensional spaces, commonly known as the "curse of dimensionality", and (2) the large number of model evaluations required to get accurate uncertainty estimates. In computational chemistry, these bottlenecks are magnified since these models are typically expensive. Therefore, rigorous and accurate uncertainty estimation is challenging, or even impossible, without accelerating the simulation prediction time. One way to achieve this speed-up is by approximating simulation outputs with an inexpensive machine learning model. These so-called surrogate models have been developed from neural networks \cite{wen_uncertainty_2020,li_rapid_2023}, polynomial chaos expansions \cite{ghanem_stochastic_2003,jacobson_how_2014}, configuration-sampling-based methods \cite{messerly_configuration-sampling-based_2018} and Gaussian processes \cite{rasmussen_gaussian_2006,nguyen-tuong_model_2009,burn_creating_2020}. 

Gaussian processes (GPs) are a compelling choice as surrogate models thanks to several distinct advantages. GPs are non-parametric, kernel-based function approximators that can interpolate function values in high-dimensional input spaces. GPs with an appropriately selected kernel also have analytical derivatives and Fourier transforms, making them well-suited for physical quantities such as potential energy surfaces \cite{dai_interpolation_2020,yang_local_2023}. Additionally, kernels can encode physics-informed prior knowledge, alleviating the "black box" nature inherent to many machine learning algorithms. In fact, a comparison of various nonlinear regressors for molecular representations of ground-state electronic properties in organic molecules demonstrated that kernel regressors drastically outperformed other techniques, including convolutional graph neural networks \cite{faber_prediction_2017}. 

Perhaps the most widely adopted application of GP surrogate models in computational chemistry is for model optimization. In the last decade, GP surrogates of simple thermophysical properties including density, heat of vaporization, enthalpy, diffusivity and pressure have been used for force field design 
\cite{angelikopoulos_bayesian_2012,cailliez_calibration_2014,kulakova_data_2017,befort_machine_2021,cmadin_using_2023,wang_machine_2023}. However, to our knowledge there are no Bayesian optimization studies that apply GP surrogate models to thermophysical properties with many independent variables, such as structural correlation functions or electromagnetic spectra. In this work, independent variables (IVs) are defined as the fixed quantities over which a measurement is made (\textit{e.g.} frequencies along a spectrum or radial positions along a radial distribution function) and the outcomes of those measurements are referred to as quantities-of-interest (QoIs). 

Measurements of complex QoIs with many IVs are often available or easily obtained, yet are rarely included as observations in Bayesian optimization of molecular models. One reason why this may be the case is that previous literature has not outlined accurate and robust approaches to design Gaussian process surrogates for such data. For example, Angelikopoulous and coworkers did not use GP surrogate models for their Bayesian analysis on the radial distribution function (RDF) of liquid Ar \cite{angelikopoulos_bayesian_2012}, despite the fact that doing so would significantly reduce computation time. It is likely that GPs have not been previously used for complex QoIs due to high training and evaluation costs. Specifically, GPs have a cubic time-complexity in the number of IVs, which quickly becomes prohibitively expensive as experimental measurements obtain higher ranges and resolutions. 

Local Gaussian processes (LGPs) are an emerging class of accelerated GP methods that are well-equipped to handle large sets of experimental data. These so-called "greedy" Gaussian process approximations are constructed by separating a GP into a subset of GPs trained at distinct locations in the input space \cite{nguyen-tuong_model_2009,das_block-gp_2010,park_patchwork_2018,terry_splitting_2021}. Computation on the LGP subset scales linearly with the number of IVs, is trivially parallelizable, and easily implemented in high-performance computing (HPC) architectures \cite{gramacy_local_2015,broad_parallel_2023}. State-of-the-art LGP models have been used to design Gaussian approximation potentials (GAPs) \cite{deringer_gaussian_2021}, a type of machine learning potential used to study atomic \cite{caro_growth_2018,deringer_realistic_2018,deringer_towards_2018} and electron structures \cite{cheng_evidence_2020,deringer_gaussian_2021}, as well as nuclear magnetic resonance chemical shifts \cite{paruzzo_chemical_2018} with uncertainty quantification \cite{musil_fast_2019}. However, to our knowledge LGPs have not been applied as surrogate models for UQ/P on complex experimental data in computational chemistry.

In this study, we detail a simple and effective surrogate modeling approach for complex experimental observables common in physical chemistry. LGPs unlock the capability for existing Bayesian optimization schemes to incorporate complex data efficiently and accurately at a previously inaccessible computational scale. The key feature of the LGP surrogate model is the reduction in time-complexity with respect to the number of QoIs from cubic to linear, resulting in orders-of-magnitude speed-ups to evaluate complex observable surrogate models and perform posterior estimation. The computational speed-up results from reducing the dimensionality of matrix operations and therefore enables Bayesian UQ/P on experimental data with many IVs. For illustration, consider that a typical Fourier transformed infrared spectroscopy (FT-IR) measurement may contain data between 4000-400 cm$^{-1}$ at a resolution of 2 cm$^{-1}$, giving a total number of QoIs around $\eta = 1800$. According to the time-complexity scaling in $\eta$, a LGP is estimated to accelerate this computation compared to a standard GP by approximately 3,240,000x. Source code and a tutorial on building LGP surrogate models is provided on GitHub.

To demonstrate the method, we trained a LGP surrogate model on the RDF of the ($\lambda$-6) Mie fluid and performed Bayesian optimization to fit the parameters of the Mie fluid model to a neutron scattering derived RDF for liquid neon (Ne). The LGP was found to accelerate the $\eta=73$ independent variable surrogate model calculation approximately 1,760,000x faster than molecular dynamics (MD) and 2100x faster than a conventional GP with accuracy comparable to the uncertainty in the reported experimental data. Bayesian posterior distributions were then calculated with Markov chain Monte Carlo (MCMC) and used to draw conclusions on model behavior, uncertainty, and adequacy. Surprisingly, we find evidence that Bayesian inference conditioned on the radial distribution function significantly constrains the ($\lambda$-6) Mie parameter space, highlighting opportunities to improve force field optimization and design based on neutron scattering experiments.      

\subsection{Computational Methods}

In the following sections, an outline of standard approaches for Bayesian inference and surrogate modeling with Gaussian processes is presented. Then, we describe the local Gaussian process approximation and highlight key differences in their implementation and computational scaling.  

\subsubsection{Bayesian Inference}

Bayes' law, derived from the definition of conditional probability, is a formal statement of revising one's prior beliefs based on new observations. Bayes' theorem for a given model, set of model input parameters, $\boldsymbol{\theta}$, and set of experimental QoIs, $\mathbf{y}$, is expressed as,  

\begin{equation}\label{eq:bayes}
    p(\boldsymbol{\theta}|\mathbf{y}) \propto p(\mathbf{y}|\boldsymbol{\theta}) p(\boldsymbol{\theta})
\end{equation}

\noindent where $p(\boldsymbol{\theta})$ is the 'prior' probability distribution over the model parameters, $p(\mathbf{y}|\boldsymbol{\theta})$ is the 'likelihood' of observing $\mathbf{y}$ given parameters $\boldsymbol{\theta}$, and $p(\boldsymbol{\theta}|\mathbf{y})$ is the 'posterior' probability that the underlying parameter $\boldsymbol{\theta}$ models or explains the observation $\mathbf{y}$. Equality holds in Eq. \eqref{eq:bayes} if the right-hand-side is normalized by the 'marginal likelihood', $p(\mathbf{y})$, but including this term explicitly is unnecessary since the posterior probability distribution can be normalized \textit{post hoc}. In molecular simulations, $\boldsymbol{\theta}$ is the set of unknown parameters in the selected model, usually the force field parameters in the Hamiltonian, to the experimental QoI that the simulation estimates. The observations, $\mathbf{y}$, can be any QoI or combination of QoIs (\textit{e.g.} RDFs, spectra, densities, diffusivities, etc). This construction, known as the standard Bayesian scheme, is generalizable to any physical model and its corresponding parameters including density functional theory (DFT), \textit{ab initio} molecular dynamics (AIMD), and path integral molecular dynamics (PIMD).

Calculating the posterior distribution then just requires prescription of prior distributions on the model input parameters and evaluation of the likelihood function. In this work, Gaussian distributions are used for both the prior and likelihood functions, which is a standard choice according to the central limit theorem. The Gaussian likelihood has the form,

\begin{equation}\label{eq:likelihood}
    p(\mathbf{y}|\boldsymbol{\theta}) = \bigg(\frac{1}{\sqrt{2 \pi}\sigma_n}\bigg)^\eta \exp\bigg[-\frac{1}{2\sigma_n^2}\sum_{i=1}^{\eta}\ [\mathbf{y}_{\boldsymbol{\theta}} - \mathbf{y}]^2\bigg]
\end{equation}

\noindent where $\eta$ is the number of observables in $\mathbf{y}$, $\mathbf{y}_{\boldsymbol{\theta}}$ is the model predicted observables at model input $\boldsymbol{\theta}$, and $\sigma_n$ is a nuisance parameter describing the unknown variance of the Gaussian likelihood. Cailliez and coworkers choose the nuisance parameter as the sum of simulation and experiment variances ($\sigma_n^2 \approx \sigma_{sim}^2 + \sigma_{exp}^2$) \cite{cailliez_calibration_2014}; however, if these variances are unknown or one wishes to explore the distribution of variances, the nuisance parameter can be inferred via the Bayesian inference. Hence, the resulting posterior distribution on the nuisance parameter includes the unknown uncertainty arising due to the sum of the model and the experimental variances. In this work, the nuisance parameter is treated as an unknown to be inferred along with the explicit model parameters. Note that in some cases a different likelihood function may be more appropriate based on physics-informed prior knowledge of the distribution of the observable of interest (\textit{e.g.} the multinomial likelihood in relative entropy minimization between canonical ensembles \cite{shell_relative_2008}).

The computationally expensive part of calculating Eq. \ref{eq:likelihood} is determining $\mathbf{y}_{\boldsymbol{\theta}}$ at a sufficient number of points in the parameter space. Generally, this can be achieved by calculating $\mathbf{y}_{\boldsymbol{\theta}}$ at dense, equally spaced points in the parameter space of interest (grid method), sampling the parameter space with Markov chain Monte Carlo (MCMC) to estimate the posterior with a histogram (approximate sampling method), or assuming that the posterior distribution has a specific functional form (\textit{i.e.} Laplace approximation). Regardless of the selected method, each of these posterior distribution characterization techniques require a prohibitive number of molecular simulations to adequately sample the parameter space (often on the order of $10^5-10^6$), which is infeasible for even modest sized molecular systems.

\subsubsection{Gaussian Process Surrogate Models}

Gaussian processes accelerate the Bayesian likelihood evaluation by approximating $\mathbf{y}_{\boldsymbol{\theta}}$ with an inexpensive matrix calculation. A Gaussian process is a stochastic process such that every finite set of random variables (position, time, etc) has a multivariate normal distribution \cite{rasmussen_gaussian_2006}. The joint distribution over all random variables in the system therefore defines a functional probability distribution. The expectation of this distribution maps a set of model parameters, $\boldsymbol{\theta}^*$, and IVs, $\mathbf{r}$, to the most probable QoI given the model parameters, $S(\mathbf{r} | \boldsymbol{\theta}^*)$, such that,

\begin{equation}
    \mathbb{E}[GP] : \boldsymbol{\theta}^* \times \mathbf{r} \mapsto S(\mathbf{r} | \boldsymbol{\theta}^*)
\end{equation}

\noindent where the expectation operator is written in terms of a kernel matrix, $\mathbf{K}$, training set parameter matrix, $\mathbf{\hat{X}}$, and training set output matrix, $\mathbf{\hat{Y}}$, according to the equation,

\begin{equation}\label{eq:surrogate}
    \mathbb{E}[\textit{GP}(\boldsymbol{\theta}^*, \mathbf{r})] = \mathbf{K}_{(\boldsymbol{\theta}^*, \mathbf{r}),\mathbf{\hat{X}}} [\mathbf{K}_{\mathbf{\hat{X}}, \mathbf{\hat{X}}} + \sigma_{noise}^2 \mathbf{I}]^{-1} \mathbf{\hat{Y}}
\end{equation}

\noindent where $\sigma_{noise}^2$ is the variance due to noise and $\mathbf{I}$ is the identity matrix. Note that in general the IVs, $\mathbf{r}$, can be multidimensional. As an example, consider the case a GP maps a set of force field parameters to the angular RDF of a liquid. We now have a 2-dimensional space of IVs since the angular RDF gives the atomic density along the radial and angular dimensions. In the following mathematical development, it is assumed that the QoI is 1-dimensional for sake of convenience and note that extending the method to higher-dimensional observables just requires redefining the IVs in accordance with Eq. \eqref{eq:surrogate}.

The kernel matrix, $\mathbf{K}$, quantifies the relatedness between input parameters and can be selected based on prior knowledge of the physical system. A standard kernel for physics-based applications is the squared-exponential (or radial basis function) since the resulting GP is infinitely differentiable, smooth, continuous, and has an analytical Fourier transform \cite{ambrogioni_integral_2018}. The squared-exponential kernel function between input points $(\boldsymbol{\theta}_m,
r_m)$ and $(\boldsymbol{\theta}_n, r_n)$ is given by,

\begin{equation}\label{eq:kernel}
    K_{mn} = \alpha^2 \exp\bigg(-\frac{(r_{m} - r_{n})^2}{2\ell_{r}^2} - \sum_{o=1}^{\text{dim}(\boldsymbol{\theta})}\frac{(\theta_{o,m} - \theta_{o,n})^2}{2\ell_{\theta_o}^2} \bigg)
\end{equation}

\noindent where $o$ indexes over dim($\boldsymbol{\theta}$) and the hyperparameters $\alpha^2$ and $\ell_A$ are the kernel variance and correlation length scale of parameter $A$, respectively. Hyperparameter optimization can be performed by log marginal likelihood maximization, $k$-fold cross validation \cite{rasmussen_gaussian_2006} or marginalization with an integrated acquisition function \cite{snoek_practical_2012}, but can be computationally expensive and is usually avoided if accurate estimates of the hyperparameters can be made from prior knowledge of the chemical system.

To train a standard GP surrogate model, $N$ training samples are generated in the input parameter space and a molecular simulation is performed for each training set sample to calculate $N$ predictions over the number of target QoIs, $\eta$. The training set, $\mathbf{\hat{X}}$, is then a ($N\eta$ $\times$ dim($\boldsymbol{\theta}$) + 1) matrix of the following form,

\begin{equation}
    \mathbf{\hat{X}} = 
        \begin{bmatrix}
        \theta_{1,1} & \theta_{2,1} & \hdots & r_1\\
        \theta_{1,1} &  \theta_{2,1} & \hdots & r_2\\
        \vdots & \vdots & \vdots & \vdots\\
        \theta_{1,1} & \theta_{2,1} & \hdots & r_\eta\\
        \theta_{1,2} & \theta_{2,2} & \hdots & r_1\\
        \vdots & \vdots & \vdots & \vdots\\
        \theta_{1,N} & \theta_{2,N} & \hdots & r_\eta\\
        \end{bmatrix}
\end{equation}

\noindent where the $\theta_{i,j}$ are the $i^{th}$ model parameter for sample index $j$ and $r_k$ are the IVs of the target QoI. Note that the training sample index, $j = 1,...,N$, is updated in the model parameters only after $\eta$ rows spanning the domain of the observable, giving $N\eta$ total rows. Therefore, the training set matrix represents all possible combinations of the training parameters in the $\boldsymbol{\theta}$ parameter input space.     
The training set observations, $\mathbf{\hat{Y}}$, are a ($N\eta$ $\times$ 1) column vector of the observable outputs from the training set,

\begin{equation}
    \mathbf{\hat{Y}} = [S(\boldsymbol{\theta}_1,r_1), ..., S(\boldsymbol{\theta}_1,r_\eta), S(\boldsymbol{\theta}_2,r_1), ..., S(\boldsymbol{\theta}_N,r_\eta)]^T
\end{equation}

\noindent where $S(\boldsymbol{\theta}_j,r_k) = y(\boldsymbol{\theta}_j,r_k) - \mu_{GP}^{prior}(\boldsymbol{\theta}_j,r_k)$ is the difference between the training set observation of model parameters $\theta_j$ at IV $r_k$ and a GP prior mean function. Of course, the GP prior mean, $\boldsymbol{\mu}_{GP}^{prior}$, is the same shape as the training set observations matrix,

\begin{equation}
    \boldsymbol{\mu}_{GP}^{prior} := [\mu(\boldsymbol{\theta}_1,r_1), ..., \mu(\boldsymbol{\theta}_1,r_\eta), \mu(\boldsymbol{\theta}_2,r_1), ..., \mu(\boldsymbol{\theta}_N,r_\eta)]^T
\end{equation}

\noindent where $\mu(\boldsymbol{\theta}_j,r_k)$ is the GP prior mean for parameter set $\boldsymbol{\theta}_j$ at $r_k$. Note that the selection of a prior mean can impact the quality of fit of the GP surrogate model and should reflect physically justified prior knowledge of the physical system. 

Conceptually, since a Gaussian process is a Bayesian model, the prior serves as a current state of knowledge that can encode an initial guess for the QoI before the GP sees any training data. The subtraction of the GP prior mean from the model output effectively shifts the QoI by this pre-specified mean function. Hence, the GP is trained on these mean shifted observations rather than the observations themselves. Although shifting the data by another function seems like it shouldn't change the ability of the GP to estimate the QoI, it actually can have an important impact on the stochastic properties of the data as a function of the IVs. By construction, GPs are stationary, meaning that the means, variances, and covariances are assumed to be equal along all QoI. But for complex data, this is often not the case. For example, it is known that the RDF is zero for small $r$ values and has asymptotic tailing behavior to unity at long-range. The GP prior mean effectively shifts this non-stationary data and makes it behave as if it were stationary by removing any $r$ dependencies.

The expectation of the GP for a new set of parameters, $S^*(\mathbf{r}|\boldsymbol{\theta}^*)$, is then a ($\eta$ x 1) column vector calculated with Eq. \eqref{eq:surrogate},

\begin{equation}
    S^*(\mathbf{r}|\boldsymbol{\theta}^*) = [S^*(r_1|\boldsymbol{\theta}^*), ..., S^*(r_\eta|\boldsymbol{\theta}^*)]^T
\end{equation}

\noindent where $S^*(\mathbf{r}|\boldsymbol{\theta}^*)$ is the most probable difference function between the model and GP prior mean. Hence, to obtain a comparison to the experimental QoI you simply add the GP prior mean at $\boldsymbol{\theta}^*$, $\boldsymbol{\mu}_{GP}^{*,prior}(\mathbf{r}|\boldsymbol{\theta}^*)$, back to $S^*(\boldsymbol{\theta}^*, \mathbf{r})$.

The GP expectation calculation is burdened by the inversion of the training-training kernel matrix with $\mathcal{O}(N^3 \eta^3)$ time complexity and the ($\eta$ $\times$ $N\eta$) $\times$ ($N\eta$ $\times$ $N\eta$) $\times$ ($N\eta$ $\times$ 1) matrix product with $\mathcal{O}(N^2 \eta^3)$ time complexity. Note that these estimates are for naive matrix multiplication. Regardless, the cubic scaling in $\eta$ dominates the time-complexity for observables with many QoIs. For example, to build a GP surrogate model for the density of a noble gas ($\eta = 1$) with Lennard-Jones interactions (dim($\boldsymbol{\theta}$) = 2) would give a training set matrix of ($2N \times 3)$. Similarly, a surrogate model for an infrared spectrum of water from 600-4000 cm$^{-1}$ at a resolution of $4$ cm$^{-1}$ ($\eta = 850$) estimated with a 3 point water model of Lennard-Jones type interactions (dim($\boldsymbol{\theta}$) = 6) would generate a training set matrix of size (850$N$ $\times$ 7). Clearly, the complexity of the output QoI causes a significant increase in the computational cost of the matrix operations.

\subsubsection{The Local Gaussian Process Surrogate Model}

The time-complexity of the training-kernel matrix inversion and the matrix product can be substantially reduced by fragmenting the full Gaussian process of Eq. \eqref{eq:surrogate} into $\eta$ Gaussian processes. This method is also referred to as the subset of regressors approximation \cite{silverman_aspects_1985} and is considered a "greedy" approximation \cite{rasmussen_gaussian_2006}. Under this construction, an individual $GP_k$ is trained to map a set of model parameters to an individual QoI,

\begin{equation}
    \mathbb{E}[GP_k] : \boldsymbol{\theta} \mapsto S(r_k)
\end{equation}

\noindent where $\mathbf{r}$ is no longer an input parameter. The training set matrix, $\mathbf{\hat{X'}}$, is now a ($N$ $\times$ dim($\boldsymbol{\theta}$)) matrix, 

\begin{equation}\label{eq:subsurrogate_training}
    \mathbf{\hat{X'}} = 
        \begin{bmatrix}
        \theta_{1,1} & \theta_{2,1} & \hdots\\
        \theta_{1,2} & \theta_{2,2} & \hdots\\
        \vdots & \vdots & \vdots\\
        \theta_{1,N} & \theta_{2,N} & \hdots\\
        \end{bmatrix}
\end{equation}

\noindent while the training set observations, $\mathbf{\hat{Y'}}_k$, is a ($N$ $\times$ 1) column vector of the QoIs from the training set at $r_k$,

\begin{equation}\label{eq:subsurrogate_observation}
    \mathbf{\hat{Y}'}_k = [S(\boldsymbol{\theta}_1,r_k), ..., S(\boldsymbol{\theta}_N,r_k)]^T
\end{equation}

\noindent where $S(\boldsymbol{\theta}_j,r_k) = y(\boldsymbol{\theta}_j,r_k) - \boldsymbol{\mu}_{LGP, k}^{prior}(r_k)$ and $k$ indexes over IVs. The LGP prior mean $\boldsymbol{\mu}_{LGP, k}^{prior}(r_k)$ is now,

\begin{equation}\label{eq:subsurrogate_mean}
    \boldsymbol{\mu}_{LGP, k}^{prior} := [\mu(\boldsymbol{\theta}_1,r_k), ..., \mu(\boldsymbol{\theta}_N,r_k)]^T
\end{equation}

\noindent such that $\mu(\boldsymbol{\theta}_j,r_k)$ is the GP prior mean for parameter $\boldsymbol{\theta}_j$ at $r_k$. The squared-exponential kernel function is now,

\begin{equation}\label{eq:LGPkernel}
    K_{mn} = \alpha^2 \exp\bigg(- \sum_{o=1}^{\text{dim}(\boldsymbol{\theta})}\frac{(\theta_{o,m} - \theta_{o,n})^2}{2\ell_{\theta_o}^2} \bigg).
\end{equation}

\noindent The LGP surrogate model expectation for the observable at $r_k$, at a new set of parameters, $\boldsymbol{\theta}^*$, is just the expectation of the $k^{th}$ Gaussian process given the training set data,

\begin{equation}\label{eq:subsurrogate}
    S_{loc}^*(r_k|\boldsymbol{\theta}^*) = \mathbb{E}[\textit{GP}_k(\boldsymbol{\theta}^*)] = \mathbf{K}_{\boldsymbol{\theta}^*,\mathbf{\hat{X'}}} [\mathbf{K}_{\mathbf{\hat{X'}}, \mathbf{\hat{X'}}} + \sigma_{noise}^2 \mathbf{I}]^{-1} \mathbf{\hat{Y'}}_k.
\end{equation}

\noindent We then just combine the local results from the subset of $\eta$ GPs to obtain a prediction for the difference between the model and LGP prior mean,

\begin{equation}\label{eq:subQOI}
    S_{loc}^*(\mathbf{r}|\boldsymbol{\theta}^*) = [S_{loc}^*(r_1|\boldsymbol{\theta}^*), ..., S_{loc}^*(r_\eta|\boldsymbol{\theta}^*)]^T.
\end{equation}

\noindent and subsequently add back the LGP prior mean to obtain the estimated QoI, $y_{loc}^*(\mathbf{r}|\boldsymbol{\theta}^*) = S_{loc}^*(\mathbf{r}|\boldsymbol{\theta}^*) + \boldsymbol{\mu}_{LGP, k}^{prior}(\boldsymbol{\theta}^*,\mathbf{r})$.

By reducing the dimensionality of the relevant matrices, the time complexity of the matrix calculations are drastically reduced compared to a standard GP. The single step inversion of the training-training kernel matrix is now of $\mathcal{O}(N^3)$ time complexity while the $\eta$ step (1 $\times$ $N$) $\times$ ($N$ $\times$ $N$) $\times$ ($N$ $\times$ 1) matrix products are reduced to $\mathcal{O}(N^2 \eta)$ time complexity. If the number of training samples, $N$, the number of IVs, $\eta$, and the number of model evaluations, $G$, are equal between the full and LGP algorithms, then a LGP approximation reduces the evaluation time complexity in a standard GP from cubic-scaling, $\eta^3$, to embarrassingly parallelizable linear-scaling, $\eta$.

In summary, a local Gaussian process is an approximation in which the QoIs are modeled as independent random variables, each described by their own Gaussian process. This amounts to assuming that the random variables are stochastically independent. For time-independent data including scattering measurements and spectroscopy, this approximation is appropriate since each observation is an independent measurement at each independent variable. Finally, it is well-established that low rank approximations of Gaussian processes can compromise the accuracy of the estimated uncertainty, so the use of LGP regressors should be carefully scrutinized based on the risk/consequences of misrepresenting the resulting functional distributions.

Complex experimental observables can be reconstructed by this set of LGPs through a series of relatively straightforward matrix operations with linear time-complexity in the number of IVs. Furthermore, the LGP has all of the primary advantages of Bayesian methods, including built-in UQ and analytical derivatives and Fourier transforms. In the following section, we demonstrate the computational enhancement and accuracy of the LGP approach by modeling the RDF of neon at 42$K$. The LGP surrogate model is then implemented within a Bayesian framework to exemplify the power of UQ/P for molecular modeling.  

\subsection{A Local Gaussian Process Surrogate for the RDF of Liquid Ne}

To explore the computational advantages of LGP surrogate models for Bayesian inference, we studied the experimental RDF of liquid Ne \cite{bellissent-funel_neutron_1992} under a ($\lambda$-6) Mie fluid model. The ($\lambda$-6) Mie force field is a flexible Lennard-Jones type potential with variable repulsive exponent, 

\begin{equation}
    v^{Mie}_2(r) = \frac{\lambda}{\lambda-6}\bigg(\frac{\lambda}{6}\bigg)^{\frac{6}{\lambda-6}} \epsilon \bigg[ \bigg(\frac{\sigma}{r}\bigg)^\lambda - \bigg(\frac{\sigma}{r}\bigg)^6 \bigg]
\end{equation}

\noindent where $\lambda$ is the short-range repulsion exponent, $\sigma$ is the collision diameter (\AA), and $\epsilon$ is the dispersion energy (kcal/mol) \cite{mie_zur_1903}.  

MD simulations were performed from a Sobol sampled set spanning a prior range based on existing force field models \cite{vrabec_set_2001,mick_optimized_2015,shanks_transferable_2022} ($\lambda = [6.1,18]$, $\sigma = [0.88, 3.32]$, and $\epsilon = [0, 0.136]$) to generate a RDF training set matrix of the form in Eq. \ref{eq:subsurrogate_training}. Prior parameter ranges were selected so that training samples were restricted to the liquid regime of the ($\lambda$-6) Mie phase diagram \cite{widom_new_1970,ramrattan_corresponding-states_2015}. A sequential sampling approach was used in which we Sobol sample the prior range of parameters, calculate the training sample with the best-fit to the experimental data (lowest root mean squared error), center the new space on this training sample, and then narrow the sample range around this center point by a user selected ratio $\gamma$. This procedure was repeated three times with 320 samples per round (960 total training simulations) with $\gamma = 0.8$. This ratio was selected so that the final range would span >3 standard deviations of the posterior distributions estimated in prior literature \cite{angelikopoulos_bayesian_2012,mick_optimized_2015}. Subsequently, 320 test simulations were randomly sampled from the final range and used to determine whether or not the surrogate model provides accurate model predictions. A visualization of this procedure is provided in the Supporting Information. 

The number of observed points $\eta$ in the radial distribution function was calculated by dividing the reported $r_{max} - r_{min} \approx 15.3$ by the effective $r$-space resolution given by, $\Delta r = \pi/Q_{max}$, where $\Delta r = 0.21$ \AA for reported $Q_{max} = 15$ \AA$^{-1}$. This relation indicates that the appropriate number of observed independent $r$-values in the RDF is $\eta = 73$.

The training set matrix and training observation matrix were then constructed from the 960 training samples according to eqs \eqref{eq:subsurrogate_training} and \eqref{eq:subsurrogate_observation}, respectively. As a prior mean, we selected the RDF determined analytically from the dilute limit potential of mean force (PMF),

\begin{equation}
    \boldsymbol{\mu}_{PMF, k}^{prior}(\boldsymbol{\theta}_j,r_k) := g(\boldsymbol{\theta}_j,r_k) = \exp{[-\beta V(\boldsymbol{\theta}_j,r_k)]}
\end{equation}

\noindent where $g(\boldsymbol{\theta}_j,r_k)$ and $V(\boldsymbol{\theta}_j,r_k)$ are the analytical dilute limit RDF and ($\lambda$-6) Mie potential for parameters $\boldsymbol{\theta}_j$ at $r_k$, respectively. A PMF prior mean yields physically realistic short-range ($g(r) = 0$) and long-range behavior ($g(r) \to 1$). The PMF prior had improved RMSE compared to an ideal gas prior ($\forall r \in \mathbb{R}_0^+$, $g(r) = 1$), but this difference did not significantly impact the Bayesian posterior estimate (see Supporting Information). Finally, LGP hyperparameter optimization was performed using brute force to minimize the LOO error \cite{sundararajan_predictive_2001} over the training set.

Quantitative analysis of model sensitivity can be performed with probabilistic derivatives of the QoI with respect to model parameters (see Supporting Information) and subsequently related to temperature derivatives of radial distribution functions \cite{piskulich_temperature_2020}.

\subsubsection{Computational Efficiency and Accuracy}

Now that we have constructed the training set matrix, we simply evaluate the expectation at each $r_k$ according to Eq. \eqref{eq:subsurrogate} and combine the results into a single array as in Eq. \eqref{eq:subQOI}. The average computational time to invert the training set matrix and evaluate the surrogate model for both a standard GP and LGP are shown in Table \ref{tab:speed}. The LGP surrogate accelerates the RDF evaluation time compared to molecular dynamics by a factor of 1,700,000 for the $\eta = 73$ independent variable QoI with 960 training simulations. This 6 orders-of-magnitude speed-up beats a standard GP by 3 orders-of-magnitude (2141x). With respect to the training-training kernel matrix inversion, the LGP wins out on the standard GP by a factor of 31,565.

\begin{table}
\centering
\caption{Average relative time and speed-up to QoI evaluation and training set matrix inversion for a standard and local Gaussian process for 960 training samples and a RDF with $\eta = 73$ points.}
\begin{tabular}{| l | c | c | c | c|}
\hline
\textrm{Model}&
\textrm{QoI Eval. Time (s)}&
\textrm{Speed Up ($t/t_{sim}$)}&
\textrm{Inv. Time (s)}\\
\hline
Simulation  & 1,251  & 1      & - \\
GP          & 1.52   & 822     & 355\\
LGP         & 0.0007 & 1,760,267 & 0.01\\
\hline
\end{tabular}
\label{tab:speed}
\end{table}

In summary, the LGP significantly accelerates both computational bottlenecks for Gaussian process surrogate modeling; namely, the training set matrix inversion and surrogate model evaluation time. Of course, the exact speed-ups depend on numerous factors including the number of IVs $\eta$, the number of training samples used to construct the training set matrix $N$, the level of code parallelization, and hyperparameter optimization procedure. Which step is rate limiting depends on the surrogate modeling application. For instance, if the surrogate model doesn't need to be evaluated a large number of times, the training set generation, matrix inversion and hyperparameter optimization will be the rate limiting steps. On the other hand, applications that require a large number of model evaluations, such as uncertainty quantification and propagation, result in the surrogate model evaluation time being rate limiting. Typically, designing a surrogate model is only necessary in the latter case.  

Clearly the LGP is fast, but is it accurate? In other words, does the LGP provide QoI predictions that are within a reasonable level of accuracy to serve as a true surrogate model for the molecular dynamics predictions? To evaluate the accuracy of the local predictions, a test set of 320 ($\lambda$-6) Mie parameters was randomly sampled from the final range of the sequential sampling method (see Supporting Information) and the RMSE computed between simulated and LGP predicted radial distribution functions along all radial positions, $r$. The results are summarized in Figure \ref{fig:rmse}.

\begin{figure}
    \centering
    \includegraphics[width = 14 cm]{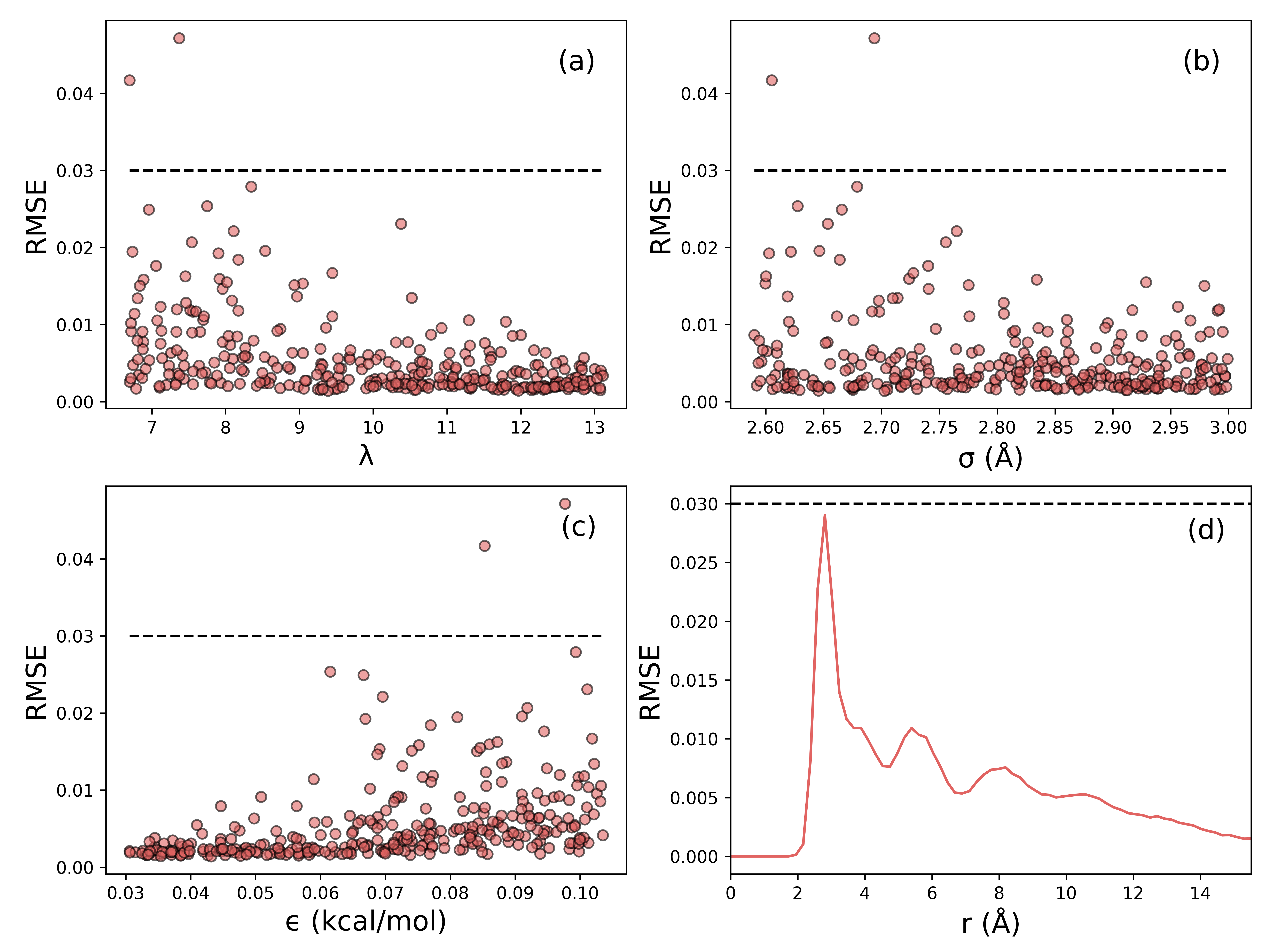}
    \caption{Visualization of LGP surrogate model testing and validation. (a)-(c) Test set samples over each parameter plotted against the RMSE between simulated and LGP data. (d) Average RMSE over the 320 test set samples as a function of $r$. The dashed line represents the reported error from the experiment.}
    \label{fig:rmse}
\end{figure}

\noindent The RMSE for all radial positions is less than 0.03, which is excellent considering that this error is smaller than the reported experimental uncertainty ($\sim$0.03). Of course, the acceptable RMSE over the QoI is user-defined and largely subjective based on the surrogate model application, but can be improved with additional training and hyperparameter optimization if necessary (an example is included in the Supporting Information). 

\subsubsection{Learning from the Ne RDF Surrogate Model with Bayesian Analysis}

Our fast and accurate LGP surrogate model now allows us to explore the underlying probability distributions on the ($\lambda$-6) Mie parameter space. This example is provided to show how one can use Bayesian analysis to learn about correlations and relationships between model parameters as well as model adequacy. This analysis can provide robust insight into the nature of the model and provide quantifiable evidence for whether or not the model is appropriate for a target application.    
Bayesian inference yields a probability distribution function over the model parameters called the joint posterior probability distribution. The maximum of the joint posterior, referred to as the \textit{maximum a posteriori} (\textit{MAP}), represents the set of parameters with the highest probability of explaining the given experimental data. In force field design, the \textit{MAP} would be an appropriate choice for an optimal set of model parameters. However, the power of the Bayesian approach lies in the fact that, not only can we identify the optimal parameters, but we can also examine the probability distribution of the parameters around these optima. For instance, the width of the distribution provides evidence for how important a parameter in the model is for representing the target data. For a given parameter, a wide distribution indicates that the parameter has little influence on the model prediction. On the other hand, a narrow distribution indicates that the parameter is critical to the model prediction. Additionally, the joint posterior may exhibit multiple peaks, or modes. A multimodal joint posterior suggests that there are multiple sets of model parameters that reproduce the target data, which may be a symptom of model inadequacy. Finally, the symmetry of the distribution provides information on relationships and correlations between parameters, providing a framework to diagnose subtle relationships that may otherwise go unnoticed.

Usually, the joint posterior distribution is a high-dimensional quantity that cannot be visualized directly. However, we can visualize the joint posterior along one dimension by integrating out the contributions over all other parameters. The resulting distributions are called marginal distributions. Marginal distributions computed over the ($\lambda$-6) Mie potential parameters optimized to the RDF of liquid Ne are shown in Figure \ref{fig:posterior}.

\begin{figure}
    \centering
    \includegraphics[width = 15cm]{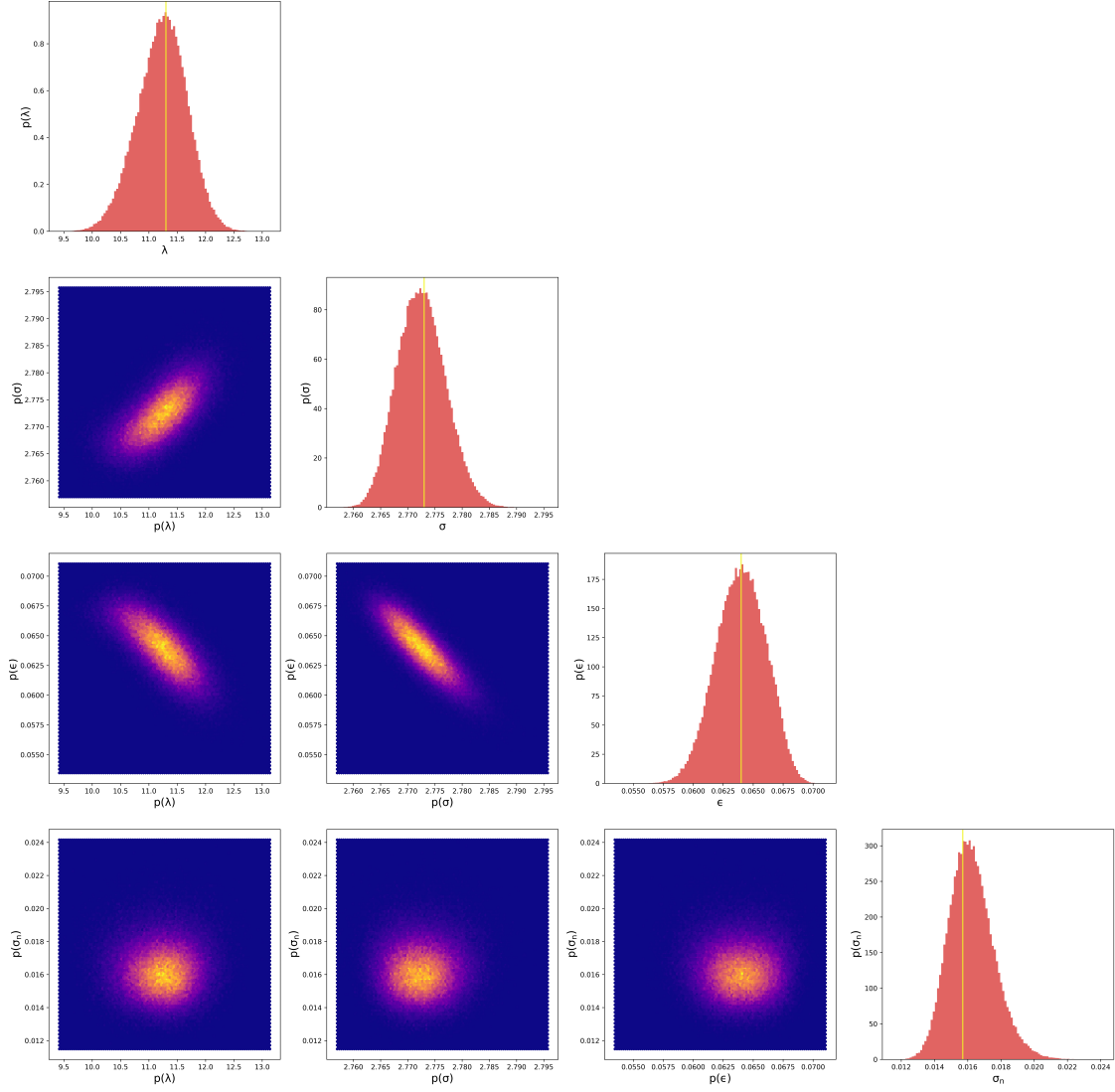}
    \caption{(Diagonal) 1D marginal distributions for the ($\lambda$-6) Mie fluid parameters. Prior distributions are not depicted since they are flat lines near 0 probability. Yellow vertical lines represent the \textit{maximum a posterior} (\textit{MAP}) estimate. (Off-Diagonal) 2D marginal histograms showing parameter correlations.}
    \label{fig:posterior}
\end{figure}

\noindent For each parameter, the resulting marginal posterior distributions are unimodal and symmetric. This result is not surprising in the context of recent results that show iterative Boltzmann inversion, which is a maximum likelihood approach to the structural inverse problem, is convex for Lennard-Jones type fluids \cite{hanke_well-posedness_2018}. Observing the 2D marginal distributions in Figure \ref{fig:posterior}, we can also see that each of the parameters are correlated to one other. For example, the negative correlation between $\sigma$ and $\epsilon$ suggests that increasing the size of the particle should be accompanied by a decrease in the effective particle attraction. Conceptually this makes sense, if the particles are larger, then they would need to have a weaker attractive force to give the same atomic structure. This result is consistent with Bayesian analysis on liquid Ar \cite{angelikopoulos_bayesian_2012}. The nuisance parameter distribution shows that the unknown standard deviation between the LGP surrogate model and the experimental data is around 0.016.

One surprising characteristic of the posterior distribution is that it is extremely narrow. Recall that narrow distributions indicate that the parameters are important, or have tight control, over the model quality-of-fit to the experimental data. From our Bayesian analysis, we can therefore confidently conclude that detailed interatomic force information is contained within the experimental RDF. This observation is in stark contrast to over 60 years of prior literature which has unanimously asserted that only the excluded volume or collision diameter can be ascertained from experimental scattering data \cite{clayton_neutron_1961,jovari_neutron_1999,hansen_theory_2013}. In fact, the Bayesian analysis shows that it is possible to determine values for $\lambda$, $\sigma$, and $\epsilon$ within $\pm 2$, $\pm 0.02 $ \AA, and $\pm 0.0075$ kcal/mol with 95\% certainty. This result leads to two important conclusions: (1) Scattering data can effectively constrain the force field model parameter space and (2) the data must be sufficiently accurate to do so. These results provide evidence that scattering data could be invaluable to inform accurate force fields, particularly for structure and self-assembly applications.

The joint posterior can also be used for model parameter selection given the experimental observation. Specifically, the optimal parameters are given by the \textit{MAP}, corresponding to the maximum of the joint posterior distribution. The \textit{MAP} is presented in Table \ref{tab:params} along with two other existing force fields for liquid Ne.

\begin{table}
\centering
\caption{Summary of ($\lambda-6)$ Mie potential parameters optimized for Ne. Values for the repulsive exponent parameter are rounded to the nearest integer.}
\begin{tabular}{| l | c | c | c | c | c |}
\hline
\textrm{Force Field}&
\textrm{QoI}&
\textrm{$\lambda$}&
\textrm{$\sigma$ (\AA)}&
\textrm{$\epsilon$ (kcal/mol)}\\
\hline
Mick (2015) & VLE & 11 & 2.794 & 0.064\\
SOPR (2022) & RDF & 11 & 2.778 & 0.063\\
This Work   & RDF & 11 & 2.773 & 0.064\\
\hline
\end{tabular}
\label{tab:params}
\end{table}

\noindent The estimated Mie parameters are in agreement with the Mick \cite{mick_optimized_2015} and structure optimized potential refinement (SOPR) \cite{shanks_transferable_2022} models. This result confirms that the radial distribution function contains sufficient information to determine transferable force field parameters in simple liquids. 

Some interesting questions arise considering that both the Mie fluid model and SOPR, which is a probabilistic iterative Boltzmann method for experimental scattering data, give similar predictions for the structure-optimized potentials. The key difference between the Bayesian optimization performed in this work and SOPR is that the former is parametric while the latter is non-parametric, both of which have strengths and weaknesses. Specifically, parametric models are less complex but may not be flexible enough to describe subtle details of the experimental observation. On the other hand, non-parametric models can describe nuanced experiments but may over-fit to non-physical features of the data. It is then natural to wonder: Is a ($\lambda$-6) Mie model adequate to describe the experimental scattering data? Or does the scattering data complexity necessitate the use of non-parametric iterative potential refinement techniques like SOPR?

We can investigate the first question of model adequacy by propagating parameter uncertainty through the LGP to construct a distribution of RDF predictions - referred to as the posterior predictive. The posterior predictive can be estimated by evaluating the LGP for all MCMC samples and computing the mean,

\begin{equation}
        \mathbb{E}[{S}_{loc}^*(\mathbf{r}_k)] \approx \frac{1}{N}\sum_{i=1}^N {S}_{loc}^*(\mathbf{r}_k|\boldsymbol{\theta}_i)
\end{equation}

\noindent and variance,

\begin{align}
    \mathbb{V}[{S}_{loc}^*(\mathbf{r}_k)] \approx \frac{1}{N}\sum_{i=1}^N ({S}_{loc}^*(\mathbf{r}_k|\boldsymbol{\theta}_i) - \mathbb{E}[{S}_{loc}^*(\mathbf{r}_k)])^2 
\end{align} 

\noindent of the resulting QoI predictions. Recall that the nuisance parameter distribution is also sampled to account for unknown uncertainties in the LGP surrogate model and experimental data. The posterior predictive therefore quantifies of how accurately we know the QoI given experimental, model, and parametric uncertainty estimated with Bayesian inference. If the model is adequate, the Bayesian credibility interval ($\mu \pm 2 \sigma$) should contain approximately 95\% of the experimental data. The posterior predictive and residuals ($g_{exp}(r) - \mu(r)$) estimated for the liquid Ne RDF are shown in Figure \ref{fig:post_pred}.

\begin{figure}
    \centering
    \includegraphics[width = 14cm]{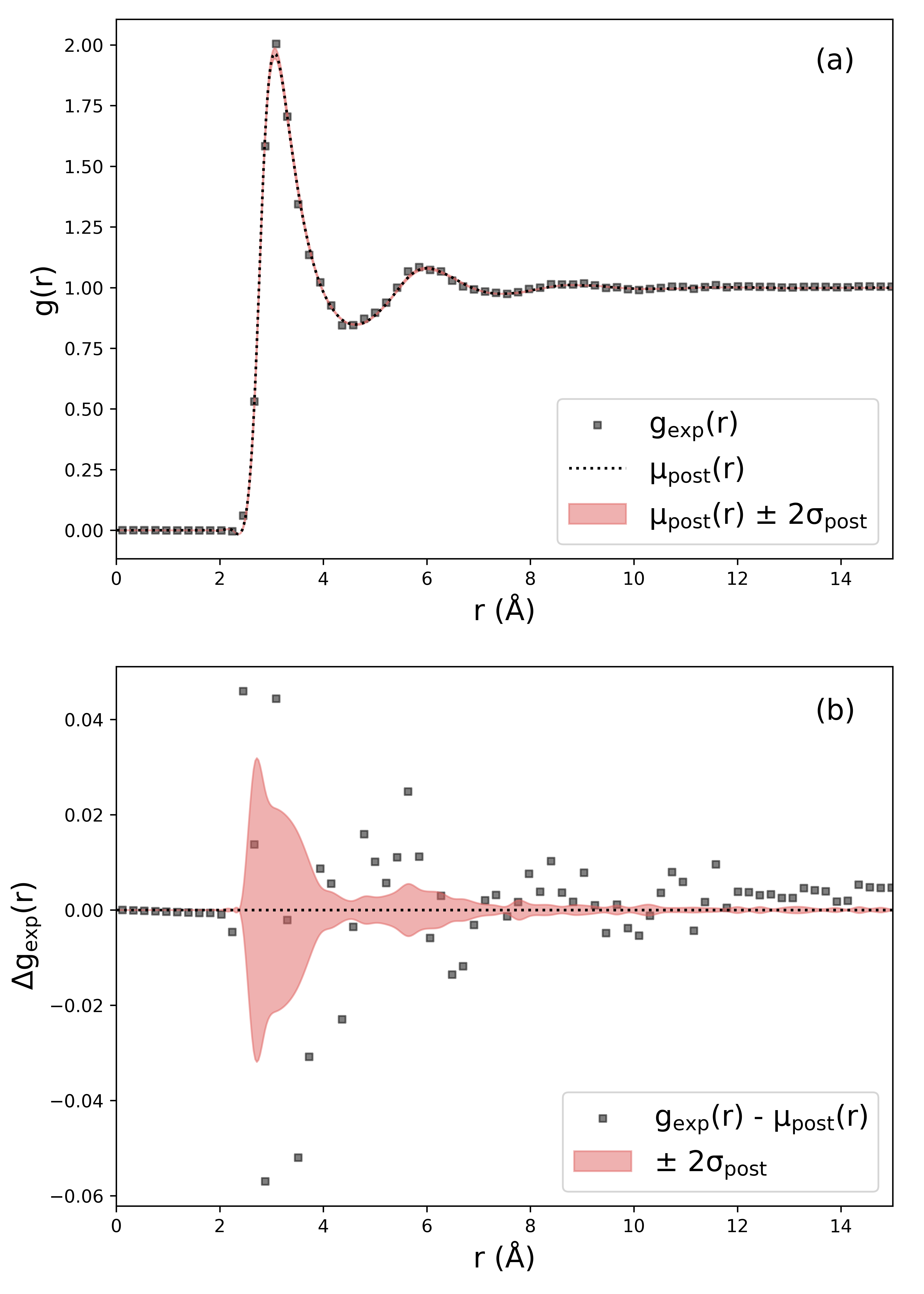}
    \caption{Posterior predictive plots for the radial distribution function. (a) RDF mean and credibility interval propagated from the parameter uncertainty quantified with Bayesian inference. (b) Residual analysis comparing the experimental data with the posterior predictive distribution.}
    \label{fig:post_pred}
\end{figure}

\noindent Clearly, the agreement between the posterior predictive mean and the experimental data is excellent. However, the residuals often lie outside of the $2\sigma_{post}$ credibility interval. These differences between the experiment and model could be explained by a number of different factors, including errors arising from Fourier transform truncation, background scattering corrections or model inadequacy, among others. However, without rigorous uncertainty quantification on the experimental scattering data, it is currently not possible to determine which factor or combination of factors results in the model disagreement. We argue that this knowledge gap necessitates rigorous UQ/P studies on scattering data as well as iterative potential refinement methods. Combining these approaches with Bayesian inference on molecular dynamics models could then shed light on what physical interactions can be learned from scattering experiments.

In summary, we have shown that a LGP surrogate model enables rapid and accurate uncertainty quantification and propagation with Bayesian inference. We then showed how the posterior distribution is an indispensable tool to learn subtle relationships between model parameters, identify how important each model parameter is to describe the outcome of experiments, and quantify our degree of belief that our model adequately describes our observations. The power of Bayesian inference is evident.

\subsection{Conclusions}

We have shown that local Gaussian process surrogate models trained on an experimental RDF of liquid neon improves the computational speed of QoI prediction 1,760,000-fold with exceptional accuracy from only 960 training simulations. The 3 orders-of-magnitude evaluation time speed-up for a local versus standard Gaussian process was shown to accelerate Bayesian inference without the need for advanced sampling techniques such as on-the-fly learning. Furthermore, since the LGP linearly scales with the number of output QoIs, significantly higher speed-ups are expected for more complex data, such as infrared spectra or high resolution scattering experiments, or for multiple data sources simultaneously (\textit{e.g.} scattering, spectra, density, diffusivity, etc). We conclude that local Gaussian processes are an accurate and reliable surrogate modeling approach that can accelerate Bayesian analysis of molecular models over a broad array of complex experimental data. 

\subsection{Supporting Information}

\subsubsection{Molecular Dynamics Simulation of Mie Fluids}

Computer generated radial distribution functions (RDFs) were calculated using molecular dynamics (MD) simulations in the HOOMD-Blue package \cite{anderson_hoomd-blue_2020}. Simulations were initiated with a lattice configuration of 864 particles and compressed to a reduced density of $\rho = 0.02477$ atom/\AA$^3$ and thermal energy $T = 42.2$ K. The HOOMD NVT integrator was used for a 0.25 nanosecond equilibration step and a 0.25 nanosecond production step (dt = 0.5 femtosecond). Potentials were truncated at $3\sigma$ with an analytical tail correction, and RDFs were calculated using the Freud package \cite{ramasubramani_freud_2020}.

\subsubsection{Training and Test Set Generation}

The first step to design a LGP surrogate model is to generate a training set of model input parameter input and QoI outputs. To generate the training set, we need dense samples of model parameters in the region of the parameter space that well-represents the target experimental data. In general, it is not known \textit{a priori} where this region is, particularly if there is no prior knowledge of what model parameters are best with respect to an experimental observation, $\mathbf{y}$. However, there are parameter regions that we can exclude \textit{a priori} based on the physics of the ($\lambda$-6) Mie fluid. For instance, Ne is a liquid at the experimental thermodynamic conditions, so we can use well-established ($\lambda$-6) Mie fluid phase diagrams and vapor-liquid transitions \cite{widom_new_1970} to restrict the parameter ranges to the liquid phase only. Specifically, given a fixed temperature (T = 42.2K) and density ($\rho$ = 0.024 \AA$^{-3}$), it is trivial to determine the $\sigma$ and $\epsilon$ parameter ranges reported in the manuscript via relations for the scaled temperature ($T^* = k_bT/\epsilon$) and scaled density ($\rho^* = \rho \sigma^3$). The parameter ranges determined using the Mie fluid phase diagram are presented in Table \ref{tab:ranges}. Restricting the parameter to physically justified ranges is important to avoid a "garbage in, garbage out" scenario for an LGP surrogate model. Given this prior range, we then performed the sequential sampling approach outlined in the manuscript. A visualization of this procedure is shown in Figure \ref{fig:trainingset}.

\begin{table}
\centering
\caption{Estimated boundaries for physics-constrained prior space based on the ($\lambda$ - 6) Mie fluid phase diagram. $m = 6$ is the attractive tail exponent of the ($\lambda$ - 6) Mie potential. $*$) The maximum $\lambda$ was selected to be substantially larger than previously reported values.}
\begin{tabular}{| c | l | l | l | l |}
\hline
\textrm{Param.}&
\textrm{Min.}&
\textrm{Min. Criteria}&
\textrm{Max.}&
\textrm{Max. Criteria}
\\
\hline
$\lambda$  & 6.1  & $m=6 \implies \lambda>6$  &  18 & Literature$^*$ \\
$\sigma$   & 2.55 & Vapor-Liquid Equil. &  3.32 & Solid-Liquid Equil.  \\
$\epsilon$ & 0.00 & $\epsilon<0$ undefined  &  0.136 & Vapor-Solid Equil. \\
\hline
\end{tabular}
\label{tab:ranges}
\end{table}

\begin{figure}
    \centering
    \includegraphics[width = 15cm]{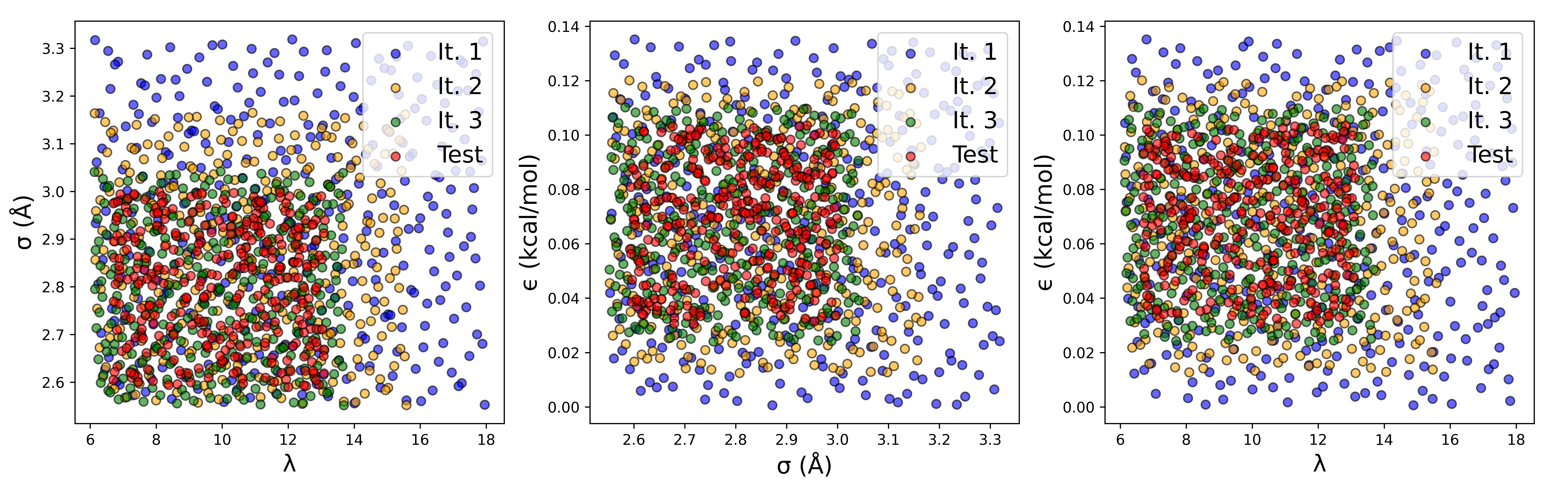}
    \caption{2D training and sample parameter set used to train and test the LGP surrogate model.}\label{fig:trainingset}
\end{figure}

\subsubsubsection{The Gaussian Process Prior Mean}

In this manuscript, an analytical solution for the RDF based on the dilute limit potential of mean force (PMF) was used as the GP prior mean. This choice is appropriate as an RDF prior since it will have the same features that we expect a liquid RDF to have, \textit{i.e.} RDF values of zero at low $r$ and a long-range tail that asymptotically approaches unity. However, note that even a prior guess that doesn't encode this information can still produce accurate LGP surrogate models for RDFs. For example, in Figure \ref{fig:rmse2} we can see that an ideal gas RDF prior, which amounts to approximating that the RDF is unity everywhere ($g^{IG}(r) = 1$), can still be learned by the local Gaussian processes with RMSE values close to the more physically justified PMF.

\begin{figure}
    \centering
    \includegraphics[width = 14cm]{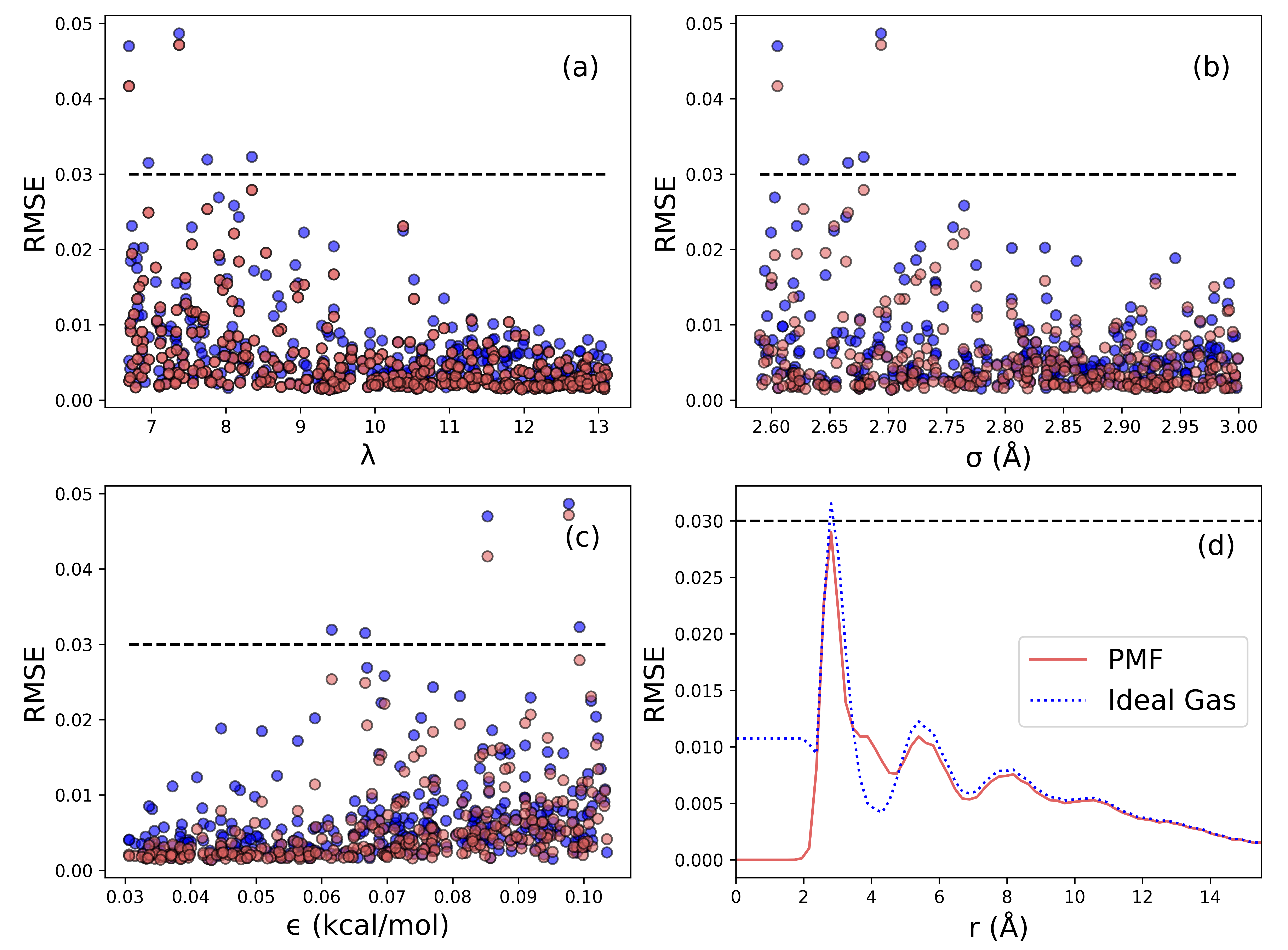}
    \caption{RMSE as a function of $r$ for an ideal gas (blue) and potential of mean force (red) prior.}
    \label{fig:rmse2}
\end{figure}

Clearly, the RMSE along $r$ is a similar magnitude for the ideal gas and PMF prior, but the $r$-dependent behavior is noticeably different. For the ideal gas prior, we see that there is high RMSE at low $r$, which is inconsistent with our intuition for a liquid RDF due to the excluded volume of atoms. What is occurring here is that the GP estimate is being "pulled" towards the prior at low $r$. On the other hand, the PMF prior exhibits behavior in line with our physical intuition; namely, near zero error at $r$ values smaller than the relative diameter of the atom. Perhaps surprisingly, we see in Figure \ref{fig:post_compare} that the choice of prior mean doesn't have a large impact on the posterior distribution or MAP estimates. We attribute this to the fact that the RMSE is sufficiently small for both the ideal gas and PMF priors that the posterior distribution isn't significantly modified. However, it does influence the posterior predictive distribution as evidenced by Figure \ref{fig:postpredictive_compare}. Specifically, note  that there is uncertainty at low $r$ for the ideal gas prior, whereas this uncertainty vanishes for the PMF prior.

\begin{figure}
    \centering
    \includegraphics[width = 15cm]{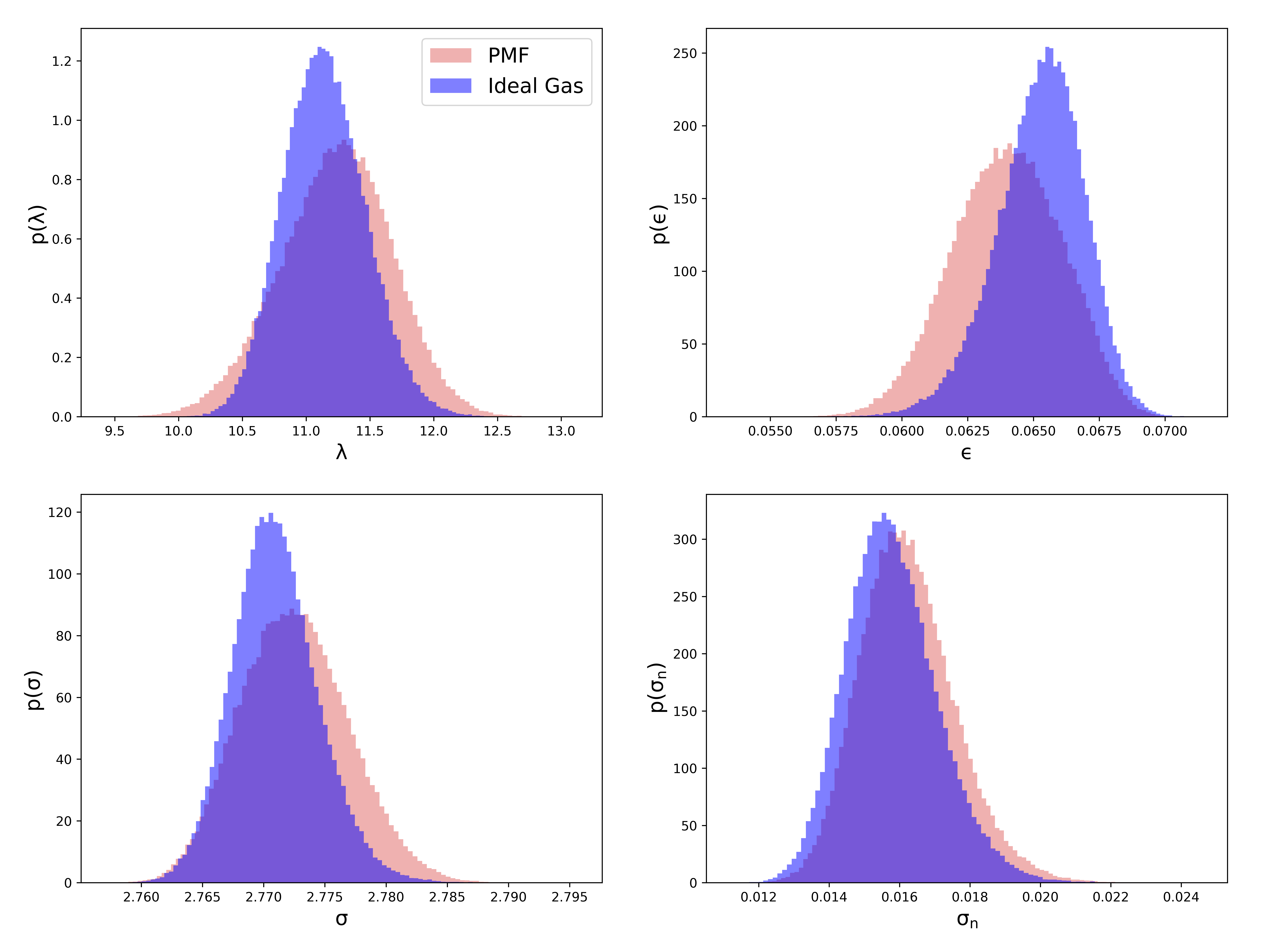}
    \caption{Marginal posteriors for the ideal gas PMF (red) and ideal gas (blue) priors.}
    \label{fig:post_compare}
\end{figure}

\begin{figure}
    \centering
    \includegraphics[width = 12cm]{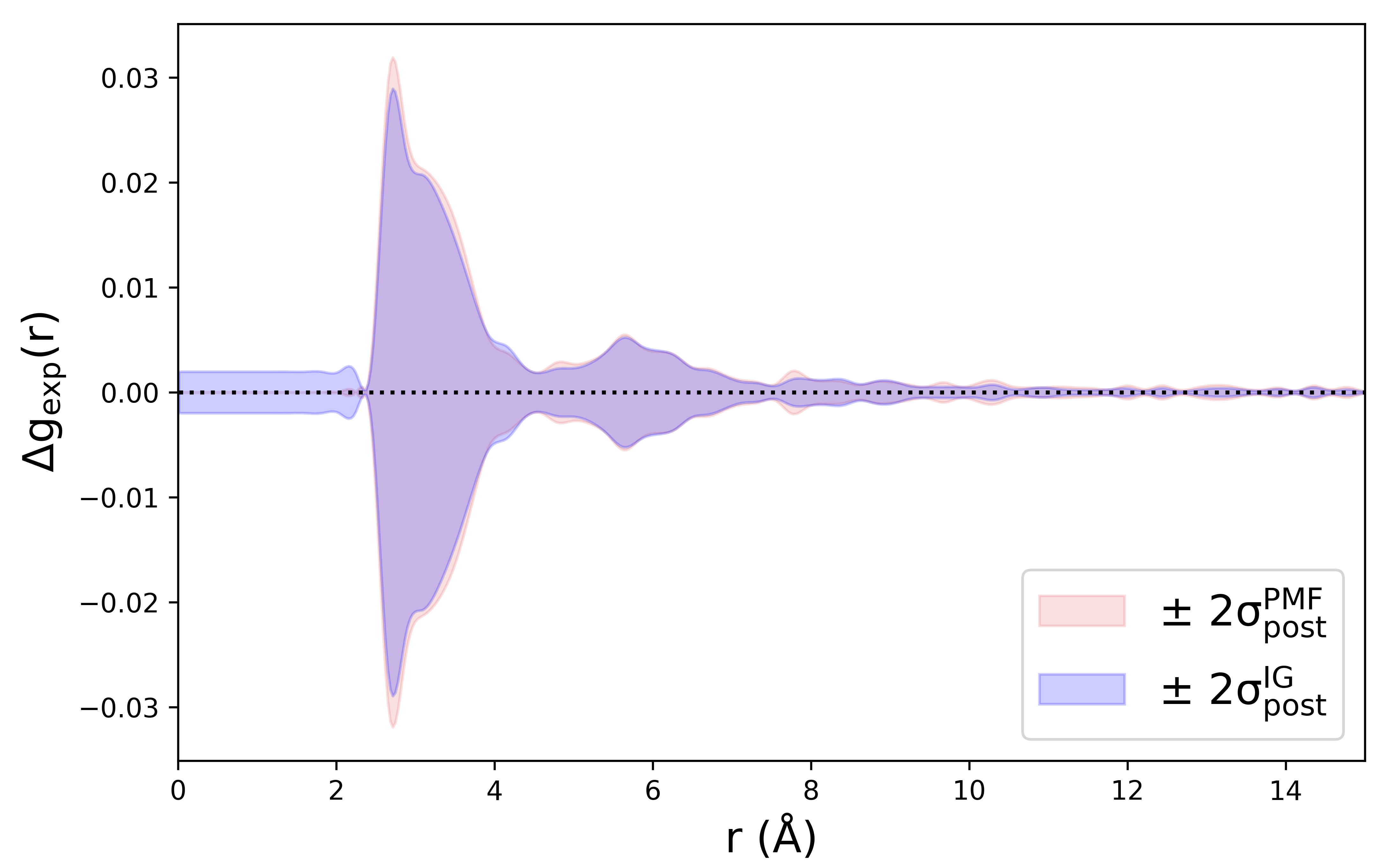}
    \caption{Posterior predictives for the ideal gas PMF (red) priors ideal gas (blue) priors.}
    \label{fig:postpredictive_compare}
\end{figure}

\subsubsubsection{Hyperparameter Selection}

The final step is to learn a set of LGP hyperparameters that provide accurate estimates of the target QoI. A standard approach to selecting hyperparameters is to maximize the model evidence \cite{rasmussen_gaussian_2006} or apply an expected improvement criterion based on an integrated acquisition function \cite{snoek_practical_2012}. Here we applied a brute force search based on minimizing the leave-one-out (LOO) error for 25,000 hyperparameter options randomly sampled over a prior range using the method of Sundararajan and coworkers \cite{sundararajan_predictive_2001} (Table \ref{tab:hyperparams}). This method gives relatively similar hyperparameter estimations for both an ideal gas and PMF GP prior. The prior range was selected based on the ($\lambda$-6) Mie parameter sensitivity analysis of Mick and coworkers \cite{mick_optimized_2015}.  

A limitation of the brute force approach to hyperparameter selection is that we don't account for potential hyperparameter uncertainty in the LGP prediction. However, the self-consistency of our predictions with existing literature on liquid neon \cite{mick_optimized_2015,shanks_transferable_2022} suggests that this uncertainty is likely insignificant. Note that one could propagate hyperparameter uncertainty by performing Bayesian optimization over the hyperparameters, sampling the resultant hyperparameter posterior distribution, and propagating the samples through the posterior predictive estimation step. Finally, hyperparameter optimization for the LGP model is non-trivial since the LGP is an approximation to a non-stationary stochastic process \cite{heinonen_non-stationary_2016}. 

What if the previously described method fails to yield an accurate surrogate model? In this case, one can repeat the sequential sampling by adding more training simulations at each range to retrain the LGP until the RMSE is sufficiently small. As an example, Figure \ref{fig:rmsevssims} demonstrates that surrogate model accuracy improves as more samples are added at each range. Note that the accuracy of the surrogate will not improve beyond the statistical uncertainty of the underlying model.

\begin{table}
\centering
\caption{Optimum hyperparameter values under the ideal gas and PMF prior computed from 25000 random samples over the reported test range.}
\begin{tabular}{| c | l | r | r |}
\hline
\textrm{Name}&
\textrm{Test Range}&
\textrm{Ideal Gas}&
\textrm{PMF}\\
\hline
$\ell_{\lambda}$  & 0.5-4      &  3.31   & 3.58\\
$\ell_{\sigma}$   & 0.01-0.05  &  0.046  & 0.048\\
$\ell_{\epsilon}$ & 0.001-0.01 &  0.0098 & 0.0093\\
$\alpha$          & 1E-4-0.1   &  0.094  & 0.095\\
$\sigma_{noise}$  & 1E-4-0.01  &  7.2E-4 & 8.3E-4\\
\hline
\end{tabular}
\label{tab:hyperparams}
\end{table}

\begin{figure}
    \centering
    \includegraphics[width = 12cm]{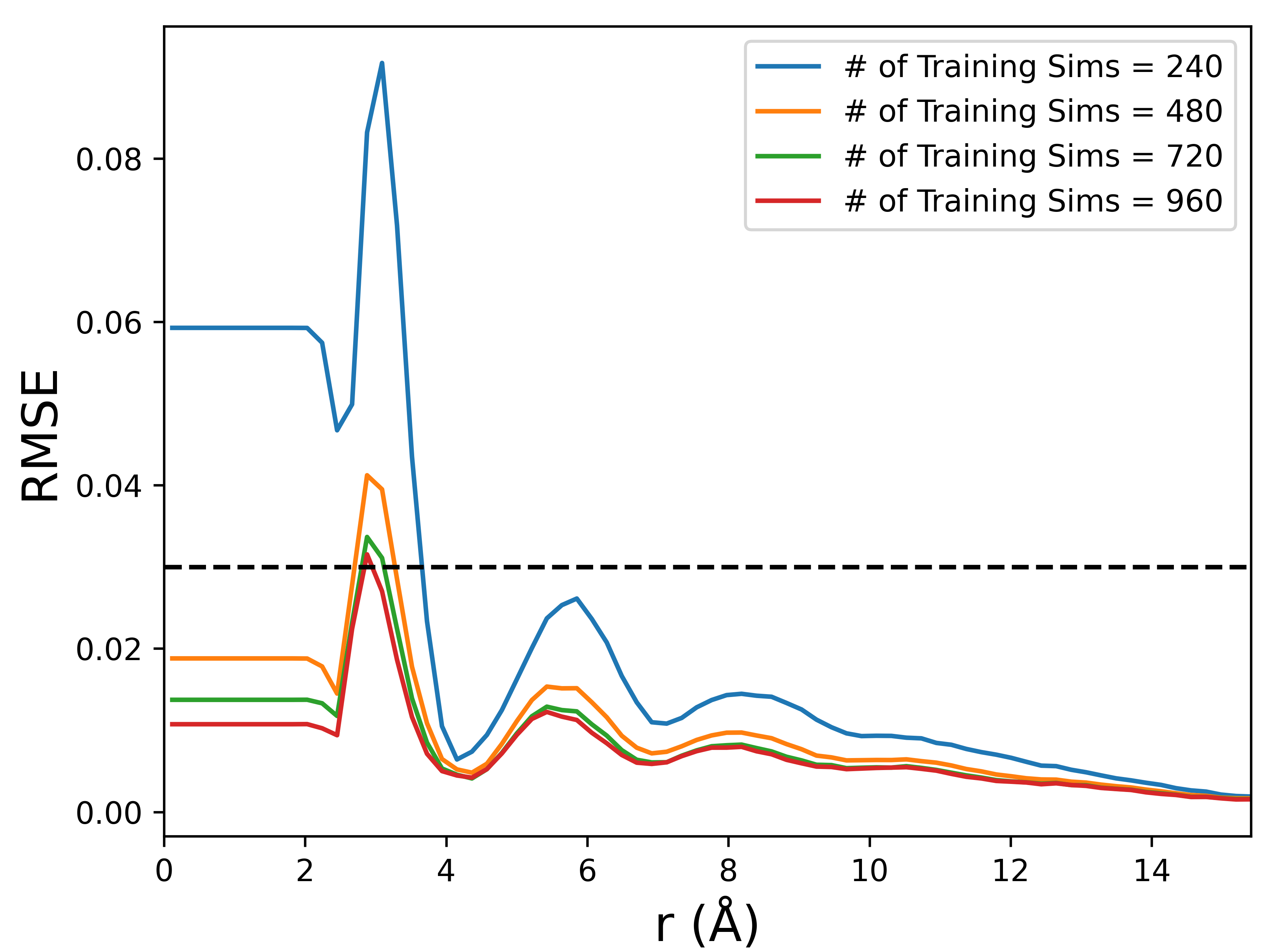}
    \caption{RMSE along $r$ for varying numbers of training simulations under an ideal gas prior.}
    \label{fig:rmsevssims}
\end{figure}

A more rigorous, but non-trivial method for surrogate model training, is to use adaptive or on-the-fly learning, in which the uncertainty in the LGP prediction is used to decide whether or not a new simulation is needed in the training set. This approach has been used in prior work\cite{angelikopoulos_bayesian_2012,angelikopoulos_x-tmcmc_2015} but was found to be unnecessary for our purposes due to the efficiency and accuracy of the LGP with relatively few training samples.

\subsubsection{Using the LGP Surrogate Model for Parameter Sensitivity Analysis}

Sensitivity of a QOI to a model parameter, $\theta_i$, can be quantified using the analytical derivative of the local GP surrogate model according to the following equation, 

\begin{equation}\label{eq:derivs}
    \frac{\partial \mathbb{E}[\textit{GP}_k(\boldsymbol{\theta}^*)]}{\partial \theta_i} = \bigg(\frac{ \theta_i - \theta_i^*}{\ell_{\theta_i}^2} \bigg) \mathbf{K}_{\boldsymbol{\theta}^*,\mathbf{\hat{X'}}} [\mathbf{K}_{\mathbf{\hat{X'}}, \mathbf{\hat{X'}}} + \sigma_{noise}^2 \mathbf{I}]^{-1}\mathbf{\hat{Y'}}_k
\end{equation}

\noindent where $k$ is the QoI index. The Gaussian process derivative is a quantitative measure of the influence of a perturbation in $\theta_i$ to the expectation of the observable $\mathbf{\hat{Y'}}_k$. Therefore, Eq. \eqref{eq:derivs} can approximate the impact of changing a molecular simulation parameter on the QOI (\textit{e.g.} how much does changing the effective particle size impact the RDF at any $r$). Figure \ref{fig:derivs} shows probabilistic local GP derivatives calculated for the ($\lambda$-6) Mie parameters. 

\begin{figure}
    \centering
    \includegraphics[width = 12cm]{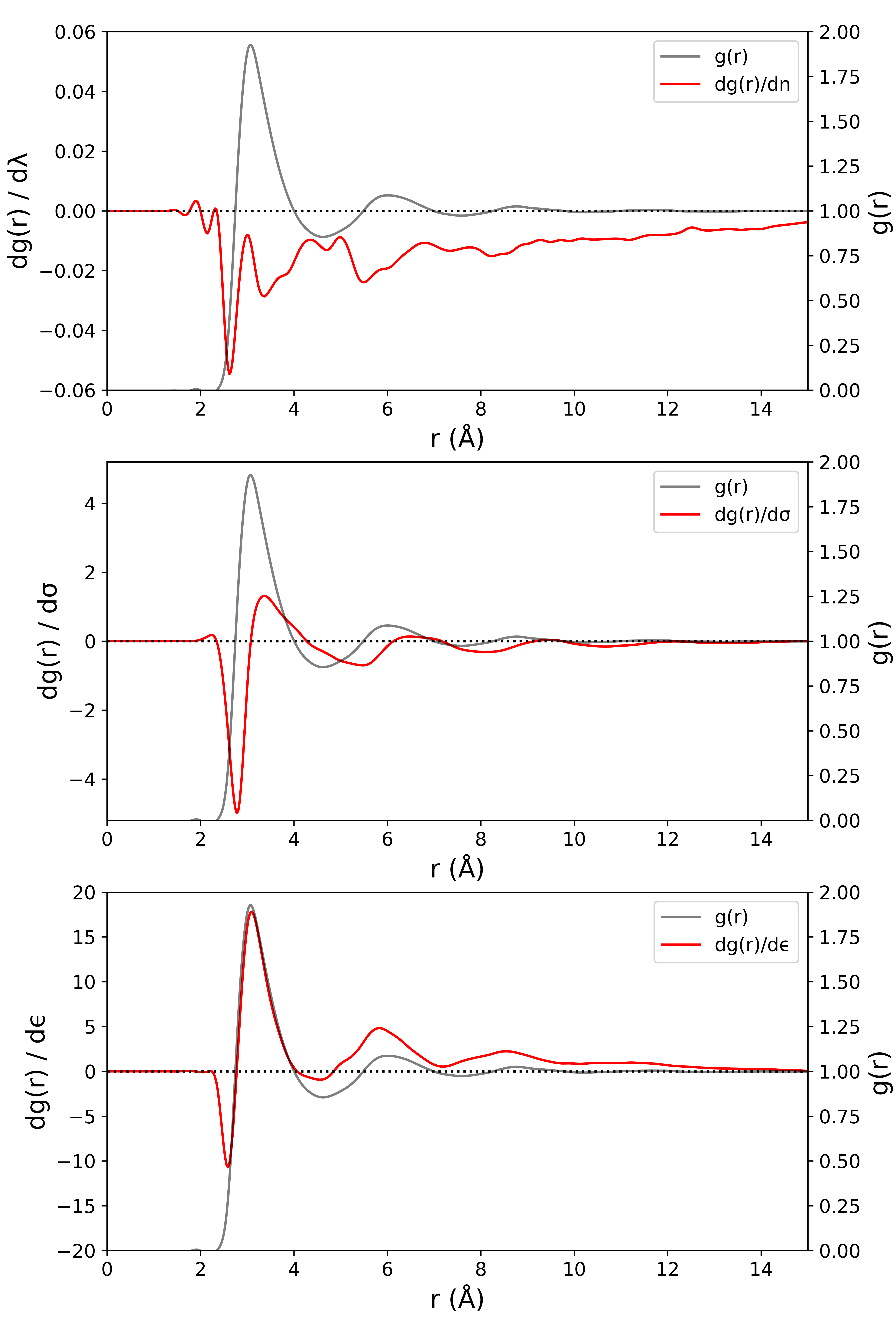}
    \caption{Derivatives of the local GP along the RDF calculated from Eq. \eqref{eq:derivs}.}
    \label{fig:derivs}
\end{figure}

The repulsive exponent derivative exhibits a small magnitude and has a minimum at the RDF half maximum. This behavior suggests that increasing the repulsive exponent, which determines the "hardness" of the particles, steepens the slope of the first peak in the RDF. This result is intuitive considering that in a hard-particle model there is a discontinuous jump at the hard-particle radius (infinite slope) that progressively softens with the introduction of an exponential repulsive decay function. In the case of the collision diameter, zeros of the derivative occur at RDF peaks and troughs, while local extrema align with the half-maximum positions. Consequently, increasing the effective particle size shifts the RDF to the right while maintaining relatively constant peak heights. Regarding the dispersion energy, its derivative displays zeros at the half-maximum positions of the RDF and local extrema at peaks and troughs. This behavior indicates that an increase in the dispersion energy leads to an increased magnitude of the RDF peaks and greater liquid structuring. 

Derivatives of structure with respect to thermodynamic state variables ($T$, $P$, $\mu$, etc) can be computed with fluctuation theory. Let's now take as an example the $\epsilon$-derivative of the RDF in Ne. We find that an increase in the dispersion energy deepens the interatomic potential well, resulting in greater attraction and a more structured liquid. Noting that the reduced temperature, $T^*$, is inversely related to $\epsilon$ by,

\begin{equation}
    T^* = \frac{k_B T}{\epsilon}
\end{equation}

\noindent then the $g(r)$ derivative with respect to $\epsilon$ at constant $T$, is equal to the $g(r)$ derivative with respect to the reduced thermodynamic beta,

\begin{equation}
    \frac{\partial g(r)}{\partial \epsilon} = \frac{\partial g(r)}{\partial \beta^*}
\end{equation}

\noindent where $\beta^* = T^*/k_BT$. In summary, an increase in $\epsilon$ is equivalent to a decrease in temperature. It is therefore expected that the $\epsilon$ derivative and temperature derivative behave the same; specifically, a decrease in temperature should increase result in greater fluid structuring without significantly impacting peak positions. Unsurprisingly, this behavior is exactly what was observed in recent work that computed temperature derivatives of the O-O pair RDF in water using a fluctuation theory approach \cite{piskulich_temperature_2020}.

\subsubsection{The Standard Bayesian Framework}

For simplicity of notation, let $\theta = \{\lambda, \sigma, \epsilon, \sigma_n\}$ represent the model parameters and $\mathcal{Y} = S_d(Q)$ be the RDF observations. The nuisance parameter, $\sigma_n$, represents the width of the Gaussian likelihood and is considered a model parameter since nothing is known about this parameter \textit{a priori}. Calculating the posterior probability distribution with Bayesian inference then requires two components: (1) prescription of prior distributions on the model parameters, $p(\theta)$, and (2) evaluation of the RDF likelihood, $p(\mathcal{Y}|\theta)$. The prior distribution over the $(\lambda-6)$ Mie parameters is assumed to be a multivariate normal distribution,

\begin{equation}
    \theta \sim \mathcal{N} (\mu_\theta, \sigma^2_{\theta})
\end{equation}

\noindent where $\mu_\theta$ and $\sigma^2_\theta$ are the prior mean and variance of each $(\lambda-6)$ Mie parameter in $\theta$, respectively. A wide, multivariate normal distribution was selected because it is non-informative and conjugate to the Gaussian likelihood equation. The prior on the nuisance parameter is assumed to be log-normal,

\begin{equation}
    \log \sigma_{n} \sim \mathcal{N} (\mu_{\sigma_n}, \sigma_{\sigma_n}^2)
\end{equation}

\noindent where $\mu_{\sigma_n}$ and $\sigma_{\sigma_n}$ are the prior mean and variance of the nuisance parameter. The log-normal prior imposes the constraint that the nuisance parameter is non-negative, which is obviously true because a negative variance in the observed data is undefined. For reference, the prior parameters used in this study are summarized in Table \ref{tab:priors}.

\begin{table}[H]
\centering
\caption{Prior parameters on the ($\lambda$-6) Mie model parameters.}
\begin{tabular}{| c | l | r r |}
\hline
\textrm{Parameter}&
\textrm{Distribution}&
\textrm{$\mu$}&
\textrm{$s$}\\
\hline
$\lambda$  &                   &  12.0    & 9\\
$\sigma$   & \text{Normal}     &  2.7     & 1.8\\
$\epsilon$ &                   &  0.112   & 0.225\\ \cline{1-4}
$\sigma_n$      & \text{Log-Normal} &  1       & 1\\
\hline
\end{tabular}
\label{tab:priors}
\end{table}

The likelihood function is assumed to be Gaussian according to the central limit theorem,

\begin{equation}
    p(\mathcal{Y}|\theta) \propto \frac{1}{\sigma_n^{n_{samples}}}\exp\bigg[-\frac{1}{2\sigma^2_{n}}\sum_i\ [{S}_{\theta_i}(Q_j) - S_d(Q_j)]^2\bigg]
\end{equation}

\noindent where ${S}_\theta(Q_i)$ is the molecular simulation predicted RDF and $j$ indexes over discrete points along the momentum vector. Bayes' theorem is then expressed as,

\begin{equation}\label{eq:inference}
    p(\theta|\mathcal{Y}) \propto p(\mathcal{Y}|\theta) p(\theta)
\end{equation}

\noindent where equivalence holds up to proportionality. This construction is acceptable since the resulting posterior distribution can be normalized \textit{post hoc} to find a valid probability distribution. 

\subsubsection{Markov Chain Monte Carlo}

To populate the Bayesian likelihood distribution, Markov Chain Monte Carlo (MCMC) samples over the model parameters $\theta = \{\lambda, \sigma, \epsilon, \sigma_n\}$  are passed to the surrogate model, evaluated, and compared to the experimental RDF. 960,000 MCMC samples were calculated using the emcee package \cite{foreman-mackey_emcee_2013} from 160 walkers (5000 samples/walker) with a 1000 sample burn-in per walker. The MCMC moves applied were differential evolution (DE) at a 0.8 ratio and DE Snooker at a 0.2 ratio, which is known to give good results for multimodal distributions. The acceptance ratio obtained from this sampling procedure was $\sim$0.27 and the autocorrelation between steps was 16 moves. 

\subsubsection{Python Notebook Tutorial}

An example code for creating a local Gaussian process surrogate model in python was created and posted on GitHub \hyperlink{https://github.com/hoepfnergroup/LGPMD}{here}. The notebook is reproduced as a pdf below and contains a detailed description of the code functions, basic theory, and results of the method. The example is for liquid neon for the same reasons as the tutorial written for the first chapter. The notebook covers how to construct a training dataset, design a kernel, include a Gaussian process prior mean, code the Gaussian process surrogate model, perform hyperparameter optimization, and validate the model with a root-mean-square-error test. 

\section{Impact of Experiment Uncertainty on Potential Estimation}

In the previous chapters, it has been shown that non-parameteric (Chapter 2) and parametric (Chapter 3) approaches can be used to learn uncertainty-aware force field parameters from liquid structure. However, one key question that still remains is whether or not experimental uncertainty significantly impacts the forces estimated from these approaches. Historically, giants in the field of liquid state theory including Levesque and Verlet \cite{levesque_note_1968} noted that systematic error in scattering data can drastically effect the force parameters that best represent a given scattering measurement, while Weeks, Chandler and Anderson concluded that the repulsive part of the interatomic forces was the key to matching structure factors in simple liquids \cite{weeks_role_1971}. However, a notable gap in the literature is the rigorous investigation of random noise in structure factor data and whether this uncertainty can hinder our ability to reconstruct local interatomic forces. Here we revisited the important question of experimental noise on force parameter reconstruction using the tools of Bayesian inference. The theme of the work is to see whether or not noise in experimental data can overwhelm our ability to reconstruct forces and identify how accurate the scattering data needs to be to prevent such pathological cases. Updated versions of this work can be accessed on arXiv with article identifier 2407.04839 here \hyperlink{https://arxiv.org/abs/2407.04839}{here} \cite{shanks_bayesian_2024} or as a published article at J. Phys. Chem. Lett. 2024, 15, 51, 12608–12618 \cite{shanks_bayesian_2024-1}.

\subsection{Abstract}

The inverse problem of statistical mechanics is an unsolved, century-old challenge to learn classical pair potentials directly from experimental scattering data. This problem was extensively investigated in the 20th century but was eventually eclipsed by standard methods of benchmarking pair potentials to macroscopic thermodynamic data. However, it is becoming increasingly clear that existing force field models fail to reliably reproduce fluid structures even in simple liquids, which can result in reduced transferability and substantial misrepresentations of thermophysical behavior and self-assembly. In this study, we revisited the structure inverse problem for a classical Mie fluid to determine to what extent experimental uncertainty in neutron scattering data influences the ability to recover classical pair potentials. Bayesian uncertainty quantification was used to show that structure factors with noise smaller than 0.005 barnes/steradian to $\sim30$ \AA$^{-1}$ are required to accurately recover pair potentials from neutron scattering. Notably, modern neutron instruments can achieve this precision to extract classical force models to within approximately $\pm$ 1.3 for the repulsive exponent, $\pm$ 0.068 \AA for atomic size, and 0.024 kcal/mol in the potential well-depth with 95\% confidence. Our results suggest the exciting possibility of improving molecular simulation accuracy through the incorporation of neutron scattering data, advancement in structural modeling, and extraction of model-independent measurements of local atomic forces in real fluids. 

\subsection{Introduction}

Reconstructing interatomic potentials from experimental scattering data is a historic inverse problem in statistical mechanics, motivated by the idea that complete knowledge of the effective interatomic potential with the atomic correlation functions allows for all thermodynamic properties of a classical liquid to be calculated \cite{kirkwood_statistical_1951}. While it has become widely accepted that liquid state systems exhibit significant many-body and quantum mechanical (both electronic and nuclear) interactions \cite{hansen_theory_2013} that influence molecular dynamics, the fact that empirical molecular simulations remain the gold-standard for efficient and accurate liquid state materials modeling has maintained the significance and impact of the inverse problem in contemporary physics. However, despite over a century of research, with seminal works by Ornstein and Zernike \cite{ornstein_accidental_1914,percus_analysis_1958,percus_approximation_1962}, Yvon, Born, and Green \cite{born_kinetic_1947}, Schommer \cite{schommers_pair_1983}, and Lyubartsev and Laaksonen \cite{lyubartsev_calculation_1995}, there is surprisingly little to no evidence that these techniques can reliably extract force field parameters from experimental scattering data \cite{toth_interactions_2007}. Furthermore, there is growing evidence that existing force fields provide inaccurate representations of fluid structure when compared to experimental estimates \cite{amann-winkel_x-ray_2016,fheaden_structures_2018}. With the advent of state-of-the-art diffractometers and the rise of machine learning and high-performance computing for robust uncertainty quantification, it is relevant to revisit and contextualize prior and current work to better understand how to resolve this longstanding challenge.

Serious attempts at determining the interatomic potential from experimental scattering data began in the 1950's. Henshaw (1958) \cite{henshaw_atomic_1958} and later Clayton and Heaton (1961) \cite{clayton_neutron_1961} speculated that the ratio between the atomic collision radius and first solvation shell radius was related to the approximate width of the interatomic potential bowl. While this concept cannot directly extract the interatomic potential from the radial distribution function, it was used to conclude that argon and krypton could be reasonably represented by a (12-6) Lennard-Jones potential. Weeks, Chandler, and Anderson then introduced a separation of the pair potential into repulsive and attractive parts, in which they concluded that the repulsive part alone produces structure factors nearly identical to the repulsive and attractive parts taken together \cite{weeks_role_1971}. Henderson (1974) then proved that for a pairwise additive and homogeneous system with equal radial distribution functions that the effective interatomic potential was unique up to an additive constant \cite{henderson_uniqueness_1974}, which was later implemented numerically by Schommers (1983) \cite{schommers_pair_1983} to study liquid gallium. Around the same time, Levesque (1985) \cite{levesque_pair_1985} proposed a modified hypernetted chain closure to the Ornstein-Zernike integral relation to calculate interatomic potentials for liquid aluminum with fast convergence. Both studies were highly influential in the study of liquid metals, but offered little in resolving the inverse problem in general since interatomic potentials derived from these methods were only shown to accurately reproduce the diffusion coefficient and not other thermodynamic properties. 

The most recent inverse methods applied to experimental data are Soper's (1996) \cite{soper_empirical_1996} empirical potential structure refinement (EPSR) and Lyabartsuv and Laaksonen's (1995) \cite{lyubartsev_calculation_1995,toth_determination_2001,toth_iterative_2003} inverse Monte Carlo (IMC). EPSR is an iterative potential refinement method that is primarily used to determine real-space structures consistent with reciprocal space scattering data in fluid and glass systems. However, Soper's work on liquid water revealed that EPSR could not be reliably implemented to determine pair interaction potentials for molecular simulation applications \cite{soper_tests_2001}. On the other hand, IMC methods have been widely adopted for coarse-graining, in which the number of degrees-of-freedom of a molecular model are reduced by mapping atomic coordinates to "beads" of atom clusters. While both methods have attracted significant research interest in recent years, with the creation of an improved EPSR software package \cite{youngs_dissolve:_2019} and applications of IMC in complex biological systems \cite{lyubartsev_systematic_2010} such as DNA \cite{sun_bottom-up_2021} and nucleosomes \cite{sun_bottom-up_2022}, the extraction of reliable and transferable interatomic potentials from experimental scattering data remains widely under-reported and unresolved, even for simple fluids such as noble gases. 

We recently proposed structure-optimized potential refinement (SOPR) \cite{shanks_transferable_2022} as an alternative approach to extract pair potentials from scattering data (2022). SOPR is a probabilistic iterative Boltzmann inversion (IBI) algorithm that uses Gaussian process regression to address challenges such as numerical instability and over-fitting to uncertain experimental data. SOPR derived potentials have demonstrated remarkable accuracy in predicting both the structural correlation functions and vapor-liquid equilibria of noble gases. Furthermore, the short-range repulsive decay rate determined with SOPR coincides with predictions from an independently optimized ($\lambda-6$) Mie force field for vapor-liquid equilibria \cite{mick_optimized_2015}. This finding represents the most complete example of the inverse problem in real systems and highlights the potential of scattering data in studying macroscopic thermophysical properties of liquids.    

The transferability of SOPR potentials raises intriguing questions regarding what factors are most important to accurately extract local forces from scattering data. One possible explanation is that the reliability of structure inversion techniques hinges on the quality of the experimental scattering data \cite{soper_uniqueness_2007}. Levesque and Verlet (1968) speculated that experimental scattering error of $<1\%$ was required to determine the interaction potential within an error of $10\%$ \cite{levesque_note_1968}, but ultimately concluded that it is not possible to obtain quantitative information on the potential from scattering data due to systematic error in the experiments. However, it is currently unclear to what extent these previous attempts have been impeded by experimental uncertainty. The reason for this knowledge gap is that rigorous uncertainty quantification and propagation (UQ/P) is computationally demanding and requires the use of machine learning surrogate models and advanced sampling methods \cite{angelikopoulos_bayesian_2012,kulakova_experimental_2017,shanks_accelerated_2024} that were not available to liquid state theorists when these questions were first investigated. To test the hypothesis that neutron instrument accuracy is essential in force field extraction therefore lies at the crossroads of theoretical statistical physics, machine learning, and high-performance computing. 

Here we assess how scattering measurement uncertainty impacts our ability to learn interatomic forces using a dataset of \textit{in silico} experimental structure factors with varying levels of noise. Bayesian optimization with a local Gaussian process (LGP) surrogate model was then applied to extract the underlying probability distributions on the force field parameters. Gaussian noise was introduced to a reduced Mie model structure factor with standard deviation $\delta S$, corresponding to data collected on various neutron instruments from 1973-2022. Constant noise at six different standard deviations, consistent with a reactor source neutron instrument \cite{willis_experimental_2017}, spanning from low to high uncertainty was added to a model structure factor (Figure \ref{fig:corrupt}). 

\begin{figure}[H]
    \centering
    \includegraphics[width = 12 cm]{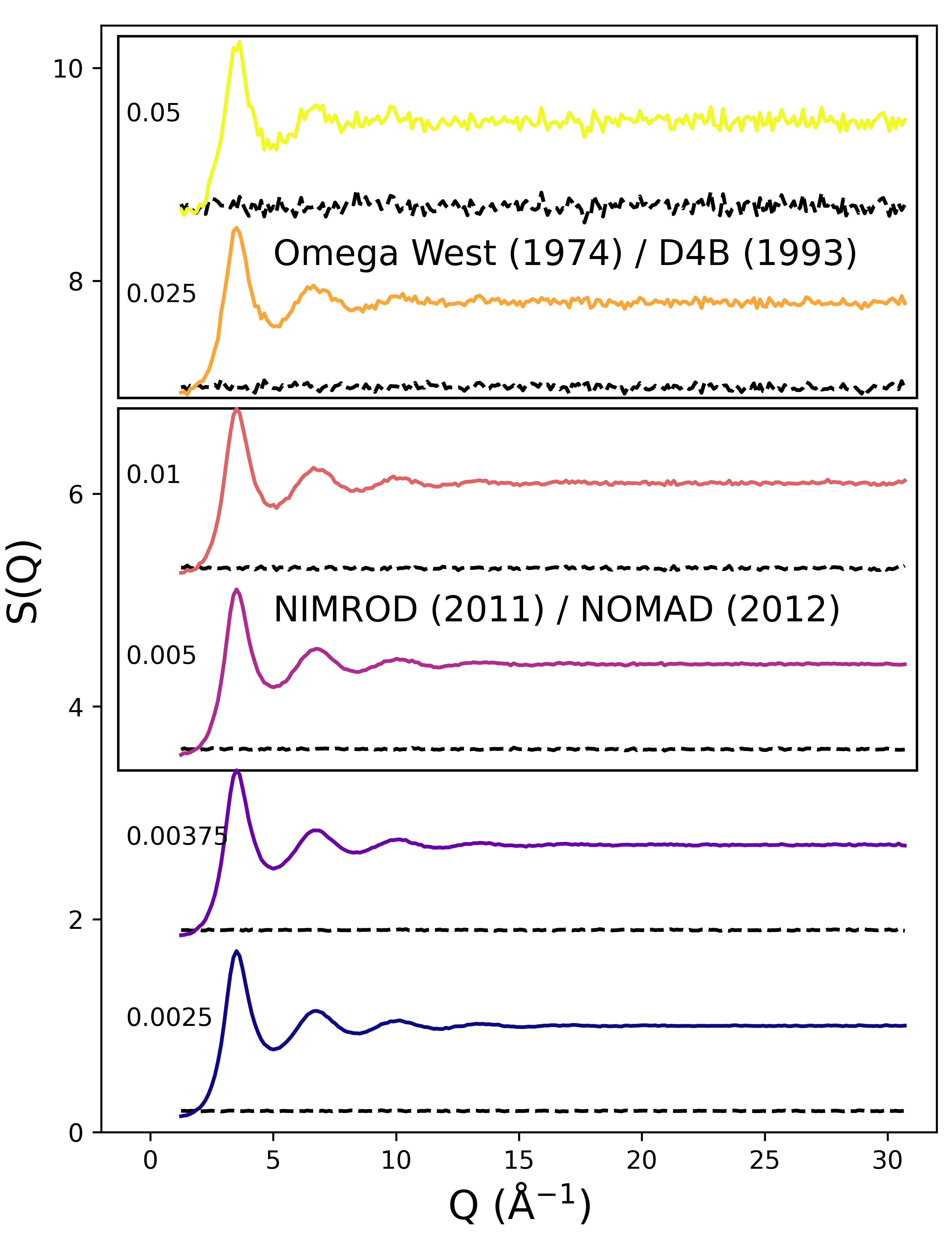}
    \caption{Static structure factors (colored lines) with introduced uncertainty (dotted black lines) for uniformly distributed noise. Measurement standard deviations $\delta S$ are labeled to the left of the structure factor.}
    \label{fig:corrupt}
\end{figure}

By studying the parameter posterior distributions as a function of introduced uncertainty, we aimed to challenge the assertion that structure factors are insensitive to the detailed form of interatomic interactions \cite{jovari_neutron_1999}. Using the ($\lambda$-6) Mie parameter Bayesian posterior distributions, we quantify how interatomic interactions such as short-range repulsion, excluded volume, and dispersion energy affect measured structure factors, shedding light on the intricate relationship between pair potentials and structural features. Surprisingly, we find that the conclusions from prior literature stating that details of the interatomic interaction could not be extracted from experimental structure factors were likely justified given the data quality available at the time, but that modern neutron instruments exceed a precision threshold where this conclusion could be overturned. These findings suggest that experimental inverse techniques were prematurely abandoned and should be revisited.

According to our results, neutron scattering measurements determined within a standard deviation of $0.005$ barnes/steradian to a $Q_{max}\sim 30$ \AA$^{-1}$ are sufficient for force parameter recovery. Fortunately, this level of precision is already available at modern diffractometers, such as the Nanoscale Ordered MAterials Diffractometer (NOMAD) \cite{neuefeind_nanoscale_2012} or at other modern instruments for sufficiently long run times. Method advancements in structure inversion, along with the improvement of neutron facilities and measurement accuracy, may therefore be the key to unlock a wealth of opportunities for improving molecular models, characterizing local atomic forces, and understanding the dynamics of atoms and molecules in relation to complex and emergent physical phenomena. 

\subsection{Computational Methods}

In this study, we aimed to model how uncertainty propagates from neutron scattering data to the estimation of force field parameters. The impact of measurement uncertainty was isolated by constraining the Bayesian analysis to a classical model fluid. While real physical systems behave quantum mechanically and are inherently many-body in nature, classical pairwise additive model fluids continue to be studied due their low computational cost and accurate predictions of complex thermodynamic properties. Furthermore, our prior work has shown that SOPR potentials exhibit potential corrections consistent with quantum mechanical calculations \cite{shanks_transferable_2022}, suggesting that effective pair interactions could be found that capture many-body and quantum mechanical contributions.

The ($\lambda$-6) Mie fluid model was selected since it is a flexible and widely successful classical model with numerous existing and developing applications for materials modeling. The pairwise, non-bonded potential energy term of the ($\lambda$-6) Mie fluid is,

\begin{equation}
    v^{Mie}_2(r) = \frac{\lambda}{\lambda-6}\bigg(\frac{\lambda}{6}\bigg)^{\frac{6}{\lambda-6}} \epsilon \bigg[ \bigg(\frac{\sigma}{r}\bigg)^\lambda - \bigg(\frac{\sigma}{r}\bigg)^6 \bigg]
\end{equation}

\noindent where $\lambda$ is the short-range repulsion exponent, $\sigma$ is the collision diameter (distance), and $\epsilon$ is the dispersion energy (energy) \cite{mie_zur_1903}. The ($\lambda$-6) Mie fluid is a flexible model potential that has been shown to accurately reproduce thermodynamic properties in real fluids \cite{mick_optimized_2015,mick_optimized_2017}. 

\subsubsection{Modeling Neutron Measurement Uncertainty in a Mie Fluid Model}

To model experimental uncertainty, a set of Mie fluids was simulated with sufficient sampling statistics to calculate a highly-accurate static structure factor ($\delta S(Q) < 0.001$) to $Q_{max} = \sim 30$ \AA$^{-1}$). Computer generated atomic trajectories were calculated in HOOMD-Blue \cite{anderson_hoomd-blue_2020}. MD simulations were initiated with a random configuration of 500 particles at reduced density $\rho^* = 0.1$ and reduced temperature $T^* = 1$ and equilibrated with Langevin dynamics for $1\times10^5$ timesteps ($dt=10$ femtoseconds). Potentials were truncated at $3\sigma$ with an analytical tail correction, and radial distribution functions were calculated with Freud \cite{ramasubramani_freud_2020}. Static structure factors are calculated via radial Fourier transform of the radial distribution function. 

Experimental measurements of structure factors are subject to uncertainty arising from various factors, including experimental, model, and numerical sources. Uncertainties in neutron flux, energy, time-of-flight, minimum and maximum momentum transfer ($Q_{min}$, $Q_{max}$), and data collection time contribute to uncertainty in neutron counting statistics and the effective resolution of the instrument. Post-processing corrections for inelastic, incoherent, multiple, and low momentum transfer scattering further contribute to uncertainty in the structure factor form \cite{soper_inelasticity_2009}. These effects result in variations in the neutron intensity that are not necessarily normally distributed \cite{barron_entropy_1986}; however, errors from neutron detection are normal due to the limiting behavior of the Poisson distribution for large number of counts \cite{heybrock_systematic_2023}. For reactor source instruments, the variance due to these random errors remains approximately constant to a limited $Q$-max (10-20 \AA$^{-1}$), while for spallation sources, the variance increases proportionally to the square of the momentum transfer to a higher $Q$-max of 50-125 \AA$^{-1}$ \cite{neuefeind_nanoscale_2012}. Currently, the extent to which this uncertainty influences the accuracy and reliability of force field reconstructions remains unknown.  

Uncertainty quantification was performed for reactor type neutron instruments by adding constant noise to the true structure factor signal with variations approximately equal to the values indicated in Figure \ref{fig:corrupt}. Bayesian analysis was performed for 16 \textit{in silico} experimental conditions over a $4 \times 4$ equal spaced grid with $\sigma = [1.85, 1.89, 1.93, 1.97]$, $\epsilon = [0.86, 0.80, 0.74, 0.70]$ and fixed $\lambda=12$. Since spallation type neutron instruments have a $Q^2$-dependent random error, UQ on the constant error can be interpreted as an upper bound for spallation instruments. Uncertainty levels were selected based on published data of structure factors measured on neutron instruments from the early 1970's to 2022. Notably, a classic argon data set collected on Omega West (1973) \cite{yarnell_structure_1973} is well approximated by constant noise with variations in $S(Q)$ of approximately $\delta S = 0.05$, as noted in Figure \ref{fig:corrupt}. Similarly, krypton data collected on D4B (1993)\cite{barocchi_neutron_1993} is approximated by the constant $\delta S = 0.025$ case, while modern instruments such as NOMAD and NIMROD (>2010) exhibit uncertainty distributions similar to the $\delta S = 0.005-0.01$ cases depending on the data collection time \cite{bowron_nimrod_2010,neuefeind_nanoscale_2012}. Two low uncertainty extremes were chosen beyond these reported values to identify measurement precision thresholds and model trends in the predictability of force parameters. Finally, to reduce the possibility of over-fitting to a single randomly generated structure factor, we used four structure factor replicates with twice the uncertainty of the reported value. This approach is consistent with the well-known relation that the standard deviation in the neutron counts is approximately proportional to the square-root of the number of counts \cite{heybrock_systematic_2023}. Assuming that neutron count is directly proportional to time, four measurements of twice the standard deviation would be approximately equivalent to one measurement with half the standard deviation.  

\subsubsection{Bayesian Uncertainty Quantification for Force Field Reconstruction}

According to the Henderson inverse theorem, it is theoretically possible to uniquely recover the underlying potential in a pairwise additive, homogeneous fluid \cite{henderson_uniqueness_1974}. In the context of Bayesian optimization, Henderson's theorem requires that there should be a global maximum in the posterior probability distribution at the unique force parameters. However, as the uncertainty in the structure factor signal increases, deviations from this unique potential are expected, causing the probability distributions to broaden. This broadening indicates a decrease in confidence in the estimation of the model parameters. In other words, as structural uncertainty increases, our ability to accurately predict the potential energy decreases, leading to a wider range of possible parameter values to explain the data. 

Bayesian inference was implemented to calculate parameter probability distributions as a function of structure factor uncertainty. For simplicity of notation, let $\boldsymbol{\theta} = \{\lambda, \sigma, \epsilon, \sigma_n\}$ represent the model parameters and $\mathcal{Y} = S_d(Q)$ be the structure factor observations. The nuisance parameter, $\sigma_n$, represents the width of the Gaussian likelihood and captures uncertainty from the experimental data and Gaussian process model, which is not known \textit{a priori}. Calculating the posterior probability distribution with Bayesian inference then requires two components: (1) prescription of prior distributions on the model parameters, $p(\theta)$, and (2) evaluation of the structure factor likelihood, $p(\mathcal{Y}|\boldsymbol{\theta})$. The prior distribution over the $(\lambda-6)$ Mie parameters is assumed to be a multivariate log-normal distribution,

\begin{equation}
     \theta - \gamma_\theta \sim \log \mathcal{N} (\mu_{\boldsymbol{\theta}} + \gamma_\theta, \sigma^2_{\boldsymbol{\theta}})
\end{equation}

\noindent where $\mu_{\boldsymbol{\theta}}$ and $\sigma^2_{\theta}$ are the prior mean and variance of each parameter in $\boldsymbol{\theta}$ and $\gamma_\theta \in \mathbb{R}$ is a real-valued parameter shift that enforces a lower bound. A wide, shifted multivariate log-normal distribution was selected because it is non-informative and imposes non-negativity constraints on the model parameters. Specifically, $\lambda-6$ (defined by the Mie type fluid), $\sigma$, $\epsilon$, and $\sigma_n$ must be positive. For reference, the prior parameters used in this study and sample range is summarized in Table \ref{tab:priors1}. Wide and flat prior distributions were selected to stress the assessment of the uncertainty contributions from the introduced structure factor noise.

The likelihood function is a Poisson distribution of the neutron counts, but we can approximate this distribution as a Gaussian because a Poisson distribution approaches a Gaussian distribution in the high count limit,

\begin{equation}\label{eq:gausslikelihood}
    p(\mathcal{Y}|\boldsymbol{\theta}) = \bigg(\frac{1}{\sqrt{2 \pi}\sigma_n}\bigg)^\eta \exp\bigg[-\frac{1}{2\sigma^2_{n}}\sum_j\ [{S}_{\boldsymbol{\theta}_i}(Q_j) - S_d(Q_j)]^2\bigg]
\end{equation}

\noindent where ${S}_{\boldsymbol{\theta}}(Q_i)$ is the molecular simulation predicted structure factor, $\eta$ is the number of observed points in the structure factor, and $j$ indexes over these points along the momentum vector. Bayes' theorem is then expressed as,

\begin{equation}\label{eq:inference1}
    p(\boldsymbol{\theta}|\mathcal{Y}) \propto p(\mathcal{Y}|\boldsymbol{\theta}) p(\mathbf{\boldsymbol{\theta}})
\end{equation}

\noindent where equivalence holds up to proportionality. This construction is acceptable since the resulting posterior distribution can be normalized \textit{post hoc} to find a valid probability distribution. For further details see the following excellent reviews of Bayesian inference \cite{gelman_bayesian_1995,bishop_pattern_2006}. 

The Bayesian likelihood distribution is estimated using Markov Chain Monte Carlo (MCMC) samples over the model parameters $\theta = \{\lambda, \sigma, \epsilon, \sigma_n\}$. Computationally, a sample of the model parameters is drawn from a Metropolis-Hastings type algorithm, passed to the surrogate model, evaluated, and compared to the \textit{in silico} structure factor. MCMC samples were calculated using the emcee package \cite{foreman-mackey_emcee_2013} from 160 walkers with dynamic burn-in and sample time based on the autocorrelation convergence criterion used in the emcee package and default stretch move with dynamic tuning.

\begin{table}[H]
\centering
\caption{\label{tab:priors1}
Prior parameters on the ($\lambda$-6) Mie model parameters.}
\begin{tabular}{| c | r | r | r |}
\hline
\textrm{Parameter}&
\textrm{$\mu$}&
\textrm{$\sigma (std.)$}&
\textrm{$\gamma_\theta$}\\
\hline
$\lambda$  &  3    & 1 & 6 \\
$\sigma$   &  2    & 1 & 0\\
$\epsilon$ &  0.7  & 1 & 0\\
$\sigma_n$ &  0.1  & 3 & 0\\
\hline
\end{tabular}
\end{table}

\subsubsection{Local Gaussian Process Surrogate Models for Structure Factors}

The process of populating the posterior distribution function necessitates the evaluation of likelihood for each condition of interest within the model parameter space, which, in turn, requires conducting an infeasible number of molecular dynamics simulations. The computational burden associated with this procedure renders the Bayesian framework impractical even for a relatively small number of samples. To illustrate this, consider the task of obtaining a collision diameter posterior distribution with a grid resolution of approximately 2\% across a wide prior range of $m_{\sigma}$ (ranging from 0.5 to 1.5). Achieving such resolution for just the $m_{\sigma}$ parameter alone would demand a minimum of 50 samples. If the same level of resolution is desired for the remaining two parameters, a staggering 125,000 molecular simulations are required to comprehensively quantify the posterior distribution space. Clearly, there is a substantial computational challenge involved in obtaining accurate and comprehensive posterior distributions within the Bayesian framework.   

To expedite the evaluation of the Bayesian likelihood, a local Gaussian process (LGP) surrogate model was trained to generate structure factors based on a training set of 960 randomly sampled ($\lambda$-6) Mie parameters in a prior range specified in Table \ref{tab:ranges1} \cite{shanks_accelerated_2024}. This range of parameters was selected to correspond with the liquid phase region of the Mie phase diagram and avoid pathological simulations that can occur near phase transitions. Note that this range will change based on the arbitrary choice of temperature and density for the Mie fluid simulation; but, since the reduced phase diagram is simply scaled from these values, the molecular dynamics simulation will have the same dynamics, average thermodynamic properties and structure.

\begin{table}[H]
\centering
\caption{Estimated boundaries for physics-constrained prior space based on the ($\lambda$ - 6) Mie fluid phase diagram. $m = 6$ is the attractive tail exponent of the ($\lambda$ - 6) Mie potential. $*$) The maximum $\lambda$ was selected to be substantially larger than previously reported values \cite{vrabec_set_2001,mick_optimized_2015,shanks_transferable_2022}.}
\begin{tabular}{| c | l | l | l | l |}
\hline
\textrm{Param.}&
\textrm{Min.}&
\textrm{Min. Criteria}&
\textrm{Max.}&
\textrm{Max. Criteria}
\\
\hline
$\lambda$  & 6.1  & $m=6 \implies \lambda>6$  &  18 & Literature$^*$ \\
$\sigma$   & 1.5 & Vapor-Liquid Equil. &  2 & Solid-Liquid Equil.  \\
$\epsilon$ & 0.2 & $\epsilon<0$ undefined  &  1.1 & Vapor-Solid Equil. \\
\hline
\end{tabular}
\label{tab:ranges1}
\end{table}

LGP surrogate models reduce the computational time complexity of standard GP regression with little loss in predictive accuracy \cite{das_block-gp_2010,chen_parallel_2013,gramacy_local_2015}. Hyperparameter selection was performed using leave-one-out marginal likelihood Bayesian optimization \cite{sundararajan_predictive_2001} with Markov chain Monte Carlo (Figure \ref{fig:hyperparams}) and surrogate model accuracy (Figure \ref{fig:rmse1}) was determined to have a root-mean-square error (RMSE) of 0.0036 for a 160 randomly sampled test set within the surrogate parameter range.

\begin{figure}
    \centering
    \includegraphics[width = 15cm]{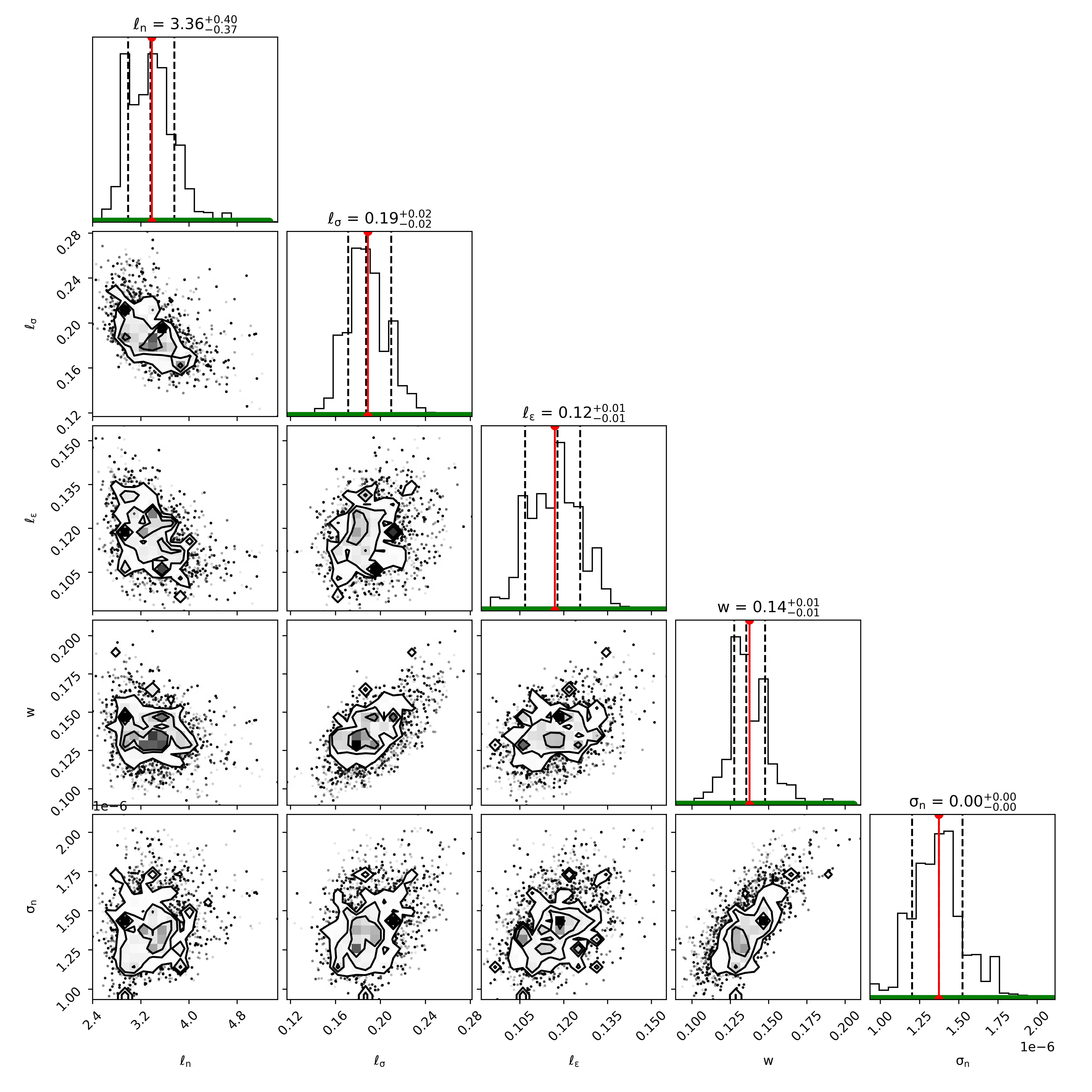}
    \caption{Bayesian posterior distribution for LGP hyperparameters according to the leave-one-out marginal likelihood.}
    \label{fig:hyperparams}
\end{figure}

\begin{figure}
    \centering
    \includegraphics[width = 15cm]{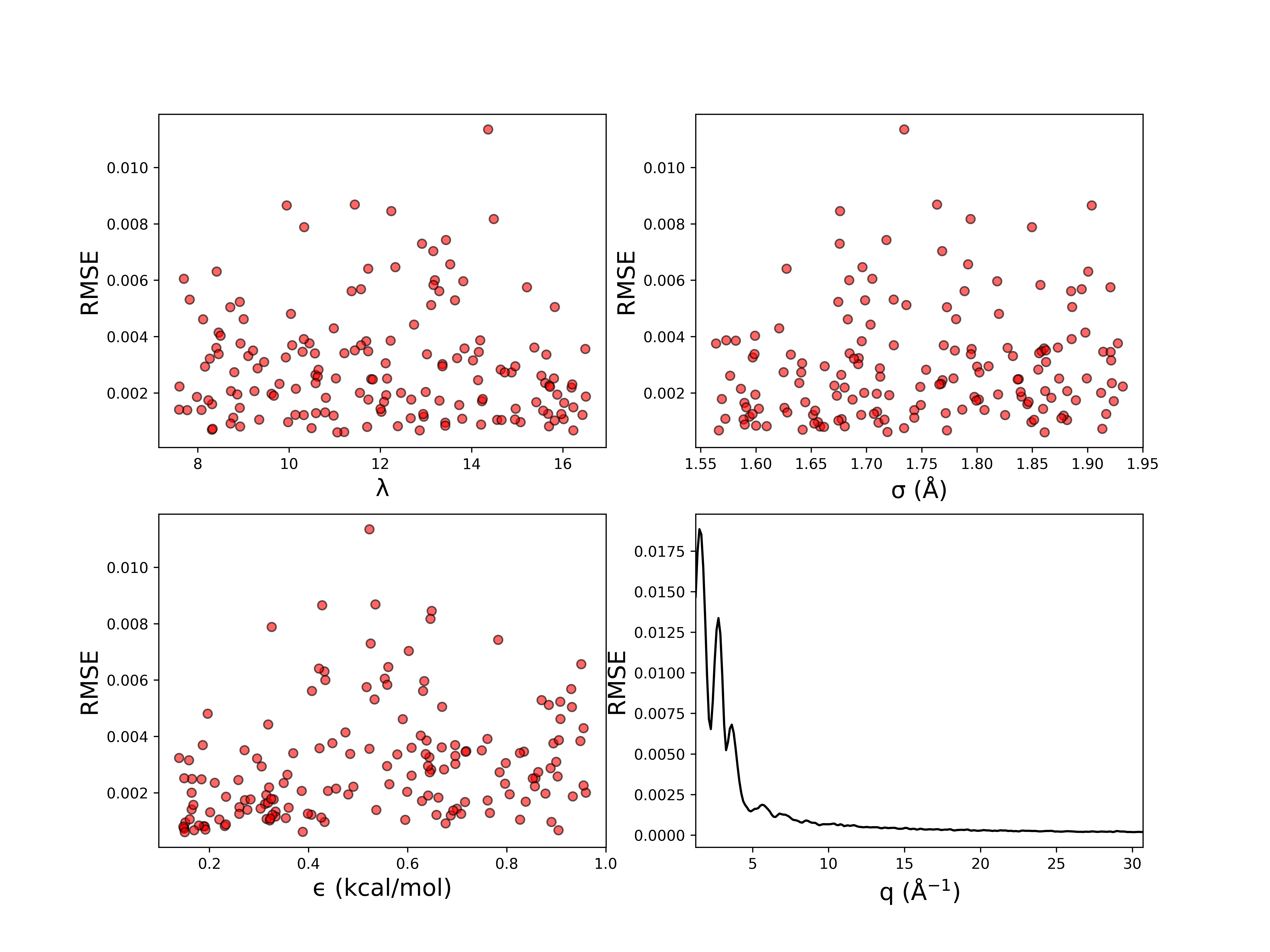}
    \caption{Root-mean-square error between LGP surrogate model prediction and training set points.}
    \label{fig:rmse1}
\end{figure}

\subsection{Results and Discussion}

Bayesian analysis was performed on structure factors with introduced uncertainty at 16 reference conditions. The aim was to determine the requisite threshold precision for extracting force field parameters from structure factors with a high level of fidelity. This threshold should allow for a clear choice of force parameters, exemplified by a sharp and unimodal distribution.  

\subsubsection{Sensitivity Analysis with Local Gaussian Process Derivatives}

The LGP surrogate model allows us to quantify the impact of varying ($\lambda$-6) Mie parameters on the structure at specific $Q$ values. One notable advantage of GPs is that they possess an analytical derivative, enabling the estimation of quantities such as $\frac{\partial S(Q)}{\partial \theta}$, where $\theta$ is some force model parameter ($\lambda, \sigma, \epsilon$) excluding the nuisance parameter $\sigma_n$. By examining the zeros and extrema of the LGP derivative, we can identify regions within the structure factor that are least and most affected by changes in the force parameters. This analysis provides valuable insight into the relationship between force parameters and the structure factor.

Figure \ref{fig:derivative} illustrates the derivatives of the surrogate model-predicted structure factor with respect to each parameter of the ($\lambda-6$) Mie force field. We find that even small changes in the Mie parameters results in substantial modification of the structure factor patterns. Consequently, provided experimental scattering results meet a necessary accuracy threshold, all three of the Mie model parameters (short-range repulsion, size, and dispersive attraction) could be extracted. Changes to the structure factor; however, do not impact all force parameters equally. The repulsive exponent derivative exhibits a small magnitude and undergoes sign changes near the full-width half maximum of the structure factor peaks. This behavior suggests that increasing the repulsive exponent, which determines the "hardness" of the particles, causes a slight increase in height and narrowing of the structure factor peaks without significantly affecting their location. In the case of the collision diameter, zeros of the derivative occur at structure factor peaks and troughs, while local extrema align with the half-maximum positions. Consequently, increasing the effective particle size shifts the structure factor towards lower $Q$ values while maintaining relatively constant peak heights. Regarding the dispersion energy, its derivative displays zeros at the half-maximum positions of the structure factor and local extrema near the peaks and troughs. This behavior indicates that an increase in the dispersion energy leads to an increased magnitude and sharpening of the structure factor peaks similar to the repulsive exponent.

\begin{figure}
    \centering
    \includegraphics[width = 15cm]{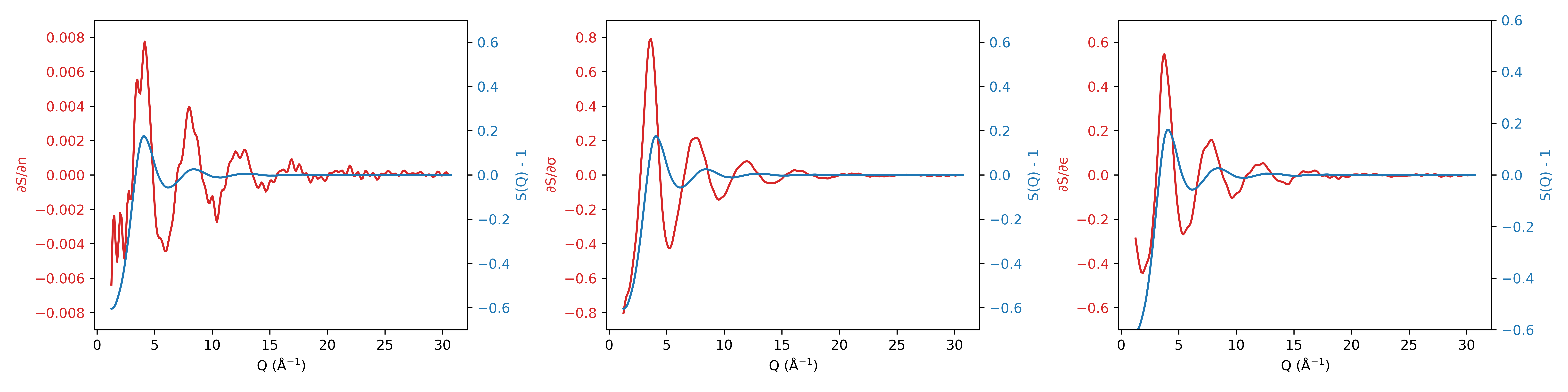}
    \caption{Derivatives of the surrogate model-predicted structure factor, S(Q), with respect to each parameter of the ($\lambda-6$) Mie force field (red line) plotted with the given structure factor (blue line).}
    \label{fig:derivative}
\end{figure}

\begin{figure}
    \centering
    \includegraphics[width = 15cm]{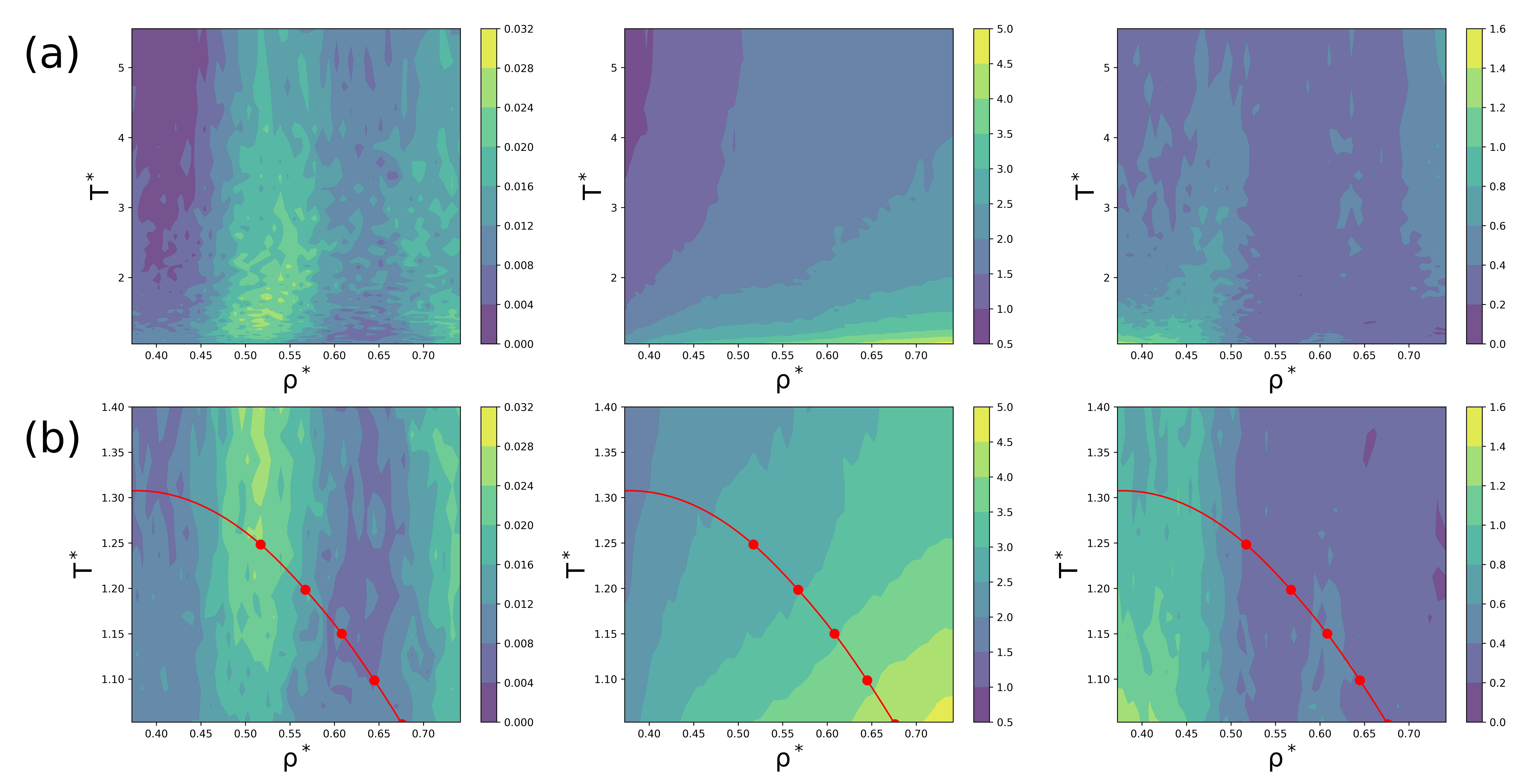}
    \caption{(a) Heat map of the maximum absolute value of the structure factor derivative with respect to each parameter of the ($\lambda-6$) Mie force field for the validated range of the surrogate model and (b) near the vapor-liquid coexistence line (red line).}
    \label{fig:derphasediag}
\end{figure}

Computation of the LGP derivative can also be performed over the entire validated range of the surrogate model. To visualize the results, we present a heat map of the maximum of the absolute value of the derivative with respect to reduced temperature $T^* = T/\epsilon$ and reduced density $\rho^* = \rho \sigma^3$ in Figure \ref{fig:derphasediag}. Higher values (yellow regions) indicate a high sensitivity of the structure factor relative to lower values (blue regions). First note that the maximum derivative estimates vary in magnitude significantly, with a two orders-of-magnitude smaller value for the repulsive exponent (0.032) compared to the collision diameter (5.0) and dispersion energy (1.6). The repulsive exponent $\lambda$ exhibits biomodality as a function of $\rho^*$, with higher sensitivity regions tending towards higher densities. This behavior could be explained by the fact that high density systems tend to collide more frequently at close range which is where the repulsive exponent strongly influences the potential energy function. The collision diameter $\sigma$ has a clear sensitivity trend with more sensitive regions being higher density and lower temperature. Observing the sensitivity over the full range (Figure \ref{fig:derphasediag}a), we can see that there also appears to be asymptotic behavior near specific densities, suggesting that the higher sensitivities correspond to closer proximity of the atoms where excluded volume effects can dominate the structure. Finally, the dispersive forces appear to become more significant to the structure within the vapor-liquid phase envelope, which makes sense considering that vapor-liquid equilibrium involves the interplay between cohesive attraction holding together the liquid phase and the equality of chemical potential driving the formation of the gas phase. The trend of higher sensitivity of the dispersive attraction towards lower densities could also be due to the fact that attractive forces play a more significant role in this region of the phase diagram in comparison to excluded volume effects. Of course, the validated range of the surrogate model is limited by the computational cost of running training simulations and further investigation of Mie parameter sensitivity to other regions of the phase diagram is reserved for future study.

\subsubsection{Force Field Parameter Posterior Distributions as a Function of Uncertainty}

Armed with a precise and rapid surrogate model for the fluid structure factor, we can now proceed to evaluate the likelihood function and, consequently, derive Bayesian posterior distributions. Figure \ref{fig:marginals} illustrates experiment averaged marginal probability distributions of $(\lambda, \sigma, \epsilon)$ as a function of noise. Note that these distributions represent an average over MCMC samples from the 16 \textit{in silico} experiments. Conceptually, these posterior distributions are an approximation to the marginal parameter posterior distributions over the joint probability density containing the model parameters, structure factor, and thermodynamic state, $p(\lambda,\sigma,\epsilon,S(Q), T^*, \rho^*)$, where $T^* = k_bT/\epsilon$ and $\rho^* = \rho \sigma^3$ are the reduced temperature and density, respectively.

\begin{figure}
    \centering
    \includegraphics[width = 15 cm]{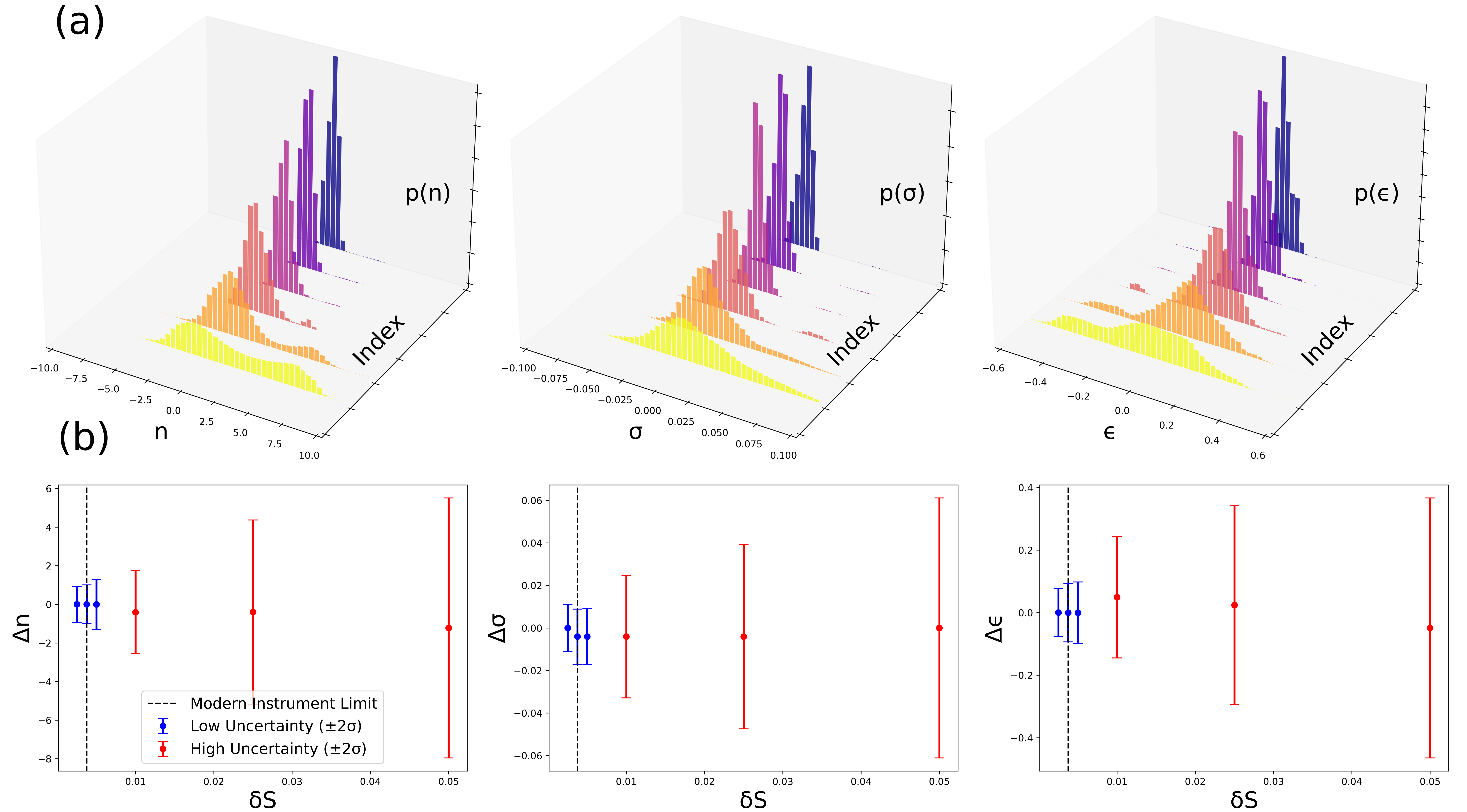}
    \caption{Mie parameter veraged marginal posterior distributions at varying structure factor noise levels. (a) 3D histogram of the average marginal distributions on the ($\lambda-6$) Mie force field parameters as a function of uncertainty in the structure factor ($\delta S$). (b) MAP estimates (points) are plotted with 2 std. dev. error bars as a function of noise. Low parameter uncertainty cases (blue) are compared to high uncertainty cases (red) and the lower limit precision of current neutron instruments (black dashed line).}
    \label{fig:marginals}
\end{figure}

The 1D Marginal distributions in Figure \ref{fig:marginals} are obtained by integrating the joint posterior probability distribution over all but one model parameter. The mode of the marginal distribution corresponds to the marginal \textit{maximum a posteriori} (\textit{MAP}). It is worth noting that as the uncertainty in the structure factor increases, the marginal distributions become wider and flatter. This behavior is expected, as greater uncertainty in the observation leads to increased uncertainty in the parameter distribution. In cases where the structure factors exhibit low uncertainty, the \textit{MAP} estimates accurately recover the unknown force field parameters. Deviation between the \textit{MAP} estimate and true parameter value is calculated as a function of uncertainty and presented in Table \ref{tab:maps}. 

First, note the drastic difference in the accuracy of the \textit{MAP} estimates for the repulsive exponent and dispersion energy parameters as we transition from an uncertainty level of $\delta S = 0.025$ to $\delta S = 0.05$. The $\sigma$ parameter is accurately estimated in both scenarios, demonstrating its reliable prediction even for low quality scattering data. In the $\delta S = 0.025$ case, the $\lambda$ and $\epsilon$ parameters are also accurately predicted. However, for the $\delta S = 0.05$ case, the $\lambda$ and $\epsilon$ parameters become unlearnable with \textit{MAP} deviations of -1.173 and -0.05, respectively. This large deviation clearly indicates that the data quality is too low to reliably extract subtle details of the potential energy function, which recent work has shown to be imperative to model thermodynamic behavior at extreme pressures \cite{messerly_uncertainty_2018}. This observation is critical in the context of prior studies in which it has been concluded that the structure factor is insensitive to the interatomic interactions beyond the excluded volume \cite{jovari_neutron_1999,hansen_theory_2013} or that uncertainty in the structure measurement impeded prediction of transferable potentials \cite{schommers_pair_1983,levesque_pair_1985,lyubartsev_calculation_1995}. In these studies, the instrument uncertainty ranged from values of 0.03-0.07, exceeding the threshold identified by our model. Barocchi's (1993) study on liquid krypton was unique in the conclusion that the neutron instrument accuracy was now high enough to elucidate detailed many-body interactions \cite{barocchi_neutron_1993}, which was concluded based on structure factors measured to a precision $\leq 0.025$. Based on our structure factor precision model, we find that these conclusions were likely appropriate for the data quality available.

The data also shows a significant change in the width of the distributions at critical uncertainty levels. The standard deviation effectively doubles for the Mie parameters between $\delta S = 0.005$ and $\delta S = 0.01$. This rapid increase in width of the posterior distribution is significant since it becomes exceedingly more difficult to estimate the potential parameters using optimization techniques. Based on these shifts in the standard deviations, we recommend that neutron scattering experiments for liquids not exceed random errors of $\delta S = 0.005$ out to $\sim$ 30 \AA$^{-1}$ if attempting to extract pair potential information from the structure factor. This level of precision is achievable on modern instruments but may require longer run times than standard neutron scattering measurements.

\begin{table}
\centering
\caption{\label{tab:maps} Error in ($\lambda-6$) Mie force field parameters determined by Bayesian inference on the structure factor. $\Delta p$ is the difference between the \textit{MAP} estimate and the underlying parameter set.}
\begin{tabular}{| c | r r r r |}
\hline
\textrm{Comparable Neutron Instrument}&
\textrm{$\delta S$}&
\textrm{$\Delta \lambda \pm 2\sigma_\lambda$}&
\textrm{$\Delta \sigma \pm 2\sigma_\sigma$}&
\textrm{$\Delta \epsilon \pm 2\sigma_\epsilon$}\\
\hline
-                                               & 0.0025  & -0.051 $\pm$ 0.9 &  0.000  $\pm$ 0.01 &  0.00 $\pm$ 0.07 \\
-                                               & 0.00375 & -0.255 $\pm$ 1.0 & -0.002  $\pm$ 0.01 &  0.00 $\pm$ 0.09 \\
NOMAD (2012)     & 0.005   &  0.051 $\pm$ 1.3 & -0.002  $\pm$ 0.01 &  0.02 $\pm$ 0.10 \\
\hline
NIMROD (2010)         & 0.01    & -0.561 $\pm$ 2.1 & -0.006  $\pm$ 0.02 &  0.03 $\pm$ 0.19 \\
D4B (1993)          & 0.025   & -0.561 $\pm$ 4.8 & -0.004  $\pm$ 0.03 &  0.03 $\pm$ 0.29 \\
Omega West (1973)  & 0.05    & -1.173 $\pm$ 6.7 & -0.002  $\pm$ 0.04 & -0.05 $\pm$ 0.36 \\
\hline
\end{tabular}
\end{table}

\subsubsection{Uncertainty Quantification on a State-of-the-Art Neutron Instrument Model}

We have demonstrated through uncertainty quantification that prior conclusions to inverse problem solutions were likely limited based on the experimental accuracy of the scattering instruments of their time. Furthermore, modern diffractometers are sufficiently precise to provide reliable inverse problem solutions to assess a variety of atomic force properties. Consequently, we further analyze the posterior distributions for structure factor results that are consistent with modern diffractometers. The highest flux instruments are spallation sources, which can measure structure factors with constant standard deviations of $\delta S = 0.005$ out to 30 \AA$^{-1}$. This condition well-represents an upper bound of  uncertainty in a structure factor measurement on the state-of-the-art NOMAD and NIMROD instruments. Posterior marginals, Markov chain Monte Carlo (MCMC) samples and heat maps of the joint posterior distribution are illustrated in Figure \ref{fig:histograms}.

The marginal \textit{MAP}, corresponding to the global maximum of the marginal distribution, accurately predicts the true parameter values (indicated by red dashed lines) with exceptional precision, exhibiting error rates below 1\% for all force parameters. The shape and width of the marginal distributions offer valuable insights into the influence of each parameter on the ensemble fluid structure. The collision diameter marginal exhibits a narrow and symmetric shape, characterized by a probability density at the \textit{MAP} that surpasses the repulsive exponent and dispersion energy by factors of 80 and 3, respectively. This symmetry and high probability density suggest a remarkable sensitivity of the structure factor to changes in the effective particle size which is consistent with the observations of Weeks, Chandler and Anderson \cite{weeks_role_1971}. However, the seemingly small difference between the repulsive structure factor alone and the true structure factor clearly contains sufficient information to determine the potential well-depth as well as the repulsive exponent and collision diameter. Therefore, we contend that the structure factor of liquids contains more information than previously believed. 

\begin{figure}
    \centering
    \includegraphics[width = 15 cm]{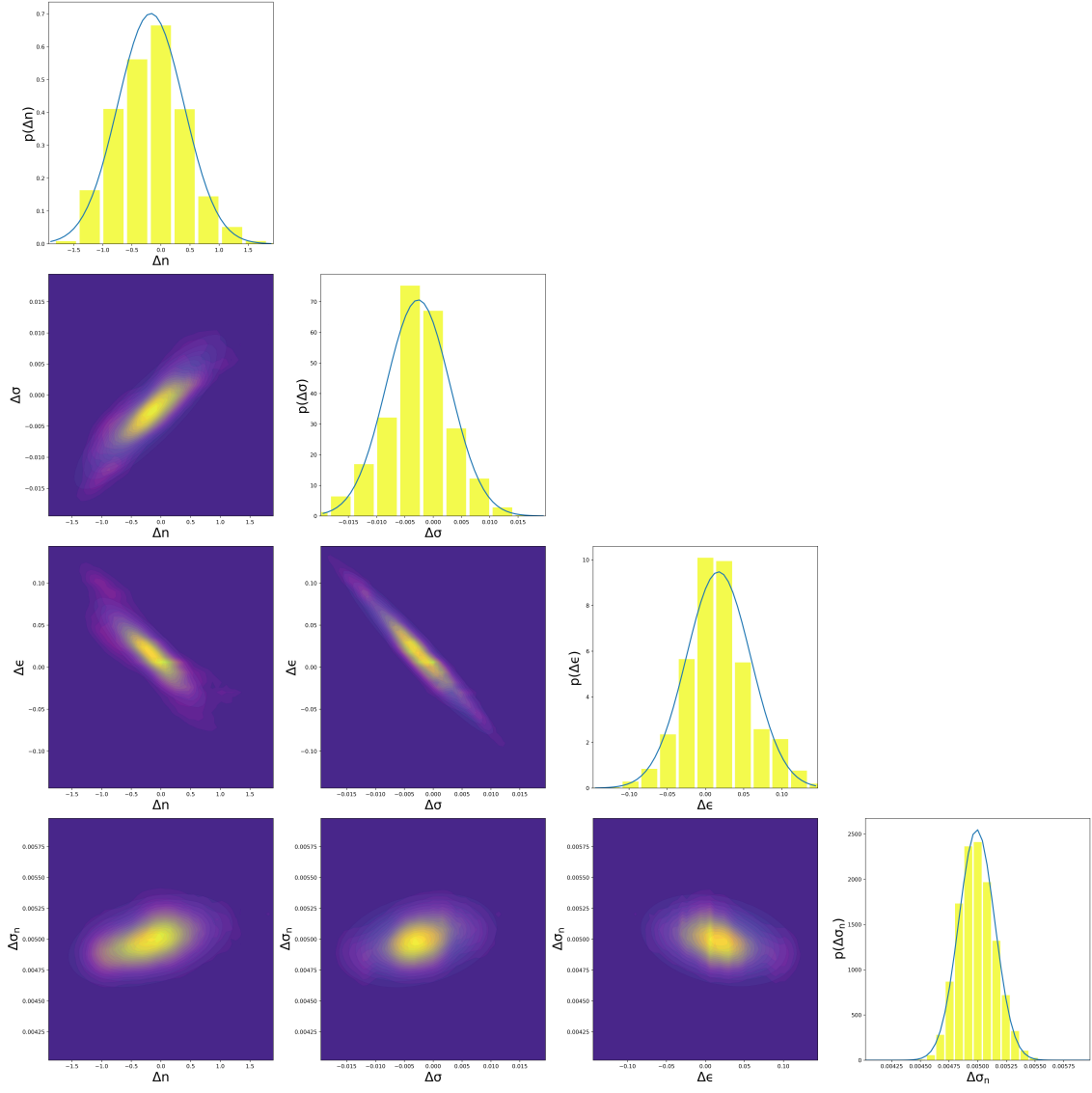}
    \caption{Marginal distributions on the ($\lambda-6$) Mie force field parameters for $\delta S = 0.005$ at 30 \AA $^{-1}$ with variance sampled with MCMC plotted with known parameter values (red dashed line).}
    \label{fig:histograms}
\end{figure}

Two standard deviations of the posterior distribution can be used as an estimate of our confidence in recovering the force parameter from the structure factor with $\sim$95\% confidence. Using this metric, we find that the force parameters can be recovered with 95\% confidence within $\pm$1.3 for the repulsive exponent, $\pm$0.02 $\sigma$, and $\pm$0.1 $\epsilon$. Of course, the posterior distributions computed here are in reduced Lennard-Jones (LJ) units and must be scaled by a known reference to approximate the credibility intervals in real units. For example, taking the LJ parameters for argon ($\sigma = 3.4$ \AA, $\epsilon = 0.24$ kcal/mol \cite{shanks_transferable_2022}) would give a real unit estimate of $\lambda \pm$1.3, $\sigma\pm$0.068 \AA, and $\epsilon\pm$0.024 kcal/mol with 95\% confidence.

Uncertainty quantification and propagation of the potential in relation to the structure factor holds the key to unlocking several capabilities of neutron scattering, including force field design, elucidation of many-body interactions, and improved understanding of structural properties in fluid systems \cite{terban_structural_2022}. While these aims have motivated research on the inverse problem for over a century, we are only now seeing evidence that accurate structure inversion on experimental data is a possibility. We argue, despite having presented a study on a simple model, that our results warrant the recommendation of revisiting inverse methods for real fluids. 

One exciting prospect for inverse problem methods is that interaction potentials derived from structure can serve as an external validation to computationally expensive bottom-up atomistic models. One example is electron structure calculations, in which a highly accurate quantum mechanical treatment of the electron structure can reveal insights into potential energy surfaces and reaction mechanisms. Electron structure methods have become faster and more robust due to quantum computing \cite{motta_emerging_2022}, clever basis set selection \cite{ratcliff_flexibilities_2020}, and machine learning \cite{duan_learning_2019,rath_framework_2023}. As these more fundamental theories of atomic structure and motion become commonplace, experimental neutron scattering data will be a crucial validate of their predictions. Indeed, we have already shown that many-body interactions in noble gases are consistent with electron structure calculations in trimeric systems \cite{guillot_triplet_1989,shanks_transferable_2022}. However, further advancement of inverse methods can provide quantifiable validation of many-body interactions in progressively complex systems.

In contrast to experimental analysis, inverse methods for coarse-graining have thrived in contemporary chemical physics. In coarse-graining, structure factors are generated from a known model where the uncertainty fluctuations are significantly smaller than that of experimental data. Bayesian analysis indicates the presence of global maxima in the posterior distributions of such systems, suggesting that well-behaved optimization schemes should be capable of reliably identifying these maxima. Since global maxima are also expected in experimental scattering measurements conducted using state-of-the-art neutron instruments, there is an opportunity to employ these maximum-likelihood methods for developing novel force fields directly from experimental data. With this evidence, we hope to bridge the gap between experiment and simulation-based inverse techniques and foster closer collaboration between these two communities.

It is important to acknowledge certain limitations in the previous analysis when considering the extension of the results to other physical systems. First, the sensitivity of more complex systems than the ($\lambda$-6) Mie model may differ from the estimates reported in our study. Therefore, it is cautioned to interpret the results of this example as a conceptual exploration of how classical two-body interactions impact the atomic organization in fluids. Hence, the specific response of complex systems to variations in interatomic forces should be studied individually. Second, if a fluid cannot be adequately described by a ($\lambda$-6) Mie model, the resulting posterior distribution may exhibit flatness or multimodality, indicating a high level of uncertainty in both the structure and model parameters. In such cases, a more accurate model of the system should be adopted to facilitate reliable parameter inference. Furthermore, systematic errors were not investigated and are certainly significant to the potential reconstruction. Therefore, further work should explore how to eliminate systematic errors in neutron scattering analysis through physics-based Gaussian process regression or analogous approaches.

\subsection{Conclusions}

Rigorous uncertainty quantification and propagation analysis has shown that modern neutron diffractometers have attained the necessary accuracy for reliable force field reconstruction. It has also been shown that neutron scattering measurements within $\leq 0.005$ at $\sim$30 \AA$^{-1}$ are sufficient for force parameter recovery in simple liquids. We stress that the structure factor contains information on the force field parameters that control the attractive as well as the repulsive part of the interatomic potential. This study highlights the exciting possibility of using neutron scattering to predict the potential energy function of Mie-type fluids, emphasizes the critical role of experimental precision in extracting potentials from scattering data, and offer svaluable insights into the nature of interatomic forces in liquids. 

The significance of these results extends beyond the field of neutron scattering analysis. They hold great potential in advancing force field design and optimization, enabling the development of effective coarse-graining techniques, and facilitating the exploration of many-body effects in fluid ensembles. The far-reaching impact of machine learning-accelerated methods in predicting interatomic forces from experimental structure measurements is evident. In summary, this research establishes the transformative potential of machine learning in extracting interatomic forces from experimental structure measurements with uncertainty quantification.

\section{Conclusions and Future Work}

\subsection{General Conclusions}

In this dissertation, it was shown that adopting a Bayesian philosophy to contemporary problems in condensed matter physics can unlock novel applications of liquid state theory to analyze scattering experiments.  For example, we have shown that function space Bayesian inference using Gaussian processes can allow for the extraction of transferable force fields for real fluids, the first result of its kind in the 100 year history of structure inversion methods. Furthermore, a practical and easy-to-implement Bayesian approach was outlined to train classical force field models on complex experimental data at a previously inaccessible computational scale. Finally, ultra-fast Bayesian methods were leveraged to revisit prior literature claiming that scattering data is unsuitable for force extraction and showed that this conclusion is likely no longer true for modern neutron instruments and computer simulations. The primary aim of this dissertation was to present novel applications of Bayesian analysis within liquid structure analysis. It is my hope that these results can form the foundation of a general Bayesian theory for the liquid state, in which we leverage all available experimental data, physical knowledge, and computational resources to link the quantum mechanical scale all the way through macroscopic scale thermodynamics.

Chapter 2 applied a Gaussian process regressor as a probabilistic interpreter of a Henderson inverse theorem potential refinement scheme, which we called structure optimized potential refinement (SOPR). Philosophically, the approach amounts to assuming that the output from iterative refinement is an observation of a potential from a functional distribution with physically-justified properties; namely, the potential function is continuous, differentiable, and varies over a length scale of $\sim$1\AA. It was shown that under these loose requirements that a nonparametric potential derived from noble gas radial distribution functions was consistent with pair potentials trained on detailed thermodynamic data with the flexible ($\lambda-6$) Mie functional form. The remarkable consistency of a fully nonparametric Bayesian scheme with a flexible parametric model indicates that there is detailed interatomic force information contained in the structure of simple liquids and that these forces can be used to predict complex thermophysical properties of the liquid state. 

One important limitation of the method from Chapter 2 is that there is a fundamental lack of uniqueness in site-site partial radial distributon functions in molecular liquids and mixtures. Without a proper characterization of these distribution functions, a nonparametric approach is difficult to implement and potentially unreliable as demonstrated by previous work \cite{soper_radial_2013}. However, providing a stronger restriction on the potential function space to an established parametric form could eliminate some of these challenges and enforce consistency with existing molecular simulation methods. Therefore, Chapter 3 explored the incorporation of Bayesian uncertainty quantification for parametric potentials trained on radial distribution functions. The main idea was to create a fast and reliable uncertainty quantification framework that could be used to train parametric force fields to structural data. In this process, we were able to speed-up the Bayesian analysis by over a million-fold by using local Gaussian processes as a surrogate model to replace the expensive molecular simulation training step. This method holds promise as a tool to learn parametric force fields in more complex liquids as it can be trained on relatively few simulations ($\sim$1000). Finally, $(\lambda-6)$ Mie potentials recovered for liquid neon, a traditionally difficult to model liquid, were consistent with both SOPR potentials and a vapor-liquid equilibrium trained model \cite{mick_optimized_2015}. Consistency across multiple training schemes bolsters the validity of the method as a force field development tool.

Chapter 4 leverages the local Gaussian process surrogate model developed in Chapter 3 to probe the importance of neutron scattering measurement accuracy on the extraction of forces from structural data. This question has deep roots in the liquid state theory community and was studied by Verlet \cite{levesque_note_1968} and Soper \cite{soper_uniqueness_2007,soper_radial_2013} over decades of structure research. The importance of this question lies in the fact that we need to be able to identify what available data is suitable for force field modeling and whether or not new neutron experiments are necessary. Furthermore, since the work uses a Bayesian uncertainty quantification framework, we identified that it is possible to recover force parameters with high certainty from neutron measurements collected on state-of-the-art neutron diffractometers. 

Taken together, the results of this dissertation paint an optimistic picture of scattering analysis in the context of liquid state theory. It has been shown that adopting a rigorous mathematical framework for uncertainty quantification can unlock novel capabilities of scattering analysis and inform modeling decisions and fundamental insight into the natural world. Finally, the methods and philosophical approaches that underscore the novel research contributions of this dissertation are already showing progress on more complex systems. Based on these preliminary successes and the strong evidence provided in this dissertation, I speculate that a Bayesian approach to liquid state theory will continue to revolutionize the field of atomic scale chemistry and physics well into the future.

\subsection{On the Horizon}

Described below are a selection of ongoing and future research directions suggested by the contents of this dissertation. The list is not exhaustive as the use of Bayesian methods in computational chemistry and materials engineering provide a lifetime of potential scientific study.

\subsubsection{Quantifying Many-Body and Quantum Effects with Neutron Scattering}

The biggest limitation of the SOPR method is its restriction to classical pairwise additive potentials. Real physical systems behave quantum mechanically and are inherently many-body in nature, leading to effects that cannot be directly modeled in a classical representation. Additionally, we have not yet verified whether decomposing the SOPR potentials into a quantum dimer and effective many-body term can capture higher-order effects (Figure \ref{fig:triplet}). One proposed test is to run a model system with 2- and 3-body interactions and then attempt SOPR on the resulting radial distribution function with a 2-body reference potential. If SOPR could recover the effect of the 3-body term, then we could verify whether or not the decomposition is appropriate for measuring many-body interactions.

\begin{figure}
    \centering
    \includegraphics[width = 14cm]{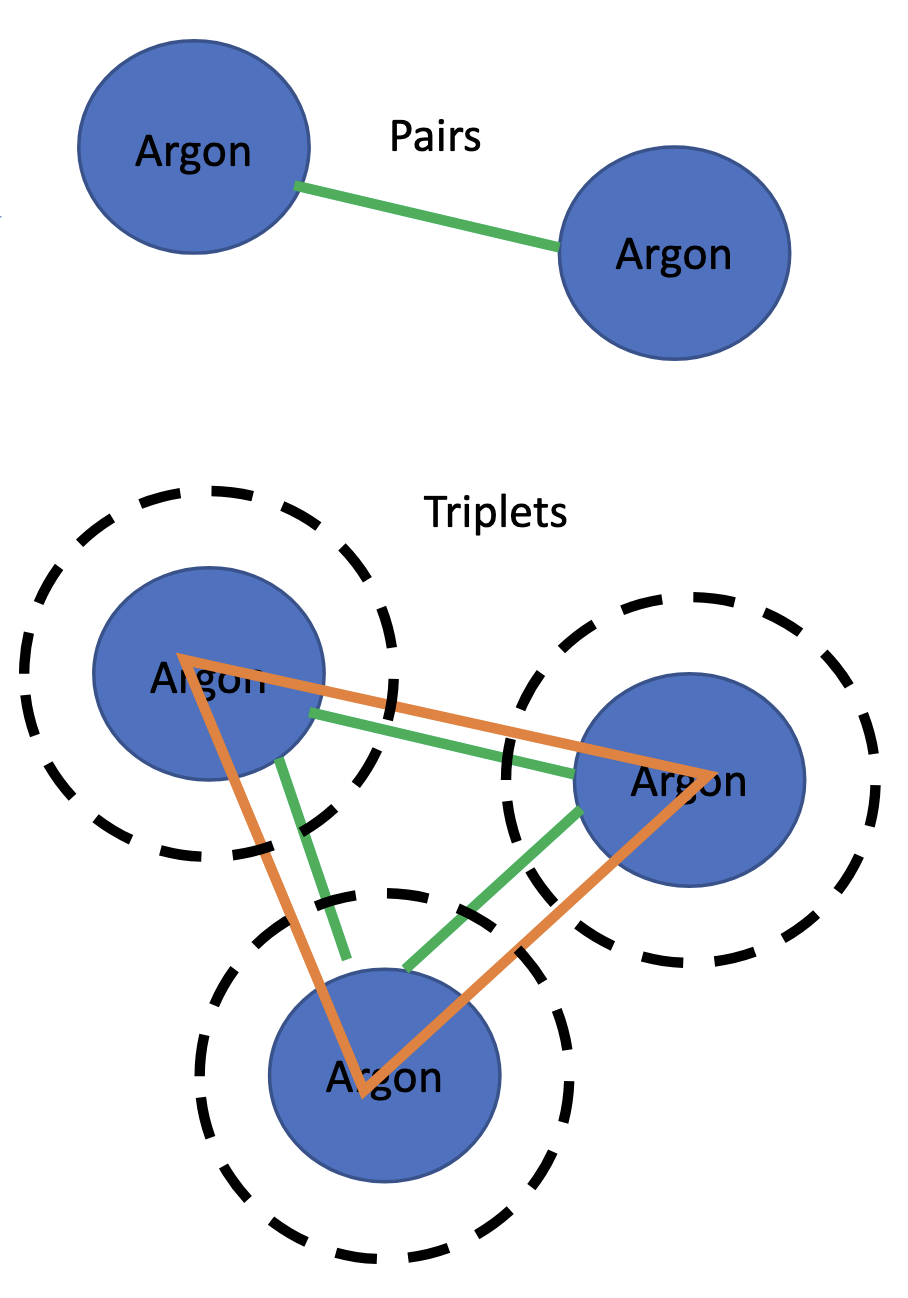}
    \caption{Recovering pair and triplet terms of the many-body potential energy expansion needs to be further explored.}
    \label{fig:triplet}
\end{figure}

Recently, we have been further exploring how SOPR potentials change at different thermodynamic states using a neutron scattering dataset of supercritical krypton at varying pressures. As discussed in Chapters 2 and 4, we have some evidence that the interatomic forces change at different thermodynamic states, and that those changes may correlate to many-body and/or quantum mechanical effects. Finally, we have preliminary results showing that SOPR corrections to \textit{ab initio} pair potentials in noble gases show changes in effective particle size consistent with the polarizability of the electron cloud, a result recently demonstrated theoretically \cite{fedorov_quantum-mechanical_2018}. An expanded discussion on SOPR corrections to quantum pair potentials and its connections to electronic polarization is included in Appendix A.

\subsubsection{Structure Optimized Potential Refinement for Molecular Liquids}

The continued success of structure optimized potential refinement hinges on its extension to molecular liquids and mixtures. The primary challenge here is that updating multiple pair potentials simultaneously can result in a loss in uniqueness of the resulting potentials and lead to unstable convergence. However, the Gaussian process can mitigate these problems by acting as a probabilistic regressor over the pair potential estimation. For example, the GP kernel can be modified to enforce a power-6 attractive tail behavior, as predicted from quantum mechanics, while allowing flexibility at short-range to accommodate subtle features present in the neutron scattering data. One kernel function that enables this flexibility is an inverse squared-exponential kernel given by,

\begin{equation}
    k(x_i, x_j) = \alpha \exp{-2l^2\sqrt{{\bigg({\frac{1}{{x_i}^2}-\frac{1}{{x_j}^2}}\bigg)^2}}}
\end{equation}

\noindent where $\alpha$ is the variance in the pair potential space and $\ell$ is a length-scale parameter that approximately controls when the kernel transitions from short to long range covariance. The GP then enforces physically-justifiable behavior, such as continuity, differentiability, and a general functional form consistent with quantum mechanics. 

\subsubsection{Spectral Gaussian Process Fourier Transforms}

As discussed in the Introduction, radial distribution functions cannot be extracted perfectly from a neutron scattering measurement due to the finite nature of the Fourier transform between momentum and real space. The biggest gap in the current literature on structure factor Fourier transforms is a lack of robustness and uncertainty quantification. Essentially, we can get estimates of the pair and total radial distribution functions, but we cannot quantify how confident to be in the results. Clearly, this is a serious problem for using scattering data to inform liquid state models as evidenced by the uncertainty quantification performed in Chapters 3 and 4 of this dissertation. Furthermore, it is equally concerning that molecular dynamicists have no credibility intervals to compare their force field structure predictions to experiments. 

To address this challenge, we have recently been applying Gaussian processes to the structure factor function space. The idea is to construct a Gaussian process regressor with a physics-informed kernel that can learn the mean and variance of the structure factor functional distribution. Using either a sampling approach or the analytical Fourier transform of the Gaussian process \cite{ambrogioni_integral_2018}, one can then construct a functional distribution of radial distribution functions in real space. Alternatively, one could try and design non-stationary spectral kernels \cite{remes_non-stationary_2017}, although we have found that this method is not restrictive enough to enforce known physical behavior.

\subsubsection{Bayesian Uncertainty Quantification for Machine Learned Potentials}

Machine learning potentials (MLPs) are becoming ubiquitous in the simulation of functional materials, with contemporary applications spanning from the design of organic semiconductor photovoltaic cells and battery electrolytes \cite{deringer_realistic_2018,zhang_why_2023} to metal organic frameworks and super low temperature hydrogen storage \cite{cheng_evidence_2020}. The main idea behind MLPs is to take state-of-the-art \textit{ab initio} electron structure methods, such as density functional theory (DFT), and learn surrogate models of interatomic forces that can be evaluated at the computational expense of classical molecular models \cite{behler_perspective_2016}. In turn, this enables the investigation of molecular dynamics at quantum mechanical accuracy with relatively low computational effort, accelerating the innovation of novel materials to address contemporary challenges from energy storage to carbon capture.

However, strong criticisms of MLPs largely focus on a lack of interpretability of the resulting force fields and a failure to include experimental data into the model training \cite{deringer_machine_2019}. Indeed, most MLP models rely solely on force estimations from electronic structure calculations and neglect important experimental information that may contain nuanced information on interatomic interactions, such as structural correlation functions and electromagnetic spectra \cite{matin_machine_2024}. Furthermore, black-box machine learning methods such as neural network potentials are often criticized as being non-interpretable and, lacking built-in uncertainty quantification and propagation (UQ/P), are too uncertain to be reliable in drawing physics based conclusions on the resulting potentials \cite{wen_uncertainty_2020}. Therefore, incorporating experimental data into an MLP framework with UQ may yield substantial improvements in molecular model predictions, enhance intepretability, and cement adoption of MLPs within the computational chemistry community at large.

Bayesian inference could be used to address these criticisms by including experimental data into MLP training. Specifically, a Bayesian committee, which is essentially just a collection of informed "voters" on the outcome of the MLP prediction, will be employed to estimate interatomic forces based on a combination of DFT calculations and experimental data \cite{willow_active_2024}. The Gaussian approximation potential (GAP) framework \cite{deringer_gaussian_2021}, a machine learning surrogate model approach that employs sparse Gaussian processes to learn interatomic forces from DFT calculations, will be the ideal foundation to incorporate the Bayesian committee over the potential field. The end product will be a method to build physics-guided and interpretable MLP potentials for contemporary problems in functional material modeling and design that will represent the state-of-the-art in both quantum physics and experiment. 

\section{Appendix}
\appendix
\section{Quantum and Many-Body Effects from Neutron Scattering}

SOPR potentials represent an estimate of the interatomic force between atoms based on a sub-angstrom length scale scattering experiment. Furthermore, since the scattering experiment obeys quantum mechanical laws and is the time-averaged measurement of a large number of atoms, we expect the measured structures to reflect both quantum mechanical and many-body interactions occurring in the liquid phase system. The hypothesis, then, is that SOPR potentials can capture these effects in a pairwise additive, mean field approximation. Although Chapter 2 provided speculative evidence that corrections to the quantum pair potential were consistent with quantum mechanical arguments proposed in prior literature, it still remains to show whether SOPR potentials definitively provide meaningful information on quantum and many-body effects in real fluids.

\subsection{Relations Between Atom Size and Electron Polarization}

Atom size is a fairly ambiguous quantity, as you can prescribe a number of different criteria to map the volume occupied by an atom or molecule to some radius. Some examples include the Van der Waal radius, defined as the radius of a hard-sphere representing the distance of "closest approach" for a neighboring atom, or the atomic radius which specifies the distance from the nucleus to the most probable outermost electron. While atomic size is not really that interesting of a property, it is still critical for atomic models of liquids and has fundamental relationships to electronic properties of atoms. For example, a recent study \cite{fedorov_quantum-mechanical_2018} showed that the van der Waal radius, $R_{vdW}$, in a quantum Drude oscillator model was related to the polarizability, $\alpha$, according to the following equation,

\begin{equation}\label{eq:vdwrad}
    R_{vdW} = 2.54 \alpha^{1/7}
\end{equation}

\noindent which deviates from the classical result of $R_{vdW}^{classical} = 1.62 \alpha^{1/3}$. Fedorov and coworkers demonstrated that this relation well-reproduced the van der Waal radius of 72 elements, providing a means to estimate the approximate hard-sphere diameter of even complex atoms such as metals. However, it is notable that the radii computed using this approach would not be suitable for liquid state modeling for more complex thermodynamic properties that require a detailed description of the short-range repulsive exponent, collision diameter and dispersion energy as shown for high pressure systems \cite{messerly_uncertainty_2018} and vapor-liquid equilibrium \cite{mick_optimized_2015}. For example, the van der Waal radius for neon calculated with \ref{eq:vdwrad} gives a value of 2.91 \AA, which is significantly outside of the parameter distributions estimated for neon for structural correlations and vapor-liquid coexistence. 

The radial distribution function provides a notion of atom size by definition. By taking any arbitrary particle in the system and counting the number of particles within a spherical shell of thickness $dr$ (spherical surfaces at $r$ and $r+dr$), we will always observe that there is an excluded volume region where the shell is inside of the arbitrarily chosen particle. At some $r$, the shell will become larger than the particle size, begin to count other atoms, and the radial distribution function will spike up and become non-zero. By virtue of structure-potential relationships in classical mechanics, this region will correlate with the steep rise of the short-range repulsive part of the interatomic potential. 

The question then is, given that we can learn the interatomic potential from a radial distribution function using SOPR, is there a way to estimate the atomic size and relate it back to fundamental quantum mechanical results on more interesting properties. One such approach was potential perturbation method proposed by Weeks, Chandler, and Anderson known as the Weeks-Chandler-Anderson (WCA) separation. The WCA separation splits the potential into a short-range repulsive wall and a long-range attractive tail at the minimum of the potential well. Figure \ref{fig:wca} shows the WCA separation for the SOPR potential derived for liquid Ne at 42K. 

\begin{figure}
    \centering
    \includegraphics[width = 14cm]{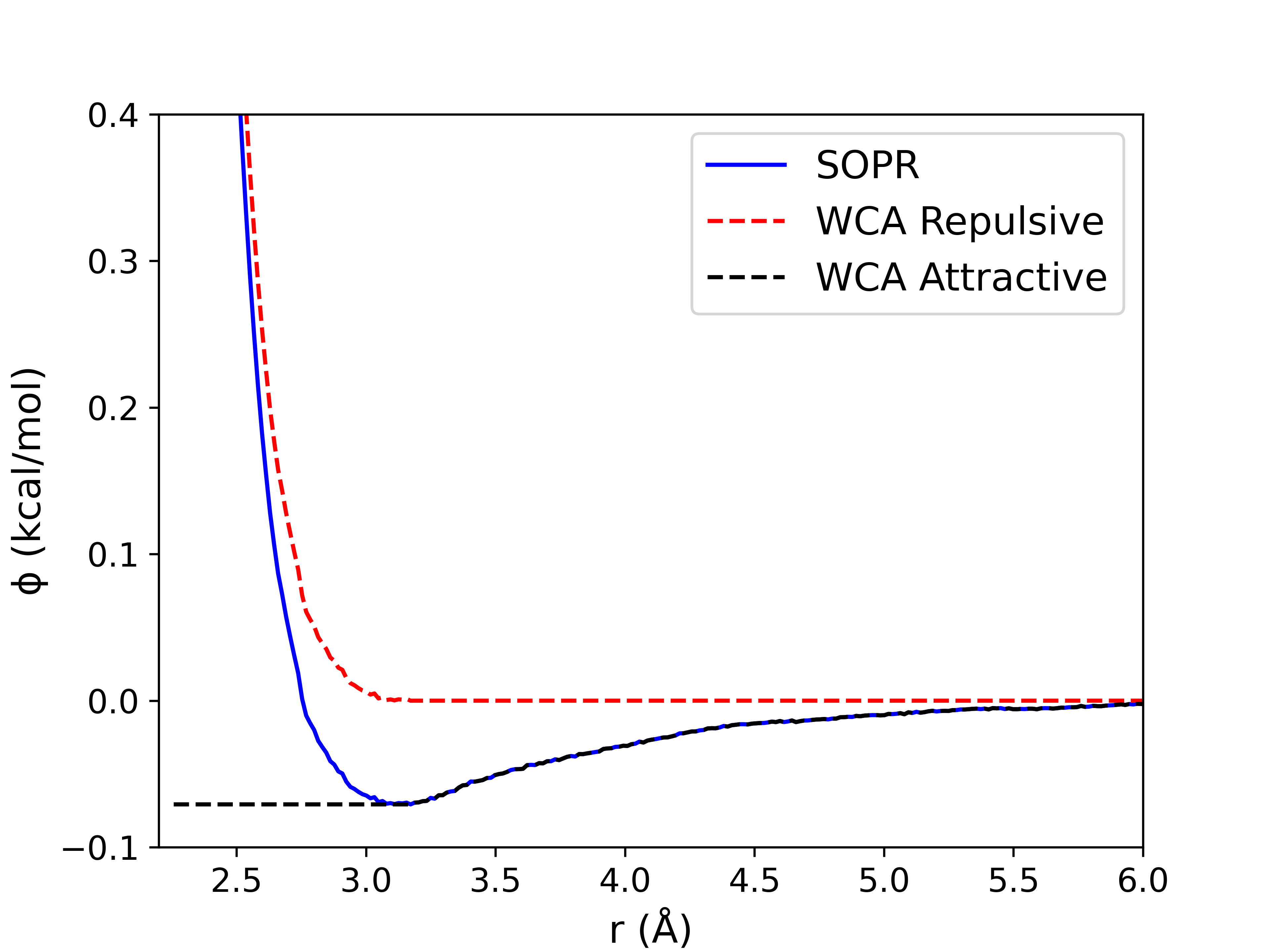}
    \caption{WCA separation of the SOPR potential for liquid Ne.}
    \label{fig:wca}
\end{figure}

Once the potentials are decomposed using the WCA separation, we can apply perturbation theory to estimate an equivalent hard-sphere diameter, $d$, from the repulsive part of the potential \cite{weeks_role_1971,hansen_theory_2013} according to the following equation,

\begin{equation}\label{eq:hardparticle}
    d = \int_0^\infty [1 - \exp(-\beta v_0(r))] dr
\end{equation}

\noindent where $v_0(r)$ is the repulsive part of the WCA separation. The integrand $[1 - \exp(-\beta v_0(r))]$ is unity at low $r$ since $\exp(-\beta v_0(r))$ is negligibly small due to the y-asymptotic behavior of $v_0(r)$. At longer range, $v_0(r)$ is zero by definition, which sends the integrand to zero as well. A plot of the integrand for Ne, Ar, Kr, and Xe WCA separated SOPR potentials is shown in Figure \ref{fig:hp}. Integrating this quantity provides an estimate of the equivalent hard-sphere diameter that can be used to study how particle sizes change with polarization of the electron cloud.

\begin{figure}
    \centering
    \includegraphics[width = 14cm]{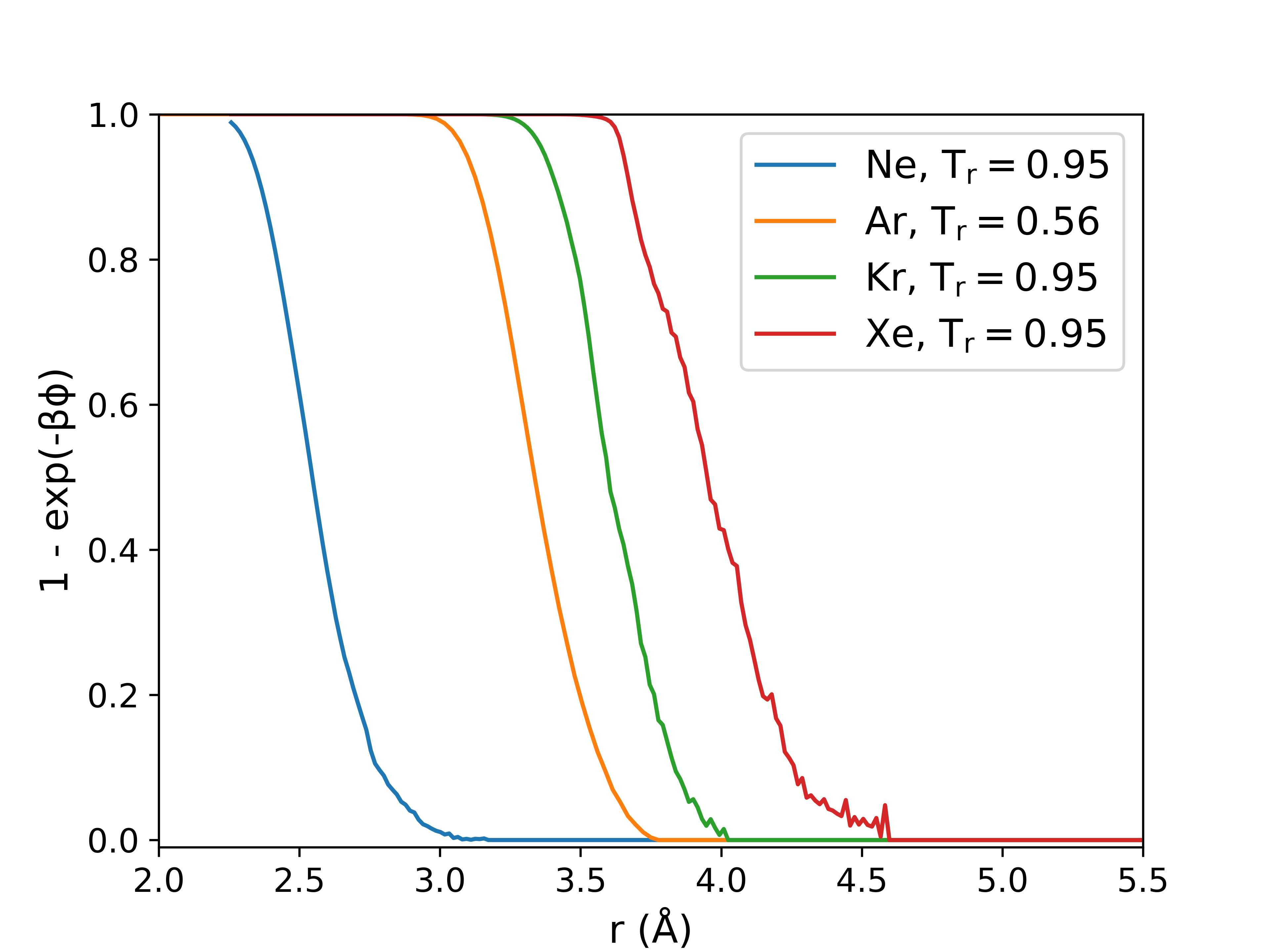}
    \caption{$[1 - \exp(-\beta v_0(r))]$ for WCA separated SOPR potentials.}
    \label{fig:hp}
\end{figure}

Figure \ref{fig:rvpolar} shows the hard-particle diameter calculated from \eqref{eq:hardparticle} versus the electron polarizability along with results compiled from Fedordov \textit{et. al.} and collision diameter estimates from quantum pair potentials. It is immediately striking that the classical prediction exhibits significantly different behavior than the other three methods. Furthermore, we can see that the two methods of estimating the atom size with quantum mechanical methods, \textit{i.e.} the quantum Drude oscillator approximation in red and advanced electron structure calculations in black, give the same scaling behavior estimated from the neutron data. This result is exciting since it provides experimental validation for the quantum mechanics calculations and demonstrates that SOPR potentials can provide estimates of fundamental atomic properties consistent with quantum theory despite the method being semi-classical. The results from this preliminary analysis can be found in a preprint arXiv article \hyperlink{https://arxiv.org/abs/2501.06501}{here} \cite{shanks_experimental_2025}.

\begin{figure}
    \centering
    \includegraphics[width = 12cm]{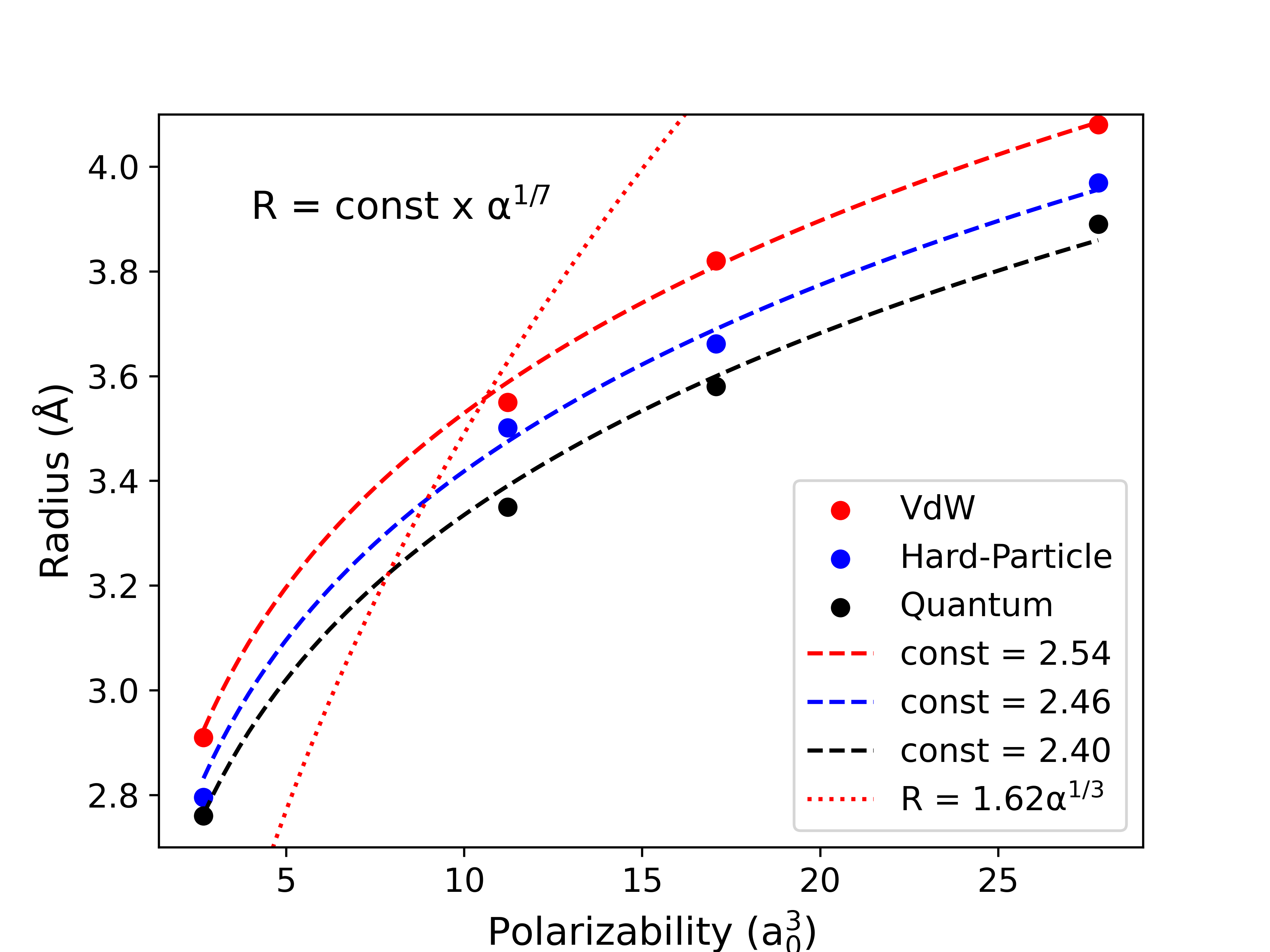}
    \caption{Particle radius determined from the van der Waal radius criterion (red dashed line and points), hard-particle radius derived from SOPR potentials (blue dashed line and points), and the quantum dimer potentials (black dashed line and points).}
    \label{fig:rvpolar}
\end{figure}

\subsection{State-Dependent Corrections in Supercritical Krypton}

Supercritical fluids have been the subject of intense investigation in recent years, owing largely to their applications as chemical solvents in industrial processes and recurrence in the physics of extreme environments \cite{mcmillan_going_2010}. At subcritical temperatures, violations of classical thermodynamic stability criterion characterize phase transitions. Specifically, vapor-liquid phase separation initiated via nucleation occurs at thermodynamic paths of constant chemical potential (binodal) while spontaneous vapor-liquid separation occurs when the curvature of the free energy with respect to composition becomes negative (spinodal). The critical point is defined as the extremum of the spinodal, beyond which there is no observable thermodynamic vapor-liquid phase separation. However, the fact that material properties of supercritical fluids are highly sensitive to changes in thermodynamic state (pressure, temperature, density, etc.) and have been shown to exhibit sharp inflections in such properties has motivated the search for a universal theory of fluids extending beyond the critical point \cite{bolmatov_phonon_2012,bolmatov_thermodynamic_2013,yoon_two-phase_2018,ha_universality_2020,cockrell_transition_2021}. Transitions in supercritical fluids have been proposed based on the behavior of thermodynamic and dynamical properties, including thermal conductivity, density fluctuations \cite{tucker_solvent_1999}, density distribution functions \cite{nishikawa_inhomogeneity_2000,cipriani_orientational_2001}, velocity autocorrelation functions, specific heat, speed of sound, and diffusion. 

One proposed theory of supercritical fluids asserts that inflections in the observed material properties can be attributed to the ability of the fluid to sustain significant transverse excitation modes \cite{brazhkin_two_2012,bryk_behavior_2017,fomin_dynamics_2018}. Conceptually, this transition occurs at the Frenkel time scale $\tau_F$, defined as the average time required for the mean squared displacement of a particle to equal its effective particle size. This transition is referred to as the Frenkel line. While a precise theoretical framework for defining the Frenkel line has not been clearly identified, simulation and experimental evidence of inflections in model and real fluid properties near the Frenkel line have been observed. Of considerable interest is a recent neutron scattering study on supercritical krypton that reported the emergence of medium range order ($\sim$ 7-10\AA) in the radial distribution function at $\sim$ 110 MPa along the 310K isotherm, indicating a transition from non-rigid to rigid dynamics that was attributed to crossing the Frenkel line \cite{pruteanu_krypton_2022}. The authors noted that empirical potentials derived from the neutron scattering analysis technique empirical potential structure refinement (EPSR) significantly deviated from a standard (12-6) Lennard-Jones potential, a model potential that is often used to approximate the Frenkel line with molecular simulation techniques. However, EPSR potentials are known to be unreliable in the prediction of thermodynamic properties as the technique was designed to determine molecular configurations consistent with experimental scattering rather than estimate physically accurate or transferable potentials. Thus, systematic analysis of the state-dependence of the interatomic potentials for supercritical krypton remains unresolved.

In addition to the proposed Frenkel line transition, there has also been work investigating supercritical phase transitions based on a maximum in the local correlation length in the pressure-temperature diagram (the Widom line) \cite{xu_relation_2005,simeoni_widom_2010,brazhkin_two_2012,banuti_crossing_2015} as well as a change in the oscillatory behavior of the radial distribution function (the Fisher-Widom line) \cite{fisher_decay_1969}. Here we will focus on supercritical transitions defined by the radial distribution function; namely, the Fisher-Widom and pair correlation definition of the Frenkel transition. Figures \ref{fig:krphase} and \ref{fig:krrdfs} show what neutron scattering data is available on the krypton phase diagram as well as a plot of radial distribution functions reported in the literature \cite{barocchi_neutron_1993,pruteanu_krypton_2022}.

\begin{figure}
    \centering
    \includegraphics[width = 11.5cm]{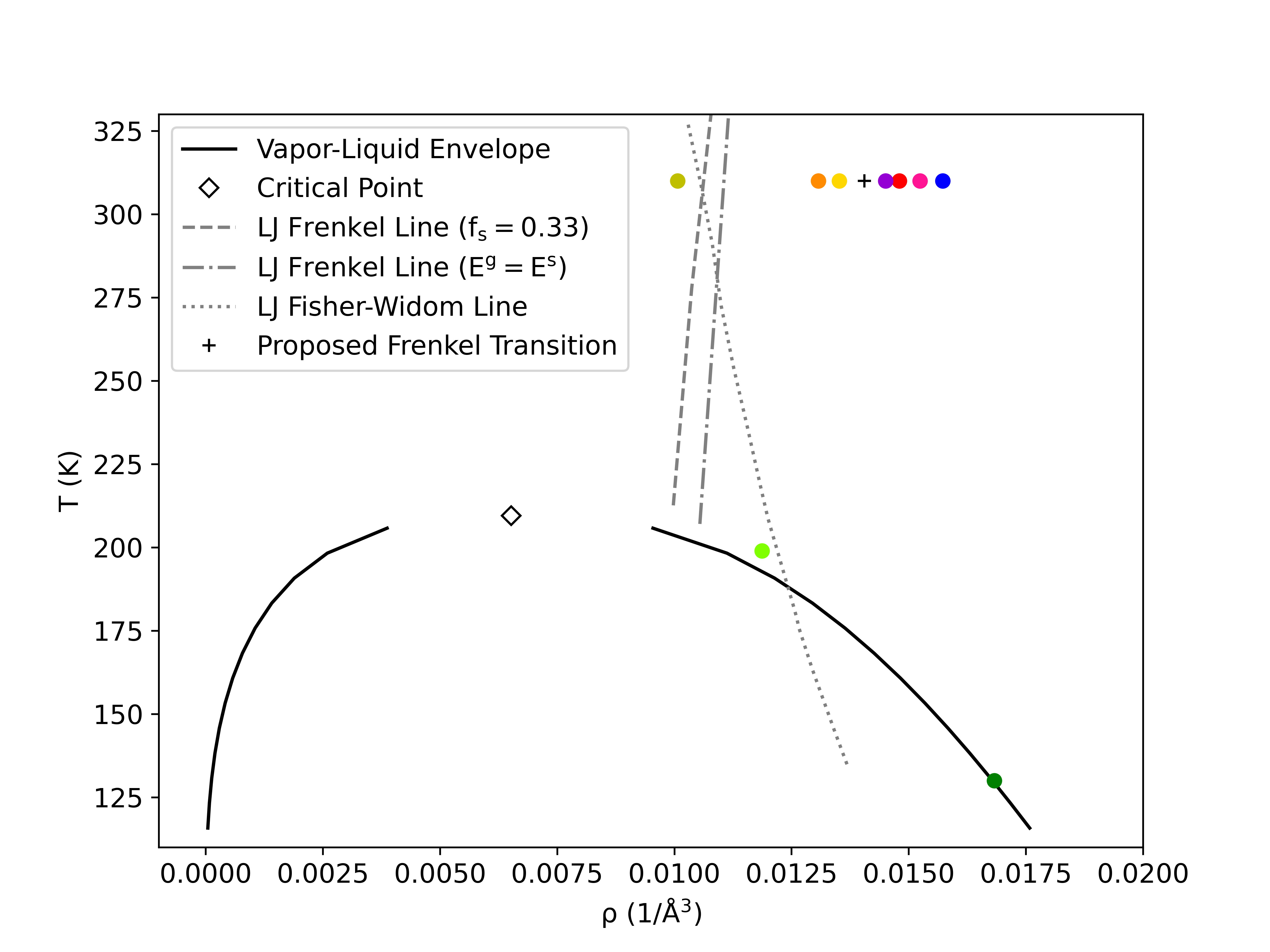}
    \caption{Phase diagram of krypton with neutron scattering experiments (colored points) and proposed supercritical transition lines drawn in based on results for the Lennard-Jones fluid.}
    \label{fig:krphase}
\end{figure}

\begin{figure}
    \centering
    \includegraphics[width = 11.5cm]{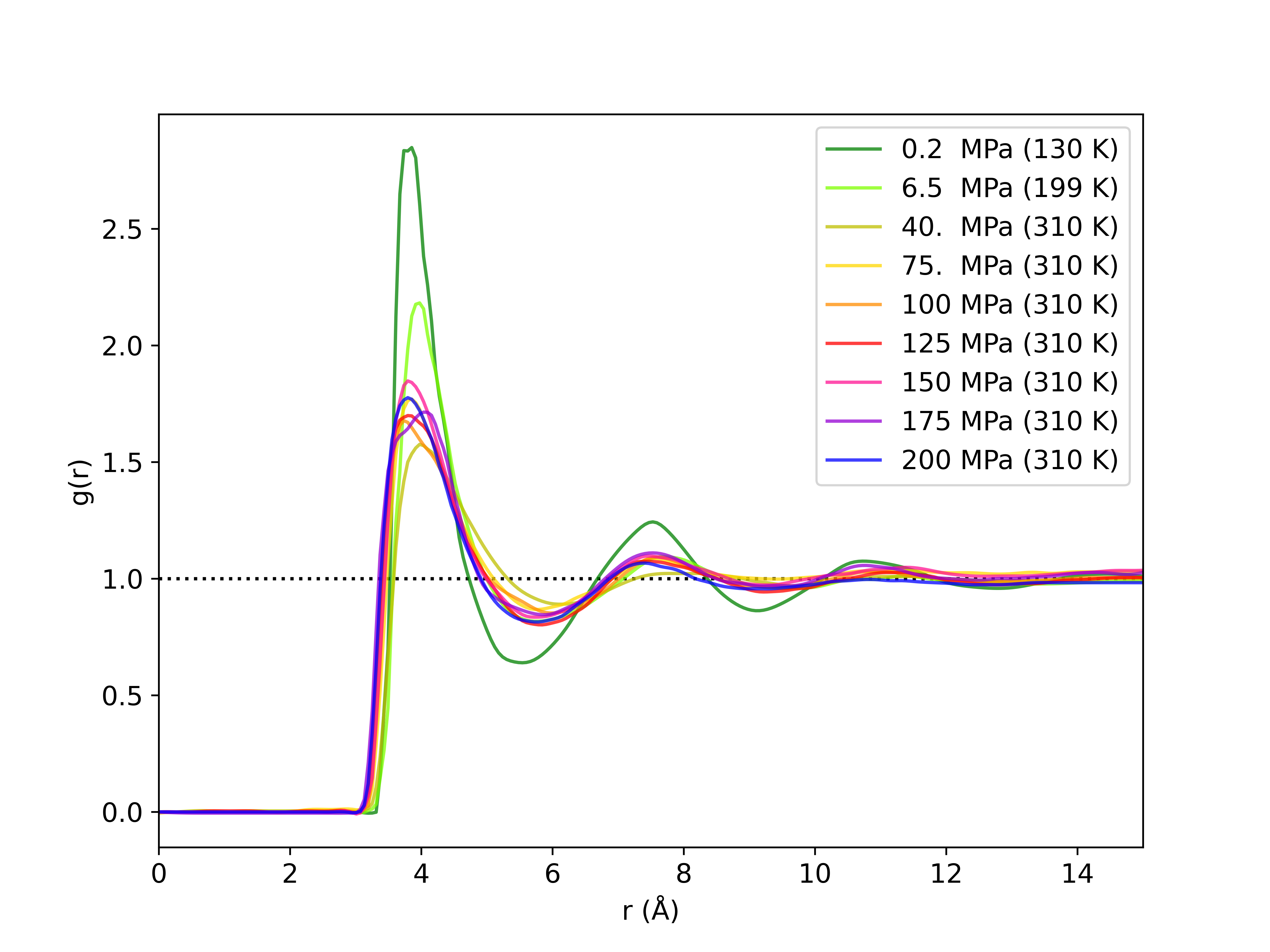}
    \caption{Available radial distribution functions derived from neutron scattering data for krypton.}
    \label{fig:krrdfs}
\end{figure}

Looking at Figure \ref{fig:krphase}, we can see that there are two data points (camo green and neon green) that fall to the left of the Fisher-Widom transition, while the rest fall to the right. According to Fisher and Widom's hypothesis, the oscillatory behavior of the radial distribution functions to the left of this transition should exhibit an always positive exponential decay whereas those to the right should be oscillatory \cite{carvalho_decay_1994}. Seeking the first experimental evidence for this transition, we computed the quantity $\log r h(r)$ from the neutron scattering derived radial distribution functions to see if this behavior is observed in Figure \ref{fig:fwidom}. 

\begin{figure}
    \centering
    \includegraphics[width = 14cm]{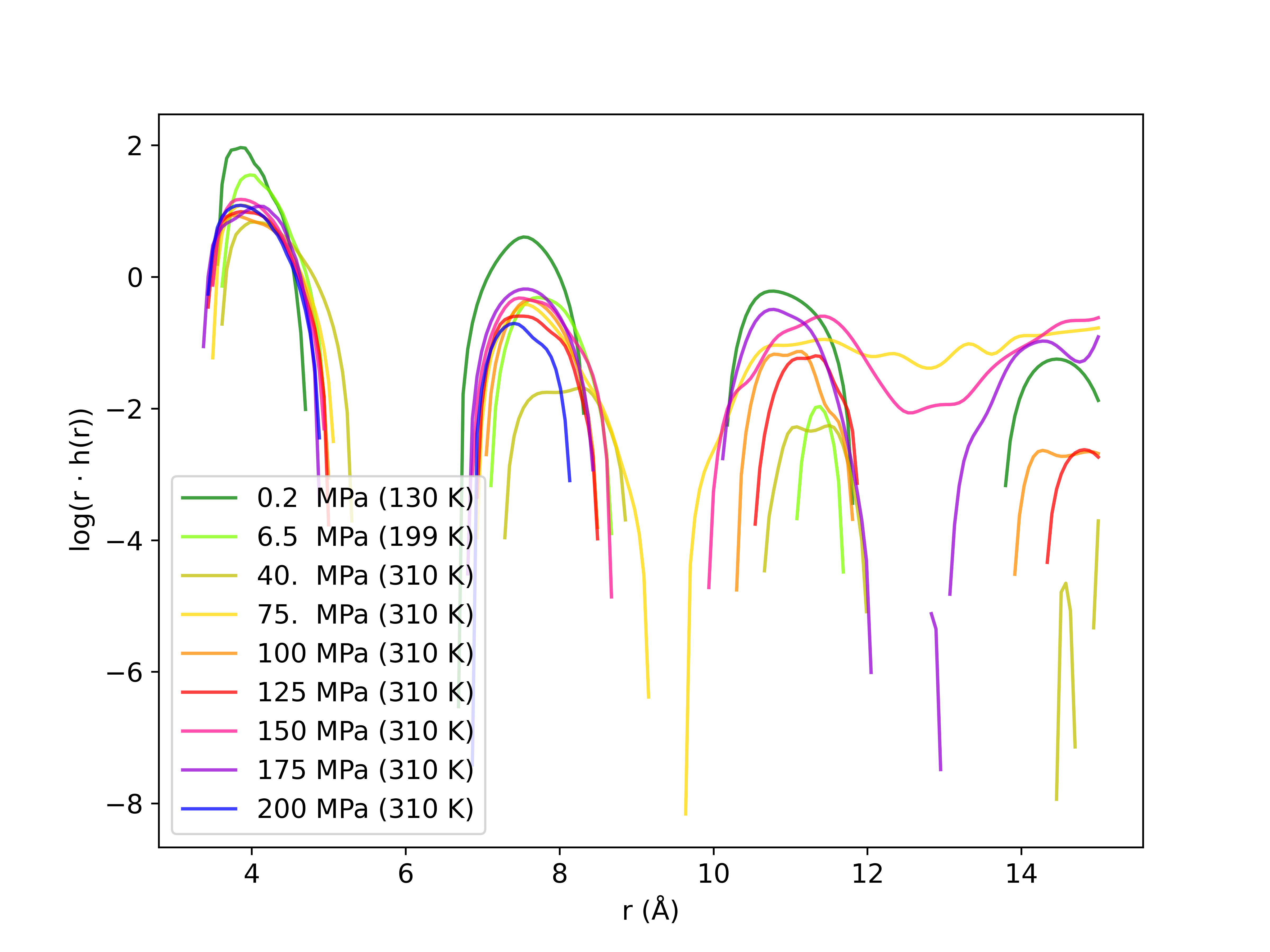}
    \caption{$\log r h(r)$ for available radial distribution function for krypton.}
    \label{fig:fwidom}
\end{figure}

Monotonic decay at long range is not observed for the two conditions left of the Fisher-Widom line but is observed for two conditions right of the transition. Therefore, it is concluded that the neutron data available does not support a supercritical transition across the Fisher-Widom line for krypton. However, it is notable that the neutron scattering measurements used to derive these radial distribution functions may not be accurate enough to observe this subtle effect. For instance, all of the krypton scattering measurements reported had random errors larger than those identified in Chapter 4 for structural inversion analysis and the 2022 dataset of high pressure supercritical appears to show systematic errors in the raw scattering data (see Supporting Information from Pruteanu \textit{et. al} \cite{pruteanu_krypton_2022}).

The next steps of this work are to apply structure-optimized potential refinement to quantify state-dependent many-body contributions to structure derived effective pair potentials. SOPR has been shown to predict transferable potentials in the subcritical region, challenging the established notion that state-dependent structure derived potentials are non-transferable to macroscopic thermodynamic property predictions. If the SOPR potentials do not significantly change along the 310 K isotherm and show no discernible change when crossing the Frenkel line, then we can conclude that there is no evidence of a transition in the interatomic forces across the proposed transition. Otherwise, we would be able to provide evidence of a transition in the interatomic forces that is independent of a thermodynamic phase transition.

\section{Statistical Mechanics of Liquid Phase Systems}

The following appendix comprises the key results used in my lectures on statistical mechanics and molecular simulations during the Fall, 2023 term at the University of Utah. The mathematical framework presented here forms the foundation of the statistical mechanical results presented in this thesis. The main references used in compiling these lectures were the following:

\begin{enumerate}
    \item Statistical Mechanics for Engineers - Isamu Kusaka \cite{kusaka_statistical_2015}
    \item Thermodynamics - Herbert Callen \cite{callen_thermodynamics_1998}
    \item Thermodynamics - Enrico Fermi \cite{fermi_thermodynamics_2012}
    \item The Collected Works of J. Willard Gibbs \cite{gibbs_collected_1928}
    \item Theory of Simple Liquids - Hansen and McDonald \cite{hansen_theory_2013}
\end{enumerate}

\subsection{Classical Thermodynamics}

An agglomerate of matter consists of an enormous amount of atoms. A single glass of water contains somewhere between $10^{24}$ and $10^{25}$ atoms alone! From a classical mechanics point of view, modeling the glass of water would require solving momentum equations for all $10^{25}$ atoms simultaneously. Clearly, this is beyond even the most powerful supercomputers that exist today. So, the question becomes, how can we study systems composed of an inconceivable number of particles without appealing to classical mechanics?

The characteristic time period of an atomic motion is on the order of $10^{-15}$ seconds. Therefore, even during a measurement of a system that is captured in a single microsecond, the atoms of a typical solid still go through ten million vibrations. This implies that a macroscopic measurement senses only averages of the atomic coordinates.

A map of how the atomic coordinates change is known as a normal mode. A normal mode is a coupled motion of atomic coordinates that include divergence (increase in particle density), convergence (decrease in particle density), and vibration. Some normal modes can be seen macroscopically, such as a change in volume or electric dipole. Others, like atomic vibrations, cannot be seen and are therefore "lost" in macroscopic observation. Classical thermodynamics is concerned with normal modes that are observable on a macroscopic scale.

By taking this macroscopic view, we lose a sense of how the motions of atomic coordinates can transfer energy. In Thermodynamics, we refer to this "invisible" mode of energy transfer as \textbf{heat}. From the Classical Mechanics perspective, heat is non-existent, since the energy of the system is completely characterized by generalized momentum and position for each particle. From this perspective, conservation of energy of a closed system requires that the change in energy be directly equal to the classical mechanical work done on the system, W'

\begin{equation}
    \Delta E = W'.
\end{equation}

To accommodate the fact that we cannot observe all forms of energy transfer macroscopically (what we will call work), we must split the W' into observable energy transfer, W, and unobservable energy transfer, Q (heat). This gives us the first law of thermodynamics

\begin{equation}
    \Delta E = W + Q.
\end{equation}

Now, in thermodynamics we are not typically concerned with the motion of a system in space or its change in position with respect to an external field, so we can simplify our energy conservation equation to simply include the internal energy, U,

\begin{equation}
    \Delta U = W + Q
\end{equation}

\noindent where W and Q are path functions, meaning that they depend on the exact way that changes are brought about by them. W is a path function because, in defining it, we have lost track of microscopic displacements, whereas classical mechanical work W' is a path-independent function. An infinitesimal view of this equation is shown below,

\begin{equation}
    dU = \dbar W + \dbar Q
\end{equation}

\noindent where $\; \dbar$ indicates an imperfect differential. The differential is imperfect since W and Q are path-dependent functions.

\subsubsection{Thermodynamic Postulates}

Thermodynamics is concerned with both reversible and irreversible processes, but for now let us consider equilibrium thermodynamics. An equilibrium state is a state of the system that, given a certain set of internal parameters U, V, and $N_i$ (i = 1, ..., n), the system tends to evolve towards. This leads us to the first postulate.

\begin{postulate}
There exist particular states, called equilibrium states, that are completely characterized by U, V, $N_i$ (i = 1, ..., n).
\end{postulate}

Note that in more complex systems, we require an inclusion of elastic strain parameters and electric dipole moment (also macroscopically measurable properties). A system at macroscopic equilibrium is a system where all representative atomic states of the system exist in the time scale of a macroscopic measurement.

\begin{postulate}
There exists a function, S, such that, S = S(U,V,N). Furthermore, the extensive parameters take values so that this function is maximized over the manifold of constrained equilibrium states.
\end{postulate}

This postulate applies only to equilibrium states, but in general not to non-equilibrium states. Essentially, we are postulating the existence of some function, called the entropy, that is maximized at equilibrium with respect to all of its dependent variables. The following two postulates apply to properties of the entropy and they are shown below.

\begin{postulate}
Entropy is additive over subsequent subsystems. Furthermore, S is a continuous, differentiable, and a monotonically increasing function of energy.
\end{postulate}

There are immediate consequences of this postulate which are listed below.

\begin{corollary}
\begin{equation}
S = \sum_{\alpha}S^{(\alpha)}
\end{equation}
\end{corollary}

\begin{corollary}
The entropy of a simple system is a homogeneous, first-order function of the extensive parameters.

\begin{equation}
S(\lambda U, \lambda V, \lambda N) = \lambda S(U, V, N)
\end{equation}
\end{corollary}

\begin{corollary}
The monotonic property implies that temperature is non-negative. In other words, 

\begin{equation}
\left(\frac{\partial S}{\partial U}\right )_{V,N} > 0.
\end{equation}

\end{corollary}

\begin{corollary}
Entropy can be inverted with respect to energy because it is a single-valued, continuous, and differentiable function with respect to S, V, N.
\end{corollary}

\begin{postulate}
The entropy of any system vanishes in the state when T = 0.
\end{postulate}

\subsection{The 2nd Law of the Thermodynamics}

Entropy can change as a result of internal or external processes. We express the differential change in entropy as, 

\begin{equation}
dS = \dbar{S_e} + \dbar{S_i}.
\end{equation}

\noindent We now accept that the expression for $\dbar{S_e}$ is given by,

\begin{equation}
\dbar{S_e} = \frac{\dbar{Q}}{T}.
\end{equation}

\noindent On the other hand, an important consequence of of the second law is that, 

\begin{equation}
\dbar{S_i} \geq 0.
\end{equation}

Equality holds in the previous equation when the process is reversible. A reversible process is such that the sequence of states visited by the system can be traversed in the opposite direction by an infinitesimal change in the boundary conditions. According to the second law, processes resulting in a decrease of the entropy are impossible for an isolated system. In terms of statistical mechanics, this is not actually the case (as will be seen later).

\subsubsection{Violating the 2nd Law: Maxwell's Demon}

Maxwell's demon is a thought experiment which shows how the second law can be violated. In Maxwell's words, "... if we conceive of a being whose faculties are so sharpened that he can follow every molecule in its course, such a being, whose attributes are as essentially finite as our own, would be able to do what is impossible to us. For we have seen that molecules in a vessel full of air at uniform temperature are moving with velocities by no means uniform, though the mean velocity of any great number of them, arbitrarily selected, is almost exactly uniform. Now let us suppose that such a vessel is divided into two portions, A and B, by a division in which there is a small hole, and that a being, who can see the individual molecules, opens and closes this hole, so as to allow only the swifter molecules to pass from A to B, and only the slower molecules to pass from B to A. He will thus, without expenditure of work, raise the temperature of B and lower that of A, in contradiction to the second law of thermodynamics."

\subsubsection{The Fundamental Relation}

To this point we have discussed entropy as a function on internal energy, volume, and the number of moles of each chemical species. We know that based on the properties of this function that we can solve for internal energy directly and it will be a function of entropy, volume, and the number of moles of each chemical species. Thus,

\begin{equation}
U = \left[U(S, V, N_i) :  i = (1, ..., n)\right].
\end{equation}

To examine infinitesimal changes in the energy, we compute the first differential of U to obtain,

\begin{equation}
dU = \left(\frac{\partial U}{\partial S}\right )_{V,N}dS + \left(\frac{\partial U}{\partial V}\right )_{S,N}dV + \sum_{i}\left(\frac{\partial U}{\partial N_i}\right )_{S,V}dN_i.
\end{equation}

The previous partial derivatives are called intensive parameters, and are given the following definitions.

\begin{equation}
    \left(\frac{\partial U}{\partial S}\right )_{V,N} \equiv T
\end{equation}

\begin{equation}
 \left(\frac{\partial U}{\partial V}\right )_{S,N} \equiv -P 
\end{equation}

\begin{equation}
 \left(\frac{\partial U}{\partial N_i}\right )_{S,V} \equiv \mu_i.
\end{equation}

\noindent With these definitions, the first differential dU can be expressed as, 

\begin{equation}
dU = TdS - PdV + \sum_{i}\mu_i dN_i.
\end{equation}

It should now be clear that $\dbar{Q} = TdS$, $\dbar{W_m} = - PdV$ (as shown previously), and that we have an additional work term known as the quasi-static chemical work, $\dbar{W_c} = \sum_{i}\mu_i dN_i$. We should note, however, that this does not hold for irreversible processes since these are path-dependent functions.

\subsubsection{Equations of State}

The intensive parameters defined previously are functions of S,V,N since they are partial derivatives of a function of those variables. Therefore, 

\begin{equation}
    T = T(S,V,N_i)
    \end{equation}
\begin{equation}
    P = P(S,V,N_i) 
\end{equation}
\begin{equation}
    \mu_i = \mu_i(S,V,N_i).
\end{equation}

These relationships that express intensive parameters in terms of the extensive parameters, are known as equations of state. In addition, since the fundamental relation is homogeneous first order, equations of state are homogeneous zero order. In other words,

\begin{equation}
    T(\lambda S, \lambda V, \lambda N_i) = T(S,V,N_i).
\end{equation}

Physically, we can interpret this as the temperature of two subsystems of a composite system is not additive. In fact, at equilibrium, any number of subsystems chosen from a composite system will have the same temperature as the composite system as a whole.

\subsubsection{The Euler Form}

Recall the first order homogeneous property of the fundamental relation.

\begin{equation}
U(\lambda S, \lambda V, \lambda N) = \lambda U(S, V, N).
\end{equation}

Differentiating with respect to $\lambda$, we find, 

\begin{equation}
    U = \left(\frac{\partial U}{\partial (\lambda S)}\right )_{V,N}\frac{\partial (\lambda S)}{\partial \lambda} + \left(\frac{\partial U}{\partial (\lambda V)}\right )_{S,N}\frac{\partial (\lambda V)}{\partial \lambda} + \sum_{i}\left(\frac{\partial U}{\partial (\lambda N_i)}\right )_{S,V}\frac{\partial (\lambda N_i)}{\partial \lambda}.
\end{equation}

This is true for any value of $\lambda$. Simplifying this expression and taking $\lambda = 1$,

\begin{equation}
    U = TS - PV  + \sum_{i}\mu_iN_i.
\end{equation}

\subsubsection{The Gibbs-Duhem Relation}

Taking the Euler form of the fundamental relation and taking an infinitesimal variation gives,

\begin{equation}
    dU = TdS + SdT - PdV - VdP + \sum_{i}\mu_idN_i + \sum_{i}N_id\mu_i.
\end{equation}

However, we know that the proper form of the equation from the fundamental relation is, 

\begin{equation}
    dU = TdS - PdV + \sum_{i}\mu_idN_i.
\end{equation}

We then take the difference between these two relations to obtain the Gibbs-Duhem Relation,

\begin{equation}
    0 = SdT - VdP + \sum_{i}N_id\mu_i.
\end{equation}

This relation demonstrates the intensive parameters (T, P, and $\mu$), cannot be varied independently. In addition, this relation shows that in a system of n components, n + 1 intensive parameters can be varied independently. This is referred to as the thermodynamic degrees of freedom of the system.

\subsubsection{Gibbs Phase Rule}

The interface between coexisting phases may be regarded as a partition that is diathermal, movable, and permeable to all species. Thus, at equilibrium, there must be a consistency in temperature, pressure, and chemical potential across all M phases.

\begin{equation}
    T_1 = T_2 = ... = T_M
\end{equation}

\begin{equation}
    P_1 = P_2 = ... = P_M
\end{equation}

\begin{equation}
    \mu_1 = \mu_2 = ... = \mu_M.
\end{equation}

If we perturb the state of the system by increasing the temperature by some $dT$, then in order to preserve the M phases, each phase parameter must be preserved after this change as well.

\begin{equation}
    T_1 + dT_1 = T_2+ dT_2 = ... = T_M+ dT_M
\end{equation}

\begin{equation}
    P_1 + dP_1 = P_2 + dP_2 = ... = P_M + dP_M
\end{equation}

\begin{equation}
    \mu_1 + d\mu_1 = \mu_2 + d\mu_2 = ... = \mu_M + d\mu_M.
\end{equation}

\noindent Thus we find that, 

\begin{equation}
    dT_1 = dT_2 = ... = dT_M
\end{equation}

\begin{equation}
   dP_1 = dP_2 = ... = dP_M
\end{equation}

\begin{equation}
     d\mu_1 =  d\mu_2 = ... = d\mu_M.
\end{equation}

Now, these relationships must be consistent with the Gibbs-Duhem relation, namely, 

\begin{equation}
    0 = SdT - VdP + \sum_{i}^{n}N_id\mu_i.
\end{equation}

If we hold pressure constant and vary only temperature, we find that, 

\begin{equation}
    0 = S^{(j)}dT + \sum_{i}^{n}N^{(j)}_id\mu_i
\end{equation}

\noindent for all phases $1,...,M$. But this requires that $dT$ and $d\mu_i$ are zero. That is, if we want to perturb the state of the system while maintaining the phase coexistence, we cannot possibly hold P constant. We now have (n + 2) infinitesimal quantities with M equations. Therefore, to ensure the system has solutions, we can only specify (n + 2 - M) of those variables independently. This is known as the Gibbs phase rule.

\subsubsection{Free Energies and the General Legendre Transform}

The Legendre transform is a mathematical tool that allows us to replace the extensive parameters as independent variables with intensive parameters. In general, the fundamental relation may be expressed as, 

\begin{equation}
    Y = Y(X_0, X_1, ... X_t).
\end{equation}

This gives a hyper-surface in a (t+2)-dimensional space with Cartesian coordinates $Y,X_0, X_1, ... X_t$. The partial slope of this hyper-surface is given by,

\begin{equation}
    P_k = \frac{\partial Y}{\partial X_k}.
\end{equation}

Now, the hyper-surface may be defined by the fundamental relation, or by the envelope of the tangent hyper-planes. If we consider the intercept of the hyper-plane, $\psi$, as a function of the slopes, then we obtain a relation,

\begin{equation}
    P_k = \frac{Y - \psi}{X_k - 0}.
\end{equation}

This is simply the "rise over run" equation for the slope. Now, this means that for any k, 

\begin{equation}
   Y - P_kX_k = \psi.
\end{equation}

\noindent Performing this calculation for all (t+1) dimensions gives the relation, 

\begin{equation}
   Y - \sum_{0}^{t+1}P_kX_k = \psi.
\end{equation}

\noindent Thus, taking the partial derivative with respect to $P_k$ gives,

\begin{equation}
    -X_k = \frac{\partial \psi}{\partial P_k}.
\end{equation}

\noindent Legendre transformed fundamental relations are called thermodynamic potentials.

\subsubsection{The Helmholtz Potential}

Suppose that we take the fundamental relation in the energy representation, $U = U(S,V,N_i)$. Then the slope of the hyper-surface formed by this fundamental relation has the usual form, 

\begin{equation}
    P_k = \frac{\partial U}{\partial X_k}.
\end{equation}

\noindent Specifically, if we consider the entropy as an extensive variable we would like to remove, we consider, 

\begin{equation}
    T = \frac{\partial U}{\partial S}.
\end{equation}

Then, we consider the intercepts of the tangent hyper-surface of this variable; namely, 

\begin{equation}
    T = \frac{U - F}{S - 0}.
\end{equation}

\noindent where $F$ is the name of the function of intercepts with zero entropy. Rearranging this equation we find the Helmholtz potential, 

\begin{equation}
    F = U - TS.
\end{equation}

Now, we have found $F = F(T,V,N_i)$. Thus, we can compute the first differential to determine, 

\begin{equation}
    dF = \left(\frac{\partial F}{\partial T}\right )_{V,N}dT + \left(\frac{\partial F}{\partial V}\right )_{T,N}dV + \sum_{i}\left(\frac{\partial F}{\partial N_i}\right )_{T,V}dN_i.
\end{equation}

\noindent But, from the relation $F = U - TS$, taking the first differential gives,

\begin{equation}
    dF = dU - TdS - SdT.
\end{equation}

\noindent Recall that the internal energy is given by,

\begin{equation}
    dU = TdS - PdV + \sum_{i}\mu_idN_i.
\end{equation}

\noindent so we substitute this expression in and find, 

\begin{equation}
    dF = - SdT - PdV + \sum_{i}\mu_idN_i.
\end{equation}

\noindent This indicates the definitions of the partial derivatives of the Helmholtz potential.

\begin{equation}
    \left(\frac{\partial F}{\partial T}\right )_{V,N}=-S
\end{equation}

\begin{equation}
    \left(\frac{\partial F}{\partial V}\right )_{T,N} = -P
\end{equation}
    
\begin{equation}
    \sum_{i}\left(\frac{\partial F}{\partial N_i}\right )_{T,V} = \mu_i.
\end{equation}

\noindent The enthalpy is found by taking the legendre transform of U by replacing V with P, and the Gibbs potential is found by replacing S and V with T and P.

\subsubsection{The Maxwell Relations}

The Maxwell relations are simply a statement of the equality of the mixed partial derivatives of the fundamental relation.

\begin{equation}
  \frac{\partial^2 U}{\partial S \partial V} =  \frac{\partial^2 U}{\partial V \partial S}
\end{equation}

\begin{equation}
  - \left(\frac{\partial P}{\partial S}\right )_{V,N_i} =  \left(\frac{\partial T}{\partial V}\right )_{S,N_i}.
\end{equation}

Given a thermodynamic potential with $(t+1)$ natural variables, there are $\frac{t(t+1)}{2}$ separate pairs of mixed partial derivatives and thus that many Maxwell relations.

\subsubsection{Jacobian Transformations}

An excellent method for manipulation of thermodynamic derivatives is based on the mathematical properties of Jacobians. If $u,v. ...w$ are functions of $x,y,...z$, the Jacobian is defined as, 

\begin{equation}
   \frac{\partial (u,v,...,w)}{\partial (x,y,...,z)} \equiv 
   \begin{vmatrix}
   \frac{\partial u}{\partial x}&\frac{\partial u}{\partial y}&...&\frac{\partial u}{\partial z}\\
   \frac{\partial v}{\partial x}&\frac{\partial v}{\partial y}&...&...\\
   \frac{\partial w}{\partial x}&...&...&\frac{\partial w}{\partial z}
   \end{vmatrix}.
\end{equation}

The property that makes this incredibly useful in thermodynamics is the relationship,

\begin{equation}
    \left(\frac{\partial u}{\partial x}\right)_{y...z} = \frac{\partial (u,y,...,z)}{\partial (x,y,...,z)}
\end{equation}

Jacobians have the following nice properties.

\begin{equation}
\frac{\partial (u,v,...,w)}{\partial (x,y,...,z)}= -\frac{\partial (v,u,...,w)}{\partial (x,y,...,z)}
\end{equation}

\begin{equation}
\frac{\partial (u,v,...,w)}{\partial (x,y,...,z)}= \frac{\partial (u,v,...,w)}{\partial (r,s,...,t)}\frac{\partial (r,s,...,t)}{\partial (x,y,...,z)}
\end{equation}

\begin{equation}
\frac{\partial (u,v,...,w)}{\partial (x,y,...,z)}= 1 / \frac{\partial (x,y,...,z)}{\partial (u,v,...,w)}.
\end{equation}

\subsubsection{The Nernst Postulate}

The third law of thermodynamics, otherwise known as the Nernst postulate, states that,

\begin{postulate}
    The entropy of a system approaches a constant value at absolute zero temperature.
\end{postulate}

From a microstate perspective, at absolute zero temperature there is either only one unique microstate, known as the ground state, or there are a set of minimum energy microstates, which are finite. An equivalent statement of the nernst postulate is,

\begin{postulate}
    It is impossible by any procedure, no matter how idealized, to reduce the temperature of any closed system to zero temperature in a finite number of finite operations.
\end{postulate}

Suppose that the temperature of a substance can be reduced in an isentropic process by changing the parameter X from X2 to X1. One can think of a multistage nuclear demagnetization setup where a magnetic field is switched on and off in a controlled way. If there were an entropy difference at absolute zero, T = 0 could be reached in a finite number of steps. However, at T = 0 there is no entropy difference so an infinite number of steps would be needed.

\subsubsection{Stability of Thermodynamic Systems}

The requirement that the second derivative of energy be positive gives rise to some interesting findings. To see this, let's first construct an intrinsic system and a complimentary subsystem from a composite, isolated system. The complimentary subsystem is assumed to be much larger than the intrinsic subsystem.  The fundamental relation is given by,

\begin{equation}
    U' = X_t u (x_0, x_1, ..., x_{t-1}+ \tilde{X}_t  \tilde{u} (\tilde{x}_0, \tilde{x}_1, ..., \tilde{x}_{t-1})
\end{equation}

\noindent where $X_t$ represents some parameter of the fundamental relation that is being held constant. Furthermore, the smallness of the intrinsic system compared to the complementary system requires that,

\begin{equation}
    |d\tilde{x}_i| << |dx_i|.
\end{equation}

Now, any changes in the $x_i$ leads to a total change in energy. We can express the energy change in terms of a Taylor expansion such that, 

\begin{equation}
   \Delta U' = X_t[du + d^2u + ...] + \tilde{X}_t d \tilde{u}
\end{equation}

\noindent where,

\begin{equation}
    du = \sum_{i}\frac{\partial u}{\partial x_i}dx_i
\end{equation}

\begin{equation}
    d^2u = \frac{1}{2!}\sum_{j}\sum_{i}\frac{\partial^2 u}{\partial x_i\partial x_j}dx_idx_j.
\end{equation}

\noindent We can neglect higher order terms of $\tilde{u}$ since we are assuming that the changes in the independent variables are very small for the complimentary system. Now, since the composite system is isolated, by our usual formalism the change in internal energy for the composite system vanishes,

\begin{equation}
    0 = X_t\sum_{i}\frac{\partial u}{\partial x_i}dx_i + \tilde{X}_t\sum_{i}\frac{\partial \tilde{u}}{\partial \tilde{x}_i}d \tilde{x}_i.
\end{equation}

This leads us to the usual results that all of the intensive parameters (T, P, $\mu_i$) are the same in each subsystem at equilibrium. In other words, the composite system is homogeneous. However, we now turn our attention to the second requirement; namely, 

\begin{equation}
    d^2u = \frac{1}{2!}\left[\sum_{0}^{t-1}\sum_{0}^{t-1}\frac{\partial^2 u}{\partial x_i\partial x_j}dx_idx_j\right] > 0.
\end{equation}

The quantity in brackets is known as the homogeneous quadratic form. The condition that the quadratic form be positive for all combinations of variables is referred to mathematically as the condition that the quadratic form be positive definite.

Notice that we have numerous cross terms in the expression for $d^2u$. The presence of these cross terms make it difficult to determine what quantities must be positive, so we need to perform a linear transformation such that the quadratic form is a sum of squares. Sylvester's law of inertia guarantees that no matter which transformation we pick, the number of positive, negative, and zero coefficients will be the same.

For this discussion, we will proceed by considering terms that contain $dx_0$. We can thus express the equation as, 

\begin{equation}
    d^2u = \frac{1}{2}\left[u_{00}^2(dx_0)^2 + +2\sum_{1}^{t-1}\frac{\partial^2 u}{\partial x_0\partial x_k}dx_0dx_k+ \sum_{1}^{t-1}\sum_{1}^{t-1}\frac{\partial^2 u}{\partial x_i\partial x_j}dx_idx_j\right] > 0.
\end{equation}

We now eliminate the cross terms by introducing the new variable, $dP_0$.

\begin{equation}
    dP_0 = u_{00}dx_0 + \sum_{1}^{t-1}u_{0k}dx_k
\end{equation}

\begin{equation}
    (dP_0)^2 = u_{00}^2(dx_0)^2 + 2 u_{00}\sum_{1}^{t-1}u_{0k}dx_k+\sum_{1}^{t-1}\sum_{1}^{t-1}u_{0j}u_{0k}dx_jdx_k
\end{equation}

\noindent which allows us to express the second derivative of internal energy as, 

\begin{equation}
    d^2u = \frac{1}{2}\left[\frac{1}{u_{00}}(dP_0)^2+ \sum_{1}^{t-1}\sum_{1}^{t-1}(u_{jk}-\frac{u_{0j}u_{0k}}{u_{00}})dx_jdx_k\right] > 0.
\end{equation}

Now, notice that the previous equation can be rewritten as, 

\begin{equation}
    d^2u = \frac{1}{2}\left[\frac{1}{u_{00}}(dP_0)^2+ \sum_{1}^{t-1}\sum_{1}^{t-1}(u_{jk})_{P_0, x_1, x_2, ...}dx_jdx_k\right] > 0.
\end{equation}

We can rewrite this expression in a helpful way. Consider the following mathematical fact, 

\begin{equation}
    \bigg(u_{jk}-\frac{u_{0j}u_{0k}}{u_{00}}\bigg) = \bigg(\frac{\partial^2 (u - P_0x_0)}{\partial x_j \partial x_k}\bigg)_{P_0, x_1, ...} = \frac{\partial^2 \psi^{(0)}}{\partial x_j \partial x_k}.
\end{equation}

The function $\psi^{(0)}$ is the Legendre transform of u with respect to $P_0$. This allows us to now write the second derivative of u as, 

\begin{equation}
    d^2u = \frac{1}{2}\left[\frac{1}{u_{00}}(dP_0)^2+ \sum_{1}^{t-1}\sum_{1}^{t-1}\psi_{jk}^{(0)}dx_jdx_k\right] > 0.
\end{equation}

Proceeding in this way, for each of the $dx_i$, we arrive at a useful form given by,

\begin{equation}
    d^2u = \frac{1}{2}\sum_{0}^{t-1}\frac{1}{\psi_{jj}^{(j-1)}}(dP_j^{(j-1)})^2 > 0.
\end{equation}

Thus, we require that all of the $\psi_{jj}^{(j-1)}$ be positive such that,

\begin{equation}
    \psi_{jj}^{(j-1)} = \bigg(\frac{\partial P_j}{\partial x_j}\bigg)_{P_0, P_1 ... P_{j-1},x_{j+1},x_{j+2}...x_{t-1}} > 0. 
\end{equation}

So, what exactly does this mean? Common terms that must be positive according to this definition are,

\begin{equation}
    -\bigg(\frac{\partial P}{\partial V}\bigg)_T > 0
\end{equation}

\begin{equation}
     \bigg(\frac{\partial U}{\partial T}\bigg)_v = C_v > 0.
\end{equation}

Of course, depending on the complexity of the system, we can have many more relations from this procedure. However, in fluid phase transitions, we are typically most concerned with the first of the two explicit equations and the second is just assumed to be true.

\subsection{Statistical Mechanics}

\subsubsection{Dynamical Variable Measurement}

We have studied how particle positions and momenta evolve in time according to Hamilton's equations of motion. Of course, when dealing with real systems, we are almost never concerned with the motion of all of the particles in the system and we know that materials can be well-described by simple thermodynamic laws. However, up to this point, you may feel a bit uneasy that we have swept the dynamics of particles under the rug when we discussed thermodynamics. So far, we collected all of the microscopic modes of energy transfer into the heat term and introduced the entropy fundamental relation. From these very simple assumptions we have found that we can derive nearly all of thermodynamics. But what is heat? Or entropy? To answer these questions, we need to study statistical mechanics. We will see that statistical mechanics is a powerful tool to unite microscopic dynamics with macroscopic theory. Statistical mechanics holds whether the system is classical or quantum mechanical, and has been applied to study systems as diverse as molecular systems, black holes, and bird flocking.

Recall that a dynamical variable was some variable $A = A(q^f, p^f, t)$. Let's suppose that we want to measure the dynamic variable in the lab. How do we actually define what the true value is? Is it the value of $A$ at the specific instant in time we measure it? How do we accommodate for the fact that $A$ could be changing extremely fast due to atomic motions?

We can define the value of $A$ as its average over the amount of time that the variable is measured. If the time scale of the measurement is significantly larger than the time scale of molecular motion, the value of A that we measure will take the form,

\begin{equation}
    A_{exp} = \lim_{\tau \rightarrow \infty}\frac{1}{\tau}\int_{t}^{t+\tau}A(q^f(t'),p^f(t'),t')dt'
\end{equation}

\noindent which is just an integral relation for the time average of $A$ over a measurement time $\tau$. To see this, just consider that the integral is equivalent to adding up all of the contributions to $A$ over time and the $1/\tau$ term is dividing by the total time of the measurement. The limit ensures that the time measurement is long enough to accurately capture the average. This is analogous to taking an average of a set of measurements in which you add up all the measurement values and divide by the number of samples.

\subsubsection{Phase Space}

The mechanical state of a system is specified by 2f variables; namely, the momentum and generalized coordinates of the $f$ particles in a system. The mechanical state comprises some specific point in this space, and we refer to this point as the phase point. Over time, the phase point changes according to the equations of motion, and forms a phase trajectory.

\subsubsection{Ensemble Average}

Suppose that we allow some 1-D system to evolve over a time $\tau$. If we choose a volume element, then the probability that the system phase is in the volume element is given by,

\begin{equation}
    \rho dq^fdp^f = \lim_{\tau \rightarrow \infty}\frac{\Delta t(q^f,p^f,t)}{\tau}
\end{equation}

\noindent where $\Delta t$ is the total time that the particle spends in the volume element. We call $\rho = \rho(q^f, p^f, t)$ the probability density function. Now, we can combine the dynamical variable measurement for the volume element with our probability density definition to see that,

\begin{equation} \label{eq:dynamical}
    A_{exp}(t) = \int A(q^f,p^f,t)\rho (q^f,p^f,t)dq^fdp^f.
\end{equation}

\noindent Note that this integration occurs over the entire phase space. Furthermore, note that this equation only holds if explicit time dependence of A occurs sufficiently slowly over $\tau$.

Since $\rho$ is a probability density function, it has a few special properties: 

\begin{enumerate}
    \item $\rho$ is a probability distribution function, so it must normalize to 1. This means that,

    \begin{equation}
        \int \rho dq^fdp^f = 1
    \end{equation}

    \item $\rho$ must be non-negative so that, $\rho \leq 0$.
\end{enumerate}

It is very reasonable to question why we would rewrite our original expression with some new, seemingly not useful function $\rho$. We haven't really made too much progress computationally, as all we have done is rearranged the problem with some new definitions. However, the hope is that by defining $\rho$ in this way that we can find a $\rho$ to use in eq \eqref{eq:dynamical} that will allow us to solve for $ A_{exp}(t)$ without ever needing to solve the equations of motion directly. Furthermore, by reformulating the motion of many-particles with respect to this probability density function, perhaps we can gain novel insight that would otherwise be difficult to describe or notice.

\subsubsection{Statistical Equilibrium}

Statistical equilibrium occurs when the change in all measured dynamic variables with respect to time is zero, unless the dynamic variable depends explicitly on time. This implies that at equilibrium $A = A(q^f, p^f)$, where the explicit time-dependence is dropped. Thus,

\begin{equation}
    \frac{dA_{exp}}{dt} = \lim_{\Delta t \rightarrow 0}\frac{1}{\Delta t}\left[\int A(q^f, p^f)[\rho(q^f, p^f, t+\Delta t)-\rho(q^f, p^f, t)]\right]dq^fdp^f
\end{equation}

\noindent and applying a first-order Taylor expansion on the probability density function gives,

\begin{equation}
    \frac{dA_{exp}}{dt}= \lim_{\Delta t \rightarrow 0}\int A(q^f, p^f)\frac{\partial \rho}{\partial t}dq^fdp^f.
\end{equation}

Therefore, the condition of statistical equilibrium is that $\rho$ does not depend explicitly on time, or,

\begin{equation}
    \frac{\partial \rho}{\partial t} = 0.
\end{equation}

\noindent Since the previous expression must hold for any choice of $A$. 

\subsubsection{The Statistical Ensemble}

The concept of a statistical ensemble was introduced by Gibbs. Let's suppose that we make $\mathcal{N}$ identical copies of a system, and we completely transfer the explicit time dependence and Hamiltonian for each copy. This is known as the statistical ensemble. Now, the number of copies of the system that fall in the volume element $dq^fdp^f$ is given by,

\begin{equation}
    \mathcal{N}\rho(q^f,p^f,t)dq^fdp^f.
\end{equation}

\noindent This represents the number of copies in the phase space around the point $(q^f,p^f)$. Noting that copies of the system are neither generated nor destroyed in this process, the total copy density $N\rho$ is conserved.

\subsubsection{Louiville's Theorem and the Canonical Ensemble}

Recall the equation of continuity from continuum mechanics. The equation of continuity amounts to writing down a mass balance on a system noting that the total mass of a system is conserved. For some fixed region of space, referred to as a control volume, the number of particles in the control volume at any moment $t$ is given by,

\begin{equation}
    \int_V \rho(\mathbf{r}, t) d \mathbf{r}
\end{equation}

\noindent where $V$ denotes the control volume and $\rho(\mathbf{r}, t)$ is the number density of particles at position $\mathbf{r}$ at time $t$. If there are no chemical reactions, the rate of change of this integral with respect to time is related to the flux of atoms across the surface of the control volume so that,

\begin{equation}
    \frac{d}{dt} \int_V \rho(\mathbf{r}, t) d \mathbf{r} = \oint_S \rho(\mathbf{r}, t) v(\mathbf{r}, t) \cdot \mathbf{n}(\mathbf{r}) dS
\end{equation}

\noindent where $v$ is the average velocity of particles through the surface element $dS$ and $\mathbf{n}(\mathbf{r})$ is the outward unit normal to the control volume surface. We can rewrite the left hand side using the Leibniz integral rule (which can be proved by invoking the bounded convergence theorem) such that,

\begin{equation}
    \frac{d}{dt} \int^{b(\mathbf{r})}_{a(\mathbf{r})} \rho(\mathbf{r}, t) d \mathbf{r} = \rho(\mathbf{r}, b) \frac{d b(\mathbf{r})}{dt} -  \rho(\mathbf{r}, a) \frac{d a(\mathbf{r})}{dt} + \int^{b(\mathbf{r})}_{a(\mathbf{r})} \frac{\partial}{\partial t} \rho(\mathbf{r}, t) d \mathbf{r}
\end{equation}

\noindent where the first two terms go to zero since $a(\mathbf{r}), b(\mathbf{r})$ are fixed in time (control volume cannot deform). We then obtain,

\begin{equation}
    \frac{d}{dt} \int_V \rho(\mathbf{r}, t) d \mathbf{r} = \int_V \frac{\partial}{\partial t} \rho(\mathbf{r}, t) d \mathbf{r}.
\end{equation}

\noindent Now, the term on the right hand side can be rewritten according to the divergence theorem,

\begin{equation}
    \oint_S \rho(\mathbf{r}, t) v(\mathbf{r}, t) \cdot \mathbf{n}(\mathbf{r}) dS = \int_V \nabla \cdot (\rho \mathbf{v}) d\mathbf{r}
\end{equation}

\noindent which after plugging into our original expression and rearranging under the integral we obtain,

\begin{equation}
    \int_V \frac{\partial \rho}{\partial t} + \nabla \cdot (\rho \mathbf{v}) d\mathbf{r} = 0.
\end{equation}

\noindent But since this holds for any $V$, we must have, 

\begin{equation}
    \frac{\partial \rho}{\partial t} + \nabla \cdot (\rho \mathbf{v}) = 0.
\end{equation}

\noindent This is an expression for the equation of continuity in continuum mechanics. By analogy, in phase space we have a control volume in 2$f$-coordinates $(q_1, ..., q_f, p_1, ..., p_f)$ and velocities $(\dot{q}_1, ..., \dot{q}_f, \dot{p}_1, ..., \dot{p}_f)$ with a total number density of $\mathcal{N}\rho$. Plugging these into the equation of continuity gives,

\begin{equation}
    \frac{\partial \mathcal{N}\rho}{\partial t} + \sum_{i=1}^f \left[\frac{\partial (\mathcal{N}\rho \dot q_i)}{\partial q_i}+\frac{\partial (\mathcal{N}\rho \dot p_i)}{\partial p_i}\right] = 0.
\end{equation}

\noindent We note that $\mathcal{N}$ is a constant and apply Hamilton's equations of motion to find,

\begin{equation}
    \frac{\partial\rho}{\partial t} + \sum_{i=1}^f  \left[ \frac{\partial}{\partial q_i}\bigg(\rho\frac{\partial\mathcal{H}}{\partial p_i}\bigg) - \frac{\partial}{\partial p_i} \bigg(\rho \frac{\partial \mathcal{H}}{\partial q_i}\bigg)\right] = 0
\end{equation}

\noindent which implies that,

\begin{equation}
    \frac{d \rho}{d t} = \frac{\partial \rho}{\partial t} + \{\rho,H\} = 0
\end{equation}

\noindent or in other words, that $\rho$ is a constant of motion. This means that the number of copies in a statistical ensemble is conserved. This is a statement that holds whether or not the system is in statistical equilibrium or not. However, we now require that the probability be a constant of motion with respect to the Hamiltonian and that it not depend explicitly on time in order for the system to be in statistical equilibrium.

\subsubsection{An Expression for Density}

So now the question becomes: how do we express $\rho$ in a productive way? What functional form does it take? We first make an assumption that $\rho = \rho(H)$ only. Now, we proceed to hypothesize that $\rho$ takes the following form,

\begin{equation}
    \rho(q^f,p^f) = \frac{1}{C}e^{-\beta H(q^f,p^f)}
\end{equation}

\noindent where C is determined by the normalization condition such that,

\begin{equation}
    1 = \int \rho (q^f,p^f) dq^fdp^f = \int\frac{1}{C}e^{-\beta H(q^f,p^f)}dq^fdp^f.
\end{equation}

The statistical ensemble characterized by a $\rho$ of this form is referred to as the canonical ensemble. The ensemble average of the energy H is called the internal energy ($U$), and is given by,

\begin{equation}
    U \equiv \langle H \rangle = \int H \rho(q^f,p^f) dq^fdp^f
\end{equation}

\noindent where all we have done is used new notation for the same integral we introduced earlier when we discussed ensemble averages,

\begin{equation}
    \langle A \rangle = \int A(q^f, p^f) \rho(q^f, p^f) dq^f dp^f = \frac{1}{C}\int A(q^f, p^f) e^{-\beta H(q^f, p^f)} dq^f dp^f
\end{equation}

\noindent for the given form of $\rho$. We now have an expression for $U$ in terms of the microscopic quantities from statistical mechanics. However, we also have an expression for $U$ in thermodynamics; namely, 

\begin{equation}
    dU = TdS - PdV
\end{equation}

\noindent which will allow us to relate the thermodynamic concepts of temperature, entropy, and work to the microscopic expression for $U$. To do this, we let $\Theta = \beta^{-1}$ and rewrite the canonical ensemble probability density in terms of $H$,

\begin{equation}
    H = -\Theta log(\rho) - \Theta log(C).
\end{equation}

\noindent Taking the ensemble average of these quantities gives,

\begin{equation}
    U = \Theta \eta + \alpha
\end{equation}

\noindent where we have defined two new variables given by,

\begin{equation}
    \eta = -\langle log(\rho) \rangle
\end{equation}

\begin{equation}
    \alpha = -\Theta log(C).
\end{equation}

\noindent Taking the total derivative of $U$ we find that we can express dU as,

\begin{equation}
    dU = \eta d\Theta + \Theta d\eta + d\alpha
\end{equation}

\subsubsection{Intuition Using a Piston Cylinder}

We now need to expand the $d\alpha$ term in a systematic way so that we can compare this expression with the one we are familiar with from thermodynamics. To this end, we consider an example system of particles confined to a piston cylinder device. How do we write $H$ for this system? $H$ is a function of the particle velocities, interactions, and piston position (since it imparts an external force on the system), which we can write as,

\begin{equation}
    H = \sum_{i=1}^N \frac{||\mathbf{p_i}||^2}{2m_i} + \phi(\mathbf{r}^N) + \psi(\mathbf{r}^N, \lambda)
\end{equation}

\noindent where $\lambda$ is the position of the piston and is common to all particles in the system. The point here is that we can perform work on the system by changing $\lambda$ and that the Hamiltonian is a function of $\lambda$. From the normalization condition we have that $C=C(\Theta, \lambda)$. The total derivative of $C$ is then,

\begin{equation}
    dC = \bigg(\frac{\partial C}{\partial \Theta}\bigg)_\lambda d\Theta + \bigg(\frac{\partial C}{\partial \lambda}\bigg)_\Theta d\lambda
\end{equation}

\noindent and we also know from our definition of $\alpha$ that,

\begin{equation}
    d \alpha = -\Theta d \log{C} - \log{C} d\Theta = -\Theta \frac{dC}{C} - \log{C} d\Theta
\end{equation}

\noindent and rewriting this expression by plugging in for $dC$ and $\log{C} = -\alpha/\Theta$ gives,

\begin{equation}
    d\alpha = \bigg[\frac{\alpha}{\Theta} - \frac{\Theta}{C}\bigg(\frac{\partial C}{\partial \Theta}\bigg)_\lambda \bigg] d \Theta - \frac{\Theta}{C}\bigg(\frac{\partial C}{\partial \lambda}\bigg)_\Theta d\lambda.
\end{equation}

\noindent Now, let's rewrite these partial derivatives with respect to the integral of the probability density function. This gives,

\begin{equation}
    \frac{1}{C}\bigg(\frac{\partial C}{\partial \Theta}\bigg)_\lambda = \frac{1}{C} \frac{\partial}{\partial \Theta} \int e^{-H/\Theta}dq^fdp^f = \frac{1}{C} \int \frac{\partial}{\partial \Theta} e^{-H/\Theta}dq^fdp^f
\end{equation}

\noindent and taking the derivative of the inside gives,

\begin{equation}
    = \frac{1}{C} \int \frac{H}{\Theta^2} e^{-H/\Theta}dq^fdp^f = \frac{U}{\Theta^2}.
\end{equation}

\noindent Similarly,

\begin{equation}
    \frac{1}{C}\bigg(\frac{\partial C}{\partial \lambda}\bigg)_\Theta = \frac{1}{C} \frac{\partial}{\partial \lambda} \int e^{-H/\Theta}dq^fdp^f = \frac{1}{C} \int \frac{\partial}{\partial \lambda} e^{-H/\Theta}dq^fdp^f
\end{equation}

\noindent which is just,

\begin{equation}
    = \frac{1}{C} \int -\frac{1}{\Theta}\frac{\partial H}{\partial \lambda} e^{-H/\Theta}dq^fdp^f =  -\frac{1}{\Theta} \bigg\langle \frac{\partial H}{\partial \lambda} \bigg\rangle.
\end{equation}

\noindent Combining these expressions into our expression for $d\alpha$,

\begin{equation}
    d\alpha = \bigg[\frac{\alpha}{\Theta} - \Theta \frac{U}{\Theta^2} \bigg] d \Theta + \Theta \frac{1}{\Theta} \bigg\langle \frac{\partial H}{\partial \lambda} \bigg\rangle d\lambda = \frac{1}{\Theta}(\alpha - U)d\Theta + \bigg\langle \frac{\partial H}{\partial \lambda} \bigg\rangle d\lambda
\end{equation}

\noindent and finally substituting in our original expression for $U = \Theta \eta + \alpha$ we obtain,

\begin{equation}
    dU = \Theta d\eta +\bigg \langle \frac{\partial H}{\partial \lambda}\bigg \rangle d\lambda.
\end{equation}

\noindent Thus, we find that the first term is related to $TdS$ and the second term is related to the work done on the system. The fact that $TdS = \Theta d\eta$ means that,

\begin{equation}
    \Theta \propto T
\end{equation}

\begin{equation}
    d \eta \propto dS.
\end{equation}

\noindent We choose the constant of proportionality as the Boltzmann constant, $k_B$, so that,

\begin{equation}
    d \eta = d\bigg(\frac{S}{k_B}\bigg)
\end{equation}

\noindent which implies that a functional form for the entropy in terms of classical statistical mechanical variables is,

\begin{equation}
    S = -k_B\langle log(\rho)\rangle.
\end{equation}

\noindent Also note that we have the relation for $\alpha$ given by,

\begin{equation}
    \alpha = U - \Theta \eta = U - TS = A.
\end{equation}

\noindent This gives us an approximate form for the Helmholtz free energy as,

\begin{equation}
    A = -k_BTlog(C).
\end{equation}

\subsubsection{Motivation for the Canonical Ensemble}

\noindent The previous section gave a definition for $\beta$ as,

\begin{equation}
    \beta = \frac{1}{k_BT}.
\end{equation}

In the canonical ensemble definition for the probability of observing a particle in some volume element around the point $(q^f,p^f)$, we now have,

\begin{equation}
    \rho (q^f,p^f) = \frac{1}{C}e^{-H(q^f,p^f)/k_BT}.
\end{equation}

\noindent But this means that the temperature of the system must be specified to determine $\rho$. To do this in practice, we need to allow for transfer of energy with a system and its surroundings. In fact, we now show that the canonical ensemble is the only possible ensemble that can account for a system in thermal contact with its surroundings. Let $\mathcal{T}$ be an isolated system containing a subsystem $\mathcal{S}$ enclosed in a rigid, impermeable wall, surrounded by subsystem $\mathcal{L}$ called the surroundings. The Hamiltonian of the system can be written as,

\begin{equation}
    H_\mathcal{T}(q^{m+n},p^{m+n}) = H_\mathcal{S}(q^m,p^m) + H_\mathcal{L}(q^{n},p^{n}) + H_{int}(q^{m+n},p^{m+n}).
\end{equation}

\noindent The last term arises from the interactions between the subsystem $\mathcal{S}$ and the surroundings $\mathcal{L}$. m,n are the mechanical degrees of freedom of $\mathcal{S}$ and $\mathcal{L}$, respectively. We need to assume that the internal interaction Hamiltonian is negligible in magnitude compared to the subsystem and surroundings Hamiltonian's; meaning,

\begin{equation}
    H_\mathcal{T}(q^{m+n},p^{m+n}) \approx H_\mathcal{S}(q^m,p^m) + H_\mathcal{L}(q^{n},p^{n}).
\end{equation}

\noindent This implies that the interaction between $\mathcal{S}$ and $\mathcal{L}$ is weak, but still sufficient to allow for a transfer of energy between the two systems over a long period of time. This transfer of energy is sufficient to maintain a constant temperature without contributing significantly to the total Hamiltonian.

We now split the subsystem $\mathcal{S}$ into two sub-subsystems A and B, with mechanical degrees of freedom $m_A$ and $m_B$, such that $m =m_A+m_B$. We now suppose the interactions between these two are sufficiently weak so that,

\begin{equation}
    H_\mathcal{S}(q^{m},p^{m}) \approx H_A(q^{m_a},p^{m_a}) + H_B(q^{m_b},p^{m_b}).
\end{equation}

\noindent This approximate independence implies that,

\begin{equation}
    \rho_{\mathcal{S}} \approx \rho_A\rho_B
\end{equation}

\noindent since independent probabilities can be multiplied together to give the total probability. Taking the logarithm of both sides gives,

\begin{equation}
    log(\rho_\mathcal{S}) \approx log(\rho_A) + log(\rho_B).
\end{equation}

\noindent However, we know from Liouville's Theorem that $\rho$ is a function of constants of motion. We also now know that the logarithm of $\rho$ is additive. This necessarily implies that $log(\rho)$ is a linear function of constants of motion that are additive. If we assume that $\rho = \rho(H)$, we will obtain the following expression,

\begin{equation}
    log(\rho) = mH + b
\end{equation}

\begin{equation}
    \rho = e^be^{mH}
\end{equation}

\noindent which is precisely the form of the canonical distribution function supposed earlier. You may wonder why we can still apply Liouville's theorem here since the system in thermal contact is not isolated. The reason for this is that over a small time interval, if the interactions between systems are sufficiently weak, then the system will behave approximately isothermally over that time. If we stitch together a large number of these time intervals, we expect to find the canonical distribution.

\subsubsection{The Maxwell-Boltzmann Distribution}

One application is to find the distribution of momentum of a particle in a system of $N$ particles held in thermal equilibrium with its surroundings. In the absence of an external field, we can write the Hamiltonian for such as a system as,

\begin{equation}
    H(\mathbf{r}^N, \mathbf{p}^N) = \sum_{i = 1}^N \frac{||\mathbf{p}_i||^2}{2m_i} + \phi(\mathbf{r}^N)
\end{equation}

\noindent where $\phi(\mathbf{r}^N$) is recognized as the particle interaction potential energy contribution. To find the probability distribution on the velocity of a single particle, we can proceed to integrate out all position and momentum coordinates except for a single arbitray particle to obtain,

\begin{equation}
    \rho(\mathbf{p}_1)d\mathbf{p}_1
\end{equation}

\noindent which is just the momentum probability density function for a single particle. Proceeding in this way, let's start with the partition function in the canonical ensemble,

\begin{equation}
    \rho d\mathbf{r}^N d\mathbf{p}^N = \frac{1}{C}\int e^{-\beta H} d\mathbf{r}^N d\mathbf{p}^N
\end{equation}

\noindent and substituting $H$ into the equation we get,

\begin{equation}
    \rho d\mathbf{r}^N d\mathbf{p}^N = \frac{e^{-\beta \bigg(\sum_{i = 1}^N \frac{||\mathbf{p}_i||^2}{2m_i} + \phi(\mathbf{r}^N)\bigg)}}{\int e^{-\beta \bigg(\sum_{i = 1}^N \frac{||\mathbf{p}_i||^2}{2m_i} + \phi(\mathbf{r}^N)\bigg)} d\mathbf{r}^N d\mathbf{p}^N}.
\end{equation}

\noindent We can begin by solving for $C$ to obtain,

\begin{equation}
    C = \int e^{-\beta \bigg(\sum_{i = 1}^N \frac{||\mathbf{p}_i||^2}{2m_i} + \phi(\mathbf{r}^N)\bigg)} d\mathbf{r}^N d\mathbf{p}^N = \int \prod_{i=1}^N e^{-\beta \frac{||\mathbf{p}_i||^2}{2m_i}} d\mathbf{p}^N \int e^{-\beta \phi(\mathbf{r}^N)} d\mathbf{r}^N
\end{equation}

\noindent which is equal to,

\begin{equation}
    C = \prod_{i=1}^N \bigg(\frac{2 \pi m_i}{\beta}\bigg)^{3/2} \int e^{-\beta \phi(\mathbf{r}^N)} d\mathbf{r}^N.
\end{equation}

\noindent Recognizing that we now just need to take the integral over $\rho d\mathbf{r}^N d\mathbf{p}^N$ with respect to all variables but a single momentum, we obtain,

\begin{equation}
   C \int \rho d\mathbf{r}^N d\mathbf{p}^{N-1} = \int e^{-\beta \bigg(\sum_{i = 1}^N \frac{||\mathbf{p}_i||^2}{2m_i} + \phi(\mathbf{r}^N)\bigg)} d\mathbf{r}^N d\mathbf{p}^{N-1}
\end{equation}

\noindent which gives,

\begin{equation}
    C \rho(\mathbf{p}_1) d\mathbf{p}_1 = \prod_{i=1}^{N-1} \bigg(\frac{2 \pi m_i}{\beta}\bigg)^{3/2} e^{-\beta \frac{||\mathbf{p}_1||^2}{2m_1}} d \mathbf{p}_1 \int e^{-\beta \phi(\mathbf{r}^N)} d\mathbf{r}^N
\end{equation}

\noindent and finally dividing by $C$ on both sides,

\begin{equation}
    \rho(\mathbf{p}_1) d\mathbf{p}_1 = \frac{\prod_{i=1}^{N-1} \bigg(\frac{2 \pi m_i}{\beta}\bigg)^{3/2} e^{-\beta \frac{||\mathbf{p}_1||^2}{2m_1}} d \mathbf{p}_1 \int e^{-\beta \phi(\mathbf{r}^N)} d\mathbf{r}^N}{\prod_{i=1}^N \bigg(\frac{2 \pi m_i}{\beta}\bigg)^{3/2} \int e^{-\beta \phi(\mathbf{r}^N)} d\mathbf{r}^N}
\end{equation}

\noindent or equivalently,

\begin{equation}
    \rho(\mathbf{p}_1) d\mathbf{p}_1 = \bigg(\frac{1}{2 m_1 \pi k_B T}\bigg)^{3/2} e^{-\beta \frac{||\mathbf{p}_1||^2}{2m_1}} d \mathbf{p}_1.
\end{equation}

\noindent This is known as the Maxwell-Boltzmann distribution.

\subsubsection{The Equipartition Theorem}

You may have noticed in previous examples that each quadratic term in positions or momenta contribute a factor of $\beta^{-1/2}$ to the partition function. This implies that each quadratic term makes a contribution of $k_BT/2$ to the total energy. Let's generalize this observation here by writing our Hamiltonian in the following form,

\begin{equation}
    H(q^f,p^f) = aq_1^2 + h(q_2, ..., q_f, p_1, ..., p_f)
\end{equation}

\noindent and computing the canonical ensemble average of $aq_1^2$. We proceed by writing,

\begin{equation}
    \langle a q_1^2 \rangle = \frac{1}{C} \int a q_1^2 e^{-\beta a q_1^2} e^{-\beta h} d\mathbf{q}^f d\mathbf{p}^f = \frac{a}{C} \sqrt{\frac{\pi}{4(\beta a)^3}} \int e^{-\beta h} d\mathbf{q}^{f\neq 1} d\mathbf{p}^f
\end{equation}

\noindent but $C$ is equal to,

\begin{equation}
    C = \int e^{-\beta a q_1^2} e^{-\beta h} d\mathbf{q}^f d\mathbf{p}^f = \sqrt{\frac{\pi}{\beta a}} \int e^{-\beta h} d\mathbf{q}^{f \neq 1} d\mathbf{p}^f 
\end{equation}

\noindent which means that,

\begin{equation}
    \langle a q_1^2 \rangle = \frac{a \sqrt{\frac{\pi}{4(\beta a)^3}} \int e^{-\beta h} d\mathbf{q}^{f\neq 1} d\mathbf{p}^f}{\sqrt{\frac{\pi}{\beta a}} \int e^{-\beta h} d\mathbf{q}^{f \neq 1} d\mathbf{p}^f} = \frac{1}{2}\beta^{-1} = \frac{1}{2} k_BT.
\end{equation}

\noindent Now, this result can be trivially generalized to Hamiltonian's of the form,

\begin{equation}
    H(q^f, p^f) = \sum_{i=1}^f (a_iq_i^2 + b_ip_i^2)
\end{equation}

\noindent to show that $U = f k_B T$. This idea that the energy of a system is equally partitioned amongst each degree-of-freedom in the system is called the equipartition theorem.

\subsubsection{Corrections from Quantum Mechanics}

We derived an equation for the Helmholtz energy using classical statistical mechanics for a particle in a box. However, we need to extend this to a general argument and introduce quantum mechanical corrections to the determined equation.

Suppose we have N identical particles ($m_i = m_j$) in a rectangular box of dimension $L_x,L_y,L_z$. The Hamiltonian is given by,

\begin{equation}
    H(\mathbf{r^N},\mathbf{p^N}) = \sum_{i=1}^N \frac{||\mathbf{p_i}||^2}{2m} + \psi_w (\mathbf{r^N})
\end{equation}

\noindent such that $\psi_w$ represents the interactions of the particles with the walls of the box. We can determine C by integrating $e^{-\beta H}$ over all of phase space according to,

\begin{equation}
    C = \int_{-\infty}^{\infty} \int_{-\infty}^{\infty} ... \int_{-\infty}^{\infty} e^{-\beta H}dx_1dy_1dz_1dx_2...dz_Ndpx_1dpy_1...dpz_N
\end{equation}

\noindent where the integrals span the 6N-dimensional phase space. Solving this integral we find that,

\begin{equation}
    C = V^N\bigg(\frac{2\pi m}{\beta}\bigg)^{3N/2}.
\end{equation}

When we apply this to our equation for the Helmholtz energy, we find that,

\begin{equation}
    F = -k_BT\log(V^N\bigg(\frac{2\pi m}{\beta}\bigg)^{3N/2}) = -k_BT\bigg[N\log(V) + \frac{3N}{2}\log(2 \pi m k_B T)\bigg].
\end{equation}

Now, we need to look at this function and see if it makes sense according to classical thermodynamics. The first thing to check is if $F$ is extensive so that $F(T, \lambda V, \lambda N) = \lambda F(T,V,N)$. Begin by substituting $V$ and $N$ for $\lambda V$ and $\lambda N$ so that, 

\begin{equation}
    F = -k_BT\bigg[\lambda N \log(\lambda V) + \frac{3 \lambda N}{2}\log(2 \pi m k_B T)\bigg]
\end{equation}

\noindent which gives,

\begin{equation}
    F = -\lambda k_BT\bigg[ N \log(\lambda V) + \frac{3 N}{2}\log(2 \pi m k_B T)\bigg] \neq \lambda F(T,V,N).
\end{equation}

Therefore, we find that $F$ is not extensive! This means that the statistical mechanics that we have developed is inconsistent with thermodynamics. The problem isn't with what we did, but rather how we started; namely, with classical mechanics. Atoms and particles are governed by quantum mechanics, and therefore we must introduce some corrections.

\subsubsubsection{Indistinguishable Particles}

In classical mechanics, every particle is distinguishable from every other particle. In fact, we reasoned that we can track the position and velocity of every particle in a system and see how those variables evolve according to Newton's equations of motion. According to quantum mechanics, however, this can't be done even in principle! In quantum mechanics, identical particles are fundamentally indistinguishable, and therefore no computation or experiment can ever be devised to distinguish between them. This means that when we perform the integral over phase space when we calculate $C$ that we are actually over counting the number of states. To adjust for this, we just need to divide by the number of ways of labeling the particles, which is just related to the factorial operator. We therefore divide C by this degeneracy factor so that,

\begin{equation}
    C' = \frac{1}{N!}\int e^{-\beta H(q^f,p^f)}dq^fdp^f.
\end{equation}

Let's now check if $F$ is extensive. We know that $F = -k_BT\log C'$, which is given by,

\begin{equation}
    F = -k_BT\bigg[ N \log(V) + \frac{3 N}{2}\log(2 \pi m k_B T) - \log N! \bigg]
\end{equation}

\noindent for which we can apply Sterling's approximation, $\log N! = N \log N - N$ for sufficiently large $N$,

\begin{equation}
    F = -k_BT\bigg[ N \log(V) + \frac{3 N}{2}\log(2 \pi m k_B T) - N \log N + N \bigg].
\end{equation}

We now make the substitution $V$ and $N$ for $\lambda V$ and $\lambda N$ so that,

\begin{equation}
    F = -k_BT\bigg[ \lambda N \log(\lambda V) + \frac{3 \lambda N}{2}\log(2 \pi m k_B T) - \lambda N \log (\lambda N) + \lambda N \bigg]
\end{equation}

\noindent which is, 

\begin{equation}
    F = -\lambda k_BT\bigg[N \log(V) + \frac{3 N}{2}\log(2 \pi m k_B T) - N \log (N) + N + \cancel{(\log \lambda - \log \lambda)} \bigg]
\end{equation}

\noindent showing that, 

\begin{equation}
    F(T, \lambda V, \lambda N) = \lambda F(T,V,N)
\end{equation}

\noindent for our new expression $C'$.

\subsubsubsection{The Heisenberg Uncertainty Principle}

To this point, we have taken integrals over infinitesimally small regions of the phase space in position and momenta, which violates the uncertainty principle. The uncertainty principle demands that a simultaneous measurement of position and momentum satisfies,

\begin{equation}
    (\Delta q)(\Delta p) \geq h
\end{equation}

\noindent where $h$ is Planck's constant in units of action (distance times momentum or equivalently energy times time). To think about how the uncertainty principle influences our normalization factor $C$, consider what happens when we take a phase space integral in classical mechanics. The integral,

\begin{equation}
    C = \int_{p} \int_q e^{-\beta H} dq dp
\end{equation}

\noindent assumes that all states inside the infinitesimal box $dq dp$ are distinguishable states and they all contribute to $C$. In reality, we only have distinguishable states up to the phase space volume $h$, so the integral should be written as,

\begin{equation}
    C = \frac{1}{h}\int_{p} \int_q e^{-\beta H} dq dp. 
\end{equation}
    
\noindent For $f$ degrees-of-freedom, the integral then becomes,

\begin{equation}
    C = \frac{1}{h^f}\int_{p^f} \int_{q^f} e^{-\beta H} dq^f dp^f
\end{equation}

\noindent which for $N$ particles in a three-dimensional space gives,

\begin{equation}
    C = \frac{1}{h^{3N}}\int_{p^f} \int_{q^f} e^{-\beta H} dq^f dp^f.
\end{equation}

Finally, combining our results from particle indistinguishability and the uncertainty principle gives the correct definition for the canonical partition function, 

\begin{equation}
    Z = \frac{1}{h^{3N}N!} \int e^{-\beta H} dq^f dp^f.
\end{equation}

These corrections essentially amount to preventing over-counting of true quantum states with classical mechanical integrals. 

\subsubsection{Canonical Ensemble}

A canonical ensemble is a collection of systems that have the same values of $N$, $V$, and $T$ and evolve according to the same Hamiltonian. It therefore represents a system at thermal equilibrium with its surroundings. In the canonical ensemble, we then argued that the equilibrium probability density for a system of identical particles is, 

\begin{equation}
    \rho(\mathbf{q}^f, \mathbf{p}^f) = \frac{1}{h^{3N}N!}\frac{1}{Z} e^{-\beta H}
\end{equation}

\noindent where $Z$ is the canonical partition function given by,

\begin{equation}
    Z = \frac{1}{h^{3N}N!}\int e^{-\beta H} d\mathbf{q}^f d\mathbf{p}^f.
\end{equation}

We then found that we could derive the free energy (Helmholtz for an $NVT$ system) in the following way,

\begin{equation}
    F = U - TS = -k_B T \log Z
\end{equation}

\noindent and derive important relationships like, 

\begin{equation}
    \langle H \rangle = U = -\bigg(\frac{\partial \log Z}{\partial \beta}\bigg)_V
\end{equation}

\noindent or,

\begin{equation}
    P = k_B T \bigg(\frac{\partial \log Z}{\partial V}\bigg)_{T,N}.
\end{equation}

The canonical ensemble describes systems that can exchange heat with the surroundings but not volume or matter. However, the canonical ensemble is just one type of system that is relevant to real systems. It doesn't take that much imagination to consider a system which can not only exchange heat (and thus equilibrate temperature with its surroundings) but also exchange volume or particles. These ensembles will have different probability density functions, free energies, and mathematical relationships.

\subsubsection{The Isothermal-Isobaric Ensemble}

The isothermal-isobaric ensemble is the same as a canonical ensemble aside from the pressure being held constant instead of volume. Therefore, we refer to the isothermal-isobaric ensemble as the $NPT$ ensemble. The free energy for an $NPT$ system is the Gibbs free energy (as an exercise, perform a partial Legendre transform of internal energy $U$ to convert $S \to T$ and $V \to P$) such that,

\begin{equation}
    G = U - TS + PV
\end{equation}

\noindent which can be related to statistical mechanics through the probability distribution function,

\begin{equation}
    \rho(\mathbf{q}^f, \mathbf{p}^f; N) = \frac{1}{h^{3N}N!}\frac{1}{\Delta_N} e^{-\beta(H + PV)}
\end{equation}

\noindent where $\Delta_N$ is the isothermal-isobaric partition function,

\begin{equation}
    \Delta_N = \frac{1}{h^{3N}N!V_0}\int^\infty_0 dV \int e^{-\beta(H + PV)} d\mathbf{q}^f d\mathbf{p}^f
\end{equation}

\noindent which can be interpreted as the volume average of canonical ensembles with weight $e^{-\beta PV}$. The Gibbs free energy is then,

\begin{equation}
    G = -k_B T \log \Delta_N.
\end{equation}

The isothermal-isobaric ensemble describes many experiments that are performed at constant temperature and pressure and is useful for determining the equation of state of fluids near atmospheric conditions.

\subsubsection{The Grand Canonical Ensemble}

The thermodynamic state of a system that is open, or able to exchange particles with its surroundings, is characterized by the grand potential (as an exercise, try the partial Legendre transform $U$ exchanging $N \to \mu$ and $S \to T$),

\begin{equation}
    \Omega = U - TS - N\mu.
\end{equation}

An ensemble of systems having the same values of chemical potential, volume and temperature ($\mu VT$) is referred to as a grand canonical ensemble. The phase space of the grand canonical ensemble is essentially just the union of canonical ensembles for every possible value of $N$. In this way, the probability distribution function is written as,

\begin{equation}
    \rho(\mathbf{q}^f, \mathbf{p}^f; N) = \frac{1}{h^{3N}N!}\frac{1}{\Xi} e^{-\beta(H - N \mu)}
\end{equation}

\noindent where $\Xi$ is the grand canonical partition function,

\begin{equation}
    \Xi = \sum_{N = 0}^\infty \frac{e^{\beta N \mu}}{h^{3N}N!} \int e^{-\beta H} d\mathbf{q}^f d\mathbf{p}^f.
\end{equation}

The grand canonical ensemble is used in molecular simulations for vapor-liquid equilibria and adsorption processes. 

\subsubsection{General Rules of Statistical Ensembles}

The problem of statistical mechanics amounts to counting - or to cleverly avoid counting - the number of equally probable ways that a system can divide up its energy. The partition function is precisely this number, so if the partition function is known for a given ensemble, we can calculate the free energy and in principle have complete knowledge of system thermodynamics. Thus, if $Q$ is some partition function for a given ensemble, we can always compute the corresponding free energy using the relation,

\begin{equation}
    FE = -k_B T \log Q.
\end{equation}

Additionally, the probability distribution function for a given ensemble can allow us to calculate moments of some dynamic variable according to the equation,

\begin{equation}
    \langle A \rangle =\frac{1}{h^{3N}N!} \frac{1}{Q}\int \rho A d\mathbf{q}^f d\mathbf{p}^f
\end{equation}

\noindent regardless of the ensemble. Although we will not explore any other ensembles here, there are many described in the literature that must be considered in certain cases. It is therefore crucial to analyze the physical system that you want to model and determine whether or not a certain ensemble is appropriate for the target application. 

\subsubsection{Inverse Kirkwood-Buff Theory}

The Kirkwood-Buff solution theory was presented in a landmark paper in 1951. The theory relates particle number fluctuations in the grand canonical ensemble to integrals of the radial distribution function. In this short introduction to the topic, we will introduce the statistical mechanics required to understand Kirkwood-Buff solution theory as well as provide details on the numerical computation of the Kirkwood-Buff integrals from experimental data. 

An ensemble of systems that have constant chemical potential, volume, and temperature belong to the so-called grand canonical ensemble. Recall that in the canonical ensemble (constant number of particles, volume, and temperature), the probability distribution function is given by the Boltzmann distribution,

\begin{equation}
    p(\mathbf{r_N}, \mathbf{p_N}) = \frac{1}{h^{3N} N!} \frac{e^{-\beta \mathcal{H}}}{Z}
\end{equation}

\noindent where $Z$ is the canonical partition function,

\begin{equation}
    Z = \frac{1}{h^{3N} N!} \int \int \exp (-\beta \mathcal{H}) d\mathbf{r} d\mathbf{p}
\end{equation}

\noindent that can be directly related to classical thermodynamics by its relation to the Helmholtz free energy,

\begin{equation}
    F = -kT\log(Z).
\end{equation}

The grand canonical ensemble is just an expansion on the concept of the canonical ensemble; in fact, the grand canonical ensemble is just the union of canonical ensembles with different values for the number of particles $N$. Therefore the probability distribution function is just a Boltzmann distribution with an additional contribution from the number of particles and chemical potential,

\begin{equation}
    p(\mathbf{r_N}, \mathbf{p_N}, N) = \frac{e^{-\beta (\mathcal{H} - N \mu)}}{\Xi}
\end{equation}

\noindent where the $\Xi$ is the grand partition function,

\begin{equation}
    \Xi = \sum_{N = 0}^\infty \frac{\exp(N \beta \mu)}{h^{3N} N!} \int \int \exp(-\beta \mathcal{H}) d\mathbf{r}^N d\mathbf{p}^N
\end{equation}

\noindent where the sum spans over all possible numbers of particles that the constant $\mu V T$ system can exist in. Just as in the canonical ensemble, we can take averages of any quantity-of-interest with respect to the grand canonical ensemble by taking the product of the actual observable $A$ by its corresponding probability and integrating over the entire phase space,

\begin{equation}\label{eq:average}
    \langle A \rangle = \sum_{N = 0}^\infty \frac{1}{h^{3N} N!} \int \int A(\mathbf{r_N}, \mathbf{p_N}) p(\mathbf{r_N}, \mathbf{p_N}, N) d\mathbf{r}^N d\mathbf{p}^N.
\end{equation}

The final piece we need is a relation between the averages of the number of particles. Looking at the partition function $\Xi$, we can see that differentiation with respect to $\mu_i$ will give us the following expression,

\begin{equation}
    \frac{\partial \Xi}{\partial \mu_i} = \sum_{N = 0}^\infty N_i \frac{\exp(N \beta \mu)}{h^{3N} N!} \int \int \exp(-\beta \mathcal{H}) d\mathbf{r}^N d\mathbf{p}^N = \beta \Xi \langle N_i \rangle.
\end{equation}

\noindent Similarly,

\begin{equation}
    \frac{\partial^2 \Xi}{\partial \mu_i \partial \mu_j} = \sum_{N = 0}^\infty N_i N_j \frac{\exp(N \beta \mu)}{h^{3N} N!} \int \int \exp(-\beta \mathcal{H}) d\mathbf{r}^N d\mathbf{p}^N = \beta^2 \Xi \langle N_i N_j \rangle.
\end{equation}

But we can also consider the second derivative of the partition function in an equivalent way,

\begin{equation}
    \frac{\partial^2 \Xi}{\partial \mu_i \partial \mu_j} = \frac{\partial}{\partial \mu_j} \beta \Xi \langle N_i \rangle = \beta\bigg(\langle N_i \rangle\frac{\partial \Xi}{\partial \mu_j} + \Xi \frac{\partial \langle N_i \rangle}{\partial \mu_j}\bigg)
\end{equation}

\begin{equation}
    = \beta \bigg(\beta \langle N_i \rangle \langle N_j \rangle \Xi + \Xi \frac{\partial \langle N_i \rangle}{\partial \mu_j}\bigg). 
\end{equation}

Equating the two expressions gives,

\begin{equation}
    \beta^2 \Xi \langle N_i N_j \rangle = \beta \bigg(\beta \langle N_i \rangle \langle N_j \rangle \Xi + \Xi \frac{\partial \langle N_i \rangle}{\partial \mu_j}\bigg) 
\end{equation}

\begin{equation}\label{gcpd}
    \langle N_i N_j \rangle - \langle N_i \rangle \langle N_j \rangle  =  \beta^{-1} \frac{\partial \langle N_i \rangle}{\partial \mu_j}.
\end{equation}

\subsubsection{The Ornstein-Zernike Equation}

In Ornstein and Zernike's seminal work, \textit{Accidental deviations of density and opalescence at the critical point of a single substance, in: KNAW, Proceedings, 17 II, 1914, pp. 793-806}, it was shown that the total correlation between molecules in a statistical mechanical system must obey an integral equation of the form, 

\begin{equation}
    h(r_{12}) = c(r_{12}) + n\int d\mathbf{r_3} c(r_{13})h(r_{23})
\end{equation}

\noindent where $h(r_{ij})$ is the total correlation function between particles $i$ and $j$ (note that $h(r_{ij}) = g(r_{ij}) - 1$, where $g(r_{ij})$ is the more familiar radial distribution function) and $c(r_{ij})$ is the 'direct' correlation between particles $i$ and $j$. In this article, I will try and demystify this equation using physical arguments from Ornstein and Zernike's original paper as well as the more recent formulations described in \textit{Liquid State Theory} from Hansen and McDonald as well as Santos' \textit{A Concise Course on the Theory of Classical Liquids}.

The motivation for Ornstein and Zernike's work was from an original paper on critical behavior of statistical mechanical systems from Smoluchowski, with the main problem being that his volume elements, of which he used to partition up a system of particles, were assumed to be mutually independent of one another. At this stage, the whole premise might already sound a little fishy. We know that liquids interact with each other strongly, so of course this mutual independence assumption must be a problem. Why are we even talking about volume elements in the first place? And what is meant by their mutual independence? 

Ornstein and Zernike implore us to consider a statistical mechanical system composed of a finite set of volume elements containing a finite number of particles so that the total volume $V$ is,

\begin{equation}
    V = v_1 + v_2 + ...
\end{equation}

\noindent and the total number of particles $N$ is,

\begin{equation}
    N = n_1 + n_2 + ...
\end{equation}

\noindent Nice! Intuitively, all we have done is taken a system of particles and discretized the position space into finite volume elements which contain a finite set of particles. Of course, in an atomic systems the particles are always moving, and indeed, always passing from one volume element to another. In general, we will say that the mean square deviation of the number of particles in a given volume element $v_i$ based on the atomic motion is, $\overline{(n_i - \bar{n_i})^2}$. If the number of particles in the element of volume is large, then this mean square deviation is approximately equal to the variance of the number of particles in the volume element. If we consider that our atomic system is a subsystem of some larger collection of particles, then we can express the mean number of particles in our system $\bar{N}$ as the sum of the mean number of particles in each volume element so that,

\begin{equation}
    \bar{N} = \bar{n_1} + \bar{n_2} + ...
\end{equation}

\noindent and therefore the mean square deviation of the number of particles in our system is,

\begin{equation}
    \overline{(N-\bar{N})^2} = \overline{[n_1 + n_2 + ... - (\bar{n_1} + \bar{n_2} + ...)]^2} = \overline{[(n_1 - \bar{n_1}) + (n_2 - \bar{n_2}) + ...]^2}
\end{equation}

\noindent which has a complicated expansion given by,

\begin{equation}\label{expansion}
    = \overline{(n_1 - \bar{n_1})^2 + (n_2 - \bar{n_2})^2 + ... + (n_1 - \bar{n_1})(n_2 - \bar{n_2}) +  (n_1 - \bar{n_1})(n_3 - \bar{n_3}) + ...}
\end{equation}

\noindent and, assuming that the cross correlations are not related and that the number of particles are equal across all of the volume elements (in which we will say there are $p$ in total) gives,

\begin{equation}
    =p\overline{(n-\bar{n})^2}
\end{equation}

\noindent The biggest problem here is the cross terms which appear in our expansion (eq \eqref{expansion}), are at odds with a mutual independence approximation. Indeed, if the system were truly independent, then there would be no contribution to the mean square deviation from pairs of different volume elements. This forms the key difference between Ornstein and Zernike's formulation and that of Smulochowski.

\paragraph{The Ornstein-Zernike Equation}

Now, consider that we take our system and subdivide into an infinite number of volume elements and select one to be our reference, $dv_0$. We want to ask how fluctuations in our local volume element $dv_0$ influence the densities of all other volume elements in the system, which is of course related to how much the particles in our reference volume interact with all the other particles. On average, we claim that the average density in our volume element, $\bar{\mathbf{v}_0}$, is linearly related to the density in all other volume elements so that,

\begin{equation}\label{linear}
    \bar{\mathbf{v}_0} = f_1\mathbf{v}_1dv_1 + f_2\mathbf{v}_2dv_2 + ...
\end{equation}

\noindent where $f_i$ are coefficients of the linear equation relating the local density to that of the other volume elements. Generally, we could write this expression as a Taylor expansion over terms of increasing order, although Ornstein and Zernike claim that the linear term is sufficient for their purposes. Although we will table this for now, it is an interesting question as to whether such an expansion is truly sufficient for statistical mechanics of liquids and whether or not this causes issues with the OZ relation methods. Moving on, we assume that this function is continuous and well-behaved over the so-called 'sphere of attraction' of the system, so we can rewrite this sum as an integral of continuous functions such that,

\begin{equation}
     \bar{\mathbf{v}_0} = \int f(\mathbf{r})\mathbf{v}(\mathbf{r})d\mathbf{r}.
\end{equation}

\noindent where the integral is understood to go over all of $\mathbf{r}$. Note that the 'sphere of attraction' simply refers to the range of direct influence of the particles of $dv_0$ on the other particles in the system. We would expect such a sphere to be finite, since the direct influence of particles in our reference volume is an interatomic potential that is at its strongest Coulombic with $\propto 1/r$ dependence. However, note that this assumption leads to difficulties with the Fourier transform (see section \ref{FT}). More generally however, the function $\mathbf{v}(\mathbf{r})$ is not actually continuous if we consider that the particles are point particles, the reason being that if we create an infinite number of volume elements for the system that there will be some volume elements with particles and others with none. Such a problem is avoided by rewriting the expressions in terms of a new function which will be introduced later. However, in Ornstein and Zernike's original work, there is not a satisfactory discussion of why such a function should be continuous at all, which would require further nuance and discussion found in the free energy functional approach discussed later.

Let us now suppose that the density in our volume element is known to be $\mathbf{v_0}$ and write the average density at any other point in space as,

\begin{equation}
    \bar{\mathbf{v}}(\mathbf{r}) = g(\mathbf{r}, \mathbf{v_0}d\mathbf{r_0})
\end{equation}

\noindent the function $g$ being dependent on, of course, the position in the system and choice of reference. What then is $g$ given that we know $f$? Selecting a known value of $\mathbf{v}_1$ at $\mathbf{r_1}$ would give us a $g$ of,

\begin{equation}
     \bar{\mathbf{v}}(\mathbf{r}) = g(\mathbf{r} - \mathbf{r_1}, \mathbf{v_1}d\mathbf{r_1})
\end{equation}

\noindent which is just the definition of $g$ selected at a point shifted in space from our original reference. The average value at $\mathbf{r_0}$ is then simply,

\begin{equation}
    \bar{\mathbf{v_0}}(\mathbf{r}) = g(\mathbf{r_1}, \mathbf{v_1}d\mathbf{r_1})
\end{equation}

\noindent Applying the original integral with an assumed known $f$ is then,

\begin{equation}
     g(\mathbf{r_1}, \mathbf{v_1}d\mathbf{r_1}) = \int_{\mathbf{r} \neq \mathbf{r_1}} f(\mathbf{r})g(\mathbf{r} - \mathbf{r_1}, \mathbf{v_1}d\mathbf{r_1})d\mathbf{r} + f(\mathbf{r_1})\mathbf{v_1}d\mathbf{r_1}.
\end{equation}

\noindent where the integral contribution at the selected particle position $\mathbf{r_1}$ is written explicitly outside of the integral (known from the linear approximation to the Taylor expansion above) since $g$ is not defined there. Finally, note that this must be true for any choice of $\mathbf{v_1}d\mathbf{r_1}$ and that $g$ is linear in $\mathbf{v_i}d\mathbf{r_i}$ (eq \eqref{linear}), meaning that we can pull out a factor of $\mathbf{v_1}d\mathbf{r_1}$ for both instances of $g$ so that,

\begin{equation}
     g(\mathbf{r_1}) = \int_{\mathbf{r} \neq \mathbf{r_1}} f(\mathbf{r})g(\mathbf{r} - \mathbf{r_1})d\mathbf{r} + f(\mathbf{r_1}).
\end{equation}

In modern symbols, the $g$ is actually just the total correlation between densities in volume elements of the system, which we write as $h$, and $f$ represents the direct coupling between volume elements, which we write as $c$. In modern symbols, the Ornstein-Zernike equation emerges,

\begin{equation} \label{OZ}
     h(\mathbf{r_1}) = \int_{\mathbf{r} \neq \mathbf{r_1}} c(\mathbf{r})h(\mathbf{r} - \mathbf{r_1})d\mathbf{r} + c(\mathbf{r_1}).
\end{equation}

\paragraph{Physical Interpretation}

What exactly does eq \eqref{OZ} mean? Starting with the total correlation function $h(\mathbf{r_1})$, we notice that this function is simply the average particle density of the system at position $\mathbf{r_1}$ within a spherical coordinate system,

\begin{figure}[H]
    \centering
    \includegraphics[width=0.5\linewidth]{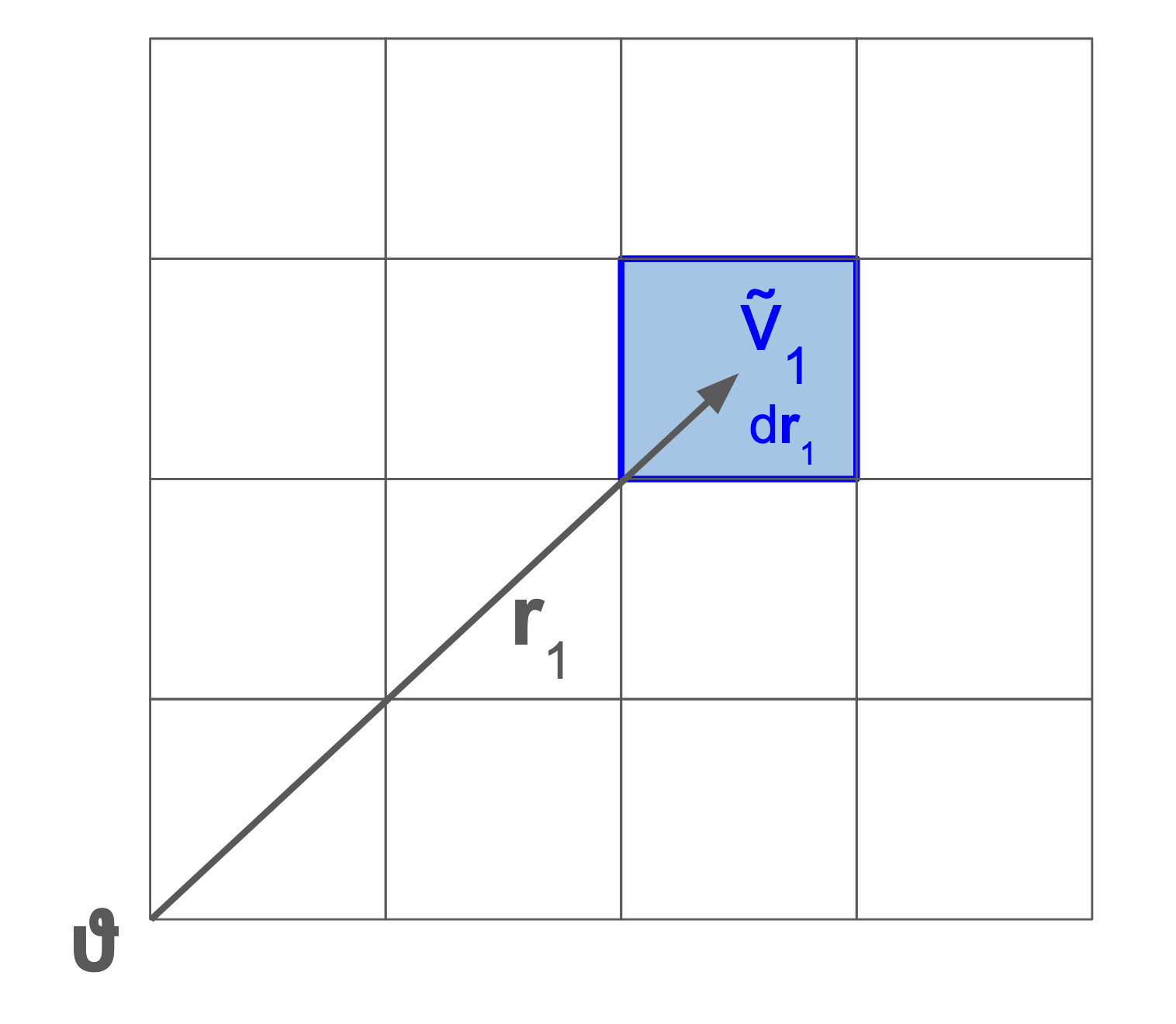}
    \caption{The total correlation function, $h(\mathbf{r_1})$, is just the average density at any point in space away from the origin, $\mathcal{O}$.}
    \label{fig:geometry}
\end{figure}

\noindent According to the Ornstein-Zernike equation, this total correlation function is fully described by three functions: the direct correlation function $c(\mathbf{r_1})$, the direct correlation function $c(\mathbf{r})$, and the total correlation function $h(\mathbf{r - r_1})$. These functions only make sense when we consider that the integral will go over all of the other volume elements of the system, specifically for all $\mathbf{r} \neq \mathbf{r_1}$. Our new picture is then,

\begin{figure}[H]
    \centering
    \includegraphics[width=0.5\linewidth]{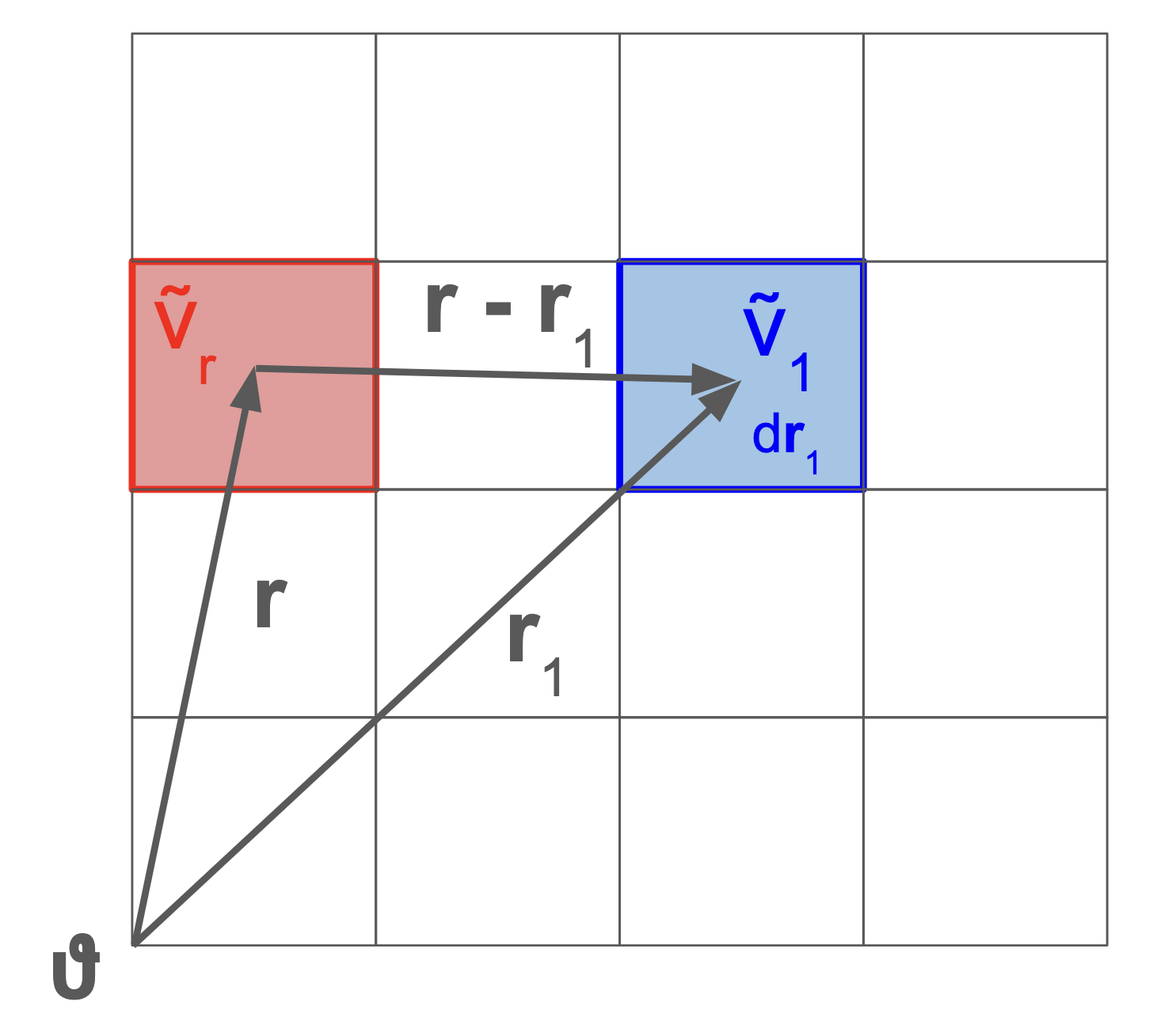}
    \caption{The new vector $\mathbf{r}$ is of course integrated over all values of $\mathbf{r}$.}
    \label{fig:integral}
\end{figure}

\noindent Now let's consider each of the terms. The direct correlation function $c(\mathbf{r_1})$ is a function of linear coupling constants which relate the particle density at the origin to all other volume elements in the system. There are no cross terms in the expansion, so this function only keeps track of the 'direct' coupling between volume elements, hence the name. $c(\mathbf{r_1})$ is therefore the direct coupling between the reference origin and some other volume element. The other direct correlation function $c(\mathbf{r})$ has the same interpretation, but now it represents the direct coupling between the reference and a different volume element not centered at $\mathbf{r_1}$. Finally, the total correlation $h(\mathbf{r - r_1})$ is just the density of a volume element from the reference volume element between the difference in the vectors $\mathbf{r - r_1}$. 

\subparagraph{Finite $f(\mathbf{r})$ Leads to Nonphysical Fourier Transforms}\label{FT}

There are some problems with the mathematical development presented by Ornstein and Zernike. In a short phone call with Harry W. Sullivan, who just traveled to University of Minnesota to start a doctoral program, it was noted that the assumption of a cut-off in $c(\mathbf{r})$ may actually be at odds with Fourier representations of the total correlation function. To see this, consider that the modified Henkel transform (or spherical Bessel transform) of the real-space total correlation function is directly related to the total structure factor, a quantity  which is observable experimentally via scattering experiments. Thus, assuming that the total correlation function is isotropic, we can write the modified Henkel transform of the total correlation function as, 

\begin{equation}
    \hat{h}(q) = \int_0^\infty h(r) J_r(qr)rdr
\end{equation}

\noindent where $J_r$ is the Bessel function of the first kind of order $r$. If there is a point at which the total correlation abruptly becomes zero at some $r_{max}$, then, expanding $h(r)$ in a Fourier Bessel series gives,

\begin{equation}
    \hat{h}(q) = \int_0^{r_{max}} \sum_{n=1}^\infty c_n J_r\bigg(\frac{U_{r,n}}{r_{max}r}\bigg) J_r(qr)rdr
\end{equation}

\noindent where $c_n$ are constants and $U_{r,n}$ is the $n^{th}$ root of the Bessel function. The sum and integral can be exchanged since they are linear, giving the following,

\begin{equation}
    \hat{h}(q) = \sum_{n=1}^\infty c_n \int_0^{r_{max}} J_r\bigg(\frac{U_{r,n}}{r_{max}r}\bigg) J_r(qr)rdr
\end{equation}

\noindent which has as a solution a delta-train,

\begin{equation}
    \hat{h}(q) = \sum_{n=1}^\infty c_n \delta\bigg(\frac{U_{r,n}}{r_{max}} - q\bigg)\frac{r_{max}}{U_{r,n}}
\end{equation}

\noindent This result, of course, is at odds with experimental measurement of the total structure factor and with physical intuition of the problem. Hence, the abrupt sending of the function $c(\mathbf{r})$ to zero is not a mathematically reasonable assumption to describe the theory of liquids.

\paragraph{Conclusions from Ornstein and Zernike's Original Work}

In summary, Ornstein and Zernike's original paper is a work of genius that resulted in an integral relation that has remained relevant to liquid state physics for over a century. Indeed, reading this manuscript one gets the sense that the authors did not understand how monumental this work would become in the development of modern physics. However, the mathematical development in their manuscript is not rigorous and therefore encounters some potential pitfalls, but nonetheless was later verified to be correct anyway without the shaky assumptions that derail the extensibility of their method. 

\paragraph{The Gibbs Interpretation}

One of the great issues with Ornstein and Zernike's original derivation is that the direct correlation function is constructed with non-rigorous assumptions and doesn't have a physical interpretation. However, these deficiencies can be addressed with the statistical mechanics methods of Gibbs. First, recall that the phase space probability density, $f^{[N]}(\mathbf{r^N}, \mathbf{p^N}; t)$, is a function so that $f^{[N]}(\mathbf{r^N}, \mathbf{p^N}) d\mathbf{r^N}\mathbf{p^N}$ is the probability at time $t$ that the system will be in the $6N$-dimensional phase space element $d\mathbf{r^N}\mathbf{p^N}$. By defining a specific physical system, we can ascribe a functional form to this probability distribution. For instance, if  we take a collection of particles at equilibrium with a heat bath of fixed temperature but that cannot exchange particles, we obtain the canonical probability distribution,

\begin{equation}
    f^{[N]}_0(\mathbf{r^N}, \mathbf{p^N}) \propto \frac{\exp(-\beta \mathcal{H})}{Q_N}
\end{equation}

\noindent where $Q_N$ is the familiar $N$-particle canonical partition function and the subscript of 0 for $f$ indicates that the system is at equilibrium and therefore does not depend on time. 

Generally, the $N$-particle phase space probability density function contains unnecessary information that can be integrated out to obtain so-called reduced phase space distribution functions in the following way,

\begin{equation}
    f^{[n]}(\mathbf{r^n}, \mathbf{p^n}; t) = \frac{N!}{(N-n)!}\iint f^{[N]}(\mathbf{r^N}, \mathbf{p^N}) d\mathbf{r^{(N-n)}}d\mathbf{p^{(N-n)}}
\end{equation}

\noindent which has an obvious physical interpretation: $f^{[n]}(\mathbf{r^n}, \mathbf{p^n}; t)d\mathbf{r^{(N-n)}}d\mathbf{p^{(N-n)}}$ gives the probability of finding the system in a reduced phase space element $d\mathbf{r^{(N-n)}}d\mathbf{p^{(N-n)}}$ irrespective of all other particle positions and momenta. Now, in general this reduced phase space probability density at equilibrium is separable into kinetic and potential parts ($\mathcal{H} = \mathcal{K} + \mathcal{P}$) so that,

\begin{equation}
    f^{[n]}(\mathbf{r^n}, \mathbf{p^n}) = \rho_N^{[n]}(\mathbf{r^n})f_M^{[n]}(\mathbf{p^n})
\end{equation}

\noindent where the $\rho_N^{[n]}(\mathbf{r^n})$ is the equilibrium $n$-particle density function so that  $\rho_N^{[n]}(\mathbf{r^n})d\mathbf{r^n}$ is the probability of finding a set of $n$ particles in reduced phase space volume element $d\mathbf{r^n}$ irrespective of the positions of all other particles and irrespective of all momenta. 

\subparagraph{Direct Correlations as Derivatives of Free Energy}

The Gibbs interpretation offers us more mathematical rigor and precision that resolves many of the issues we have seen in Ornstein and Zernike's original manuscript. First, we will discuss the physical interpretation of the total correlation function in the Gibbs interpretation. Intuitively, the total correlation function, $h^{(2)}(\mathbf{r,r'})$, which appears in the OZ equation, describes the ensemble average distribution function of any two particle densities in the system. However, this quantity is defined in terms of the equilibrium $n$-particle density function in the following way,

\begin{equation}
    g_N^{[n]}(\mathbf{r^N}) = \frac{\rho_N^{[n]}(\mathbf{r^n})}{\prod_{i=1}^n\rho_N^{[1]}(\mathbf{r_i})}
\end{equation}

\noindent which, for a reduced 2-particle density function, can be written as,

\begin{equation}
    g_N^{[2]}(\mathbf{r_1},\mathbf{r_2}) = \frac{\rho_N^{[2]}(\mathbf{r_1},\mathbf{r_2})}{\rho_N^{[1]}(\mathbf{r_1})\rho_N^{[1]}(\mathbf{r_2})}
\end{equation}

\noindent and finally as our function,

\begin{equation}
    h_N^{[2]}(\mathbf{r_1},\mathbf{r_2}) = g_N^{[2]}(\mathbf{r_1},\mathbf{r_2}) - 1
\end{equation}

\noindent Notice that this function is defined in terms of a 2-particle density function, which contains contributions from both individual particle densities and the intermediate particles as in our original picture from the Ornstein-Zernike manuscript. Notice we have completely avoided the difficulties of assuming that the density function is continuous (when in fact, it is not), since $h_N^{[2]}(\mathbf{r_1},\mathbf{r_2})$ is equally described by an ensemble average that is continuous for liquids at equilibrium.

So what of Ornstein and Zernike's coupling constant function, $f$, under the Gibbs interpretation? Perhaps surprisingly, $f$ is exactly equal to a second partial functional derivative of the excess free energy functional, $\mathcal{F}_{ex}[\rho^{(1)}]$ (recall that the excess part of the intrinsic free energy is just the contribution to the free energy due to particle interactions), so that,

\begin{equation}
    c(\mathbf{r},\mathbf{r'}) = -\beta \frac{\delta^2 \mathcal{F}_{ex}[\rho^{(1)}]}{\delta \rho^{(1)}(\mathbf{r}) \delta \rho^{(1)}(\mathbf{r'})}
\end{equation}

\noindent Such a physical interpretation of the direct correlation is extremely satisfying for the following reasons. First, we see that indeed the direct correlation does not involve reduced particle density distributions higher than order 1, meaning that there is no contribution from intermediate particles (we had the same situation in Ornstein and Zernike's picture previously). However, now $c(\mathbf{r},\mathbf{r'})$ can be thought of as a response of the excess intrinisic free energy to changes in the single particle density distributions, rather than just the linear coupling terms of a truncated Taylor expansion. Such a construction does not require $c(\mathbf{r},\mathbf{r'})$ to abruptly vanish, resolving the issue with the modified Henkel transform in section \ref{FT}. Finally, we now see how $c(\mathbf{r},\mathbf{r'})$ could even be used to approximate the intrinsic free energy functional and hence describe entirely the thermodynamics of a model fluid. 

\subparagraph{The OZ Equation in the Gibbs Interpretation}

See Section 3.5 of \textit{Theory of Simple Liquids} by Hansen and McDonald. 

\paragraph{Numerical Implementations}

\subparagraph{The Isotropic OZ Relation and the Bridge Function}

In practice, we usually consider liquids to be isotropic, allowing us to reduce the vector quantities introduced earlier into just the distance between particles, $r$, to give an isotropic Ornstein-Zernike relation,

\begin{equation}
    h(r) = c(r) + \rho \int c(|\mathbf{r} - \mathbf{r'}|)h(r')d\mathbf{r'}
\end{equation}

\noindent which you may notice is just a convolution of $h(r')$ with $c(|\mathbf{r} - \mathbf{r'}|)$,

\begin{equation}
    h(r) = c(r) + \rho [h(r')*c(|\mathbf{r} - \mathbf{r'}|)](\mathbf{r})
\end{equation}

\noindent which under Fourier transform is simply,

\begin{equation}
    \hat{h}(q) = \hat{c}(q) + \rho \hat{h}(k) \hat{c}(k)
\end{equation}

\noindent or rearranging, 

\begin{equation}\label{OZFT}
    \hat{h}(q) = \frac{\hat{c}(q)}{1-\rho \hat{c}(q)}
\end{equation}

Finally, diagrammatic techniques (see \textit{A Concise Course on the Theory of Classical Liquids} by Santos and Section) and the introduction of the indirect correlation function ($\gamma(r) = h(r) - c(r)$), allow us to write the radial distribution as,

\begin{equation}
    g(r) = \exp[-\beta u(r) + h(r) - c(r) + b(r)]
\end{equation}

\noindent where $b(r)$ is the so-called, and unknown to us, bridge function arising from elementary diagrams in the integral expansion of $g(r)$. Using the definition of the indirect correlation, we find that,

\begin{equation}
    g(r) = \exp[-\beta u(r) + \gamma(r) + b(r)]
\end{equation}

\noindent Noting that $h(r) = g(r) - 1$, we can subtract the direct correlation $c(r)$ from both sides of the equation to obtain,

\begin{equation}
    h(r) - c(r) = g(r) - 1 - c(r)
\end{equation}

\noindent and using the definition of the indirect correlation and the diagram expansion of $g(r)$,

\begin{equation}
    \gamma(r) = \exp[-\beta u(r) + \gamma(r) + b(r)] - 1 - c(r)
\end{equation}

\noindent which, upon rearranging, gives,

\begin{equation}\label{diagdirect}
    c(r) = \exp[-\beta u(r) + \gamma(r) + b(r)] - \gamma(r) - 1
\end{equation}

\noindent Here a closure relation is introduced to write $b(r)$ in terms of the indirect correlation function. Specific examples will be provided in section \ref{numerical}.

\paragraph{OZ Fourier Transform in Terms of the Indirect Correlation}

In terms of the indirect correlation, the OZ equation becomes,

\begin{equation}
    \gamma(r) = \rho [(\gamma(r) + c(r))*c(|\mathbf{r} - \mathbf{r'}|)](\mathbf{r})
\end{equation}

\noindent, which, upon taking the Fourier transform, gives,

\begin{equation}
    \hat{\gamma}(q) = \rho (\hat{\gamma}(q) + \hat{c}(q))(\hat{c}(q)) = \rho \hat{\gamma}(q) \hat{c}(q) + \rho \hat{c}^2(q)
\end{equation}

\noindent and rearranging,

\begin{equation}\label{indirectOZ}
    \hat{\gamma}(q) = \frac{\rho \hat{c}^2(q)}{1 - \rho \hat{c}(q)}
\end{equation}

\subparagraph{Iterative Numerical Solvers}\label{numerical}

We now have everything we need to solve for the pair correlation function $h(r)$, given a potential. The steps we use are:

\begin{enumerate}
    \item Guess the indirect correlation function $\gamma(r)$.
    \item Compute the direct correlation function, $c(r)$, using eq \eqref{diagdirect} and a suitable closure relation for the bridge function.
    \item Fourier transform $c(q)$.
    \item Compute Fourier transform of the indirect correlation, $\hat{\gamma}(q)$, using eq \eqref{indirectOZ}.
    \item Fourier transform $\hat{\gamma}(q)$ to get a new guess for $\gamma(r)$.
    \item Repeat until convergence.
\end{enumerate}

\noindent Now, what if we have the total correlation function from experiment but do not know the potential? In this case, we change the algorithm as follows:

\begin{enumerate}
    \item Guess the pair potential $\beta u^{(0)}(r)$.
    \item Compute the direct correlation function, $c(r)$, using eq \eqref{diagdirect} and a suitable closure relation for the bridge function.
    \item Fourier transform $c(r)$.
    \item Compute Fourier transform of the total correlation, $\hat{h}(q)$, using eq \eqref{OZFT}.
    \item Fourier transform $\hat{h}(q)$ to get $h(r)$.
    \item Perform iterative refinement scheme to update potential to $\beta u^{(1)}(r)$.
    \item Repeat until the computed $h(r)$ matches the experimental data.
\end{enumerate}

\subsubsection{Density Functionals}

Consider a region of some volume, $v$, which contains $N_1, ..., N_k$ molecules of $k$-species. First, define the single particle density functional for species $\alpha$ as,

\begin{equation}
    \rho_\alpha^{(1)}(\mathbf{r_1}) = \sum^{N_\alpha}_{i_\alpha = 1} \delta(\mathbf{r_{i_\alpha}} - \mathbf{r_1})
\end{equation}

\noindent where the $\delta$ is understood as the Dirac-$\delta$ function. Let's consider the content of this equation fully before proceeding. The term $\rho_\alpha^{(1)}(\mathbf{r_1})$ is explicitly referring to the number density (in atoms/volume units) as a function of position ($\mathbf{r_1}$, also known as configuration space), of a specific species $\alpha$. We can evaluate this functional in the following way. Suppose we are given some vector $\mathbf{r}$, defined with respect to some pre-defined (and arbitrary) coordinate system. Then we just check if that vector points to (or corresponds with) the position of a particle with label $\alpha$. If it does, the functional returns a $\delta$ distribution at that vector, and if not, a zero. Then, the integral of the single particle density functional over the entire volume is just exactly equal to the number of particles of species $\alpha$ such that,

\begin{equation}
    \int_v \rho_\alpha^{(1)}(\mathbf{r_1}) dv = N_\alpha
\end{equation}

\noindent which essentially just amounts to counting all of the atoms of species $\alpha$ in the given system.

Now that we have introduced the singlet particle density functional, we will proceed to the pair density functional, which by a similar definition is given as,

\begin{equation}
    \rho_{\alpha, \beta}^{(2)}(\mathbf{r_1}, \mathbf{r_2}) = \sum^{N_\alpha}_{i_\alpha = 1} \sum^{N_\beta}_{k_\beta = 1} \delta(\mathbf{r_{i_\alpha}} - \mathbf{r_1})\delta(\mathbf{r_{k_\beta}} - \mathbf{r_2})
\end{equation}

\noindent which can be understood in a similar way as the single particle density functional. First, give the functional two vectors and then determine if (1) vector 1 points to the position of a particle with label $\alpha$ and (2) vector 2 points to the position of a particle with label $\beta$. If both statements are true then the functional returns a $\delta$ distribution at that pair of vectors, and if not it returns a zero. As in the previous equation, the summation goes over all of the known positions of both particles. In this case, the integral over the entire space is,

\begin{equation}
    \int_{v_1} \int_{v_2} \rho_{\alpha, \beta}^{(2)}(\mathbf{r_1}, \mathbf{r_2}) dv_1 dv_2 = N_\alpha N_\beta - N_\alpha \delta_{\alpha  \beta}
\end{equation}

\noindent since in the first integral we will find all particles of label $\beta$ given a specific vector for label $\alpha$, then the second integral will find $N_\beta$ particles for all positions of particles $\alpha$. Thus, the total number of counts where $\rho_{\alpha, \beta}^{(2)}(\mathbf{r_1}, \mathbf{r_2})$ is non-zero is $N_\alpha N_\beta$. However, if $\alpha = \beta$ then we will double count the vectors $N_\alpha$ times, so we need to subtract $N_\alpha \delta_{\alpha \beta}$ in this case. 

Note that to this point we have simply looked a system of particles with fixed positions. Of course, in real physical systems the particles are always moving and we observe the averages of the motions. Therefore, we need to consider an ensemble of systems that represent the average behavior of the system, which amounts to taking the ensemble average of the density functionals.

We then need to evaluate the average of the density functionals in the grand canonical ensemble. Rather than write these explicitly, we just apply Equation \eqref{average} to the integrals of the singlet and pair density functionals to obtain,

\begin{equation}
    \hat{\rho}_\alpha^{(1)}(\mathbf{r_1}) = \langle \rho_\alpha^{(1)}(\mathbf{r_1}) \rangle
\end{equation}

\begin{equation}
    \hat{\rho}_{\alpha, \beta}^{(2)}(\mathbf{r_1}, \mathbf{r_2}) = \langle \rho_{\alpha, \beta}^{(2)}(\mathbf{r_1}, \mathbf{r_2}) \rangle
\end{equation}

\noindent which by linearity of the expectation gives,

\begin{equation}
    \int_v \hat{\rho}_\alpha^{(1)}(\mathbf{r_1}) dv = \langle N_\alpha \rangle
\end{equation}

\begin{equation}
    \int_{v_1} \int_{v_2} \hat{\rho}_{\alpha, \beta}^{(2)}(\mathbf{r_1}, \mathbf{r_2}) dv_1 dv_2 = \langle N_\alpha N_\beta \rangle - \langle N_\alpha \rangle \delta_{\alpha  \beta}.
\end{equation}

Furthermore, by linearity of the expectation we can combine these two equations in the following clever way,

\begin{equation}
\begin{split}
        \int_{v_1} \int_{v_2} [\hat{\rho}_{\alpha, \beta}^{(2)}(\mathbf{r_1}, \mathbf{r_2}) - \hat{\rho}_\alpha^{(1)}(\mathbf{r_1}) \hat{\rho}_\beta^{(1)}(\mathbf{r_2})] dv_1 dv_2 = & \\ [\langle N_\alpha N_\beta \rangle - \langle N_\alpha \rangle \langle N_\beta \rangle]  - \langle N_\alpha \rangle \delta_{\alpha  \beta}.
\end{split}
\end{equation}

We can further simplify this expression by noting that the means of the density functionals take on specific forms in fluids. For example, the mean of the single density functional of a species $\alpha$ is just the concentration (in atoms/volume) of that species $c_\alpha$. The mean of the pair density functional is given a special definition in terms of the radial distribution function, which is just,

\begin{equation}
    \hat{\rho}_{\alpha, \beta}^{(2)}(\mathbf{r_1}, \mathbf{r_2}) = c_\alpha c_\beta g_{\alpha, \beta}(r).
\end{equation}

Plugging these definitions into our integral equation, we obtain,

\begin{equation}\label{kbintegral}
\begin{split}
     \int_{v} [g_{\alpha, \beta}(r) - 1] dv = & \\ v\frac{\langle N_\alpha N_\beta \rangle - \langle N_\alpha \rangle \langle N_\beta \rangle}{\langle N_\alpha \rangle \langle N_\beta \rangle}  -  \frac{\delta_{\alpha  \beta}}{\langle N_\alpha \rangle}
\end{split}
\end{equation}

\noindent which is precisely the relationship needed to connect the integrals of the radial distribution function with thermodynamic properties from the grand canonical ensemble. Just take eq \eqref{kbintegral} and substitute in eq \eqref{gcpd} and we obtain,

\begin{equation}
    c_\alpha c_\beta G_{\alpha, \beta} + \delta_{\alpha, \beta}c_\alpha = \frac{\beta}{v}\bigg(\frac{\partial \mu_\alpha}{\partial N_\beta}\bigg)_{T, V, N_{i \neq \beta}}
\end{equation}

\noindent where we have defined,

\begin{equation}
    G_{\alpha, \beta} = \int_{v} [g_{\alpha, \beta}(r) - 1] dv.
\end{equation}

\section{Principles of Bayesian Statistics}

\subsection{Probability Theory}

Probability theory is the foundation of statistics, population modeling, and making inferences from experimental data. Here we begin our discussion with set theory, which is fundamental in the study of probability.

\subsection{Set Theory}

The possible outcomes from an experiment are known as a sample space. This sample space is represented by a set, $\mathcal{S}$. Any subset of $\mathcal{S}$ is known as an event including $\mathcal{S}$ itself and the number of unique elements in the set can either be countable or uncountable depending on the possible experimental outcomes.

\begin{theorem}
    For any three events $A$, $B$, and $C$ contained in $\mathcal{S}$, the following relationships are true:
    
    \begin{equation}
        A \cup B = B \cup A
    \end{equation}
    
    \begin{equation}
        A \cap B = B \cap A
    \end{equation}
    
    \begin{equation}
        A \cup (B \cup C) = (A \cup B) \cup C
    \end{equation}
    
    \begin{equation}
        A \cap (B \cap C) = (A \cap B) \cap C
    \end{equation}
    
    \begin{equation}
        A \cap (B \cup C) = (A \cap B) \cup (A \cap C)
    \end{equation}
    
    \begin{equation}
        A \cup (B \cap C) = (A \cup B) \cap (A \cup C)
    \end{equation}
    
    \begin{equation}
        (A \cup B)^c = A^c \cap B^c
    \end{equation}
    
    \begin{equation}
        (A \cap B)^c = A^c \cup B^c.
    \end{equation}
\end{theorem}

We define two events $A$ and $B$ to be mutually exclusive if $A \cap B = 0$ and a collection of events ${A_i}$ to form a partition if $\bigcup_i A_i = \mathcal{S}$ and the set of ${A_i}$ are mutually exclusive.

\subsubsection{Probability}

For every event $A$ we associate a number between 0 and 1 as the probability of event $A$ as $P(A)$. This will only hold if the subsets of $\mathcal{S}$ form a Borel field $\mathcal{B}$, defined so that the elements of the field are closed under the union operation, all compliments of a subset also exist in $\mathcal{B}$, and that the null set is contained in $\mathcal{B}$. We can then define a probability function $P(A)$ with Borel field domain $\mathcal{B}$ such that,

\begin{equation}
    P(A) \geq 0
\end{equation}

\begin{equation}
    P(\mathcal{S}) = 1
\end{equation}

\begin{equation}
    P(\bigcup_i A_i) = \sum_i P(A_i)
\end{equation}

\noindent where the last equality holds only if the set ${A_i}$ are mutually exclusive. These are referred to as Kolmogorov's Axioms and are fundamental axioms from which we will construct probability theory.

It is also useful to define probability in a functional analytic sense. A probability space is defined as a measure space $(\Omega, \mathcal{S}, P)$ where $\Omega$ is a non-empty set of all possible events, $\mathcal{S}$ is a $\sigma$-algebra, and $P$ is a probability measure such that $P(\Omega) = 1$ (\textit{c.f.} Appendix B).

\subsubsection{The Fundamental Theorem of Counting}

\begin{theorem}
    If a job consists of $k$ separate tasks, the ith of which can be done in $n_i$ ways, then the entire job can be done in $n_1 \times n_2 \times ... \times n_k$ ways.
\end{theorem}

The Fundamental Theorem of Counting (FTC) seems quite obvious but may be difficult to implement in practice. For example, we may have situations where we cannot count the same thing twice, known as counting without replacement. Or, we may have situations where the order that we count matters, which is known as ordered counting.

Let's consider an example where a person can select six numbers from a set of 44 different numbers for a lottery ticket. We will examine each of the four different cases for counting the total number of possible ways for a ticket to be selected. \textit{Ordered, without replacement}: Using the FTC, we note that we can pick the first number in 44 ways, the second in 43 ways, etc, so that we have the total number of ways as $44 \times 43 \times 42 \times 41 \times 40 \times 39 = \frac{44!}{38!} = 5,082,517,440$ possible ways. \textit{Ordered, with replacement}: We now have 44 choices every time so the total number is $44 \times 44 \times 44 \times 44 \times 44 \times 44 = 44^6 = 7,256,313,856$ possible ways. \textit{Unordered, without replacement}: This case is interesting because there will be degenerate cases where the six chosen numbers are the same but selected in a different order. The FTC tells us that the six numbers can be arranged in $6!$ ways, which means the total number of unordered tickets are

\begin{equation}
    \frac{44 \times 43 \times 42 \times 41 \times 40 \times 39}{6!} = \frac{44!}{6!38!} = 7,059,052.
\end{equation}

\noindent We can define the previous operation as n choose r where 

\begin{equation}
    \binom{n}{r} = \frac{n!}{r!(n-r)!}
\end{equation}

\noindent which is a binomial coefficient. The unordered, without replacement case gave a total number of choices of $\binom{44}{6}$. \textit{Unordered, with replacement}: This is the most difficult case to count. We can think of this case as putting "markers" for the chosen numbers into 44 "bins" (which have 45 walls separating them). We should ignore the end walls since no matter how we rearrange the walls the ends will never contribute to the partition, so we have 43 total walls that can move around the 44 numbers. We also have 6 markers that can move around, leaving us with 49 different objects that can move around in $49!$ different ways. Since we are looking for the unordered ways to do this, we have to eliminate redundant orderings of the walls (43!) and the markers (6!), leaving us with a total count of,

\begin{equation}
    \binom{n+r-1}{r} = \frac{49!}{6!43!} = 13,983,816.
\end{equation}

\subsubsection{Conditional Probability}

The conditional probability is the probability of an event $A$ occurring given some other event(s) $B$ and is denoted by $P(A|B)$. In essence, the conditional probability takes the sample space into event(s) $B$ and then specifies the probability that $A$ is true,

\begin{equation}
    P(A|B) = \frac{P(A\cap B)}{P(B)}.
\end{equation}

The definition of conditional probability leads to Bayes' Rule, for which the proof is trivial and left as an exercise.

\begin{equation}
    P(A|B) = \frac{P(B|A) P(A)}{P(B)}.
\end{equation}

\subsubsection{Sample Spaces and Random Variables}

A mapping between a sample space to the real numbers is called a random variable. For example, if we roll two dice there are 36 possible outcomes for the two numbers that appear on the dice. However, we can define a random variable $X$ as the sum of the numbers on the two dice, leaving us with a variable that ranges from [2,12] in whole number increments. Every random variable has a cumulative distribution function which essentially describes the probability of obtaining a result that is less than or equal to some specified value of the random variable. We write

\begin{equation}
    F_X(x) = P_X(X \leq x)
\end{equation}

\noindent for all $x \in X$. This means that at the right most limit of our sample space (12 for the sum of two dice random variable) that $F_X(x) = 1$ and the left most limit that $F_X(x) = 0$. As we move from the left limit to the right limit, $F_X(x)$ cannot decrease, because it represents a sum of the probabilities of all the conditions before it. We can define our random variable as continuous or discontinuous based on whether or not $F_X(x)$ is continuous or discontinuous, respectively. Random variables $X,Y$ are said to be identically distributed if $F_X(x) = F_Y(x)$, which just means that their probability distributions are the same as long as they belong to the same subset of a Borel field.

\subsubsection{Probability Mass and Density Functions}

The probability mass and density functions are probability distributions on a random variable in the discrete and continuous cases respectively. They can be defined from the cumulative density distribution such that

\begin{equation}
    F_X(b) = \sum_{k=1}^b f_X(k)
\end{equation}

\noindent in the discrete case and

\begin{equation}
    F_X(x) = \int_{-\infty}^xf_X(t) dt
\end{equation}

\noindent in the continuous case. All this means is that the pmf and pdf are functions that, when summed up from left to right, give the cumulative distribution function. Moving away from the mathematical definition, the pmf(x) gives the probability of random variable $x$ and the integral of pdf(x + dx) in the neighborhood of random variable $x$ gives the probability of observing the random variable in the interval $[x,x+dx]$. PDFs must have finite integrals over the interval $(-\infty,\infty)$ which means that the right and left limits must converge to zero.

\subsection{Expectation Values and Moments}

Suppose we construct a map between our random variable $X$ to a new space $Y$ that conserves our distribution on $x$, named $g(x)$. The expected value is the mean or average value of a distribution in a sample space that is given by $Eg(x)$. Defined in terms of the cumulative density function we have for continuous distributions,

\begin{equation}
    Eg(x) = \int_{-\infty}^{\infty}g(x)f_X(x) dx
\end{equation}

\noindent and for discontinuous distributions

\begin{equation}
    Eg(x) = \sum_{x \in X} g(x) f_X(x).
\end{equation}

\noindent The expectation value can be generalized to moments by the following definition

\begin{equation}
    \mu_n' = EX^n
\end{equation}

\noindent where $\mu_n'$ is known as the $n^{th}$ moment of $X$. If we centralize the $n^{th}$ moment by the mean $\mu$, we obtain what are known as the $n^{th}$ central moments 

\begin{equation}
    \mu_n = E(X - \mu)^n
\end{equation}

\noindent where $\mu = EX$. The second central moment is commonly known as the variance and is important to describing the distribution of a random variable.

\subsubsection{Exponential Distribution Example}

Suppose that $X$ has an exponential distribution such that

\begin{equation}
    f_X(x) = \frac{1}{\lambda}e^{-x/\lambda}
\end{equation}

\noindent where $0 \leq x \leq \infty$ and $\lambda > 0$. We can calculate the mean of the distribution $EX$ as

\begin{equation}
    EX = \int_0^\infty x \frac{1}{\lambda}e^{-x/\lambda}dx = \int_0^\infty e^{-x/\lambda}dx - xe^{-x/\lambda}|^\infty_0 = \int_0^\infty e^{-x/\lambda} dx = \lambda.
\end{equation}

\noindent And now calculate the variance $E(X-\mu)^2 = E(X - \lambda)^2$

\begin{equation}
    E(X-\lambda)^2 = \int_0^\infty (x-\lambda)^2\frac{1}{\lambda}e^{-x/\lambda}dx = \int_0^\infty (x^2 - 2x\lambda + \lambda^2)\frac{1}{\lambda}e^{-x/\lambda}dx.
\end{equation}

\noindent We can just integrate each term at a time, starting with the $\lambda^2$ term

\begin{equation}
    \int_0^\infty \frac{\lambda^2}{\lambda}e^{-x/\lambda} dx = \lambda \int_0^\infty e^{-x/\lambda} dx = \lambda^2
\end{equation}

\begin{equation}
    -2\lambda \int_0^\infty x\frac{1}{\lambda}e^{-x/\lambda} dx = -2 \lambda^2
\end{equation}

\begin{equation}
    \int_0^\infty x^2 \frac{1}{\lambda}e^{-x/\lambda} dx = -x^2 e^{-x/\lambda}|^\infty_0 + 2\lambda \int_0^\infty x\frac{1}{\lambda} e^{-x/\lambda}dx = 2\lambda^2
\end{equation}

\noindent which gives a total variance of 

\begin{equation}
     E(X-\lambda)^2 = \int_0^\infty (x^2 - 2x\lambda + \lambda^2)\frac{1}{\lambda}e^{-x/\lambda}dx = \lambda^2 - 2\lambda^2 + 2\lambda^2 = \lambda^2.
\end{equation}

\subsubsection{Moment Generating Functions}

Moment generating functions are general functions that can generate the $n^{th}$ moments of a distribution. The moment generating function only exists if the expectation of a function exists in a neighborhood of zero. We define the moment generating function as

\begin{equation}
    M_X(t) = Ee^{tx}.
\end{equation}

\noindent The $n^{th}$ derivative of the moment generating function at $t=0$ is the $n^th$ moment. We can see this by considering the following argument

\begin{equation}
    \frac{d}{dt}M_X(t) = \frac{d}{dt}\int_{-\infty}^\infty e^{tx}f_X(x) dx = \int_{-\infty}^\infty xe^{tx}f_X(x) dx = EXe^{tx}
\end{equation}

\noindent assuming that we can move the derivative into the integral sign, which is a special case of the Leibniz rule that applies in this case. Now, notice that at $t=0$ we have

\begin{equation}
    \frac{d}{dt}M_X(t)|_{t=0} = EXe^{tx}|_{t=0} = EX.
\end{equation}

\noindent Now suppose we take $n$ partial derivatives of $M_X(t)$, We obtain something that looks like

\begin{equation}
    \frac{d^n}{dt^n}M_X(t) = \frac{d}{dt}\int_{-\infty}^\infty e^{tx}f_X(x) dx = \int_{-\infty}^\infty x^ne^{tx}f_X(x) dx = EX^n e^{tx}
\end{equation}

\noindent which gives us the same thing as before when evaluated at $t=0$ but for the $n^{th}$ moment

\begin{equation}
    \frac{d^n}{dt^n}M_X(t)|_{t=0} = EX^n e^{tx}|_{t=0} = EX^n.
\end{equation}

\subsection{Common Families and Distributions}

Families and distributions represent functional forms that can be used to model population distributions of random variables. In this section, we explore common distributions and their properties.

\subsubsection{Discrete Distributions}

The uniform discrete distribution is a simple distribution of discrete values with probability distribution given by

\begin{equation}
    P(X = x | N) = \frac{1}{N}.
\end{equation}

The hypergeometric distribution is best conceptualized by considering the case where you are looking for the probability of selecting $x$ red balls in $k$ choices from a jar filled with $M$ red balls and $N-M$ green balls. The fundamental theorem of counting gives us the probability distribution function,

\begin{equation}
    P(X=x|N,M,k) = \frac{\binom{M}{k}\binom{N-M}{k-x}}{\binom{N}{k}}.
\end{equation}

The binomial distribution is a distribution of independent Bernoulli trials that have only two possible outcomes that occur with probability $p$ and $1-p$ respectively. A conceptual example here is to consider the case where we want to know the probability of an event occurring $y$ times in $n$ trials, such as the number of heads in a coin flip experiment. Since the events are independent, we have a probability given by,

\begin{equation}
    P(Y=y|n,p) = \binom{n}{y}p^y(1-p)^{n-y}.
\end{equation}

\noindent Poisson distributions are useful in counting experiments where the probability of observing an event increases with time (which is intuitive for most cases). Poisson distributions have the form,

\begin{equation}
    P(X=x|\lambda) = \frac{e^{-\lambda}\lambda^x}{x!}.
\end{equation}

\subsubsection{Continuous Distributions}

The first common continuous distribution is the gamma distribution given by,

\begin{equation}
    f(t) = \frac{t^{\alpha-1}e^{-t}}{\Gamma(\alpha)}
\end{equation}

\begin{equation}
    \Gamma(\alpha) = \int_0^\infty t^{\alpha-1}e^{-t} dt.
\end{equation}

\noindent The normal distribution has probability distribution function given by,

\begin{equation}
    f(x|\mu, \sigma^2) = \frac{1}{\sqrt{2\pi} \sigma}e^{-(x - \mu^2)/(2\sigma^2)}.
\end{equation}

\noindent The mean and standard deviation give complete information about the shape and distribution of the random value and therefore belong to a special class of distributions known as the location-scale families. 

\subsubsection{Normal Variable Theorems}

The first important result to consider is the normal linear transform theorem which states that,

\begin{equation}
    \alpha + \beta \mathcal{N}(m, a^2) = \mathcal{N}(\alpha + \beta m, \beta^2a^2)
\end{equation}

\noindent which can be proved by finding the moment generating function of both sides of the equation. The second important theorem is the normal sum theorem which states that,

\begin{equation}
    \mathcal{N}(m_1 + m_2, a_1^2 + a_2^2) = \mathcal{N}(m_1, a_1^2) + \mathcal{N}(m_2, a_2^2)
\end{equation}

\noindent as long as the two variables on the right hand side are statistically independent. This can also be proved by showing that the moment generating functions on the right and left hand sides are equal.

\subsection{Bayesian Inference}

In Bayesian inference over a model $M$, we are interested in the posterior probability distribution of the model parameters $\theta$ given some observed data, which we denote $p(\theta, M | D)$. According to Bayes' theorem, this quantity can be expressed as,

\begin{equation}
    p(\theta, M | D) = \frac{p(D|\theta, M) p(\theta, M)}{p(D)}
\end{equation}

\noindent where $p(D|\theta, M)$ is the likelihood that data $D$ is explained by the model and its parameters, $p(\theta, M)$ is our prior knowledge (a preselected parameter distribution based on expert beliefs), and $p(D)$ is the probability of observing data $D$ at all. The latter quantity need not be computed since we can often normalize the probability distribution \textit{post hoc}, allowing us to compute only the quantity,

\begin{equation}
    p(\theta, M | D) \propto p(D|\theta, M) p(\theta, M).
\end{equation}

In practice, then, all we need to do is choose prior distributions and a likelihood equation to perform Bayesian inference. If the solution to the problem is not analytical, we can then compute the posterior by sampling this distribution with Markov chain Monte Carlo (MCMC). 

\subsubsection{Bayesian Model Averaging}

What happens if we have multiple models that provide information about the same quantity-of-interest and we want to combine those predictions? Bayesian model averaging allows us to infer the probability distribution function on a set of model parameters by combining results from different models in a rigorous way. The idea is basically to compute an average of the parameter probability distributions calculated from Bayesian inference weighted by the model probability and model evidence. Note that in the case that you are estimating a model parameter that the parameter must have the exact same interpretation in both models. If the parameter has a different interpretation, then Bayesian model averaging will not apply.

In Bayesian model averaging, the posterior distribution for the model parameters $\theta$ given experimental data $D$, $p(\theta|D)$, is equal to a weighted sum of posterior distributions over a model set $M = [M_1, ..., M_n]$ such that,

\begin{equation}\label{eq:bma}
    p(\theta|D) = \sum_M p(M|D) p(\theta|M, D)
\end{equation}

\noindent where $p(\theta|M, D)$ is the posterior distribution of model parameters for model $M$ and $p(M|D)$ is the posterior model probability computed via,

\begin{equation}\label{eq:postmodelprob}
    p(M|D) = \frac{p(D|M)p(M)}{\sum_{M'} p(D|M') p (M')}.
\end{equation}

The quantity $p(D|M)$ is known as the marginal likelihood of model $M$ and $p(M)$ is the prior model probability.

Assuming that the probability of each model is the same leads to a simplification of eq \eqref{eq:postmodelprob} such that,

\begin{equation}\label{eq:postmodelsimplified}
    p(M|D) = \frac{p(D|M)}{\sum_{M'} p(D|M')}.
\end{equation}

The tricky step is determining the marginal likelihood of model $M$, $p(D|M)$. However, notice that we can compute this quantity using the likelihood computed for that model and its parameters in the following way,

\begin{equation}
    p(D|M) = \int p(D|\theta, M) p(\theta, M) d\theta
\end{equation}

\noindent which is just the integral over the parameter space of the likelihood multiplied by the prior for model $M$. This is why $p(D|M)$ is known as a marginal likelihood, since it is a marginal distribution of the likelihood without the model parameters.

At this stage, we simply need to run Bayesian inference over both models, compute the marginal likelihood of each model, plug those results into eq \eqref{eq:postmodelsimplified} to obtain the posterior model probabilities, and finally compute the Bayesian model average in eq \eqref{eq:bma} using the Bayesian posterior probabilities computed over all models weighted by the posterior model probabilities. The result will be a probability distribution on the model parameters that takes into account the fit quality of both models in a rigorous way.

\subsection{Bayesian Nonparametrics}

Bayesian nonparametrics is the application of Bayesian inference to models with no fixed number of parameters \cite{hjort_bayesian_2010}. For example, Bayesian nonparametrics can be used to learn functional distributions for infinite dimensional continuous and differentiable functions defined on the real line or even multidimensional fields. The main idea is that we are extending the earlier notion of Bayesian inference over a set of model parameters to an entire distribution of functions. Detailed and mathematically rigorous formulations can be found in Lemm's \textit{Bayesian Field Theory} \cite{lemm_bayesian_2003}.

\subsubsection{Gaussian Processes}

Of the various nonparameteric Bayesian methods, Gaussian processes are the most widely used for physics based applications. A Gaussian process is a stochastic process such that every finite set of random variables (position, time, etc) has a multivariate normal distribution. In other words, suppose we take some real world process represented as a function $f(\mathbf{x})$. We say that $f(\mathbf{x})$ is distributed as a Gaussian process with mean $\mu(\mathbf{x})$ and covariance $k(\mathbf{x}, \mathbf{x'})$ by writing,

\begin{equation}
    f(\mathbf{x}) \sim \mathcal{GP}(\mu(\mathbf{x}), k(\mathbf{x}, \mathbf{x'}))
\end{equation}

\noindent where the $\sim$ symbol can be interpreted as 'distributed as'. The predictive mean of this distribution can be computed using the complete-the-square trick to obtain,

\begin{equation}
    \bar{\mathbf{f_*}} = \mu(X_*) + K(X,X_*)K_y^{-1}[\mathbf{y} - \mu(X)]
\end{equation}

\noindent with covariance,

\begin{equation}
    \text{cov}(\mathbf{f_*}) = K(X_*,X_*) - K(X,X_*)[K(X,X) +\sigma_n^2I]^{-1}K(X,X_*)
\end{equation}

\noindent where $K$ is the covariance matrix constructed from $k(\mathbf{x}, \mathbf{x'})$, $K_y = K(X,X) - \sigma_n^2I$, $I$ is the identity matrix, $\mathbf{y}$ is a given set of training observations at training matrix inputs $X$ and model output matrix $X_*$. Often, the mean function is chosen to be zero, but choosing a non-zero mean function can be beneficial to enforce physics based behavior in the resulting function and/or to enhance interpretability (\textit{c.f.} Chapter 3).

Note that Gaussian processes are also used for classification problems \cite{rasmussen_gaussian_2006}, including one of the authors recent papers on tree species classification from airborne imaging spectroscopy data \cite{seeley_classifying_2023}.

\subsubsection{Physics-Informed Kernel Design}

Kernel design is arguably the most important aspect of Gaussian process regression and classification. The kernel, along with the Gaussian process prior mean, specify a prior state-of-knowledge that can be leveraged to improve the physical reliability of the Gaussian process prediction. Kernel selection can enforce highly general notions such as continuity and differentiability, function behavior including periodicity, and also non-stationary covariances through methods such as the Gibbs kernel \cite{gibbs_bayesian_1997,heinonen_non-stationary_2016}. Note that knowing how to construct physics-informed kernels for a learning problem of interest can be extremely valuable and restrict the function space to only physically realizable predictions.

From here, we can use thermodynamic relationships to derive a number of properties of multi-component systems in terms of the Kirkwood-Buff integrals since we know the relationship between thermodynamic derivatives and measurable thermodynamic properties.

\section{Introduction to Functional Analysis}

Quantum theory developed around the same time as the theory of structural correlations in condensed matter systems. Unlike classical mechanics, which was used to construct much of the known statistical theory of liquids, quantum mechanics is a deeper and more fundamental physical description of the behavior of atoms. Early work by Heisenberg, Schrodinger, and Planck was based on evidence from famous experiments from the early 20$^{th}$ century; however, a satisfactory description of quantum theory was not presented until John von-Neumann, a chemical engineer by training and a mathematician at heart, proposed that quantum mechanics could be described by an emerging mathematical field: functional analysis. The content of this appendix is necessary to understand the notation and mathematical objects that describe the functional analysis interpretation of the Henderson inverse theorem for Lennard-Jones type fluids as well as Bayesian field theory.

\subsection{Introduction to Functional Analysis}

Functional analysis will provide us with the mathematical rigor to describe physical observables such as the momentum, spin, or position with both discrete and continuous spectra. This section uses rigorous mathematical notation to properly define the quantum theory within the context of the mathematics of infinite dimensional vector spaces and is based on the excellent lecture series from Frederich Schuller and the book \textit{Mathematical Foundations of Quantum Mechanics} by John von Neumann.

\subsubsection{Banach and Separable Hilbert Spaces}

A \textbf{Banach space} is a vector space $(V, +, \cdot, ||\cdot||)$ equipped with a non-negative, linear norm that satisfies the triangle inequality and that is complete with respect to this norm, meaning that any Cauchy sequence in $V$ converges to an element in $V$. A sequence $\{f_i\}$ for $f_i \in V$ is Cauchy if,

\begin{equation}
    \forall \epsilon > 0 \;  \exists N \in \mathbb{N} \; \ni n,m \geq N \; \norm{f_n - f_m} < \epsilon
\end{equation}

\noindent which means that for any real number $\epsilon$ greater than 0 (including those that are arbitrarily small) that there exists some natural number $N$ such that for $n,m$ larger than $N$ that the norm of $f_n - f_m$ is smaller than $\epsilon$. A sequence $\{f_i\} \subseteq V$ is said to strongly converge to $f \in V$ if,

\begin{equation}
     \forall \epsilon > 0 \; \exists n \in \mathbb{N} \ni \forall n>N \; \norm{f_n - f} < \epsilon
\end{equation}

\noindent and converge weakly if,

\begin{equation}
    \forall \ell \in V^*
\end{equation}

\noindent that $\ell(f_i)$ strongly converges to $\ell(f)$. Here $V^*$ is defined as the dual space of a normed space $V$, which is the space of bounded linear maps, $\mathcal{L}(V,\mathbb{C})$, equipped with the norm $\norm{}_\mathcal{L}$.

Now, a bounded map $A: V \rightarrow W$ that takes an element of a normed (and not necessarily complete vector space) to a complete Banach space satisfies,

\begin{equation}
    \norm{A} = \sup_{\norm{f}_V = 1} \norm{Af}_W < \infty
\end{equation}

\noindent where $\norm{A}$ is called the \textbf{operator norm}. Hence, the operator norm measures the "size" of $A$ by determining the biggest change in the $\norm{Af}_W$ and clearly depends on $f$ and the norms in the $V$ and $W$ spaces.

\begin{theorem}
    The set of all linear and bounded maps defined as, \newline  $\mathcal{L}(V,W) := \{A:V \rightarrow W\}$ with $A$ linear and bounded, is a Banach space if equipped with pointwise addition and the operator norm $\norm{\cdot}_{\mathcal{L}(V,W)}$.
\end{theorem}

\begin{proof}
We begin by defining addition and scalar multiplication on $\mathcal{L}(V,W)$ so that,

\begin{equation}
    +_{\mathcal{L}}: \mathcal{L} \times \mathcal{L} \rightarrow \mathcal{L}
\end{equation}

\begin{equation}
    \cdot_{\mathcal{L}}: \mathbb{C} \times \mathcal{L} \rightarrow \mathcal{L} \; \forall \alpha \in \mathbb{C}.
\end{equation}

\noindent Commutativity, associativity, the neutral element and an inverse element of the $+_{\mathcal{L}}$ and associativity, distributivity, and the scalar identity are given since $V$ is a vector space and $\mathcal{L}$ are linear maps. Since $\mathcal{L}$ is bounded and W is a Banach space, we must have,

\begin{equation}
    \norm{\cdot}_{\mathcal{L}(V,W)} = \sup_{f \in V: \norm{f}_V = 1} \norm{Af}_W < \infty.
\end{equation}

\noindent It suffices to show that $(\mathcal{L}(V,W), +_{\mathcal{L}}, \cdot_{\mathcal{L}}, \norm{\cdot}_{\mathcal{L}})$ is complete with respect to the operator norm. Choose a candidate map as the limit of a Cauchy sequence in $\mathcal{L}$ such that,

\begin{equation}
    A: V \rightarrow W
\end{equation}

\begin{equation}
    f \in V \rightarrow \sum^\infty_{n=1} A_nf \in W.
\end{equation}

\noindent $A$ is linear and bounded since $A_n$ is linear and Cauchy in $\mathcal{L}$. Now, we just need to show that $A_n$ converges against $A$. Since $A_n$ is Cauchy,

\begin{equation}
    \forall \epsilon > 0 \; \exists N \in \mathbb{N} \; \ni n,m \geq N \; \norm{A_n - A_m} < \epsilon.
\end{equation}

\noindent By the definition of the operator norm, we have the estimate,

\begin{equation}
     \forall f \in V \; \frac{\norm{(A_n - A_m)f}_W}{\norm{f}_V} \leq \norm{A_n - A_m} < \epsilon.
\end{equation}

\noindent Noting the linearity of addition and taking the limit as $m \to \infty$ gives,

\begin{equation}
    \lim_{m \to \infty} \norm{A_nf - A_mf} = \norm{A_nf - \lim_{m \to \infty} A_mf} = \norm{A_nf - Af} < \epsilon \norm{f}_V.
\end{equation}

\noindent Thus,

\begin{equation}
     \forall f \in V \; \forall \epsilon > 0 \; \frac{\norm{A_nf - Af}_W}{\norm{f}_V} < \epsilon.
\end{equation}

\noindent But since this is true for all $f \in V$, it is certainly true for $f$ that give the supremum of this quantity,

\begin{equation}
    \forall f \in V \; \forall \epsilon > 0 \; \exists N \in \mathbb{N} \; \ni n \geq N \; \norm{A_n - A} \leq \epsilon.
\end{equation}

\noindent Hence, $A_n$ converges against $A \in \mathcal{L}(V,W)$ and therefore $\mathcal{L}(V,W)$ is complete.
\end{proof}

\begin{theorem}
    The Bounded Linear Transformation (BLT) Theorem: There is a unique extension, $\hat{A}$, of a bounded linear map,
    
    \begin{equation}
        A: \mathcal{D}_A \subseteq V \rightarrow W \ni \overline{\mathcal{D}_A} = V
    \end{equation}
    
    \noindent where $\mathcal{D}_A$ is densely defined subset of a normed vector space $V$. This extension is defined as,
    
    \begin{equation}
        \hat{A}: V \rightarrow W \ni \hat{A}(\alpha) \; \forall \alpha \in \mathcal{D}_A \implies \hat{A}(\alpha) = A(\alpha).
    \end{equation}
\end{theorem}

\begin{proof}

\noindent Let $f \in V$ and consider a sequence $\{f_n\} \subseteq \mathcal{D}_A$ such that $\lim_{n \to \infty} f_n = f$. Now, this means that $f_n$ and $Af_n$ is Cauchy since $A$ is bounded and further that $\lim_{n \to \infty} A f_n \in W$ because $W$ is Banach. We can then propose an extension,

\begin{equation}
    \hat{A}:V \rightarrow W
\end{equation}

\begin{equation}
    f \in V \rightarrow \hat{A}f := \lim_{n \to \infty} A f_n
\end{equation}

\noindent which is clearly an extension since, if we let $f \in \mathcal{D}_A \subseteq V$ such that $\lim_{n \to \infty} f_n = f$,

\begin{equation}
    \hat{A}(f) = \lim_{n \to \infty}Af_n = A(f).
\end{equation}

\noindent Now, the extension $\hat{A}$ is linear due to the continuity of addition and scalar multiplication. We now just need to show that $\hat{A}$ is unique. Suppose that there is a second extension $\hat{B}$ and consider $\hat{A} - \hat{B}$ acting on a Cauchy sequence $f_n$,

\begin{equation}
    (\hat{A} - \hat{B})(f_n) = 0.
\end{equation}

\end{proof}

To understand the definition of a separable Hilbert space, we must first consider a few important theorems and definitions of bases in infinite dimensional vector spaces. A Hilbert space is a $\mathbb{C}$-vector space ($\mathcal{H}$, +, $\cdot$) equipped with a sesqui-linear inner product that induces a norm, with respect to which ($\mathcal{H}$, +, $\cdot$) is complete. A \textbf{sesqui-linear inner product} is an inner product $(\cdot, \cdot): \mathcal{H} \times \mathcal{H} \rightarrow \mathbb{C}$ that is Hermitian,

\begin{equation}
    (\phi, \psi) = \overline{(\psi, \phi)}
\end{equation}

\noindent and linear in the second argument,

\begin{equation}
    (\phi, \psi_1 + \alpha \psi_2) =  (\phi, \psi_1) + \alpha(\phi, \psi_2)
\end{equation}

\noindent with the additional condition that $(\psi, \psi) \geq 0$ with equality when $\psi = 0 \in \mathcal{H}$.

\begin{theorem}
    A norm $\norm{\cdot}: V \rightarrow \mathbb{R}$ is induced by a sesqui-linear inner product if and only if the parallelogram identity holds,
    
    \begin{equation}
        \forall f,g \in V \; \norm{f+g}^2 + \norm{f-g}^2 = 2\norm{f} + 2\norm{g}.
    \end{equation}
\end{theorem}

\begin{proof}

By the definition of the norm we have,

\begin{equation}
    \norm{f+g}^2 + \norm{f-g}^2 = \sup_{v \in V} \bigg(\frac{\abs{f(v)+g(v)}^2}{\norm{v}^2} + \frac{\abs{f(v)+g(v)}^2}{\norm{v}^2}\bigg)
\end{equation}
  
\begin{equation}
    \sup_{v \in V} \frac{2f\bar{f} + 2g\bar{g}}{\norm{v}} = 2\norm{f} + 2\norm{g}.
\end{equation}

\end{proof}

The notion of a basis set for a naked vector space $(V, +, \cdot)$ is familiar from linear algebra in the form of a Hamel basis. A subset $B \subseteq V$ is called a \textbf{Hamel basis} if,

\begin{enumerate}
    \item Any finite subset $\{\mathbf{e}_1, \mathbf{e}_2, ..., \mathbf{e}_n\} \subseteq B$ is linearly independent. This means that,
    
    \begin{equation}
        \sum_{i = 1}^n \lambda_i \mathbf{e}_i = 0
    \end{equation}
    
    \noindent implies that $\lambda_1 = \lambda_2 = ... = \lambda_n = 0$.
    
    \item For any $v \in V$ there exists a finite subset $\{\mathbf{e}_1, \mathbf{e}_2, ..., \mathbf{e}_n\} \subseteq B$ and complex numbers $v_1, v_2, ..., v_n\}$ so that,
    
    \begin{equation}
        v = \sum_{i=1}^n v_i \mathbf{e}_i.
    \end{equation}
\end{enumerate}

\noindent The dimension of the vector space is then defined based on the number of elements of the basis set (it can be finite or infinite). In quantum mechanics, the notion of a basis is actually different since we are working with Hilbert spaces. The \textbf{Schauder basis} $S \subseteq \mathcal{H}$ on a Hilbert space ($\mathcal{H}, +, \cdot, \bra{\cdot}\ket{\cdot})$ is defined as,

\begin{enumerate}
    \item Same condition as (1) for a Hamel basis.
    
    \item For any $\psi \in \mathcal{H}$ there exists a unique, countable subset $\{\mathbf{e}_1, \mathbf{e}_2, ...\} \subseteq S$ and a unique sequence $\{\psi_i\}_{i = \{1,2,...\}} \in \mathbb{C}$ such that $\psi$ can be written like (2) for finite dimensions and if $|S| = \infty$,
    
    \begin{equation}
        \psi = \sum_{i=1}^\infty \psi_i \mathbf{e}_i.
    \end{equation}
\end{enumerate}

A Schauder basis $S$ is a more powerful means to generate elements of a Hilbert space, since in general we need less elements to write any element of the Hilbert space than in the more familiar Hamel basis, $B$. A Hilbert space is called separable if it has a countable and orthonormal Schauder basis. Surprisingly, all infinite dimensional separable Hilbert spaces are unitarily equivalent to the set of all square-summable sequences in the complex numbers. The definition of a unitary map $\mathcal{U} \in \mathcal{L}(\mathcal{H}, \mathcal{G})$ between two Hilbert spaces $\mathcal{H}$ and $\mathcal{G}$ is that the map preserves the inner product such that,

\begin{equation}
    \forall \phi,\psi \in \mathcal{H} \; (\mathcal{U} \psi, \mathcal{U} \phi)_{\mathcal{G}} = (\psi, \phi)_{\mathcal{H}}.
\end{equation}

\noindent By the polarization formula, it suffices to check whether $\mathcal{U} \in \mathcal{L}(\mathcal{H}, \mathcal{G})$ preserves norms.

\begin{theorem}
    $(\mathcal{H}, +, \cdot, (\cdot,\cdot))$ is unitarily equivalent to $\ell^2(\mathbb{N},(\cdot,\cdot))$.
\end{theorem}

\begin{proof}
We need to show that there exists a bounded linear map that preserves norms such that,

\begin{equation}
    \norm{\mathcal{U} \psi}_{\ell^2} = \norm{\psi}_\mathcal{H}.
\end{equation}

\noindent Let's choose an orthonormal Schauder basis on $\mathcal{H}$ and construct a map that takes an element of $\mathcal{H}$ into $\ell^2(\mathbb{N})$ such that $\psi \rightarrow \{(\mathbf{e}_n, \psi)\}_{n \in \mathbb{N}}$. Since $\{\mathbf{e}_i\}$ is a Schauder basis we have,

\begin{equation}
    \psi = \sum_{i=1}^n (\mathbf{e}_i, \psi) \mathbf{e}_i
\end{equation}

\noindent and $\norm{\psi}^2_\mathcal{H} = \sum_{n=1}^\infty |(\mathbf{e}_n, \psi)|^2 = \norm{\mathcal{U}\psi}^2_{\ell^2(\mathbb{N})}$. We can show that this map is linear and bounded trivially, completing the proof.

\end{proof}

In this section, we have introduced the definitions of separable Hilbert spaces and proved two very important results in quantum theory. The first is the Bounded Linear Transformation theorem which shows that for a bounded linear map $A$ defined on a dense subset of a normed vector space $V$ that there exists a unique extension $\hat{A}$ defined on $V$. The second is that all infinite-dimensional separable Hilbert spaces are unitarily equivalent to the set of square-summable sequences in the complex numbers. 

\subsubsection{Measure Theory}

Here a short introduction to measure theory is provided without proofs since (1) the spectral theorem for self-adjoint operators requires the notion of mathematical objects called projection-valued measures and (2) the most common $\infty$-dimensional Hilbert space in quantum mechanics is the $L^2(\mathbb{R}^d)$ space equipped with the Lebesgue measure.

Let $M$ be a non-empty set. Then the collection of subsets $\sigma \subseteq \mathcal{P}(M)$ (the power set of $M$) is called a $\sigma$-algebra for $M$ if,

\begin{enumerate}
    \item $M \in \sigma$
    \item For $A \in \sigma$, then $M \backslash A \in \sigma$
    \item For a sequence $A_1, A_2, ... \in \sigma$, then $\bigcup_{n \geq 1} A_n \in \sigma$
\end{enumerate}

\noindent A set $A \in \sigma$ is called \textbf{measurable} since it can be assigned a "volume". A \textbf{measurable space} is the original set $M$ taken with the additional structure of the $\sigma$-algebra, $(M, \sigma)$. However, at this point it is not clear how to actually measure the "volume" of elements in the set, so we add an additional map known as a measure that takes an element of the $\sigma$-algebra to an element of the extended, positive real numbers,

\begin{equation}
    \mu: \sigma \rightarrow \bar{\mathbb{R}^+_0}
\end{equation}

\noindent that satisfies the following properties:

\begin{enumerate}
    \item $\mu(\emptyset) = 0$
    \item For a pairwise disjunct sequence $A_1, A_2, ... \in \sigma$ such that $A_i \cap A_j = \emptyset$ if $i \neq j$ then,
    
    \begin{equation}
        \mu\bigg(\bigcup_{n \geq 1} A_n\bigg) = \sum_{n \geq 1}\mu(A_n)
    \end{equation}
\end{enumerate}

\noindent Of course, if the sequence in (2) above is not pairwise disjunct we have the estimate,

\begin{equation}
    \mu\bigg(\bigcup_{n \geq 1} A_n\bigg) \leq \sum_{n \geq 1}\mu(A_n).
\end{equation}

\noindent The \textbf{measure} is the map that takes an element of the $\sigma$-algebra and "measures" it, or actually assigns it a "volume". Taking the set $M$ along with its $\sigma$-algebra and a measure defines a \textbf{measure space}, $(M, \sigma, \mu)$. 

Up to this point, the definition of a measure space has been quite abstract, with no clear prescription for how to choose the $\sigma$-algebra or measure. In quantum mechanics, there are actually very few $\sigma$-algebras and measures that will be important to our study of the spectral theorem of self-adjoint operators, but we will first need to introduce the concept of a topology to motivate our choices. A \textbf{topology} on $M$ is the collection of subsets $\mathcal{O}\subseteq P(M)$ with the following properties:

\begin{enumerate}
    \item $\emptyset, M \in \mathcal{O}$
    \item Closure under countable unions of $u_i \in \mathcal{O}$ implies that,
    
    \begin{equation}
        \bigcup_{i = 1}^\infty u_i \in \mathcal{O}.
    \end{equation}
    
    \item Closure under finite intersections such that if $u_i \in \mathcal{O}$ then,
    
    \begin{equation}
        \bigg(\bigcap_{i=1}^\infty u_i \bigg) \in \mathcal{O}.
    \end{equation}
\end{enumerate}

\noindent and a space taken with its topology is termed a topological space, $(M, \mathcal{O})$. The open sets on a topological space $(M,\mathcal{O})$ are an interesting generating set for a $\sigma$-algebra known as the Borel-$\sigma$-algebra. When $M = \mathbb{R}$, the Borel-$\sigma$-algebra generates the "standard topology" such that,

\begin{equation}
    \forall a,b \in \overline{\mathbb{R}}, \; (a,b) \in \sigma
\end{equation}

\noindent which can then be equipped with a special measure known as the \textbf{Lebesgue measure},

\begin{equation}
    \mu_L = \lambda^d:\sigma(\mathcal{O}) \rightarrow \overline{\mathbb{R}^+_0}
\end{equation}

\noindent where,

\begin{equation}
    \lambda^d([a_1, b_1) \times [a_2, b_2) \times ... \times [a_d, b_d)) = (b_1-a_1)...(b_d-a_d)
\end{equation}

\noindent which just defines volume as the product of the lengths of the sides of a d-dimensional rectangular prism in $\mathbb{R}^d$. In quantum theory, we will almost always take the vector space of real numbers, $\mathbb{R}^d$, equipped with Borel-$\sigma$-algebra, $\sigma(\mathcal{O})$, and Lebesgue measure, $\lambda^d$, as our measure space, $(\mathbb{R}^d, \sigma(\mathcal{O}), \lambda^d)$.

\subsubsubsection{The Lebesgue Integral}

By definition, a map $f:M \rightarrow N$ is called measurable (w.r.t. $\sigma$-algebras on each set) if $\forall A \in \sigma_N: \text{preim}_f(A) \in \sigma_m$. Conceptually, all this means is that a map between two measurable spaces is measurable if the preimage of any element of the target space $\sigma$-algebra is in the original space $\sigma$-algebra. Clearly, we then have that any continuous map is measurable w.r.t. the Borel-$\sigma$-algebra and any monotonous map is measurable.  

The \textbf{push-forward} of a measure is a way that a measurable space can inherit a measure from a measure space. Let $(M, \sigma, \mu)$ be a measure space, $(N, \tau)$ be a measurable space, and let $f:M \rightarrow N$ be a measurable map so that $\forall A \in \tau: \text{preim}_f(A) \in \sigma$. The so-called push-foward $f_*\mu$ is a measure on $(N, \tau)$ such that,

\begin{equation}
    f_*\mu: \tau \rightarrow \bar{\mathbb{R}}_0^+
\end{equation}

\begin{equation}
    \forall A \in \tau: \; f_*\mu(A) = \mu(\text{preim}_f(A))
\end{equation}

\noindent which means that the push-forward allows us to construct a measure space with measure $f_*\mu$; namely, $(N, \tau, f_*\mu)$.

We begin our discussion of Lebesgue integrals by studying only non-negative measurable functions. We will see that in general measurable functions can be integrated so long as the integral of the absolute value of the function is finite, but it will require decomposing the general measurable function into two (or four in the case of complex numbers) integrals over non-negative measurable functions. 

We will begin by defining a \textbf{simple function}, which is a measurable function $s:M \rightarrow \mathbb{R_0^+}$ that satisfies,

\begin{equation}
    s(M) = \{s_1, ..., s_N\} \; \text{for some} \; N\in\mathbb{N}
\end{equation}

\noindent which by definition means that $M$ and $\mathbb{R_0^+}$ are measurable spaces, ($M, \sigma(\mathcal{O_M})$) and ($\mathbb{R_0^+}, \sigma(\mathcal{O})$). Hence, we know that,

\begin{equation}
    \text{preim}_s(\{s_i\}) \in \sigma(\mathcal{O_M})
\end{equation}

\noindent and,

\begin{equation}
    s = \sum_{z \in s(M)} z \cdot \chi_{\text{preim}_s(\{z\})}
\end{equation}

\noindent where the $\chi_{\text{preim}_s(\{z\})}$ is a map $\chi_A:M \rightarrow \mathbb{R}$ such that,

\begin{equation}
    \chi_A(m \in M) = \left\{
        \begin{array}{ll}
            1 & \quad m \in A \\
            0 & \quad m \notin A
        \end{array}
    \right.
\end{equation}

\noindent known as the \textbf{characteristic function}. Therefore, the characteristic function $\chi_{\text{preim}_s(\{z\})}$ is a map from $M$ to $\mathbb{R}$ that is zero unless it is evaluated for some $m \in M$ where $\{z\} = s(m)$ and the preimage of $s(m)$ is in the $\sigma$-algebra of $M$. 

A non-negative measurable function is a measurable function $f:M \rightarrow \bar{\mathbb{R}}$ such that,

\begin{enumerate}
    \item $\forall m \in M$: $f(m) \geq 0$
    \item f is measurable on $\sigma(\mathcal{O}_{\bar{\mathbb{R}}})$
\end{enumerate}

\noindent Now, the \textbf{Lebesgue integral} of a non-negative measurable function $f:(M, \sigma_M, \mu) \rightarrow \bar{\mathbb{R}}$ is given by,

\begin{equation}
    \int f(x) \mu(dx) = \int f d\mu := \sup{\bigg[\sum_{z \in s(M)} z \cdot_{\bar{\mathbb{R}}} \mu(\text{preim}_f(\{z\})\bigg]}
\end{equation}

\noindent where $s(M) = \{s_1, s_2, ..., s_n\}$ is a finite sequence of values from simple measurable functions. Notice we do not require a measure on the target space, but only on the domain (owing to the use of the supremum). The Lebesgue integral of a non-negative and measurable function has the following nice properties:

\begin{enumerate}
    \item Markov inequality: $z \in \mathbb{R}_0^+$ then $\int f d\mu \geq z \cdot \mu(\text{preim}_f\{z\})$.
    \item $f \leq g$ implies $\int f d\mu \leq \int g d\mu$.
    \item $f =_{a.e.} g$ implies $\int f d\mu = \int g d\mu$.
    \item $\int f d\mu = 0$ implies $f =_{a.e.} 0$.
    \item $\int f d\mu < \infty$ implies $f <_{a.e.} \infty$.
\end{enumerate}

\noindent where the abbreviation a.e. stands for "almost everywhere". If an expression holds \textbf{almost everywhere} then it holds on the measure space for all elements where the measure is non-zero, $\mu(A) \neq 0$ for $A \in \sigma$. The following three theorems will be stated without proof.

\begin{theorem}
    Theorem of Monotone Convergence: Let $0 \leq f_1 \leq f_2 ...$ be a sequence of measurable functions such that $f_n:M \rightarrow \bar{\mathbb{R}}$. Provided that $f := \sup_{n \geq 1}{f_n}$ pointwise, then we have,
    
    \begin{equation}
        \lim_{n \to \infty} \int f_n d\mu = \int f d\mu.
    \end{equation}
\end{theorem}

\begin{theorem}
    Let $f,g \geq 0$ be measurable functions and $\alpha \in \mathbb{R}_0^+$. Then,
    
    \begin{equation}
        \int (f + \alpha g) d\mu = \int f d\mu + \alpha \int g d\mu.
    \end{equation}
\end{theorem}

\begin{theorem}
    For any sequence $f_n: M \rightarrow \bar{\mathbb{R}}^+$ we have,
    
    \begin{equation}
        \int \bigg(\sum_{n = 1}^\infty f_n \bigg) d\mu = \sum_{n = 1}^\infty \int f_n d\mu.
    \end{equation}
\end{theorem}

Now, extending the previous notions to general functions $f:M \rightarrow \bar{\mathbb{R}}$ we define that the function is \textbf{integrable} if:

\begin{enumerate}
    \item $f$ is measurable 
    \item $\int |f| d\mu$ is finite
\end{enumerate}

\noindent Remark: Condition 2 is equivalent to requiring that the decomposition of $f = f^+ - f^-$ into $f^+ := \max (f,0)$ and $f^- := \max (-f,0)$ gives integrals over $f^+$ and $f^-$ that are finite. The integral of an integrable function is then defined as (with a visualization provided in Figure \ref{fig:lebintegral_1}),

\begin{equation}
    \int f d\mu := \int f^+ d\mu - \int f^- d\mu
\end{equation}

\noindent Note that this is easily generalized into complex spaces since we can write,

\begin{equation}
    \int f d\mu := \int \Re (f) d\mu + i\int \Im (f) d\mu 
\end{equation}

\noindent for both $f^+$ and $f^-$. 

The Lebesgue integral of integrable functions $f,g$ has the following properties:

\begin{enumerate}
    \item $f \leq_{a.e.} g$ implies $\int f d \mu \leq \int g d\mu$
    \item $\forall \alpha \in \mathbb{R}$: $\int(f+\alpha g) d\mu = \int f d \mu +\alpha \int g d\mu$
\end{enumerate}

\begin{theorem}
    Theorem of Dominant Convergence: Let $f_1, f_2, ...$ be a sequence of non-negative functions such that $f \rightarrow_{a.e.} f$ pointwise for some function $f$. Let $g$ be a non-negative, measurable function with $\int g d\mu < \infty$ s.t. $\forall n \in \mathbb{N}$: $|f_n| \leq_{a.e.} g$. Then the following properties hold:
    
    \begin{enumerate}
        \item $f$ and all $f_i$ are integrable including the limit function.
        \item $\lim_{n \to \infty} \int |f_n - f| d\mu = 0$
        \item $\lim_{n \to \infty} \int f_n d\mu = \int f d\mu$
    \end{enumerate}

    \noindent A visualization of the theorem of dominant convergence is provided in Figure \ref{fig:domconverge_1}.
\end{theorem}

\begin{figure}
    \centering
    \includegraphics{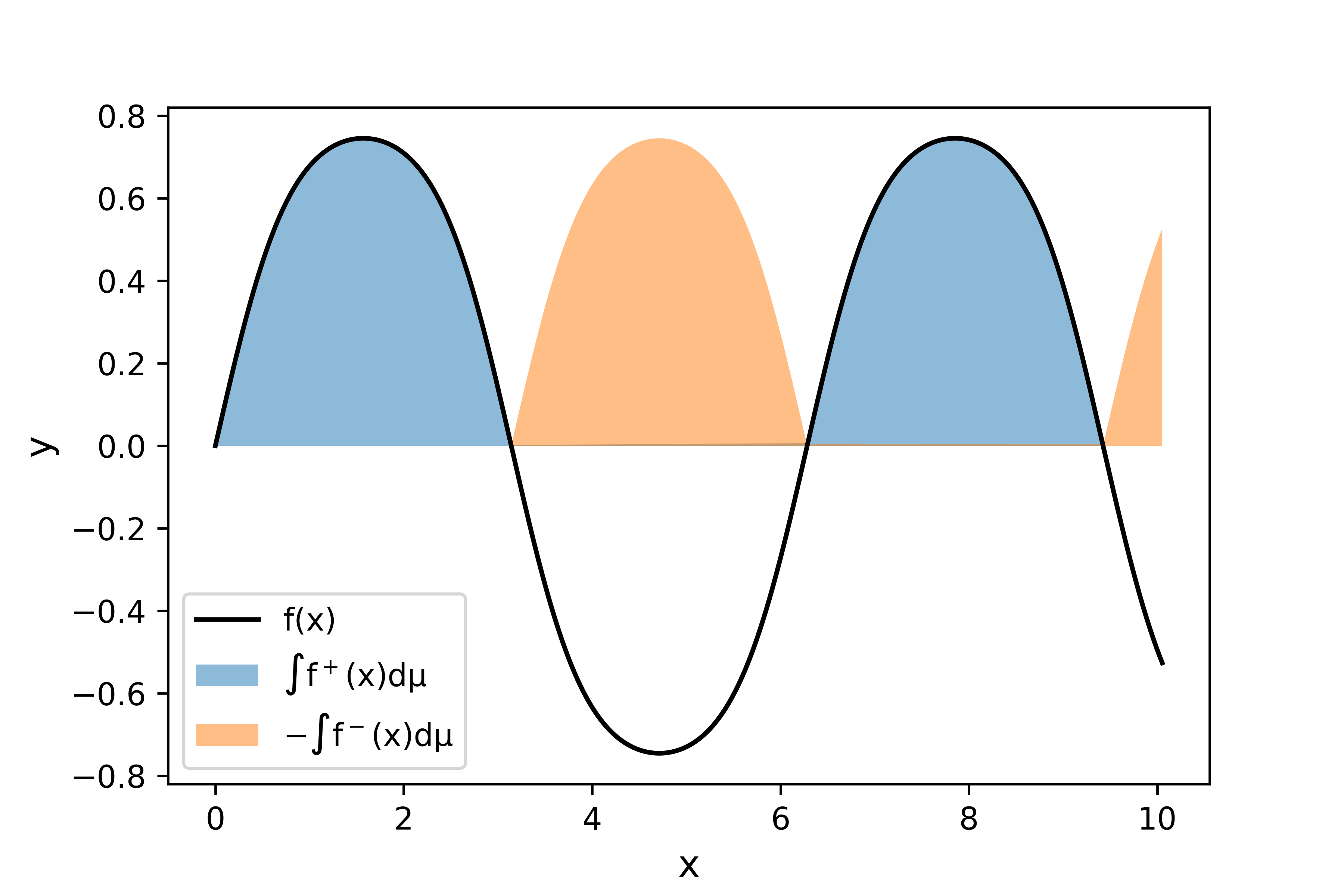}
    \caption{Example of a decomposition of a bounded, measurable function into $f^+$ and $f^-$. The Lebesgue integral is then just the integral of the blue region minus the integral of the orange region.}
    \label{fig:lebintegral_1}
\end{figure}

\begin{figure}
    \centering
    \includegraphics{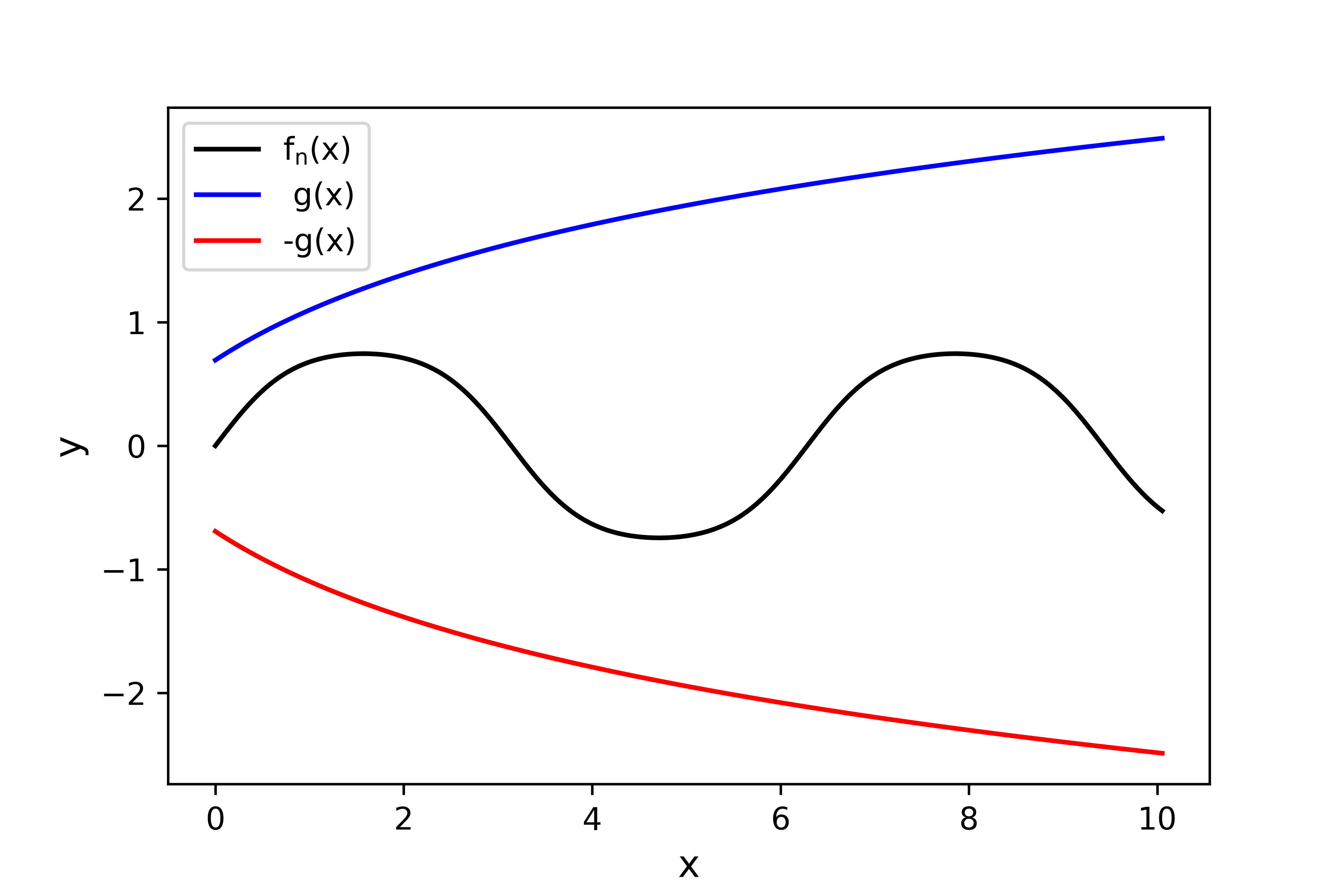}
    \caption{If $g$ is a non-negative, measurable function such that $\int g d\mu < \infty$ and $g$ dominates $f_n \rightarrow f$ pointwise almost everywhere ($\abs{f_n} \leq_{a.e.} g$), then $f$ and even $f_i$ are integrable.}
    \label{fig:domconverge_1}
\end{figure}

\subsubsubsection{$L^P (M, \sigma, \mu)$ Function Spaces}

Define the set, $\mathcal{L}^P := \{f:M \rightarrow \mathbb{C} | f \; \text{measurable}, \int |f|^p d\mu < \infty \}$. If we equip $\mathcal{L}^P$ with pointwise addition and scalar multiplication then clearly $\mathcal{L}^P$ is a $\mathbb{C}$ vector space. One can then construct a map $\{\{\cdot \}\}_p: \mathcal{L}^P \rightarrow \mathbb{R}$ such that,

\begin{equation}
    \{\{f \}\}_p = \bigg(\int |f|^p d\mu \bigg)^{1/p}
\end{equation}

\noindent where $p \in \mathbb{R}$ and $1 \leq p < \infty$ that is a \textbf{semi-norm} with the following properties:

\begin{enumerate}
    \item $\forall \alpha \in \mathbb{C}$ then $\{\{\alpha f\}\}_p = |\alpha| \cdot_\mathbb{C} \{\{f \}\}_p$
    \item $\forall f,g \in \mathcal{L}^P$ then $\{\{f +_\mathcal{L} g \}\}_p \leq \{\{f\}\}_p +_\mathbb{C} \{\{g\}\}_p$
    \item $\{\{f\}\}_p \geq 0$ and $\{\{f\}\}_p = 0$ implies $f =_{a.e.} 0$
\end{enumerate}

\noindent Notice that for a norm, we must have that the last condition is satisfied everywhere (not almost everywhere). Therefore, $\mathcal{L}^P$ is not even a normed space. However, we can use \textbf{equivalence classes} to remove troublesome functions in $\mathcal{L}^P$ that are zero almost everywhere while refusing to be zero everywhere. First, we can define the equivalence relation between $f,g \in \mathcal{L}^P$ as equivalent if,

\begin{equation} \label{equivalence}
    f \sim g :<=> f =_{a.e.} g.
\end{equation}

\noindent Equivalence relations are interesting since they allow us to split up a space according to equivalent classes of objects within the set. We define the \textbf{quotient space} of a set, $M/ \sim$ to be the set of all $m \in M$ for which $m_i \sim m_j$. If we take an equivalence class on $\mathcal{L}^P$ to be those functions satisfying Equation \eqref{equivalence}, then we define a new space $L^P$ such that,

\begin{equation}
    L^P = \mathcal{L}^P/\sim = \{ [f]_\sim | f \in \mathcal{L}^P \}.
\end{equation}

\noindent Although we will not explicitly prove this, $L^P$ inherits a norm from $\mathcal{L}^P$ such that,

\begin{equation}
    \norm{\cdot}_p:L^P \rightarrow \mathbb{R}
\end{equation}

\begin{equation}
    [f]_\sim \rightarrow \norm{[f]_\sim}_p := \{\{f\}\}_p.
\end{equation}

\noindent It can the be shown that $L^P$ is in fact a Banach space, and by the Holder inequality it is even a Hilbert space for $p = 2$ (the inner product is sesqui-linear). In fact, all of quantum mechanics of a particle in three dimensions rests on the $L^2(\mathbb{R}^3, \sigma(\mathcal{O}_{\mathbb{R}}), \lambda^3)$ Hilbert space.

\subsubsection{Adjoints}

The \textbf{adjoint} $A^*: \mathcal{D}_{A^*} \rightarrow \mathcal{H}$ of a densely defined operator $A: \mathcal{D}_{A} \rightarrow \mathcal{H}$ is defined such that,

\begin{equation}
    \mathcal{D}_{A^*} := \{\psi \in \mathcal{H} | \forall \alpha \in \mathcal{D}_A \; \exists \eta \in \mathcal{H}: \; (\psi, A \alpha) = (\eta, \alpha)\}
\end{equation}

\noindent where the operator $A^* \psi = \eta$. Prove that $A \subseteq B$ implies that $B^* \subseteq A^*$.
    
    \begin{proof}
        $A \subseteq B$ implies that $\mathcal{D}_A \subseteq \mathcal{D}_B$ and $\forall \alpha \in \mathcal{D}_A$ $B$ acts on $\alpha$ in the same manner: $A\alpha = B\alpha$. Let $\psi \in \mathcal{D}_{B^*}$. Then $\forall \alpha \in \mathcal{D}_{A} \subseteq \mathcal{D}_{B}$ we have,
        
        \begin{equation}
            (\psi, B\alpha) = (\psi, A\alpha).
        \end{equation}
        
        \noindent But by the definition of $\mathcal{D}_{A^*}$,
        
        \begin{equation}
            \mathcal{D}_{A^*} := \{\psi \in \mathcal{H} | \forall \alpha \in \mathcal{D}_A \; \exists \eta \in \mathcal{H}: \; (\psi, A \alpha) = (\eta, \alpha)\}
        \end{equation}
        
        \noindent we must have that $\psi \in \mathcal{D}_{A^*}$. This implies that $\mathcal{D}_{B^*} \subseteq \mathcal{D}_{A^*}$ and $B^* \subseteq A^*$.
    \end{proof}
    
\noindent A densely defined operator is called \textbf{symmetric} (in physics literature, Hermitian) if $\forall \alpha, \beta \in \mathcal{D}_{A}: \; (\alpha, A\beta) = (A\alpha, \beta)$. We can easily show that if a densely defined operator $A$ is symmetric that its adjoint is an extension $A \subseteq A^*$.
    
    \begin{proof}
        Let $A: \mathcal{D}_{A} \rightarrow \mathcal{H}$ be a symmetric opertor such that, $\forall \alpha, \beta \in \mathcal{D}_{A}: \; (\alpha, A\beta) = (A\alpha, \beta)$. Then let $\psi \in \mathcal{D}_A$. By the definition of $\mathcal{D}_{A^*}$,
        
        \begin{equation}
            \mathcal{D}_{A^*} := \{\psi \in \mathcal{H} | \forall \alpha \in \mathcal{D}_A \; \exists \eta \in \mathcal{H}: \; (\psi, A \alpha) = (\eta, \alpha)\}
        \end{equation}
        
        \noindent we must have that $\psi \in \mathcal{D}_{A^*}$ so that $\mathcal{D}_A \subseteq \mathcal{D}_{A^*}$. Since $A$ is symmetric, we also have that $A^*\psi = A\psi$ so that $A \subseteq A^*$.
    \end{proof}
    
A \textbf{self-adjoint} operator is an operator that is equal to its adjoint, $A = A^*$. An interesting property of self-adjoint operators are that they cannot be further extended to some other operator. To see this, 
    
    \begin{proof}
        Assume that $B$ is a self-adjoint extension of $A$. Then,
        
        \begin{equation}
            A \subseteq B = B^* \subseteq A^* = A.
        \end{equation}
    \end{proof}
    
The notion of \textbf{closable}, \textbf{closure}, and \textbf{closed} are important to understanding the behavior of operators. We introduce the following three definitions:

\begin{enumerate}
    \item Closable: A densely defined operator is closeable if its adjoint $A^*$ is also densely defined.
    \item The closure of a closeable operator $A$ is $A^{**}$.
    \item An operator $A$ is closed if $A = A^{**}$.
\end{enumerate}

\noindent Show that if $A$ is symmetric that $A^{**} \subseteq A^*$.

    \begin{proof}
        A symmetric means that $A \subseteq A^*$. But by our previous result we know that $A^{**} \subseteq A^*$ and we are done.
    \end{proof}
    
\noindent A symmetric operator $A$ is called \textbf{essentially self-adjoint} if its closure $A^{**}$ is self-adjoint. An interesting result is that for an essentially self-adjoint operator $A$ the closure $A^{**}$ is the unique self-adjoint extension of $A$.

    \begin{proof}
        We know that $A \subseteq A^{**}$ by our previous theorem. Assume there is some other self-adjoint extension $B$ so that $A \subseteq B = B^*$. But then $B^{**} \subseteq A^*$ or equivalently $A^{**} \subseteq B^{***} = B$. But since $A^{**}$ cannot be further extended we must have that $B = A^{**}$.
    \end{proof}

\noindent We can quickly check if a symmetric operator $A$ has self-adjoint extensions by evaluating its \textbf{defect indices}. If the defect indices coincide, then there are self-adjoint extensions of $A$, and if they are different there are none. For the special case where both defect indices are zero, there exists a unique such extension.

    \begin{equation}
        d_+ = \text{dim}(\text{ker}(A^* - i))
    \end{equation}
    
    \begin{equation}
        d_- = \text{dim}(\text{ker}(A^* + i)).
    \end{equation}
    
\subsubsection{The Spectral Theorem}

As we will see in the axioms of quantum mechanics, the possible measurement values are those in the spectrum, $\sigma(A)$ of an observable $A$. A common task in any quantum mechanical calculation is to determine the spectra of one or several self-adjoint operators. In this section, we will first introduce important definitions related to the spectra of operators and then proceed to the spectral theorem.

\subsubsubsection{The Resolvent Set and Spectrum}

First, define the \textbf{resolvent map} of a closed operator $A = A^{**}$ as the map,

\begin{equation}
    R_A: \rho(A) \rightarrow \mathcal{L}(\mathcal{H}, \mathcal{H})
\end{equation}

\begin{equation}
    z \in \mathbb{C} \longmapsto (A - z \cdot id_{\mathcal{D}_A})^{-1}
\end{equation}

\noindent where $\rho(A) \subseteq \mathbb{C}$ is known as the \textbf{resolvent set}. Hence, the resolvent set is all the complex numbers where $(A - z \cdot id_{\mathcal{D}_A})^{-1}$ exists and is a bounded linear map from $\mathcal{H} \rightarrow \mathcal{H}$. The complement of the resolvent set is called the \textbf{spectrum}, $\sigma(A) = \mathbb{C} \backslash \rho(A)$. This means that the spectrum is just the set of complex numbers where $(A - z \cdot id_{\mathcal{D}_A})^{-1}$ does not exist. The spectrum of an operator on a finite dimensional vector space is precisely equal to the set of eigenvalues, but for infinite dimension spaces this is not necessarily the case. All that we know is that if the eigenvalues of the self-adjoint operator exist, they must be in the spectrum. 

We define the total spectra $\sigma(A)$ as the union of pure point, point imbedded in a continuum and purely continuous spectra such that,

\begin{equation}
    \sigma(A) = \sigma_{pp} \cup \sigma_{pic} \cup \sigma_{pc}
\end{equation}

\noindent with precise definitions given by,

\begin{enumerate}
    \item $\sigma_{pp} := \{z \in \mathbb{C} | ran(A-z) = \overline{ran(A-z)} \neq \mathcal{H} \}$
    \item $\sigma_{pic} := \{z \in \mathbb{C} | ran(A-z) \neq \overline{ran(A-z)} \neq \mathcal{H} \}$
    \item $\sigma_{pc} := \{z \in \mathbb{C} | ran(A-z) \neq \overline{ran(A-z)} = \mathcal{H} \}$
\end{enumerate} 

\noindent We now arrive at an important theorem.

\begin{theorem}
    The point spectrum of a self-adjoint operator $A$ has as its elements precisely the eigenvalues of $A$.
    
    \begin{proof}
        First, show that if $\lambda$ is an eigenvalue of $A$ then $\lambda \in \mathbb{R}$. Consider, for $A\psi = \lambda \psi$, that,
        
        \begin{equation}
            \lambda (\psi, \psi) = (\psi, \lambda \psi) = (\psi, A\psi).
        \end{equation}
        
        \noindent But if $A$ is self-adjoint, we have,
        
        \begin{equation}
            (\psi, A\psi) = (A\psi, \psi) = \overline{(\psi, A\psi)} = \overline{(\psi, \lambda\psi)} = \overline{\lambda} \overline{(\psi, \psi)} = \overline{\lambda}(\psi, \psi).
        \end{equation}
        
        \noindent Of course, since $\psi \neq 0$ we must have $\lambda = \overline{\lambda}$ which implies $\lambda \in \mathbb{R}$. Now suppose that $\lambda$ is indeed an eigenvalue of $A$. Then,
        
        \begin{equation}
            \text{ker}(A-\lambda) \neq \{0\}_{\mathcal{H}}.
        \end{equation}
        
        \noindent But then,
        
        \begin{equation}
            \text{ker}(A-\lambda) = \text{ker}(A^*-\lambda) = \text{ker}((A-\overline{\lambda})^*) = \text{ran}(A-\overline{\lambda})^\perp = \text{ran}(A-\lambda)^\perp
        \end{equation}
        
        \noindent which implies that,
        
        \begin{equation}
            (\text{ran}(A-\lambda)^\perp)^\perp = \overline{\text{ran}(A-\lambda)} \neq \{0\}^\perp_{\mathcal{H}} = \mathcal{H}  
        \end{equation}
        
        \noindent which by definition means that $\lambda \in \sigma_{pp}(A)$. Conversely, suppose that $\lambda \in \mathbb{C}$ is not an eigenvalue of $A$. Then $\overline{\lambda}$ is not an eigenvalue either. By a manner similar to the previous statement we obtain,
        
        \begin{equation}
            \text{ker}(A-\lambda) = \{0\}_{\mathcal{H}}
        \end{equation}
        
        \begin{equation}
            (\text{ran}(A-\lambda)^\perp)^\perp = \overline{\text{ran}(A-\lambda)} = \{0\}^\perp_{\mathcal{H}} = \mathcal{H}  
        \end{equation}
        
        \noindent which means that $\lambda \notin \sigma_{pp}(A)$ and we are done.
    \end{proof}
\end{theorem}

\subsubsubsection{Projection-Valued Measures}

A projection-valued measure (PVM) is a map from the Borel $\sigma$-algebra on the reals to the bounded linear maps on a Hilbert space, $P:\sigma(\mathcal{O}_{\mathbb{R}}) \rightarrow \mathcal{L}(\mathcal{H}, \mathcal{H})$, that satisfies the following properties:

\begin{enumerate}
    \item $\forall \Omega \in \sigma(\mathcal{O}_{\mathbb{R}})$, $P(\Omega)^* = P(\Omega)$
    \item $\forall \Omega \in \sigma(\mathcal{O}_{\mathbb{R}})$, $P(\Omega) \circ P(\Omega) = P(\Omega)$
    \item $P(\mathbb{R}) = id_{\mathcal{H}}$
    \item $\forall \Omega \in \sigma(\mathcal{O}_{\mathbb{R}}) \ni \Omega = \sqcup_{n \geq 1} \Omega_n$ implies $\forall \psi \in \mathcal{H}: \sum_{n \geq 1}(P(\Omega_n)\psi) = P(\Omega)\psi$
\end{enumerate}

\noindent Many of the useful properties of PVMs are listed below:

\begin{enumerate}
    \item $P(\emptyset) = 0_\mathcal{H}$ where $0_\mathcal{H}: \mathcal{H} \rightarrow \mathcal{H}$ which takes an element in the Hilbert space and returns the 0 on $\mathcal{H}$.
    \item $P(\mathbb{R}\backslash\Omega) = id_{\mathcal{H}} - P(\Omega)$
    \item $P(\Omega_1 \cup \Omega_2) + P(\Omega_1 \cap \Omega_2) = P(\Omega_1) + P(\Omega_2)$
    \item $P(\Omega_1 \cap \Omega_2) = P(\Omega_1) \circ P(\Omega_2)$
    \item $\Omega_1 \subseteq \Omega_2$ implies $\text{ran}(P(\Omega_1)) \subseteq \text{ran}(P(\Omega_2))$
\end{enumerate}

\noindent Finally, one piece of notation that we will find useful is that $P(\lambda) := P((-\infty, \lambda])$. Recall that $P$ takes an element of the Borel $\sigma$-algebra which is generated by the half open intervals $(-\infty, \lambda]$.

Now, we introduce the concept that a PVM can be used to define complex- and real-valued Borel measures. For instance, a complex-valued Borel measure can be induced such that,

\begin{equation}
    \forall \psi,\phi \in \mathcal{H}, \; \mu_{\psi, \phi}:\sigma(\mathcal{O}_{\mathbb{R}}) \rightarrow \mathbb{C}
\end{equation}

\noindent where $\mu_{\psi, \phi}$ maps an element $\Omega$ of the Borel-$\sigma$-algebra into the inner product $(\psi, P(\Omega)\phi)$,

\begin{equation}
    \Omega \rightarrow \mu_{\psi, \phi}:= (\psi, P(\Omega)\phi).
\end{equation}

\noindent Trivially, a real-valued Borel measure can be induced $\forall \psi \in \mathcal{H}$ so that, $\mu_{\psi} := \mu_{\psi, \psi}$.

We now proceed to discuss integration of PVMs for simple, bounded and unbounded measurable functions $f:\mathbb{R} \rightarrow \mathbb{C}$.

\textit{Integration of Simple Functions}: Recall that $f:\mathbb{R} \rightarrow \mathbb{C}$ is called simple if it has a decomposition $f(\mathbb{R}) =\{f_1, ..., f_n\} \subseteq \mathbb{C}$ such that,

\begin{equation}
    f = \sum_{n=1}^N f_n\cdot\chi_{preim_f}(\{f_n\})
\end{equation}

\noindent where clearly the function $f$ can be written as a sum of complex-values by their respective characteristic functions on the pre-image of $f$. The integral of such a function with respect to a PVM is defined as,

\begin{equation}
    \int_{\mathbb{R}}fdP:=\sum_{n=1}^N f_n\cdot P(\text{preim}_f(\{f_n\}))
\end{equation}

\noindent which is easily observed to have the following properties:

\begin{enumerate}
    \item For any $\chi_{\Omega}$ with $\Omega \in \sigma(\mathcal{O}_{\mathbb{R}})$ the integral over $\chi_{\Omega}$ is just the PVM acting on $\Omega$ so that,
    
    \begin{equation}
        \int_{\mathbb{R}}\chi_{\Omega}dP = P(\Omega).
    \end{equation}
    
    \item $\forall \psi, \phi \in \mathcal{H}$,
    
    \begin{equation}
        \begin{split}
            (\psi, \bigg(\int_{\mathbb{R}}fdP\bigg) \phi) = & (\psi, \sum_{n=1}^N f_n \cdot P(\Omega_n)\phi) \\
            = & \sum_{n=1}^N f_n (\psi, P(\Omega_n)\phi) \\
            = & \sum_{n=1}^N f_n \cdot \mu_{\psi,\phi}(\Omega_n) \\
            = & \int_{\mathbb{R}} f d\mu_{\psi,\phi}.
        \end{split}
    \end{equation}
\end{enumerate}

\textit{Integration of Bounded Borel Functions}: Consider the Banach space of measurable functions $B(\mathbb{R}):=\{f:\mathbb{R}\rightarrow \mathbb{C}| \norm{f}_\infty < \infty\}$. Note that the space of simple measurable functions from the real to complex numbers is a dense subset of $B(\mathbb{R})$, which means that there exists a unique extension of $\int dP$ to $B(\mathbb{R})$ by the Bounded Linear Transformation theorem,

\begin{equation}
    \int dP: \overbrace{B(\mathbb{R})}^{\text{Banach}} \rightarrow \overbrace{\mathcal{L}(\mathcal{H}, \mathcal{H})}^{\text{Banach}}
\end{equation}

\noindent with the following properties:

\begin{enumerate}
    \item $\int_{\mathbb{R}}1dP = id_{\mathcal{H}}$
    \item $\int_{\mathbb{R}}(f \cdot g)dP = \int_{\mathbb{R}}fdP \circ \int_{\mathbb{R}}gdP$
    \item $\int_{\mathbb{R}}\overline{f}dP = (\int_{\mathbb{R}}fdP)^*$
\end{enumerate}

\noindent The first two properties can be called a $C^*$-algebra homomorphism.

\textit{Integration of Unbounded Borel Functions}: The integral over the PVM has been shown to send a bounded Borel function to a bounded, linear operator $\int f dP$. However, we will now define a construction that also applies to unbounded Borel functions. Let $f:\mathbb{R}\rightarrow\mathbb{C}$ be a measureable, unbounded Borel function. Define the linear map,

\begin{equation}
    \int f dP: \mathcal{D}_{\int_{\mathbb{R}fdP}} \subseteq \mathcal{H} \rightarrow \mathcal{H}
\end{equation}

\noindent be defined on the dense linear subset $\mathcal{D}_{\int_{\mathbb{R}fdP}}:= \{\psi \in \mathcal{H}|\int \abs{f}^2d\mu_\psi < \infty\}$ which takes an element $\psi \in \mathcal{H}$ to its integral,

\begin{equation}
    \psi \mapsto \bigg(\int f dP\bigg)\psi := \lim_{n \to \infty}\bigg[ \bigg(\int_{\mathbb{R}}f_ndP\bigg) \psi\bigg]
\end{equation}

\noindent where the $f_n:= \chi_{\{x\in\mathbb{R}|\abs{f(x)}\leq n\}}f$. Since the integral is Cauchy in $\mathcal{H}$, the integral has the following properties:

\begin{enumerate}
    \item $\bigg(\int f dP\bigg)^* = \int \overline{f} dP$.
    \item $\forall \alpha \in \mathbb{C}$, we have,
    
    \begin{equation}
        \alpha \int f dP + \int g dP \subseteq \int (\alpha f + g) dP
    \end{equation}
    
    \noindent with equality when $f,g$ are bounded. If $f,g$ are unbounded then the right-hand side is an extension of the left-hand side.
    
    \item $(\int f dP) \circ (\int g dP) \subseteq \int(f\circ g)dP$ with equality when $f,g$ are bounded. 
\end{enumerate}

\textit{The Inverse Spectral Theorem}: The inverse spectral theorem can now be stated as a precursor to the proof of the spectral theorem. Given a PVM, $P$, we can construct a self-adjoint operator,

\begin{equation}
    A_P := \int id_{\mathbb{R}}dP
\end{equation}

\noindent since,

\begin{equation}
    A_P^* = \bigg(\int id_{\mathbb{R}}dP\bigg)^* = \int \overline{id_{\mathbb{R}}}dP = \int id_{\mathbb{R}}dP = A_P.
\end{equation}

Now we want to go backwards from the previous result; namely, we want to actually find this $P$ given some self-adjoint operator $A$. First we introduce the concept of spectrally decomposable self-adjoint operators. $A$ self-adjoint is called spectrally decomposable if there exists a projection-valued measure $P$ such that,

\begin{equation}
    A = \int id_{\mathbb{R}}dP
\end{equation}

\noindent like what we have in the inverse spectral theorem. Then for any measurable $f:\mathbb{R} \rightarrow \mathbb{R}$ we define the application of the real-valued function to the spectrally decomposable self-adjoint operator as,

\begin{equation}
    f(A) := \int f(\lambda)P(d\lambda) = \int f \circ id_\mathbb{R} dP
\end{equation}

\noindent with,

\begin{equation}
    f(A): \mathcal{D}_{\int f dP} \rightarrow \mathcal{H}
\end{equation}

\noindent where the domain is given by,

\begin{equation}
    \mathcal{D}_{\int f dP} = \{\psi \in \mathcal{H} | \int |f|^2 d\mu_\psi\}
\end{equation}

\noindent and $\mu_\psi$ is the real-valued Borel measure induced by $(\psi, P(\Omega) \psi)$. 

But how do we reconstruct a projection-valued measure from a spectrally decomposed self-adjoint operator? Let $A = \int_\mathbb{R} \lambda P(d\lambda)$ for some projection-valued measure $P$. Then the resolvent $R_A(z)$ for $z \in \rho(A) \subseteq \mathbb{C}$ can be written as, 

\begin{equation}
    R_A(z) = (A - z id_\mathcal{H})^{-1} = r_z(A) = \int_\mathbb{R} r_z dP
\end{equation}

\noindent where $r_z(A)$ is a function from $r_z(A):\mathbb{C} \to \mathbb{C}$ takes some complex number $\lambda \to \frac{1}{\lambda - z}$. Then $\forall \psi \in \mathcal{H}$, we can write,

\begin{equation}
    (\psi, R_A(z) \psi)  = \int_\mathbb{R} r_z d \mu_\psi = \int_\mathbb{R} \frac{1}{\lambda - z} d \mu_\psi.
\end{equation}

\noindent Such a function has special properties since it is a Herglotz function, that is, that it maps the upper half of the complex plane into itself,

\begin{equation}
    ( \psi, R_A(z) \psi): \mathbb{C}^+ \to \mathbb{C}^+
\end{equation}

\noindent where $\mathbb{C}^+ := \{z \in \mathbb{C} | \Im(z) \geq 0\}$. This allows us to reconstruct a real-valued Borel measure from the object $(\psi, R_A(z) \psi) \geq 0$. To see this, consider that,

\begin{equation}
   \lim_{\epsilon \to 0^+} \frac{1}{\pi} \int_{t_1}^{t_2} dt \Im ( \psi, R_A(t + i\epsilon) \psi) = \lim_{\epsilon \to 0^+} \frac{1}{\pi} \int_{t_1}^{t_2} dt \int_\mathbb{R} \frac{\epsilon}{\abs{\lambda - t - i\epsilon}^2} \mu_\psi (d\lambda)
\end{equation}

\noindent where Fubini's theorem allows us to exchange the integrals so that,

\begin{equation}
    = \lim_{\epsilon \to 0^+} \frac{1}{\pi} \int_\mathbb{R}\int_{t_1}^{t_2} dt  \frac{\epsilon}{\abs{(\lambda - t)^2 +\epsilon^2}} \mu_\psi (d\lambda) = \lim_{\epsilon \to 0^+} \int_\mathbb{R}\frac{1}{\pi} \arctan{\bigg(\frac{t - \lambda}{\epsilon}}\bigg)\bigg|^{t_2}_{t_1} \mu_\psi (d\lambda)
\end{equation}

\noindent which in the limit as $\epsilon \to 0$ is simply the sum of two characteristic functions 
(by virtue of the two competing $\arctan$ terms for $t_1$ and $t_2$),

\begin{equation}
   \lim_{\epsilon \to 0^+} \frac{1}{\pi} \int_{t_1}^{t_2} dt \Im ( \psi, R_A(t + i\epsilon) \psi)  = \int_\mathbb{R} \frac{1}{2}(\chi_{(t_1, t_2)} + \chi_{[t1,t2]})\mu_\psi (d\lambda).
\end{equation}

\noindent Thus, we have the following theorem:

\begin{theorem}
    Stieltjes Inversion Formula: For Herglotz functions $\mu_\psi:\sigma(\mathcal{O}_\mathbb{R}) \to \mathbb{R}^+_0$, we can recover the real Borel measure from the resolvent set using,

    \begin{equation}
        \mu((-\infty, \lambda]) = \lim_{\delta \to 0^+} \lim_{\epsilon \to 0^+} \frac{1}{\pi} \int_{-\infty}^{\lambda + \delta} dt \Im ( \psi, R_A(t + i\epsilon) \psi).
    \end{equation}

\begin{proof}
    The right hand side is just,

    \begin{equation}
        = \lim_{\delta \to 0^+} \int_\mathbb{R} \frac{1}{2}(\chi_{(-\infty, \lambda + \delta)} + \chi_{[-\infty,\lambda + \delta]})\mu_\psi (d\lambda)
    \end{equation}

    \noindent which by the theorem of dominant convergence is just,

    \begin{equation}
        = \int_\mathbb{R} (\chi_{(-\infty, \lambda]})\mu_\psi (d\lambda)
    \end{equation}

    \noindent just gives us $\mu((-\infty, \lambda])$ by definition.
\end{proof}
\end{theorem}

Up to this point, we know that given some spectrally decomposable, self-adjoint operator $A = \int_\mathbb{R} id_{\mathbb{R}} dP$, that $P$ can be recovered from $A$ by virtue of $(\psi, P(\Omega) \phi) = \int \chi_{\Omega} d \mu_{\psi,\phi}$ where $\mu_{\psi,\phi}$ is the complex Borel measure.

\begin{theorem}

The Spectral Theorem: For any self-adjoint operator $A:\mathcal{D}_A \to \mathcal{H}$, there is a unique projection-valued measure $P_A:\sigma(\mathcal{O}_\mathbb{R}) \to \mathcal{L}(\mathcal{H} \to \mathcal{H})$, such that $A = \int_\mathbb{R} id_{\mathcal{H}} dP_A$.

\begin{proof}

To construct a projection-valued measure from a self-adjoint operator, it thus remains to show that $( \psi, R_A(\cdot) \psi )$ is Herglotz for any self-adjoint operator $A$, that $( \psi, P(\Omega) \phi ) := \int \chi_\Omega d \mu_{\psi,\phi}$ is indeed a projection-valued measure, and that this projection-valued measure is unique. We will first begin by proving a few lemmas that will be useful later in the proof.

\begin{lemma}
\label{lemma:Herglotz}
    The First Resolvent Formula: For any operator $A:\mathcal{D}_A \to \mathcal{H}$ and $a,b \in \mathbb{C} \in \rho(A)$, we have,

    \begin{equation}
        R_A(a) - R_B(b) = (a-b)R_A(a) \circ R_A(b)
    \end{equation}

    \noindent or equivalently,

    \begin{equation}
        R_A(a) - R_B(b) = (a-b)R_A(b) \circ R_A(a).
    \end{equation}

\begin{proof}

    Note that,

    \begin{equation}
        (A-a)^{-1} - (a-b)(A-a)^{-1}(A-b)^{-1} = (A-a)^{-1}(id_{\mathcal{H}} - (a-b)(A-b)^{-1})
    \end{equation}

    \noindent which is of course,

    \begin{equation}
        (A-a)^{-1}(id_{\mathcal{H}} - (a + A - A -b)(A-b)^{-1}) = (A-a)^{-1}(A-a)(A-b)^{-1} = (A-b)^{-1}
    \end{equation}

    \noindent which proves this first assertion known as the "first resolvent formula".
    
\end{proof}

\end{lemma}

\begin{lemma}
\label{lemma:uniqueness}
    $(\psi, R_A(\cdot) \psi )$ is Herglotz for any self-adjoint operator $A$.

    \begin{proof}
        Let $z \in \mathbb{C}^+$ and $F(z) = (\psi, R_A(\cdot) \psi ) = \int_\mathbb{R} \frac{1}{\lambda - z} d \mu_\psi$. Then just notice that,

        \begin{equation}
            \Im F(z) = \Im z \int_\mathbb{R} \frac{1}{\abs{\lambda - z}^2} d \mu_\psi \in \mathbb{C}^+.
        \end{equation}
    \end{proof}
\end{lemma}

\begin{lemma}
    $(\psi, P(\Omega) \phi ) := \int \chi_\Omega d \mu_{\psi,\phi}$ is a unique projection-valued measure constructed from $A$.

    \begin{proof}
        The Stieltjes inversion formula can provide a unique real Borel measure from the resolvent set, from which $(\psi, P(\Omega) \phi )$ is uniquely determined. 
    \end{proof}
\end{lemma}

Putting all of this together, we can prove the spectral theorem in its entirety. First, we can define a real-valued Borel measure $\forall \psi \in \mathcal{H}$ such that,

\begin{equation}
    \mu_\psi^A: \sigma(\mathcal{O}_\mathbb{R}) \to \mathbb{R}
\end{equation}

\begin{equation}
    \mu_\psi^A((-\infty, \lambda]) := \lim_{\delta \to 0^+} \lim_{\epsilon \to 0^+} \frac{1}{\pi} \int_{-\infty}^{\lambda + \delta} dt \Im ( \psi, R_A(t + i\epsilon) \psi)
\end{equation}

\noindent since $\mu_\psi^A$ is Herglotz (Lemma \ref{lemma:Herglotz}) and where $R_A$ is the resolvent map which takes an complex element of the resolvent set to a bounded linear map on $\mathcal{H}$ in the following way,

\begin{equation}
    R_A:\rho(A) \to \mathcal{L}(\mathcal{H})
\end{equation}

\begin{equation}
    R_A(z) := (A-z id_{\mathcal{H}})^{-1}.
\end{equation}

Then, we construct a complex-valued Borel measure via the polarization formula such that $\forall \psi,\phi \in \mathcal{H}$ and all Borel sets $\Omega \in \sigma(\mathcal{O}_\mathbb{R})$,

\begin{equation}
    \mu_{\psi,\phi}: \sigma(\mathcal{O}_\mathbb{R}) \to \mathbb{C}
\end{equation}

\begin{equation}
    \mu_{\psi,\phi}(\Omega) := \frac{1}{4}[\mu_{\psi+\phi}(\Omega) - \mu_{\psi-\phi}(\Omega) + i\mu_{\psi-i\phi}(\Omega) - i\mu_{\psi+i\phi}(\Omega)].
\end{equation}

\noindent Finally, we can define the unique projection-valued measure (Lemma \ref{lemma:uniqueness})  such that $\forall \psi, \phi \in H$,

\begin{equation}
    P_A:\sigma(\mathcal{O}_\mathbb{R}) \to \mathcal{L}(\mathcal{H})
\end{equation}

\begin{equation}
    (\psi, P_A(\Omega) \phi):= \int \chi_{\Omega} d \mu_{\psi,\phi}
\end{equation}

\noindent where $\chi$ is the characteristic function. Thus, we have constructed the unique projection-valued measure from a self-adjoint operator and we are done.

\end{proof}

\end{theorem}

The consequences of the spectral theorem to quantum theory are immensely important. First, we will take as an axiom that the observables of a quantum mechanical system are given by the self-adjoint linear maps on the Hilbert space associated with a quantum system. Then, the probability of the system being in some state $\rho$ is defined as the trace of the unique projection-valued measure $P_A$ induced by self-adjoint operator $A$ composed with $\rho$. The only remaining axioms related to unitary and projective quantum dynamics serve to evolve the quantum system in time and "collapse" the system into some observation at measurement time $t_m$, respectively. 

\subsection{An Axiomatic Approach to Quantum Mechanics}

\begin{axiom}
    With every quantum system there is associated a complex Hilbert space $(\mathcal{H}, +, \cdot, (\cdot, \cdot))$, with states that are all positive, trace-class, linear maps $\rho: \mathcal{H} \rightarrow \mathcal{H}$ for which $Tr (\rho) = 1$
\end{axiom}

The Hilbert space is equipped with addition, scalar multiplication, and a sesqui-linear inner product (Hermitian, linear, and non-negative) and is complete.

\begin{axiom}
    The observables of a quantum system are the self-adjoint linear maps $A:\mathcal{D}_A \rightarrow \mathcal{H}$ where $\mathcal{D}_A$ is some densely defined subset of $\mathcal{H}$. Self-adjointness is defined so that the map $A$ coincides with its adjoint map $A^*$ so that,

\begin{equation}
    \mathcal{D}_{A^*} := \{\psi \in \mathcal{H} | \forall \alpha \in \mathcal{D}_A \; \exists \nu \in \mathcal{H}: \; (\psi, A \alpha) = (\nu, \alpha)\}
\end{equation}

\begin{equation}
    A^*(\psi) = \nu.
\end{equation}
\end{axiom}

\begin{axiom}
    The probability that a measurement of an observable A on a system is in the state $\rho$ yields the result in the Borel set $E \subset \mathbb{R}$ is given by, 
    
    \begin{equation}
        \mu_\rho^A(E) := Tr(P_A(E) \circ \rho)
    \end{equation}
    
    \noindent where $P_A(E)$ is the unique projection-valued measure that is associated with a self-adjoint linear map $A$ according to the spectral theorem. Furthermore, the composition $P_A(E) \circ \rho$ is again trace-class. 
\end{axiom}

\begin{axiom}
    Unitary dynamics occur during time intervals $(t_1, t_2)$ during which no measurement occurs such that,
    
    \begin{equation}
        \rho(t_2) = \mathcal{U}(t_2 - t_1) \rho(t_1) \mathcal{U}^{-1}(t_2 - t_1)
    \end{equation}
    
    \noindent where the time evolution operator is written as $\mathcal{U}(t) = e^{\frac{-i\mathcal{H}t}{\hbar}}$.
\end{axiom}

\begin{axiom}
    Projective dynamics occur when a measurement is made at some time, $t_m$. Then the state immediately after the measurement of an observable $A$ is,
    
    \begin{equation}
        \rho_{after} := \frac{P_A(E) \rho_{before} P_A(E)}{Tr(P_A(E) \rho_{before} P_A(E))}.
    \end{equation}
\end{axiom}

The idea is just to find some self-adjoint operator that corresponds to a physical observable, use the spectral theorem to find a unique projection-valued measure, and then calculate the probability of the observation, time-evolution of the system, and projective dynamics of the system according to these axioms. The content of quantum theory then is to construct self-adjoint operators that correspond to physical observables / experimental measurements. This task is non-trivial and forms the basis of all of quantum theory.

\printbibliography

\typeout{get arXiv to do 4 passes: Label(s) may have changed. Rerun}

\end{document}